\documentclass{imforarxiv}%{ourelsarticle}
\pdfoutput=1
\usepackage{amssymb}
\usepackage{array}

\usepackage{mathpartir}
\usepackage{hyperref}

\usepackage{multicol}
\title{Lewis meets Brouwer: constructive strict implication}
\author[tl]{Tadeusz Litak}
\address[tl]{Informatik 8, FAU Erlangen-N\"{u}rnberg, Martensstra\ss e 3, 91058 Erlangen \\
\href{mailto:tadeusz.litak@fau.de}{tadeusz.litak@fau.de}}

\author[av]{Albert Visser}
\address[av]{Philosophy, Faculty of Humanities,
                Utrecht University,
               Janskerkhof 13,
                3512BL~~Utrecht \\% The Netherlands \\
\href{mailto:a.visser@uu.nl}{a.visser@uu.nl}}

\date{Draft of \today}                                           % Activate to display a given date or no date
\usepackage{tabularx}

\newcolumntype{L}[1]{>{\raggedright\let\newline\\\arraybackslash\hspace{0pt}}m{#1}}
\newcolumntype{C}[1]{>{\centering\let\newline\\\arraybackslash\hspace{0pt}}m{#1}}
\newcolumntype{R}[1]{>{\raggedleft\let\newline\\\arraybackslash\hspace{0pt}}m{#1}}

\usepackage{ifthen}

\makeatletter
\newif\if@full
 \@fulltrue    %%%%%%%%%%% <--- switch between \@fullfalse and \@fulltrue to compactify bibliography etc!
\relax
\let\iffull\if@full

\makeatletter
\newif\if@deriv
\@derivfalse    %%%%%%%%%%% <--- CHANGE TO \@derivtrue to display derivations!
\relax
\let\ifderiv\if@deriv

\makeatletter
\newif\if@bibt
\@bibtfalse  %%%%%%%%%%% <--- CHANGE TO \@bibttrue to compile bibliography with latex!
\relax
\let\ifbibt\if@bibt

\usepackage{etoolbox}
\ifdef{\c@author}{
\let\c@author\relax
}{}
\makeatother
\iffull
\usepackage[
    backend=biber,
    bibstyle=alphabetic,% formerly used IEEE citation style
    citestyle=alphabetic,% citing multiple papers will produce format similar to [2,4-8,12] instead of [2,4,5,6,7,8,12] 
    sorting=nty,
    natbib=false,
    isbn=false,
    uniquename=false,
    firstinits=true,
    backref=true,
    backrefstyle=two,
    abbreviate=true,
    bibwarn=false,
    url=true,     
    doi=false,
    safeinputenc,
    bibencoding=auto,
    maxnames=5,
    maxalphanames=5    
]{biblatex}
\AtEveryBibitem{%
  \clearname{editor}%
  \clearlist{language}%
  \clearfield{month}%
  \clearlist{location}%
  \clearfield{pagetotal}%
}
\else
\usepackage[
    backend=biber,
    bibstyle=alphabetic,% formerly used IEEE citation style
    citestyle=alphabetic,% citing multiple papers will produce format similar to [2,4-8,12] instead of [2,4,5,6,7,8,12] 
    sorting=nty,
    natbib=false,
    isbn=false,
    uniquename=false,
    firstinits=true,
    backref=false,
    abbreviate=true,
    bibwarn=false,
    url=true,     
    doi=false,
    safeinputenc,
    bibencoding=auto,
    maxnames=3,
    maxalphanames=3    
]{biblatex}
\AtEveryBibitem{%
  \clearname{editor}%
  \clearlist{language}%
  \clearfield{month}%
  \clearlist{location}%
  \clearfield{pagetotal}%
}
\fi

\ifbibt
\makeatletter
\def\blx@bblfile@biber{%
  \blx@secinit
  \begingroup
  \blx@bblstart
\begingroup
\makeatletter
\@ifundefined{ver@biblatex.sty}
  {\@latex@error
     {Missing 'biblatex' package}
     {The bibliography requires the 'biblatex' package.}
      \aftergroup }
  {}
\endgroup

\refsection{0}
  \sortlist[entry]{nty/global/}
    \entry{AbelV14:aplas}{incollection}{}
      \name{author}{2}{}{%
        {{hash=492fee9410a0fd601511d2201c985f28}{%
           family={Abel},
           familyi={A\bibinitperiod},
           given={Andreas},
           giveni={A\bibinitperiod}}}%
        {{hash=38b551d0b5ec5e2b4c915bd08649e102}{%
           family={Vezzosi},
           familyi={V\bibinitperiod},
           given={Andrea},
           giveni={A\bibinitperiod}}}%
      }
      \name{editor}{1}{}{%
        {{hash=afd798828c50cc237b473922c42b9658}{%
           family={Garrigue},
           familyi={G\bibinitperiod},
           given={Jacques},
           giveni={J\bibinitperiod}}}%
      }
      \list{publisher}{1}{%
        {Springer International Publishing}%
      }
      \strng{namehash}{168dcaf224a645f74d9f2ad0681bdd2d}
      \strng{fullhash}{168dcaf224a645f74d9f2ad0681bdd2d}
      \strng{authornamehash}{168dcaf224a645f74d9f2ad0681bdd2d}
      \strng{authorfullhash}{168dcaf224a645f74d9f2ad0681bdd2d}
      \strng{editornamehash}{afd798828c50cc237b473922c42b9658}
      \strng{editorfullhash}{afd798828c50cc237b473922c42b9658}
      \field{labelalpha}{AV14}
      \field{sortinit}{A}
      \field{sortinithash}{3248043b5fe8d0a34dab5ab6b8d4309b}
      \field{labelnamesource}{author}
      \field{labeltitlesource}{title}
      \field{booktitle}{Proceedings of {APLAS}}
      \field{isbn}{978-3-319-12735-4}
      \field{series}{{LNCS}}
      \field{title}{A Formalized Proof of Strong Normalization for Guarded Recursive Types}
      \field{volume}{8858}
      \field{year}{2014}
      \field{pages}{140\bibrangedash 158}
      \range{pages}{19}
      \verb{doi}
      \verb 10.1007/978-3-319-12736-1_8
      \endverb
      \verb{url}
      \verb http://dx.doi.org/10.1007/978-3-319-12736-1_8
      \endverb
    \endentry
    \entry{Abramsky93:tcs}{article}{}
      \name{author}{1}{}{%
        {{hash=f4721a56c3b94b002e7ef8e0ef2b89fa}{%
           family={Abramsky},
           familyi={A\bibinitperiod},
           given={Samson},
           giveni={S\bibinitperiod}}}%
      }
      \strng{namehash}{f4721a56c3b94b002e7ef8e0ef2b89fa}
      \strng{fullhash}{f4721a56c3b94b002e7ef8e0ef2b89fa}
      \strng{authornamehash}{f4721a56c3b94b002e7ef8e0ef2b89fa}
      \strng{authorfullhash}{f4721a56c3b94b002e7ef8e0ef2b89fa}
      \field{labelalpha}{Abr93}
      \field{sortinit}{A}
      \field{sortinithash}{3248043b5fe8d0a34dab5ab6b8d4309b}
      \field{labelnamesource}{author}
      \field{labeltitlesource}{title}
      \field{abstract}{We study Girard's linear logic from the point of view of giving a concrete computational interpretation of the logic, based on the Curry\x{fffd}\x{fffd}\x{fffd}Howard isomorphism. In the case of Intuitionistic linear logic, this leads to a refinement of the lambda calculus, giving finer control over order of evaluation and storage allocation, while maintaining the logical content of programs as proofs, and computation as cut-elimination. In the classical case, it leads to a concurrent process paradigm with an operational semantics in the style of Berry and Boudol's chemical abstract machine. This opens up a promising new approach to the parallel implementation of functional programming languages; and offers the prospect of typed concurrent programming in which correctness is guaranteed by the typing.}
      \field{issn}{0304-3975}
      \field{journaltitle}{Theoretical Computer Science}
      \field{number}{1--2}
      \field{title}{Computational interpretations of linear logic}
      \field{volume}{111}
      \field{year}{1993}
      \field{pages}{3\bibrangedash 57}
      \range{pages}{55}
      \verb{doi}
      \verb http://dx.doi.org/10.1016/0304-3975(93)90181-R
      \endverb
      \verb{url}
      \verb http://www.sciencedirect.com/science/article/pii/030439759390181R
      \endverb
    \endentry
    \entry{Amerbauer96:sl}{article}{}
      \name{author}{1}{}{%
        {{hash=bd9d8c212734c86c7f3481883c4493cb}{%
           family={Amerbauer},
           familyi={A\bibinitperiod},
           given={Martin},
           giveni={M\bibinitperiod}}}%
      }
      \list{publisher}{1}{%
        {Springer}%
      }
      \strng{namehash}{bd9d8c212734c86c7f3481883c4493cb}
      \strng{fullhash}{bd9d8c212734c86c7f3481883c4493cb}
      \strng{authornamehash}{bd9d8c212734c86c7f3481883c4493cb}
      \strng{authorfullhash}{bd9d8c212734c86c7f3481883c4493cb}
      \field{labelalpha}{Ame96}
      \field{sortinit}{A}
      \field{sortinithash}{3248043b5fe8d0a34dab5ab6b8d4309b}
      \field{labelnamesource}{author}
      \field{labeltitlesource}{title}
      \field{abstract}{We give sound and complete tableau and sequent calculi for the propositional normal modal logics S4.04, K4B and $\text{G}_{0}$ (these logics are the smallest normal modal logics containing K and the schemata $\square A\rightarrow \square \square A,\square A\rightarrow A$ and $\square \diamond \square A\rightarrow (A\rightarrow \square A)$ ; $\square A\rightarrow \square \square A$ and $A\rightarrow \square \diamond A;\square A\rightarrow \square \square A$ and $\square (\square (A\rightarrow \square A)\rightarrow A)\rightarrow \square A\text{resp.}$ ) with the following properties: the calculi for S4.04 and $\text{G}_{0}$ are cut-free and have the interpolation property, the calculus for K4B contains a restricted version of the cut-rule, the so-called analytical cut-rule. In addition we show that $\text{G}_{0}$ is not compact (and therefore not canonical), and we proof with the tableau-method that $\text{G}_{0}$ is characterised by the class of all finite, (transitive) trees of degenerate or simple clusters of worlds; therefore $\text{G}_{0}$ is decidable and also characterised by the class of all frames for $\text{G}_{0}$ .}
      \field{issn}{00393215}
      \field{journaltitle}{Studia Logica}
      \field{number}{2/3}
      \field{title}{Cut-Free Tableau Calculi for Some Propositional Normal Modal Logics}
      \field{volume}{57}
      \field{year}{1996}
      \field{pages}{359\bibrangedash 372}
      \range{pages}{14}
      \verb{url}
      \verb http://www.jstor.org/stable/20015881
      \endverb
    \endentry
    \entry{AppelMRV07:popl}{inproceedings}{}
      \name{author}{4}{}{%
        {{hash=448045e4531c99f2a2552ae1affe65fe}{%
           family={Appel},
           familyi={A\bibinitperiod},
           given={Andrew\bibnamedelima W.},
           giveni={A\bibinitperiod\bibinitdelim W\bibinitperiod}}}%
        {{hash=a9f37916182d59fd60115ad8a650c06d}{%
           family={Melli\`{e}s},
           familyi={M\bibinitperiod},
           given={Paul-Andr\'{e}},
           giveni={P\bibinithyphendelim A\bibinitperiod}}}%
        {{hash=586cb12157524735b38a2dc2a67b4b34}{%
           family={Richards},
           familyi={R\bibinitperiod},
           given={Christopher\bibnamedelima D.},
           giveni={C\bibinitperiod\bibinitdelim D\bibinitperiod}}}%
        {{hash=7f094af8d8fd00381650074f87b06f85}{%
           family={Vouillon},
           familyi={V\bibinitperiod},
           given={J\'{e}r\^{o}me},
           giveni={J\bibinitperiod}}}%
      }
      \name{editor}{2}{}{%
        {{hash=d3c25d820c67e8530324893291b911d5}{%
           family={Hofmann},
           familyi={H\bibinitperiod},
           given={Martin},
           giveni={M\bibinitperiod}}}%
        {{hash=f2dd06499ca403079167fe9967155025}{%
           family={Felleisen},
           familyi={F\bibinitperiod},
           given={Matthias},
           giveni={M\bibinitperiod}}}%
      }
      \list{organization}{1}{%
        {ACM SIGPLAN-SIGACT}%
      }
      \strng{namehash}{a6da2eb93f5402ea8df679b20c24cb6a}
      \strng{fullhash}{a6da2eb93f5402ea8df679b20c24cb6a}
      \strng{authornamehash}{a6da2eb93f5402ea8df679b20c24cb6a}
      \strng{authorfullhash}{a6da2eb93f5402ea8df679b20c24cb6a}
      \strng{editornamehash}{5bdb6a2226c8898014fbf46b7ca7cc95}
      \strng{editorfullhash}{5bdb6a2226c8898014fbf46b7ca7cc95}
      \field{labelalpha}{AMRV07}
      \field{sortinit}{A}
      \field{sortinithash}{3248043b5fe8d0a34dab5ab6b8d4309b}
      \field{labelnamesource}{author}
      \field{labeltitlesource}{title}
      \field{booktitle}{Proceedings of {POPL}}
      \field{isbn}{1-59593-575-4}
      \field{title}{A very modal model of a modern, major, general type system}
      \field{year}{2007}
      \field{pages}{109\bibrangedash 122}
      \range{pages}{14}
    \endentry
    \entry{arde:sigm14}{article}{}
      \name{author}{2}{}{%
        {{hash=9d51539569ddbe71309385f0d8091518}{%
           family={Ardeshir},
           familyi={A\bibinitperiod},
           given={M.},
           giveni={M\bibinitperiod}}}%
        {{hash=2281f3af07dbcc1d485dcb11f041cef1}{%
           family={Mojtahedi},
           familyi={M\bibinitperiod},
           given={S.M.},
           giveni={S\bibinitperiod}}}%
      }
      \strng{namehash}{1b88a4e523af535801f5bc0452d39d7f}
      \strng{fullhash}{1b88a4e523af535801f5bc0452d39d7f}
      \strng{authornamehash}{1b88a4e523af535801f5bc0452d39d7f}
      \strng{authorfullhash}{1b88a4e523af535801f5bc0452d39d7f}
      \field{labelalpha}{AM14}
      \field{sortinit}{A}
      \field{sortinithash}{3248043b5fe8d0a34dab5ab6b8d4309b}
      \field{labelnamesource}{author}
      \field{labeltitlesource}{title}
      \field{journaltitle}{arXiv preprint arXiv:1409.5699}
      \field{title}{The {$\Sigma_1$}-Provability Logic of {{\sf HA}}}
      \field{year}{2014}
    \endentry
    \entry{ArtemovP16:rsl}{article}{}
      \name{author}{2}{}{%
        {{hash=1841fe17ddbeb4863d02f07ce78da70b}{%
           family={Artemov},
           familyi={A\bibinitperiod},
           given={Sergei},
           giveni={S\bibinitperiod}}}%
        {{hash=9dc13e90c848df4322c6a5bb59e02221}{%
           family={Protopopescu},
           familyi={P\bibinitperiod},
           given={Tudor},
           giveni={T\bibinitperiod}}}%
      }
      \list{publisher}{1}{%
        {Cambridge University Press}%
      }
      \strng{namehash}{80f38770b47d3cb0f09adba3688fdd45}
      \strng{fullhash}{80f38770b47d3cb0f09adba3688fdd45}
      \strng{authornamehash}{80f38770b47d3cb0f09adba3688fdd45}
      \strng{authorfullhash}{80f38770b47d3cb0f09adba3688fdd45}
      \field{labelalpha}{AP16}
      \field{sortinit}{A}
      \field{sortinithash}{3248043b5fe8d0a34dab5ab6b8d4309b}
      \field{labelnamesource}{author}
      \field{labeltitlesource}{title}
      \field{journaltitle}{The Review of Symbolic Logic}
      \field{number}{2}
      \field{title}{Intuitionistic Epistemic Logic}
      \field{volume}{9}
      \field{year}{2016}
      \field{pages}{266\bibrangedash 298}
      \range{pages}{33}
      \verb{doi}
      \verb 10.1017/S1755020315000374
      \endverb
    \endentry
    \entry{arte:prov04}{incollection}{}
      \name{author}{2}{}{%
        {{hash=e89f88161784433c4a82d2fcaa62dfd6}{%
           family={Artemov},
           familyi={A\bibinitperiod},
           given={S.N.},
           giveni={S\bibinitperiod}}}%
        {{hash=f8878e25fde2ea6d38da5bec072a0a34}{%
           family={Beklemishev},
           familyi={B\bibinitperiod},
           given={L.D.},
           giveni={L\bibinitperiod}}}%
      }
      \name{editor}{2}{}{%
        {{hash=d82591404c108a8a2a6835e75d5d92e0}{%
           family={Gabbay},
           familyi={G\bibinitperiod},
           given={D.},
           giveni={D\bibinitperiod}}}%
        {{hash=5da0854c8cdce485a3fd9dcf28b01a8f}{%
           family={Guenthner},
           familyi={G\bibinitperiod},
           given={F.},
           giveni={F\bibinitperiod}}}%
      }
      \list{location}{1}{%
        {Dordrecht}%
      }
      \list{publisher}{1}{%
        {Springer}%
      }
      \strng{namehash}{9dc4f9b7005a02d875316a2e5ca4b6ab}
      \strng{fullhash}{9dc4f9b7005a02d875316a2e5ca4b6ab}
      \strng{authornamehash}{9dc4f9b7005a02d875316a2e5ca4b6ab}
      \strng{authorfullhash}{9dc4f9b7005a02d875316a2e5ca4b6ab}
      \strng{editornamehash}{d2c0e896668cb1b6e69ea35c85a805c9}
      \strng{editorfullhash}{d2c0e896668cb1b6e69ea35c85a805c9}
      \field{labelalpha}{AB04}
      \field{sortinit}{A}
      \field{sortinithash}{3248043b5fe8d0a34dab5ab6b8d4309b}
      \field{labelnamesource}{author}
      \field{labeltitlesource}{title}
      \field{booktitle}{{Handbook of Philosophical Logic}, 2nd ed.}
      \field{title}{Provability Logic}
      \field{volume}{13}
      \field{year}{2004}
      \field{pages}{229\bibrangedash 403}
      \range{pages}{175}
    \endentry
    \entry{AtkeyMB13:icfp}{inproceedings}{}
      \name{author}{2}{}{%
        {{hash=b5507ca1177b01ddc89cb764e5fbf287}{%
           family={Atkey},
           familyi={A\bibinitperiod},
           given={Robert},
           giveni={R\bibinitperiod}}}%
        {{hash=2c2c2406f3f7d52d846f02ca883d9b3d}{%
           family={McBride},
           familyi={M\bibinitperiod},
           given={Conor},
           giveni={C\bibinitperiod}}}%
      }
      \name{editor}{2}{}{%
        {{hash=33b4dd1b5ab1beaed368ad65da17385e}{%
           family={Morrisett},
           familyi={M\bibinitperiod},
           given={Greg},
           giveni={G\bibinitperiod}}}%
        {{hash=ee50513d1a9288cce2900f00506e9fe4}{%
           family={Uustalu},
           familyi={U\bibinitperiod},
           given={Tarmo},
           giveni={T\bibinitperiod}}}%
      }
      \list{organization}{1}{%
        {ACM SIGPLAN}%
      }
      \strng{namehash}{a38c055325e87463ec50f6cb0ddbb150}
      \strng{fullhash}{a38c055325e87463ec50f6cb0ddbb150}
      \strng{authornamehash}{a38c055325e87463ec50f6cb0ddbb150}
      \strng{authorfullhash}{a38c055325e87463ec50f6cb0ddbb150}
      \strng{editornamehash}{465f69f8fbaf3ea353be4436ef6a4005}
      \strng{editorfullhash}{465f69f8fbaf3ea353be4436ef6a4005}
      \field{labelalpha}{AM13}
      \field{sortinit}{A}
      \field{sortinithash}{3248043b5fe8d0a34dab5ab6b8d4309b}
      \field{labelnamesource}{author}
      \field{labeltitlesource}{title}
      \field{booktitle}{International Conference on Functional Programming, (ICFP)}
      \field{isbn}{978-1-4503-2326-0}
      \field{title}{Productive coprogramming with guarded recursion}
      \field{year}{2013}
      \field{pages}{197\bibrangedash 208}
      \range{pages}{12}
    \endentry
    \entry{Barcan53:jsl}{article}{}
      \name{author}{1}{}{%
        {{hash=ac1a9cd17bfa1d02b2d413f5bcabe328}{%
           family={{Barcan Marcus}},
           familyi={B\bibinitperiod},
           given={Ruth},
           giveni={R\bibinitperiod}}}%
      }
      \list{publisher}{1}{%
        {Association for Symbolic Logic}%
      }
      \strng{namehash}{ac1a9cd17bfa1d02b2d413f5bcabe328}
      \strng{fullhash}{ac1a9cd17bfa1d02b2d413f5bcabe328}
      \strng{authornamehash}{ac1a9cd17bfa1d02b2d413f5bcabe328}
      \strng{authorfullhash}{ac1a9cd17bfa1d02b2d413f5bcabe328}
      \field{labelalpha}{Bar53}
      \field{sortinit}{B}
      \field{sortinithash}{5f6fa000f686ee5b41be67ba6ff7962d}
      \field{labelnamesource}{author}
      \field{labeltitlesource}{title}
      \field{issn}{00224812}
      \field{journaltitle}{The Journal of Symbolic Logic}
      \field{number}{3}
      \field{title}{Strict Implication, Deducibility and the Deduction Theorem}
      \field{volume}{18}
      \field{year}{1953}
      \field{pages}{234\bibrangedash 236}
      \range{pages}{3}
      \verb{url}
      \verb http://www.jstor.org/stable/2267407
      \endverb
    \endentry
    \entry{Barcan46:jsl}{article}{}
      \name{author}{1}{}{%
        {{hash=e0f500c3ad17dffdd6f52857e4b1cf74}{%
           family={Barcan},
           familyi={B\bibinitperiod},
           given={Ruth\bibnamedelima C.},
           giveni={R\bibinitperiod\bibinitdelim C\bibinitperiod}}}%
      }
      \list{publisher}{1}{%
        {Association for Symbolic Logic}%
      }
      \strng{namehash}{e0f500c3ad17dffdd6f52857e4b1cf74}
      \strng{fullhash}{e0f500c3ad17dffdd6f52857e4b1cf74}
      \strng{authornamehash}{e0f500c3ad17dffdd6f52857e4b1cf74}
      \strng{authorfullhash}{e0f500c3ad17dffdd6f52857e4b1cf74}
      \field{labelalpha}{Bar46}
      \field{sortinit}{B}
      \field{sortinithash}{5f6fa000f686ee5b41be67ba6ff7962d}
      \field{labelnamesource}{author}
      \field{labeltitlesource}{title}
      \field{issn}{00224812}
      \field{journaltitle}{The Journal of Symbolic Logic}
      \field{number}{4}
      \field{title}{The Deduction Theorem in a Functional Calculus of First Order Based on Strict Implication}
      \field{volume}{11}
      \field{year}{1946}
      \field{pages}{115\bibrangedash 118}
      \range{pages}{4}
      \verb{url}
      \verb http://www.jstor.org/stable/2268309
      \endverb
    \endentry
    \entry{Bazhanov03}{article}{}
      \name{author}{1}{}{%
        {{hash=985c0e441ed5d81fc82527fde33e9ed1}{%
           family={Bazhanov},
           familyi={B\bibinitperiod},
           given={Valentin\bibnamedelima A.},
           giveni={V\bibinitperiod\bibinitdelim A\bibinitperiod}}}%
      }
      \list{publisher}{1}{%
        {Cambridge University Press}%
      }
      \strng{namehash}{985c0e441ed5d81fc82527fde33e9ed1}
      \strng{fullhash}{985c0e441ed5d81fc82527fde33e9ed1}
      \strng{authornamehash}{985c0e441ed5d81fc82527fde33e9ed1}
      \strng{authorfullhash}{985c0e441ed5d81fc82527fde33e9ed1}
      \field{labelalpha}{Baz03}
      \field{sortinit}{B}
      \field{sortinithash}{5f6fa000f686ee5b41be67ba6ff7962d}
      \field{labelnamesource}{author}
      \field{labeltitlesource}{title}
      \field{journaltitle}{Science in Context}
      \field{number}{4}
      \field{title}{The Scholar and the ``Wolfhound Era'': The Fate of Ivan E. Orlov's Ideas in Logic, Philosophy, and Science}
      \field{volume}{16}
      \field{year}{2003}
      \field{pages}{535\bibrangedash 550}
      \range{pages}{16}
      \verb{doi}
      \verb 10.1017/S0269889703000954
      \endverb
    \endentry
    \entry{Becker30}{book}{}
      \name{author}{1}{}{%
        {{hash=feaef9ed0cefe0209c3958bd49a7a7eb}{%
           family={Becker},
           familyi={B\bibinitperiod},
           given={Oskar},
           giveni={O\bibinitperiod}}}%
      }
      \list{publisher}{1}{%
        {Halle}%
      }
      \strng{namehash}{feaef9ed0cefe0209c3958bd49a7a7eb}
      \strng{fullhash}{feaef9ed0cefe0209c3958bd49a7a7eb}
      \strng{authornamehash}{feaef9ed0cefe0209c3958bd49a7a7eb}
      \strng{authorfullhash}{feaef9ed0cefe0209c3958bd49a7a7eb}
      \field{labelalpha}{Bec30}
      \field{sortinit}{B}
      \field{sortinithash}{5f6fa000f686ee5b41be67ba6ff7962d}
      \field{labelnamesource}{author}
      \field{labeltitlesource}{title}
      \field{series}{Jahrbuch f\"{u}r Philosophie und ph\"{a}\-no\-me\-no\-lo\-gische Forschung}
      \field{title}{Zur Logik der Modalit\"{a}ten}
      \field{year}{1930}
    \endentry
    \entry{bees:nond75}{article}{}
      \name{author}{1}{}{%
        {{hash=4c6fe4bf7d421e0d0693b2329613439a}{%
           family={Beeson},
           familyi={B\bibinitperiod},
           given={M.},
           giveni={M\bibinitperiod}}}%
      }
      \strng{namehash}{4c6fe4bf7d421e0d0693b2329613439a}
      \strng{fullhash}{4c6fe4bf7d421e0d0693b2329613439a}
      \strng{authornamehash}{4c6fe4bf7d421e0d0693b2329613439a}
      \strng{authorfullhash}{4c6fe4bf7d421e0d0693b2329613439a}
      \field{labelalpha}{Bee75}
      \field{sortinit}{B}
      \field{sortinithash}{5f6fa000f686ee5b41be67ba6ff7962d}
      \field{labelnamesource}{author}
      \field{labeltitlesource}{title}
      \field{journaltitle}{The Journal of Symbolic Logic}
      \field{title}{The nonderivability in intuitionistic formal systems of theorems on the continuity of effective operations}
      \field{volume}{40}
      \field{year}{1975}
      \field{pages}{321\bibrangedash 346}
      \range{pages}{26}
    \endentry
    \entry{bekl:limi05}{article}{}
      \name{author}{2}{}{%
        {{hash=f8878e25fde2ea6d38da5bec072a0a34}{%
           family={Beklemishev},
           familyi={B\bibinitperiod},
           given={L.D.},
           giveni={L\bibinitperiod}}}%
        {{hash=7c59a81f2d67e5dbdf77cc8061473f2d}{%
           family={Visser},
           familyi={V\bibinitperiod},
           given={A.},
           giveni={A\bibinitperiod}}}%
      }
      \strng{namehash}{2dfc128e9b36a361a53b140b9f438751}
      \strng{fullhash}{2dfc128e9b36a361a53b140b9f438751}
      \strng{authornamehash}{2dfc128e9b36a361a53b140b9f438751}
      \strng{authorfullhash}{2dfc128e9b36a361a53b140b9f438751}
      \field{labelalpha}{BV05}
      \field{sortinit}{B}
      \field{sortinithash}{5f6fa000f686ee5b41be67ba6ff7962d}
      \field{labelnamesource}{author}
      \field{labeltitlesource}{title}
      \field{journaltitle}{Annals of Pure and Applied Logic}
      \field{number}{1--2}
      \field{title}{On the Limit Existence Principles in Elementary Arithmetic and {$\Sigma^0_n$}-consequences of Theories}
      \field{volume}{136}
      \field{year}{2005}
      \field{pages}{56\bibrangedash 74}
      \range{pages}{19}
    \endentry
    \entry{BentonBP98:jfp}{article}{}
      \name{author}{3}{}{%
        {{hash=3b57216091a3ec53b2beec3ffc9e8751}{%
           family={Benton},
           familyi={B\bibinitperiod},
           given={P.\bibnamedelimi N.},
           giveni={P\bibinitperiod\bibinitdelim N\bibinitperiod}}}%
        {{hash=aa5401d298a45e90389918daca5a6f8b}{%
           family={Bierman},
           familyi={B\bibinitperiod},
           given={Gavin\bibnamedelima M.},
           giveni={G\bibinitperiod\bibinitdelim M\bibinitperiod}}}%
        {{hash=e8301f8c4b3b95126e1606835de6bc12}{%
           family={Paiva},
           familyi={P\bibinitperiod},
           given={Valeria},
           giveni={V\bibinitperiod},
           prefix={de},
           prefixi={d\bibinitperiod}}}%
      }
      \strng{namehash}{f927d20b32c13d927647f09f7dd3dc5b}
      \strng{fullhash}{f927d20b32c13d927647f09f7dd3dc5b}
      \strng{authornamehash}{f927d20b32c13d927647f09f7dd3dc5b}
      \strng{authorfullhash}{f927d20b32c13d927647f09f7dd3dc5b}
      \field{labelalpha}{BBP98}
      \field{sortinit}{B}
      \field{sortinithash}{5f6fa000f686ee5b41be67ba6ff7962d}
      \field{labelnamesource}{author}
      \field{labeltitlesource}{title}
      \field{journaltitle}{J. Funct. Programming}
      \field{number}{2}
      \field{title}{Computational Types from a Logical Perspective}
      \field{volume}{8}
      \field{year}{1998}
      \field{pages}{177\bibrangedash 193}
      \range{pages}{17}
    \endentry
    \entry{bera:inte90}{article}{}
      \name{author}{1}{}{%
        {{hash=dd60f4c0e2dacf636c541ce67c6cfbbc}{%
           family={Berarducci},
           familyi={B\bibinitperiod},
           given={A.},
           giveni={A\bibinitperiod}}}%
      }
      \strng{namehash}{dd60f4c0e2dacf636c541ce67c6cfbbc}
      \strng{fullhash}{dd60f4c0e2dacf636c541ce67c6cfbbc}
      \strng{authornamehash}{dd60f4c0e2dacf636c541ce67c6cfbbc}
      \strng{authorfullhash}{dd60f4c0e2dacf636c541ce67c6cfbbc}
      \field{labelalpha}{Ber90}
      \field{sortinit}{B}
      \field{sortinithash}{5f6fa000f686ee5b41be67ba6ff7962d}
      \field{labelnamesource}{author}
      \field{labeltitlesource}{title}
      \field{journaltitle}{The Journal of Symbolic Logic}
      \field{title}{The interpretability logic of {P}eano arithmetic}
      \field{volume}{55}
      \field{year}{1990}
      \field{pages}{1059\bibrangedash 1089}
      \range{pages}{31}
    \endentry
    \entry{BezhanishviliJ12:acs}{article}{}
      \name{author}{2}{}{%
        {{hash=a7c042aaa1a3fa49ceb88ffd5870c794}{%
           family={Bezhanishvili},
           familyi={B\bibinitperiod},
           given={Guram},
           giveni={G\bibinitperiod}}}%
        {{hash=3c0295e0ef5b742d981462b229350392}{%
           family={Jansana},
           familyi={J\bibinitperiod},
           given={Ramon},
           giveni={R\bibinitperiod}}}%
      }
      \list{language}{1}{%
        {English}%
      }
      \list{publisher}{1}{%
        {Springer Netherlands}%
      }
      \strng{namehash}{9881fa507f7165b5151d8b3853cc2df0}
      \strng{fullhash}{9881fa507f7165b5151d8b3853cc2df0}
      \strng{authornamehash}{9881fa507f7165b5151d8b3853cc2df0}
      \strng{authorfullhash}{9881fa507f7165b5151d8b3853cc2df0}
      \field{labelalpha}{BJ13}
      \field{sortinit}{B}
      \field{sortinithash}{5f6fa000f686ee5b41be67ba6ff7962d}
      \field{labelnamesource}{author}
      \field{labeltitlesource}{title}
      \field{abstract}{We develop a new duality for implicative semilattices, generalizing Esakia duality for Heyting algebras. Our duality is a restricted version of generalized Priestley duality for distributive semilattices, and provides an improvement of Vrancken-Mawet and Celani dualities. We also show that Heyting algebra homomorphisms can be characterized by means of special partial functions between Esakia spaces. On the one hand, this yields a new duality for Heyting algebras, which is an alternative to Esakia duality. On the other hand, it provides a natural generalization of K\x{fffd}\x{fffd}hler\x{fffd}\x{fffd}\x{fffd}s partial functions between finite posets to the infinite case.}
      \field{issn}{0927-2852}
      \field{journaltitle}{Appl. Categor. Struct.}
      \field{note}{10.1007/s10485-011-9265-0}
      \field{title}{Esakia Style Duality for Implicative Semilattices}
      \field{volume}{21}
      \field{year}{2013}
      \field{pages}{181\bibrangedash 208}
      \range{pages}{28}
      \verb{url}
      \verb http://dx.doi.org/10.1007/s10485-011-9265-0
      \endverb
    \endentry
    \entry{Bierman94}{report}{}
      \name{author}{1}{}{%
        {{hash=0b0ca41321c70d04661fe69322aa43fa}{%
           family={Bierman},
           familyi={B\bibinitperiod},
           given={G.M.},
           giveni={G\bibinitperiod}}}%
      }
      \list{institution}{1}{%
        {University of Cambridge, Computer Laboratory}%
      }
      \list{location}{1}{%
        {15 JJ Thomson Avenue, Cambridge CB3 0FD, United Kingdom, phone +44 1223 763500}%
      }
      \strng{namehash}{0b0ca41321c70d04661fe69322aa43fa}
      \strng{fullhash}{0b0ca41321c70d04661fe69322aa43fa}
      \strng{authornamehash}{0b0ca41321c70d04661fe69322aa43fa}
      \strng{authorfullhash}{0b0ca41321c70d04661fe69322aa43fa}
      \field{labelalpha}{Bie94}
      \field{sortinit}{B}
      \field{sortinithash}{5f6fa000f686ee5b41be67ba6ff7962d}
      \field{labelnamesource}{author}
      \field{labeltitlesource}{title}
      \field{month}{8}
      \field{number}{UCAM-CL-TR-346}
      \field{title}{{On intuitionistic linear logic}}
      \field{type}{techreport}
      \field{year}{1994}
    \endentry
    \entry{BirkedalM13:lics}{inproceedings}{}
      \name{author}{2}{}{%
        {{hash=f4714354745ea31cf6c543434102f444}{%
           family={Birkedal},
           familyi={B\bibinitperiod},
           given={Lars},
           giveni={L\bibinitperiod}}}%
        {{hash=87f3fdf5910beec6d8dba13dfb21e1be}{%
           family={M{\o{}}gelberg},
           familyi={M\bibinitperiod},
           given={Rasmus\bibnamedelima Ejlers},
           giveni={R\bibinitperiod\bibinitdelim E\bibinitperiod}}}%
      }
      \list{organization}{1}{%
        {ACM/IEEE}%
      }
      \strng{namehash}{3086f870f6b94dded414fce663296736}
      \strng{fullhash}{3086f870f6b94dded414fce663296736}
      \strng{authornamehash}{3086f870f6b94dded414fce663296736}
      \strng{authorfullhash}{3086f870f6b94dded414fce663296736}
      \field{labelalpha}{BM13}
      \field{sortinit}{B}
      \field{sortinithash}{5f6fa000f686ee5b41be67ba6ff7962d}
      \field{labelnamesource}{author}
      \field{labeltitlesource}{title}
      \field{booktitle}{Proceedings of {LiCS}}
      \field{isbn}{978-1-4799-0413-6}
      \field{title}{Intensional Type Theory with Guarded Recursive Types qua Fixed Points on Universes}
      \field{year}{2013}
      \field{pages}{213\bibrangedash 222}
      \range{pages}{10}
    \endentry
    \entry{BirkedalMSS12:lmcs}{article}{}
      \name{author}{4}{}{%
        {{hash=f4714354745ea31cf6c543434102f444}{%
           family={Birkedal},
           familyi={B\bibinitperiod},
           given={Lars},
           giveni={L\bibinitperiod}}}%
        {{hash=87f3fdf5910beec6d8dba13dfb21e1be}{%
           family={M{\o{}}gelberg},
           familyi={M\bibinitperiod},
           given={Rasmus\bibnamedelima Ejlers},
           giveni={R\bibinitperiod\bibinitdelim E\bibinitperiod}}}%
        {{hash=82fe256cf7f8e5dfa92bd8bf21a3c9f3}{%
           family={Schwinghammer},
           familyi={S\bibinitperiod},
           given={Jan},
           giveni={J\bibinitperiod}}}%
        {{hash=be0cfa690f9269274456598c103e4a4f}{%
           family={St{\o{}}vring},
           familyi={S\bibinitperiod},
           given={Kristian},
           giveni={K\bibinitperiod}}}%
      }
      \strng{namehash}{6000d29f39c0372c76f24cf18718bab7}
      \strng{fullhash}{6000d29f39c0372c76f24cf18718bab7}
      \strng{authornamehash}{6000d29f39c0372c76f24cf18718bab7}
      \strng{authorfullhash}{6000d29f39c0372c76f24cf18718bab7}
      \field{labelalpha}{BMSS12}
      \field{sortinit}{B}
      \field{sortinithash}{5f6fa000f686ee5b41be67ba6ff7962d}
      \field{labelnamesource}{author}
      \field{labeltitlesource}{title}
      \field{issue}{4}
      \field{journaltitle}{{LMCS}}
      \field{title}{First Steps in Synthetic Guarded Domain Theory: Step-Indexing in the Topos of Trees}
      \field{volume}{8}
      \field{year}{2012}
      \field{pages}{1\bibrangedash 45}
      \range{pages}{45}
    \endentry
    \entry{BizjakBGCM16:fossacs}{inproceedings}{}
      \name{author}{5}{}{%
        {{hash=34ed4cfc95e28238c0fabb0b8e292362}{%
           family={Bizjak},
           familyi={B\bibinitperiod},
           given={A.},
           giveni={A\bibinitperiod}}}%
        {{hash=9d3b4ae0575e402ab19716f09d7be0b3}{%
           family={Grathwohl},
           familyi={G\bibinitperiod},
           given={H.B.},
           giveni={H\bibinitperiod}}}%
        {{hash=a335e2b40cadabe119918a701153ab34}{%
           family={Clouston},
           familyi={C\bibinitperiod},
           given={R.},
           giveni={R\bibinitperiod}}}%
        {{hash=0a92bc213687f05d1d19634a32bdf1cd}{%
           family={M{\o{}}gelberg},
           familyi={M\bibinitperiod},
           given={R.E.},
           giveni={R\bibinitperiod}}}%
        {{hash=9d96263a8f3ea9129a610e2480687427}{%
           family={Birkedal},
           familyi={B\bibinitperiod},
           given={L.},
           giveni={L\bibinitperiod}}}%
      }
      \strng{namehash}{76a9ce6d0038860a0cb3b7ad4ce2ce1f}
      \strng{fullhash}{76a9ce6d0038860a0cb3b7ad4ce2ce1f}
      \strng{authornamehash}{76a9ce6d0038860a0cb3b7ad4ce2ce1f}
      \strng{authorfullhash}{76a9ce6d0038860a0cb3b7ad4ce2ce1f}
      \field{labelalpha}{BGCMB16}
      \field{sortinit}{B}
      \field{sortinithash}{5f6fa000f686ee5b41be67ba6ff7962d}
      \field{labelnamesource}{author}
      \field{labeltitlesource}{title}
      \field{booktitle}{Proceedings of {FoSSaCS}}
      \field{title}{Guarded Dependent Type Theory with Coinductive Types}
      \field{year}{2016}
    \endentry
    \entry{BizjakBM14:rtlc}{incollection}{}
      \name{author}{3}{}{%
        {{hash=817e5b781022c3cc6410d371130a87e5}{%
           family={Bizjak},
           familyi={B\bibinitperiod},
           given={Ale\v{s}},
           giveni={A\bibinitperiod}}}%
        {{hash=f4714354745ea31cf6c543434102f444}{%
           family={Birkedal},
           familyi={B\bibinitperiod},
           given={Lars},
           giveni={L\bibinitperiod}}}%
        {{hash=86bf5aa5120acb7a2dad8e2a1adf1e2d}{%
           family={Miculan},
           familyi={M\bibinitperiod},
           given={Marino},
           giveni={M\bibinitperiod}}}%
      }
      \name{editor}{1}{}{%
        {{hash=3b369ef275beda75fb5310cca3a8326a}{%
           family={Dowek},
           familyi={D\bibinitperiod},
           given={Gilles},
           giveni={G\bibinitperiod}}}%
      }
      \list{language}{1}{%
        {English}%
      }
      \list{publisher}{1}{%
        {Springer International Publishing}%
      }
      \strng{namehash}{fda58b3645e06e5e1b578322edb922c4}
      \strng{fullhash}{fda58b3645e06e5e1b578322edb922c4}
      \strng{authornamehash}{fda58b3645e06e5e1b578322edb922c4}
      \strng{authorfullhash}{fda58b3645e06e5e1b578322edb922c4}
      \strng{editornamehash}{3b369ef275beda75fb5310cca3a8326a}
      \strng{editorfullhash}{3b369ef275beda75fb5310cca3a8326a}
      \field{labelalpha}{BBM14}
      \field{sortinit}{B}
      \field{sortinithash}{5f6fa000f686ee5b41be67ba6ff7962d}
      \field{labelnamesource}{author}
      \field{labeltitlesource}{title}
      \field{booktitle}{Proceedings of RTA-TLCA}
      \field{isbn}{978-3-319-08917-1}
      \field{series}{{LNCS}}
      \field{title}{A Model of Countable Nondeterminism in Guarded Type Theory}
      \field{volume}{8560}
      \field{year}{2014}
      \field{pages}{108\bibrangedash 123}
      \range{pages}{16}
      \verb{doi}
      \verb 10.1007/978-3-319-08918-8_8
      \endverb
      \verb{url}
      \verb http://dx.doi.org/10.1007/978-3-319-08918-8_8
      \endverb
    \endentry
    \entry{BizjakM15:entcs}{article}{}
      \name{author}{2}{}{%
        {{hash=c246861cdfff84700e5f8605b429159f}{%
           family={Bizjak},
           familyi={B\bibinitperiod},
           given={Ales},
           giveni={A\bibinitperiod}}}%
        {{hash=87f3fdf5910beec6d8dba13dfb21e1be}{%
           family={M{\o{}}gelberg},
           familyi={M\bibinitperiod},
           given={Rasmus\bibnamedelima Ejlers},
           giveni={R\bibinitperiod\bibinitdelim E\bibinitperiod}}}%
      }
      \strng{namehash}{27d1a7ce32f665409859bb3712e32f9a}
      \strng{fullhash}{27d1a7ce32f665409859bb3712e32f9a}
      \strng{authornamehash}{27d1a7ce32f665409859bb3712e32f9a}
      \strng{authorfullhash}{27d1a7ce32f665409859bb3712e32f9a}
      \field{labelalpha}{BM15}
      \field{sortinit}{B}
      \field{sortinithash}{5f6fa000f686ee5b41be67ba6ff7962d}
      \field{labelnamesource}{author}
      \field{labeltitlesource}{title}
      \field{journaltitle}{{ENTCS}}
      \field{title}{A Model of Guarded Recursion With Clock Synchronisation}
      \field{volume}{319}
      \field{year}{2015}
      \field{pages}{83\bibrangedash 101}
      \range{pages}{19}
      \verb{doi}
      \verb 10.1016/j.entcs.2015.12.007
      \endverb
      \verb{url}
      \verb http://dx.doi.org/10.1016/j.entcs.2015.12.007
      \endverb
    \endentry
    \entry{Boolos1993}{book}{}
      \name{author}{1}{}{%
        {{hash=2a7b53bfa4244e6c5b7e509cef7233ae}{%
           family={Boolos},
           familyi={B\bibinitperiod},
           given={G.},
           giveni={G\bibinitperiod}}}%
      }
      \list{location}{1}{%
        {Cambridge}%
      }
      \list{publisher}{1}{%
        {Cambridge University Press}%
      }
      \strng{namehash}{2a7b53bfa4244e6c5b7e509cef7233ae}
      \strng{fullhash}{2a7b53bfa4244e6c5b7e509cef7233ae}
      \strng{authornamehash}{2a7b53bfa4244e6c5b7e509cef7233ae}
      \strng{authorfullhash}{2a7b53bfa4244e6c5b7e509cef7233ae}
      \field{labelalpha}{Boo93}
      \field{sortinit}{B}
      \field{sortinithash}{5f6fa000f686ee5b41be67ba6ff7962d}
      \field{labelnamesource}{author}
      \field{labeltitlesource}{title}
      \field{title}{The Logic of Provability}
      \field{year}{1993}
    \endentry
    \entry{bool:emer91}{article}{}
      \name{author}{2}{}{%
        {{hash=2a7b53bfa4244e6c5b7e509cef7233ae}{%
           family={Boolos},
           familyi={B\bibinitperiod},
           given={G.},
           giveni={G\bibinitperiod}}}%
        {{hash=2c81ad77fdcc492add5526215bdb6d25}{%
           family={Sambin},
           familyi={S\bibinitperiod},
           given={G.},
           giveni={G\bibinitperiod}}}%
      }
      \strng{namehash}{ce8d42dd8aa7fb9c83952d0d4a156dd1}
      \strng{fullhash}{ce8d42dd8aa7fb9c83952d0d4a156dd1}
      \strng{authornamehash}{ce8d42dd8aa7fb9c83952d0d4a156dd1}
      \strng{authorfullhash}{ce8d42dd8aa7fb9c83952d0d4a156dd1}
      \field{labelalpha}{BS91}
      \field{sortinit}{B}
      \field{sortinithash}{5f6fa000f686ee5b41be67ba6ff7962d}
      \field{labelnamesource}{author}
      \field{labeltitlesource}{title}
      \field{journaltitle}{Studia Logica}
      \field{title}{Provability: the emergence of a mathematical modality}
      \field{volume}{50}
      \field{year}{1991}
      \field{pages}{1\bibrangedash 23}
      \range{pages}{23}
    \endentry
    \entry{BozicD84:sl}{article}{}
      \name{author}{2}{}{%
        {{hash=36547857bec46a50594fbd8e24a0cae7}{%
           family={Bo\v{z}i\'{c}},
           familyi={B\bibinitperiod},
           given={Milan},
           giveni={M\bibinitperiod}}}%
        {{hash=c4ab19be43d195339fde731718dcf134}{%
           family={Do\v{s}en},
           familyi={D\bibinitperiod},
           given={Kosta},
           giveni={K\bibinitperiod}}}%
      }
      \list{publisher}{1}{%
        {Springer}%
      }
      \strng{namehash}{40e57ae3f65df8447c9a7c66605332d8}
      \strng{fullhash}{40e57ae3f65df8447c9a7c66605332d8}
      \strng{authornamehash}{40e57ae3f65df8447c9a7c66605332d8}
      \strng{authorfullhash}{40e57ae3f65df8447c9a7c66605332d8}
      \field{labelalpha}{BD84}
      \field{sortinit}{B}
      \field{sortinithash}{5f6fa000f686ee5b41be67ba6ff7962d}
      \field{labelnamesource}{author}
      \field{labeltitlesource}{title}
      \field{abstract}{Kripke-style models with two accessibility relations, one intuitionistic and the other modal, are given for analogues of the modal system K based on Heyting's propositional logic. It is shown that these two relations can combine with each other in various ways. Soundness and completeness are proved for systems with only the necessity operator, or only the possibility operator, or both. Embeddings in modal systems with several modal operators, based on classical propositional logic, are also considered. This paper lays the ground for an investigation of intuitionistic analogues of systems stronger than K. A brief survey is given of the existing literature on intuitionistic modal logic.}
      \field{issn}{00393215}
      \field{journaltitle}{Studia Logica}
      \field{number}{3}
      \field{title}{Models for Normal Intuitionistic Modal Logics}
      \field{volume}{43}
      \field{year}{1984}
      \field{pages}{217\bibrangedash 245}
      \range{pages}{29}
      \verb{url}
      \verb http://www.jstor.org/stable/20015164
      \endverb
    \endentry
    \entry{CelaniJ01:ndjfl}{article}{}
      \name{author}{2}{}{%
        {{hash=346270a05c86fffb695dfa2787d67103}{%
           family={Celani},
           familyi={C\bibinitperiod},
           given={Sergio},
           giveni={S\bibinitperiod}}}%
        {{hash=3c0295e0ef5b742d981462b229350392}{%
           family={Jansana},
           familyi={J\bibinitperiod},
           given={Ramon},
           giveni={R\bibinitperiod}}}%
      }
      \list{publisher}{1}{%
        {Duke University Press}%
      }
      \strng{namehash}{fb1973141a684ba647862f531665c112}
      \strng{fullhash}{fb1973141a684ba647862f531665c112}
      \strng{authornamehash}{fb1973141a684ba647862f531665c112}
      \strng{authorfullhash}{fb1973141a684ba647862f531665c112}
      \field{labelalpha}{CJ01}
      \field{sortinit}{C}
      \field{sortinithash}{095692fd22cc3c74d7fe223d02314dbd}
      \field{labelnamesource}{author}
      \field{labeltitlesource}{title}
      \field{journaltitle}{Notre Dame J. Formal Logic}
      \field{month}{10}
      \field{number}{4}
      \field{title}{A Closer Look at Some Subintuitionistic Logics}
      \field{volume}{42}
      \field{year}{2001}
      \field{pages}{225\bibrangedash 255}
      \range{pages}{31}
      \verb{doi}
      \verb 10.1305/ndjfl/1063372244
      \endverb
      \verb{url}
      \verb http://dx.doi.org/10.1305/ndjfl/1063372244
      \endverb
    \endentry
    \entry{CelaniJ05:mlq}{article}{}
      \name{author}{2}{}{%
        {{hash=346270a05c86fffb695dfa2787d67103}{%
           family={Celani},
           familyi={C\bibinitperiod},
           given={Sergio},
           giveni={S\bibinitperiod}}}%
        {{hash=3c0295e0ef5b742d981462b229350392}{%
           family={Jansana},
           familyi={J\bibinitperiod},
           given={Ramon},
           giveni={R\bibinitperiod}}}%
      }
      \list{publisher}{1}{%
        {WILEY-VCH Verlag}%
      }
      \strng{namehash}{fb1973141a684ba647862f531665c112}
      \strng{fullhash}{fb1973141a684ba647862f531665c112}
      \strng{authornamehash}{fb1973141a684ba647862f531665c112}
      \strng{authorfullhash}{fb1973141a684ba647862f531665c112}
      \field{labelalpha}{CJ05}
      \field{sortinit}{C}
      \field{sortinithash}{095692fd22cc3c74d7fe223d02314dbd}
      \field{labelnamesource}{author}
      \field{labeltitlesource}{title}
      \field{issn}{1521-3870}
      \field{journaltitle}{Mathematical Logic Quarterly}
      \field{number}{3}
      \field{title}{Bounded distributive lattices with strict implication}
      \field{volume}{51}
      \field{year}{2005}
      \field{pages}{219\bibrangedash 246}
      \range{pages}{28}
      \verb{doi}
      \verb 10.1002/malq.200410022
      \endverb
      \verb{url}
      \verb http://dx.doi.org/10.1002/malq.200410022
      \endverb
      \keyw{Subintuitionistic logic,distributive lattice with a binary operation,subresiduated lattice,basic algebra,Heyting algebra,duality theory}
    \endentry
    \entry{CloustonBGB15:fossacs}{inproceedings}{}
      \name{author}{4}{}{%
        {{hash=2dc908c0aac7961293eb598f4585b1c9}{%
           family={Clouston},
           familyi={C\bibinitperiod},
           given={Ranald},
           giveni={R\bibinitperiod}}}%
        {{hash=c246861cdfff84700e5f8605b429159f}{%
           family={Bizjak},
           familyi={B\bibinitperiod},
           given={Ales},
           giveni={A\bibinitperiod}}}%
        {{hash=6eb7a3363d2f1eed08000448233917bd}{%
           family={Grathwohl},
           familyi={G\bibinitperiod},
           given={Hans\bibnamedelima Bugge},
           giveni={H\bibinitperiod\bibinitdelim B\bibinitperiod}}}%
        {{hash=f4714354745ea31cf6c543434102f444}{%
           family={Birkedal},
           familyi={B\bibinitperiod},
           given={Lars},
           giveni={L\bibinitperiod}}}%
      }
      \name{editor}{1}{}{%
        {{hash=2e998c3985e6a8efab31950621ac04ff}{%
           family={Pitts},
           familyi={P\bibinitperiod},
           given={Andrew\bibnamedelima M.},
           giveni={A\bibinitperiod\bibinitdelim M\bibinitperiod}}}%
      }
      \list{publisher}{1}{%
        {Springer}%
      }
      \strng{namehash}{a03ba48c5c273fad7689aa14d1cc6c52}
      \strng{fullhash}{a03ba48c5c273fad7689aa14d1cc6c52}
      \strng{authornamehash}{a03ba48c5c273fad7689aa14d1cc6c52}
      \strng{authorfullhash}{a03ba48c5c273fad7689aa14d1cc6c52}
      \strng{editornamehash}{2e998c3985e6a8efab31950621ac04ff}
      \strng{editorfullhash}{2e998c3985e6a8efab31950621ac04ff}
      \field{labelalpha}{CBGB15}
      \field{sortinit}{C}
      \field{sortinithash}{095692fd22cc3c74d7fe223d02314dbd}
      \field{labelnamesource}{author}
      \field{labeltitlesource}{title}
      \field{booktitle}{Proceedings of {FoSSaCS}}
      \strng{crossref}{DBLP:conf/fossacs/2015}
      \field{series}{{LNCS}}
      \field{title}{Programming and Reasoning with Guarded Recursion for Coinductive Types}
      \field{volume}{9034}
      \field{year}{2015}
      \field{pages}{407\bibrangedash 421}
      \range{pages}{15}
      \verb{doi}
      \verb 10.1007/978-3-662-46678-0_26
      \endverb
      \verb{url}
      \verb http://dx.doi.org/10.1007/978-3-662-46678-0_26
      \endverb
    \endentry
    \entry{CloustonG15:fossacs}{inproceedings}{}
      \name{author}{2}{}{%
        {{hash=2dc908c0aac7961293eb598f4585b1c9}{%
           family={Clouston},
           familyi={C\bibinitperiod},
           given={Ranald},
           giveni={R\bibinitperiod}}}%
        {{hash=8c4ba0a2bbed11bccf9d81ded27b3f1f}{%
           family={Gor\'{e}},
           familyi={G\bibinitperiod},
           given={Rajeev},
           giveni={R\bibinitperiod}}}%
      }
      \name{editor}{1}{}{%
        {{hash=2e998c3985e6a8efab31950621ac04ff}{%
           family={Pitts},
           familyi={P\bibinitperiod},
           given={Andrew\bibnamedelima M.},
           giveni={A\bibinitperiod\bibinitdelim M\bibinitperiod}}}%
      }
      \list{publisher}{1}{%
        {Springer}%
      }
      \strng{namehash}{c18892566e35d117d2ba813d350e9344}
      \strng{fullhash}{c18892566e35d117d2ba813d350e9344}
      \strng{authornamehash}{c18892566e35d117d2ba813d350e9344}
      \strng{authorfullhash}{c18892566e35d117d2ba813d350e9344}
      \strng{editornamehash}{2e998c3985e6a8efab31950621ac04ff}
      \strng{editorfullhash}{2e998c3985e6a8efab31950621ac04ff}
      \field{labelalpha}{CG15}
      \field{sortinit}{C}
      \field{sortinithash}{095692fd22cc3c74d7fe223d02314dbd}
      \field{labelnamesource}{author}
      \field{labeltitlesource}{title}
      \field{booktitle}{Proceedings of {FoSSaCS}}
      \strng{crossref}{DBLP:conf/fossacs/2015}
      \field{series}{{LNCS}}
      \field{title}{Sequent Calculus in the Topos of Trees}
      \field{volume}{9034}
      \field{year}{2015}
      \field{pages}{133\bibrangedash 147}
      \range{pages}{15}
    \endentry
    \entry{Corsi87:mlq}{article}{}
      \name{author}{1}{}{%
        {{hash=6cf33341658de1045024e7fa43199cd3}{%
           family={Corsi},
           familyi={C\bibinitperiod},
           given={Giovanna},
           giveni={G\bibinitperiod}}}%
      }
      \list{publisher}{1}{%
        {WILEY-VCH Verlag Berlin GmbH}%
      }
      \strng{namehash}{6cf33341658de1045024e7fa43199cd3}
      \strng{fullhash}{6cf33341658de1045024e7fa43199cd3}
      \strng{authornamehash}{6cf33341658de1045024e7fa43199cd3}
      \strng{authorfullhash}{6cf33341658de1045024e7fa43199cd3}
      \field{labelalpha}{Cor87}
      \field{sortinit}{C}
      \field{sortinithash}{095692fd22cc3c74d7fe223d02314dbd}
      \field{labelnamesource}{author}
      \field{labeltitlesource}{title}
      \field{issn}{1521-3870}
      \field{journaltitle}{Mathematical Logic Quarterly}
      \field{number}{5}
      \field{title}{Weak Logics with Strict Implication}
      \field{volume}{33}
      \field{year}{1987}
      \field{pages}{389\bibrangedash 406}
      \range{pages}{18}
      \verb{doi}
      \verb 10.1002/malq.19870330503
      \endverb
      \verb{url}
      \verb http://dx.doi.org/10.1002/malq.19870330503
      \endverb
    \endentry
    \entry{cotr:tran13}{article}{}
      \name{author}{2}{}{%
        {{hash=f0b66a8c73f3822a239f0463a3870899}{%
           family={Cotrini},
           familyi={C\bibinitperiod},
           given={C.},
           giveni={C\bibinitperiod}}}%
        {{hash=5acf9628f31347408bc2dddcb8459cd3}{%
           family={Gurevich},
           familyi={G\bibinitperiod},
           given={Y.},
           giveni={Y\bibinitperiod}}}%
      }
      \list{publisher}{1}{%
        {Cambridge University Press}%
      }
      \strng{namehash}{9f26b7d88ee1aa77bdeb9ff344a0f2fb}
      \strng{fullhash}{9f26b7d88ee1aa77bdeb9ff344a0f2fb}
      \strng{authornamehash}{9f26b7d88ee1aa77bdeb9ff344a0f2fb}
      \strng{authorfullhash}{9f26b7d88ee1aa77bdeb9ff344a0f2fb}
      \field{labelalpha}{CG13}
      \field{sortinit}{C}
      \field{sortinithash}{095692fd22cc3c74d7fe223d02314dbd}
      \field{labelnamesource}{author}
      \field{labeltitlesource}{title}
      \field{journaltitle}{The Review of Symbolic Logic}
      \field{number}{02}
      \field{title}{Transitive primal infon logic}
      \field{volume}{6}
      \field{year}{2013}
      \field{pages}{281\bibrangedash 304}
      \range{pages}{24}
    \endentry
    \entry{Dosen93}{incollection}{}
      \name{author}{1}{}{%
        {{hash=c4ab19be43d195339fde731718dcf134}{%
           family={Do\v{s}en},
           familyi={D\bibinitperiod},
           given={Kosta},
           giveni={K\bibinitperiod}}}%
      }
      \name{editor}{1}{}{%
        {{hash=ce8500c4435316ca20758d1f88f330bb}{%
           family={Rijke},
           familyi={R\bibinitperiod},
           given={Maarten},
           giveni={M\bibinitperiod},
           prefix={de},
           prefixi={d\bibinitperiod}}}%
      }
      \list{location}{1}{%
        {Dordrecht}%
      }
      \list{publisher}{1}{%
        {Springer Netherlands}%
      }
      \strng{namehash}{c4ab19be43d195339fde731718dcf134}
      \strng{fullhash}{c4ab19be43d195339fde731718dcf134}
      \strng{authornamehash}{c4ab19be43d195339fde731718dcf134}
      \strng{authorfullhash}{c4ab19be43d195339fde731718dcf134}
      \strng{editornamehash}{ce8500c4435316ca20758d1f88f330bb}
      \strng{editorfullhash}{ce8500c4435316ca20758d1f88f330bb}
      \field{labelalpha}{Do\v{s}93}
      \field{sortinit}{D}
      \field{sortinithash}{d10b5413de1f3d197b20897dd0d565bb}
      \field{labelnamesource}{author}
      \field{labeltitlesource}{title}
      \field{booktitle}{Diamonds and Defaults: Studies in Pure and Applied Intensional Logic}
      \field{isbn}{978-94-015-8242-1}
      \field{title}{Modal Translations in {K} and {D}}
      \field{year}{1993}
      \field{pages}{103\bibrangedash 127}
      \range{pages}{25}
    \endentry
    \entry{Dosen1992:orlov}{article}{}
      \name{author}{1}{}{%
        {{hash=c4ab19be43d195339fde731718dcf134}{%
           family={Do\v{s}en},
           familyi={D\bibinitperiod},
           given={Kosta},
           giveni={K\bibinitperiod}}}%
      }
      \strng{namehash}{c4ab19be43d195339fde731718dcf134}
      \strng{fullhash}{c4ab19be43d195339fde731718dcf134}
      \strng{authornamehash}{c4ab19be43d195339fde731718dcf134}
      \strng{authorfullhash}{c4ab19be43d195339fde731718dcf134}
      \field{labelalpha}{Do\v{s}92}
      \field{sortinit}{D}
      \field{sortinithash}{d10b5413de1f3d197b20897dd0d565bb}
      \field{labelnamesource}{author}
      \field{labeltitlesource}{title}
      \field{abstract}{This is a review, with historical and critical comments, of a paper by I. E. Orlov from 1928, which gives the oldest known axiomatization of the implication-negation fragment of the relevant logic R. Orlov's paper also foreshadows the modal translation of systems with an intuitionistic negation into S4-type extensions of systems with a classical, involutive, negation. Orlov introduces the modal postulates of S4 before Becker, Lewis and G\"{o}del. Orlov's work, which seems to be nearly completely ignored, is related to the contemporancous work on the axiomatization of intuitionistic logic.}
      \field{issn}{1573-0433}
      \field{journaltitle}{Journal of Philosophical Logic}
      \field{month}{11}
      \field{number}{4}
      \field{title}{The first axiomatization of relevant logic}
      \field{volume}{21}
      \field{year}{1992}
      \field{pages}{339\bibrangedash 356}
      \range{pages}{18}
      \verb{doi}
      \verb 10.1007/BF00260740
      \endverb
      \verb{url}
      \verb https://doi.org/10.1007/BF00260740
      \endverb
    \endentry
    \entry{Dragalin88:trams}{book}{}
      \name{author}{1}{}{%
        {{hash=99b46c29553cec302374a1d09e91d53a}{%
           family={Dragalin},
           familyi={D\bibinitperiod},
           given={Albert\bibnamedelima Grigorevich},
           giveni={A\bibinitperiod\bibinitdelim G\bibinitperiod}}}%
      }
      \list{publisher}{1}{%
        {American Mathematical Society}%
      }
      \strng{namehash}{99b46c29553cec302374a1d09e91d53a}
      \strng{fullhash}{99b46c29553cec302374a1d09e91d53a}
      \strng{authornamehash}{99b46c29553cec302374a1d09e91d53a}
      \strng{authorfullhash}{99b46c29553cec302374a1d09e91d53a}
      \field{labelalpha}{Dra88}
      \field{sortinit}{D}
      \field{sortinithash}{d10b5413de1f3d197b20897dd0d565bb}
      \field{labelnamesource}{author}
      \field{labeltitlesource}{title}
      \field{note}{Translated by E. Mendelson, ed. by B. Silver}
      \field{series}{Translations of Mathematical Monographs}
      \field{title}{Mathematical Intuitionism: Introduction to Proof Theory}
      \field{volume}{67}
      \field{year}{1988}
    \endentry
    \entry{DreyerAB11:lmcs}{article}{}
      \name{author}{3}{}{%
        {{hash=3c61a075f82133ed00b1cf07a5b5c194}{%
           family={Dreyer},
           familyi={D\bibinitperiod},
           given={Derek},
           giveni={D\bibinitperiod}}}%
        {{hash=94a245e2a7f48b654d74d292a96172d0}{%
           family={Ahmed},
           familyi={A\bibinitperiod},
           given={Amal},
           giveni={A\bibinitperiod}}}%
        {{hash=f4714354745ea31cf6c543434102f444}{%
           family={Birkedal},
           familyi={B\bibinitperiod},
           given={Lars},
           giveni={L\bibinitperiod}}}%
      }
      \strng{namehash}{87ce23b74c8f4ef7e8a1b3e1842ea08d}
      \strng{fullhash}{87ce23b74c8f4ef7e8a1b3e1842ea08d}
      \strng{authornamehash}{87ce23b74c8f4ef7e8a1b3e1842ea08d}
      \strng{authorfullhash}{87ce23b74c8f4ef7e8a1b3e1842ea08d}
      \field{labelalpha}{DAB11}
      \field{sortinit}{D}
      \field{sortinithash}{d10b5413de1f3d197b20897dd0d565bb}
      \field{labelnamesource}{author}
      \field{labeltitlesource}{title}
      \field{journaltitle}{{LMCS}}
      \field{number}{2}
      \field{title}{Logical Step-Indexed Logical Relations}
      \field{volume}{7}
      \field{year}{2011}
      \verb{doi}
      \verb 10.2168/LMCS-7(2:16)2011
      \endverb
      \verb{url}
      \verb http://dx.doi.org/10.2168/LMCS-7(2:16)2011
      \endverb
    \endentry
    \entry{Dunn72}{article}{}
      \name{author}{1}{}{%
        {{hash=b8fc70824855efc7d89bc98323f8f3ec}{%
           family={Dunn},
           familyi={D\bibinitperiod},
           given={J.\bibnamedelimi Michael},
           giveni={J\bibinitperiod\bibinitdelim M\bibinitperiod}}}%
      }
      \list{publisher}{1}{%
        {Duke University Press}%
      }
      \strng{namehash}{b8fc70824855efc7d89bc98323f8f3ec}
      \strng{fullhash}{b8fc70824855efc7d89bc98323f8f3ec}
      \strng{authornamehash}{b8fc70824855efc7d89bc98323f8f3ec}
      \strng{authorfullhash}{b8fc70824855efc7d89bc98323f8f3ec}
      \field{labelalpha}{Dun72}
      \field{sortinit}{D}
      \field{sortinithash}{d10b5413de1f3d197b20897dd0d565bb}
      \field{labelnamesource}{author}
      \field{labeltitlesource}{title}
      \field{journaltitle}{Notre Dame J. Formal Logic}
      \field{month}{04}
      \field{number}{2}
      \field{title}{A modification of {P}arry's analytic implication.}
      \field{volume}{13}
      \field{year}{1972}
      \field{pages}{195\bibrangedash 205}
      \range{pages}{11}
      \verb{doi}
      \verb 10.1305/ndjfl/1093894715
      \endverb
      \verb{url}
      \verb http://dx.doi.org/10.1305/ndjfl/1093894715
      \endverb
    \endentry
    \entry{FairtloughM97:ic}{article}{}
      \name{author}{2}{}{%
        {{hash=b67dd7eb5995a6de4e6249b0fc708c5e}{%
           family={Fairtlough},
           familyi={F\bibinitperiod},
           given={Matt},
           giveni={M\bibinitperiod}}}%
        {{hash=8e080dab84b81c2d36b0172b0d64e46b}{%
           family={Mendler},
           familyi={M\bibinitperiod},
           given={Michael},
           giveni={M\bibinitperiod}}}%
      }
      \strng{namehash}{10047853bd0a54a43f4b0c4aab290742}
      \strng{fullhash}{10047853bd0a54a43f4b0c4aab290742}
      \strng{authornamehash}{10047853bd0a54a43f4b0c4aab290742}
      \strng{authorfullhash}{10047853bd0a54a43f4b0c4aab290742}
      \field{labelalpha}{FM97}
      \field{sortinit}{F}
      \field{sortinithash}{276475738cc058478c1677046f857703}
      \field{labelnamesource}{author}
      \field{labeltitlesource}{title}
      \field{journaltitle}{Inform. and Comput.}
      \field{number}{1}
      \field{title}{Propositional Lax Logic}
      \field{volume}{137}
      \field{year}{1997}
      \field{pages}{1\bibrangedash 33}
      \range{pages}{33}
    \endentry
    \entry{Fine86}{article}{}
      \name{author}{1}{}{%
        {{hash=68c361f515af4ba746fa726b79ce7825}{%
           family={Fine},
           familyi={F\bibinitperiod},
           given={Kit},
           giveni={K\bibinitperiod}}}%
      }
      \list{publisher}{1}{%
        {Duke University Press}%
      }
      \strng{namehash}{68c361f515af4ba746fa726b79ce7825}
      \strng{fullhash}{68c361f515af4ba746fa726b79ce7825}
      \strng{authornamehash}{68c361f515af4ba746fa726b79ce7825}
      \strng{authorfullhash}{68c361f515af4ba746fa726b79ce7825}
      \field{labelalpha}{Fin86}
      \field{sortinit}{F}
      \field{sortinithash}{276475738cc058478c1677046f857703}
      \field{labelnamesource}{author}
      \field{labeltitlesource}{title}
      \field{journaltitle}{Notre Dame J. Formal Logic}
      \field{month}{04}
      \field{number}{2}
      \field{title}{Analytic implication.}
      \field{volume}{27}
      \field{year}{1986}
      \field{pages}{169\bibrangedash 179}
      \range{pages}{11}
      \verb{doi}
      \verb 10.1305/ndjfl/1093636609
      \endverb
      \verb{url}
      \verb http://dx.doi.org/10.1305/ndjfl/1093636609
      \endverb
    \endentry
    \entry{Fine75:connections}{inproceedings}{}
      \name{author}{1}{}{%
        {{hash=68c361f515af4ba746fa726b79ce7825}{%
           family={Fine},
           familyi={F\bibinitperiod},
           given={Kit},
           giveni={K\bibinitperiod}}}%
      }
      \name{editor}{1}{}{%
        {{hash=e602da6416052b54dfc1959cb174bd16}{%
           family={Kanger},
           familyi={K\bibinitperiod},
           given={S.},
           giveni={S\bibinitperiod}}}%
      }
      \list{publisher}{1}{%
        {North-Holland Publishing Company}%
      }
      \strng{namehash}{68c361f515af4ba746fa726b79ce7825}
      \strng{fullhash}{68c361f515af4ba746fa726b79ce7825}
      \strng{authornamehash}{68c361f515af4ba746fa726b79ce7825}
      \strng{authorfullhash}{68c361f515af4ba746fa726b79ce7825}
      \strng{editornamehash}{e602da6416052b54dfc1959cb174bd16}
      \strng{editorfullhash}{e602da6416052b54dfc1959cb174bd16}
      \field{labelalpha}{Fin75}
      \field{sortinit}{F}
      \field{sortinithash}{276475738cc058478c1677046f857703}
      \field{labelnamesource}{author}
      \field{labeltitlesource}{title}
      \field{booktitle}{Proceedings of the Third Scandinavian Logic Symposium}
      \field{title}{Some Connections between Elementary and Modal Logic}
      \field{year}{1975}
    \endentry
    \entry{garg:prov84}{inproceedings}{}
      \name{author}{1}{}{%
        {{hash=8f21fc6e98230cecd0b686aaebcccfb9}{%
           family={Gargov},
           familyi={G\bibinitperiod},
           given={G.\bibnamedelimi K.},
           giveni={G\bibinitperiod\bibinitdelim K\bibinitperiod}}}%
      }
      \list{publisher}{1}{%
        {Bulgarian Academy of Sciences}%
      }
      \strng{namehash}{8f21fc6e98230cecd0b686aaebcccfb9}
      \strng{fullhash}{8f21fc6e98230cecd0b686aaebcccfb9}
      \strng{authornamehash}{8f21fc6e98230cecd0b686aaebcccfb9}
      \strng{authorfullhash}{8f21fc6e98230cecd0b686aaebcccfb9}
      \field{labelalpha}{Gar84}
      \field{sortinit}{G}
      \field{sortinithash}{618d986594b7198ba52cf8b00d348f3f}
      \field{labelnamesource}{author}
      \field{labeltitlesource}{title}
      \field{booktitle}{Mathematical {L}ogic, proceedings of the conference on {M}athematical {L}ogic, dedicated to the memory of {A}. {A}. {M}arkov (1903--1979), {S}ofia, {S}eptember 22--23, 1980}
      \field{title}{A note on the provability logics of certain extensions of {H}eyting's {A}rithmetic}
      \field{year}{1984}
      \field{pages}{20\bibrangedash 26}
      \range{pages}{7}
    \endentry
    \entry{GHV06}{article}{}
      \name{author}{3}{}{%
        {{hash=7050d6d1f17087a7a5f3147e5b86a1f6}{%
           family={Gehrke},
           familyi={G\bibinitperiod},
           given={M.},
           giveni={M\bibinitperiod}}}%
        {{hash=aa4539700054af76f08b5a2c2d629048}{%
           family={Harding},
           familyi={H\bibinitperiod},
           given={J.},
           giveni={J\bibinitperiod}}}%
        {{hash=f8c3d122473d3f9cc001bd94a34e334d}{%
           family={Venema},
           familyi={V\bibinitperiod},
           given={Y},
           giveni={Y\bibinitperiod}}}%
      }
      \strng{namehash}{2ef99519cbe91d115fd0a329afbc7e5a}
      \strng{fullhash}{2ef99519cbe91d115fd0a329afbc7e5a}
      \strng{authornamehash}{2ef99519cbe91d115fd0a329afbc7e5a}
      \strng{authorfullhash}{2ef99519cbe91d115fd0a329afbc7e5a}
      \field{labelalpha}{GHV06}
      \field{sortinit}{G}
      \field{sortinithash}{618d986594b7198ba52cf8b00d348f3f}
      \field{labelnamesource}{author}
      \field{labeltitlesource}{title}
      \field{journaltitle}{Transactions of the American Mathematical Society}
      \field{title}{Mac{N}eille completions and canonical extensions}
      \field{volume}{358}
      \field{year}{2006}
      \field{pages}{573\bibrangedash 590}
      \range{pages}{18}
    \endentry
    \entry{Girard87}{article}{}
      \name{author}{1}{}{%
        {{hash=b2918bf11715bb95276ff254fb247d4b}{%
           family={Girard},
           familyi={G\bibinitperiod},
           given={Jean-Yves},
           giveni={J\bibinithyphendelim Y\bibinitperiod}}}%
      }
      \strng{namehash}{b2918bf11715bb95276ff254fb247d4b}
      \strng{fullhash}{b2918bf11715bb95276ff254fb247d4b}
      \strng{authornamehash}{b2918bf11715bb95276ff254fb247d4b}
      \strng{authorfullhash}{b2918bf11715bb95276ff254fb247d4b}
      \field{labelalpha}{Gir87}
      \field{sortinit}{G}
      \field{sortinithash}{618d986594b7198ba52cf8b00d348f3f}
      \field{labelnamesource}{author}
      \field{labeltitlesource}{title}
      \field{journaltitle}{Theoretical Computer Science}
      \field{title}{Linear Logic}
      \field{volume}{50}
      \field{year}{1987}
      \field{pages}{1\bibrangedash 102}
      \range{pages}{102}
    \endentry
    \entry{Goedelv1}{book}{}
      \name{author}{1}{}{%
        {{hash=741b5eae8f0486fe4cfb19725eb596a0}{%
           family={G\"{o}del},
           familyi={G\bibinitperiod},
           given={Kurt},
           giveni={K\bibinitperiod}}}%
      }
      \list{publisher}{1}{%
        {OUP USA}%
      }
      \strng{namehash}{741b5eae8f0486fe4cfb19725eb596a0}
      \strng{fullhash}{741b5eae8f0486fe4cfb19725eb596a0}
      \strng{authornamehash}{741b5eae8f0486fe4cfb19725eb596a0}
      \strng{authorfullhash}{741b5eae8f0486fe4cfb19725eb596a0}
      \field{labelalpha}{G\"{o}d86}
      \field{sortinit}{G}
      \field{sortinithash}{618d986594b7198ba52cf8b00d348f3f}
      \field{labelnamesource}{author}
      \field{labeltitlesource}{title}
      \field{note}{ed. by S. Feferman et al.}
      \field{series}{Collected Works}
      \field{title}{Kurt G\"{o}del: Publications 1929--1936}
      \field{volume}{1}
      \field{year}{1986}
    \endentry
    \entry{Goldblatt81:mlq}{article}{}
      \name{author}{1}{}{%
        {{hash=67bb8e8058d2f020f3e5302702a18037}{%
           family={Goldblatt},
           familyi={G\bibinitperiod},
           given={Robert\bibnamedelima I.},
           giveni={R\bibinitperiod\bibinitdelim I\bibinitperiod}}}%
      }
      \list{publisher}{1}{%
        {WILEY-VCH Verlag Berlin GmbH}%
      }
      \strng{namehash}{67bb8e8058d2f020f3e5302702a18037}
      \strng{fullhash}{67bb8e8058d2f020f3e5302702a18037}
      \strng{authornamehash}{67bb8e8058d2f020f3e5302702a18037}
      \strng{authorfullhash}{67bb8e8058d2f020f3e5302702a18037}
      \field{labelalpha}{Gol81}
      \field{sortinit}{G}
      \field{sortinithash}{618d986594b7198ba52cf8b00d348f3f}
      \field{labelnamesource}{author}
      \field{labeltitlesource}{title}
      \field{issn}{1521-3870}
      \field{journaltitle}{{MLQ}}
      \field{number}{31--35}
      \field{title}{Grothendieck Topology as Geometric Modality}
      \field{volume}{27}
      \field{year}{1981}
      \field{pages}{495\bibrangedash 529}
      \range{pages}{35}
      \verb{doi}
      \verb 10.1002/malq.19810273104
      \endverb
      \verb{url}
      \verb http://dx.doi.org/10.1002/malq.19810273104
      \endverb
    \endentry
    \entry{haje:cons90}{article}{}
      \name{author}{2}{}{%
        {{hash=264d232906477f45b77faf271ac79f37}{%
           family={H\'{a}jek},
           familyi={H\bibinitperiod},
           given={P.},
           giveni={P\bibinitperiod}}}%
        {{hash=5239d615f497bac4c8d2489df911d02d}{%
           family={Montagna},
           familyi={M\bibinitperiod},
           given={F.},
           giveni={F\bibinitperiod}}}%
      }
      \strng{namehash}{fb05b53e566a8ea159e6bbcdb26c2af5}
      \strng{fullhash}{fb05b53e566a8ea159e6bbcdb26c2af5}
      \strng{authornamehash}{fb05b53e566a8ea159e6bbcdb26c2af5}
      \strng{authorfullhash}{fb05b53e566a8ea159e6bbcdb26c2af5}
      \field{labelalpha}{HM90}
      \field{sortinit}{H}
      \field{sortinithash}{2f664b453ec75da1fe3804ca92633405}
      \field{labelnamesource}{author}
      \field{labeltitlesource}{title}
      \field{journaltitle}{Archiv f\"{u}r Mathematische Logik und Grundlagenforschung}
      \field{title}{The Logic of {$\Pi_1$}-conservativity}
      \field{volume}{30}
      \field{year}{1990}
      \field{pages}{113\bibrangedash 123}
      \range{pages}{11}
    \endentry
    \entry{halb:henk14}{incollection}{}
      \name{author}{2}{}{%
        {{hash=3e061fe1ac2ef476241b38ee00db1f15}{%
           family={Halbach},
           familyi={H\bibinitperiod},
           given={V.},
           giveni={V\bibinitperiod}}}%
        {{hash=7c59a81f2d67e5dbdf77cc8061473f2d}{%
           family={Visser},
           familyi={V\bibinitperiod},
           given={A.},
           giveni={A\bibinitperiod}}}%
      }
      \name{editor}{3}{}{%
        {{hash=d2969faf2d2f928408fa8ac38657ef9a}{%
           family={Manzano},
           familyi={M\bibinitperiod},
           given={M.},
           giveni={M\bibinitperiod}}}%
        {{hash=acbc53751758de102e47e9eb8edb4566}{%
           family={Sain},
           familyi={S\bibinitperiod},
           given={I.},
           giveni={I\bibinitperiod}}}%
        {{hash=ccc4b0ec47e22b0699221ccc45b81659}{%
           family={Alonso},
           familyi={A\bibinitperiod},
           given={E.},
           giveni={E\bibinitperiod}}}%
      }
      \list{publisher}{1}{%
        {Springer}%
      }
      \strng{namehash}{900bcff3bad122c7bc5afb43fa505aca}
      \strng{fullhash}{900bcff3bad122c7bc5afb43fa505aca}
      \strng{authornamehash}{900bcff3bad122c7bc5afb43fa505aca}
      \strng{authorfullhash}{900bcff3bad122c7bc5afb43fa505aca}
      \strng{editornamehash}{a91ee17dd37175889db1bbd4a887c9a8}
      \strng{editorfullhash}{a91ee17dd37175889db1bbd4a887c9a8}
      \field{labelalpha}{HV14}
      \field{sortinit}{H}
      \field{sortinithash}{2f664b453ec75da1fe3804ca92633405}
      \field{labelnamesource}{author}
      \field{labeltitlesource}{title}
      \field{booktitle}{The Life and Work of {L}eon {H}enkin}
      \field{title}{The {H}enkin sentence}
      \field{year}{2014}
      \field{pages}{249\bibrangedash 263}
      \range{pages}{15}
    \endentry
    \entry{hodg:foun96}{proceedings}{}
      \name{editor}{4}{}{%
        {{hash=c1bb1fc061237f46cc9852530e09d7e6}{%
           family={Hodges},
           familyi={H\bibinitperiod},
           given={W.},
           giveni={W\bibinitperiod}}}%
        {{hash=dc5936afe8fedbfe955e61816bb3223c}{%
           family={Hyland},
           familyi={H\bibinitperiod},
           given={M.},
           giveni={M\bibinitperiod}}}%
        {{hash=a41f373941d989f3413e06ecdacce707}{%
           family={Steinhorn},
           familyi={S\bibinitperiod},
           given={C.},
           giveni={C\bibinitperiod}}}%
        {{hash=97dc326461152543c695748917c3b100}{%
           family={Truss},
           familyi={T\bibinitperiod},
           given={J.},
           giveni={J\bibinitperiod}}}%
      }
      \list{publisher}{1}{%
        {Clarendon Press, Oxford}%
      }
      \strng{namehash}{2ede39bca1fb2203c381a0cfc8e65eb5}
      \strng{fullhash}{2ede39bca1fb2203c381a0cfc8e65eb5}
      \strng{editornamehash}{2ede39bca1fb2203c381a0cfc8e65eb5}
      \strng{editorfullhash}{2ede39bca1fb2203c381a0cfc8e65eb5}
      \field{labelalpha}{HHST96}
      \field{sortinit}{H}
      \field{sortinithash}{2f664b453ec75da1fe3804ca92633405}
      \field{labelnamesource}{editor}
      \field{labeltitlesource}{title}
      \field{title}{Logic: from foundations to applications}
      \field{year}{1996}
    \endentry
    \entry{HollidayL16}{misc}{}
      \name{author}{2}{}{%
        {{hash=34b0bd1f18eafa2894e9adba9a23b205}{%
           family={Holliday},
           familyi={H\bibinitperiod},
           given={Wesley\bibnamedelima H.},
           giveni={W\bibinitperiod\bibinitdelim H\bibinitperiod}}}%
        {{hash=8f5431086357f9f5a0eef79eb62f74f4}{%
           family={Litak},
           familyi={L\bibinitperiod},
           given={Tadeusz},
           giveni={T\bibinitperiod}}}%
      }
      \strng{namehash}{0ac8d8f9b968c914d4eaa1eb99bb5668}
      \strng{fullhash}{0ac8d8f9b968c914d4eaa1eb99bb5668}
      \strng{authornamehash}{0ac8d8f9b968c914d4eaa1eb99bb5668}
      \strng{authorfullhash}{0ac8d8f9b968c914d4eaa1eb99bb5668}
      \field{labelalpha}{HL16}
      \field{sortinit}{H}
      \field{sortinithash}{2f664b453ec75da1fe3804ca92633405}
      \field{labelnamesource}{author}
      \field{labeltitlesource}{title}
      \field{note}{eScholarship system of UC Berkeley, Working Papers series. Accepted by ``The Review of Symbolic Logic''}
      \field{title}{Complete Additivity and Modal Incompleteness}
      \field{year}{2016}
      \verb{url}
      \verb http://www.escholarship.org/uc/item/8pp4d94t
      \endverb
    \endentry
    \entry{Hughes00:scp}{article}{}
      \name{author}{1}{}{%
        {{hash=05db43db1a50c29c0eca01fd35ec8758}{%
           family={Hughes},
           familyi={H\bibinitperiod},
           given={John},
           giveni={J\bibinitperiod}}}%
      }
      \strng{namehash}{05db43db1a50c29c0eca01fd35ec8758}
      \strng{fullhash}{05db43db1a50c29c0eca01fd35ec8758}
      \strng{authornamehash}{05db43db1a50c29c0eca01fd35ec8758}
      \strng{authorfullhash}{05db43db1a50c29c0eca01fd35ec8758}
      \field{labelalpha}{Hug00}
      \field{sortinit}{H}
      \field{sortinithash}{2f664b453ec75da1fe3804ca92633405}
      \field{labelnamesource}{author}
      \field{labeltitlesource}{title}
      \field{journaltitle}{Science of Computer Programming}
      \field{number}{1-3}
      \field{title}{Generalising monads to arrows}
      \field{volume}{37}
      \field{year}{2000}
      \field{pages}{67\bibrangedash 111}
      \range{pages}{45}
      \verb{doi}
      \verb 10.1016/S0167-6423(99)00023-4
      \endverb
      \verb{url}
      \verb http://dx.doi.org/10.1016/S0167-6423(99)00023-4
      \endverb
    \endentry
    \entry{iemh:moda01}{inproceedings}{}
      \name{author}{1}{}{%
        {{hash=3fec7473b9b7d0881d58711403e21a4e}{%
           family={Iemhoff},
           familyi={I\bibinitperiod},
           given={R.},
           giveni={R\bibinitperiod}}}%
      }
      \list{location}{1}{%
        {Uppsala}%
      }
      \strng{namehash}{3fec7473b9b7d0881d58711403e21a4e}
      \strng{fullhash}{3fec7473b9b7d0881d58711403e21a4e}
      \strng{authornamehash}{3fec7473b9b7d0881d58711403e21a4e}
      \strng{authorfullhash}{3fec7473b9b7d0881d58711403e21a4e}
      \field{labelalpha}{Iem01}
      \field{sortinit}{I}
      \field{sortinithash}{a3dcedd53b04d1adfd5ac303ecd5e6fa}
      \field{extraalpha}{1}
      \field{labelnamesource}{author}
      \field{labeltitlesource}{title}
      \field{booktitle}{Proceedings of AiML'98}
      \field{title}{A Modal Analysis of Some Principles of the Provability Logic of {H}eyting {A}rithmetic}
      \field{volume}{2}
      \field{year}{2001}
    \endentry
    \entry{iemh:pres03}{article}{}
      \name{author}{1}{}{%
        {{hash=3fec7473b9b7d0881d58711403e21a4e}{%
           family={Iemhoff},
           familyi={I\bibinitperiod},
           given={R.},
           giveni={R\bibinitperiod}}}%
      }
      \list{publisher}{1}{%
        {Wiley Online Library}%
      }
      \strng{namehash}{3fec7473b9b7d0881d58711403e21a4e}
      \strng{fullhash}{3fec7473b9b7d0881d58711403e21a4e}
      \strng{authornamehash}{3fec7473b9b7d0881d58711403e21a4e}
      \strng{authorfullhash}{3fec7473b9b7d0881d58711403e21a4e}
      \field{labelalpha}{Iem03}
      \field{sortinit}{I}
      \field{sortinithash}{a3dcedd53b04d1adfd5ac303ecd5e6fa}
      \field{labelnamesource}{author}
      \field{labeltitlesource}{title}
      \field{journaltitle}{Mathematical Logic Quarterly}
      \field{number}{3}
      \field{title}{Preservativity logic: An analogue of interpretability logic for constructive theories}
      \field{volume}{49}
      \field{year}{2003}
      \field{pages}{230\bibrangedash 249}
      \range{pages}{20}
    \endentry
    \entry{iemh:prop05}{article}{}
      \name{author}{3}{}{%
        {{hash=3fec7473b9b7d0881d58711403e21a4e}{%
           family={Iemhoff},
           familyi={I\bibinitperiod},
           given={R.},
           giveni={R\bibinitperiod}}}%
        {{hash=52216a6faee5a0bd06313992005d48c4}{%
           family={De\bibnamedelima Jongh},
           familyi={D\bibinitperiod\bibinitdelim J\bibinitperiod},
           given={D.H.J.},
           giveni={D\bibinitperiod}}}%
        {{hash=f45998c28e76ee1438ff765f242dcfc1}{%
           family={Zhou},
           familyi={Z\bibinitperiod},
           given={C.},
           giveni={C\bibinitperiod}}}%
      }
      \list{publisher}{1}{%
        {Oxford Univ Press}%
      }
      \strng{namehash}{a51f9ae56ce2966c8e49d3d0ad8d2c3e}
      \strng{fullhash}{a51f9ae56ce2966c8e49d3d0ad8d2c3e}
      \strng{authornamehash}{a51f9ae56ce2966c8e49d3d0ad8d2c3e}
      \strng{authorfullhash}{a51f9ae56ce2966c8e49d3d0ad8d2c3e}
      \field{labelalpha}{IDZ05}
      \field{sortinit}{I}
      \field{sortinithash}{a3dcedd53b04d1adfd5ac303ecd5e6fa}
      \field{labelnamesource}{author}
      \field{labeltitlesource}{title}
      \field{journaltitle}{Logic Journal of IGPL}
      \field{number}{6}
      \field{title}{Properties of intuitionistic provability and preservativity logics}
      \field{volume}{13}
      \field{year}{2005}
      \field{pages}{615\bibrangedash 636}
      \range{pages}{22}
    \endentry
    \entry{Iemhoff01:phd}{thesis}{}
      \name{author}{1}{}{%
        {{hash=14eccfc6a4675763cd48755206aa7593}{%
           family={Iemhoff},
           familyi={I\bibinitperiod},
           given={Rosalie},
           giveni={R\bibinitperiod}}}%
      }
      \list{institution}{1}{%
        {University of Amsterdam}%
      }
      \strng{namehash}{14eccfc6a4675763cd48755206aa7593}
      \strng{fullhash}{14eccfc6a4675763cd48755206aa7593}
      \strng{authornamehash}{14eccfc6a4675763cd48755206aa7593}
      \strng{authorfullhash}{14eccfc6a4675763cd48755206aa7593}
      \field{labelalpha}{Iem01}
      \field{sortinit}{I}
      \field{sortinithash}{a3dcedd53b04d1adfd5ac303ecd5e6fa}
      \field{extraalpha}{2}
      \field{labelnamesource}{author}
      \field{labeltitlesource}{title}
      \field{title}{Provability Logic and Admissible Rules}
      \field{type}{phdthesis}
      \field{year}{2001}
    \endentry
    \entry{Japaridze1988}{incollection}{}
      \name{author}{1}{}{%
        {{hash=170b7393797b79442e45772612c70fe9}{%
           family={Japaridze},
           familyi={J\bibinitperiod},
           given={G.\bibnamedelimi K.},
           giveni={G\bibinitperiod\bibinitdelim K\bibinitperiod}}}%
      }
      \name{editor}{2}{}{%
        {{hash=edde96e6a69e131dd31dedb8ff273577}{%
           family={Smirnov},
           familyi={S\bibinitperiod},
           given={V.\bibnamedelimi A.},
           giveni={V\bibinitperiod\bibinitdelim A\bibinitperiod}}}%
        {{hash=eef280ff290f27c1db4d71dfb323fe87}{%
           family={Bezhanishvili},
           familyi={B\bibinitperiod},
           given={M.\bibnamedelimi N.},
           giveni={M\bibinitperiod\bibinitdelim N\bibinitperiod}}}%
      }
      \list{location}{1}{%
        {Tbilisi}%
      }
      \list{publisher}{1}{%
        {Metsniereba}%
      }
      \strng{namehash}{170b7393797b79442e45772612c70fe9}
      \strng{fullhash}{170b7393797b79442e45772612c70fe9}
      \strng{authornamehash}{170b7393797b79442e45772612c70fe9}
      \strng{authorfullhash}{170b7393797b79442e45772612c70fe9}
      \strng{editornamehash}{7549ffb49d42adf744f8615010d75bd4}
      \strng{editorfullhash}{7549ffb49d42adf744f8615010d75bd4}
      \field{labelalpha}{Jap88}
      \field{sortinit}{J}
      \field{sortinithash}{c86bd6cced82a15683b396c2169909ef}
      \field{labelnamesource}{author}
      \field{labeltitlesource}{title}
      \field{booktitle}{Intensional Logics and the Logical Structure of Theories: Proceedings of the Fourth Soviet-Finnish Symposium on Logic, Telavi, May 1985}
      \field{title}{The polymodal logic of provability}
      \field{year}{1988}
      \field{pages}{16\bibrangedash 48}
      \range{pages}{33}
    \endentry
    \entry{japa:logi98}{incollection}{}
      \name{author}{2}{}{%
        {{hash=545b29311514fec87e359471730dbf37}{%
           family={Japaridze},
           familyi={J\bibinitperiod},
           given={G.},
           giveni={G\bibinitperiod}}}%
        {{hash=e9ecdbe168f7c409d1dcd29f24b6fcd9}{%
           family={{de}\bibnamedelima Jongh},
           familyi={d\bibinitperiod\bibinitdelim J\bibinitperiod},
           given={D.},
           giveni={D\bibinitperiod}}}%
      }
      \name{editor}{1}{}{%
        {{hash=66ecca62bae230c1116c46f9f65929cc}{%
           family={Buss},
           familyi={B\bibinitperiod},
           given={S.},
           giveni={S\bibinitperiod}}}%
      }
      \list{location}{1}{%
        {Amsterdam}%
      }
      \list{publisher}{1}{%
        {North-Holland Publishing Co.}%
      }
      \strng{namehash}{b78a9f77ce6e04b3e076c3c712b99a56}
      \strng{fullhash}{b78a9f77ce6e04b3e076c3c712b99a56}
      \strng{authornamehash}{b78a9f77ce6e04b3e076c3c712b99a56}
      \strng{authorfullhash}{b78a9f77ce6e04b3e076c3c712b99a56}
      \strng{editornamehash}{66ecca62bae230c1116c46f9f65929cc}
      \strng{editorfullhash}{66ecca62bae230c1116c46f9f65929cc}
      \field{labelalpha}{Jd98}
      \field{sortinit}{J}
      \field{sortinithash}{c86bd6cced82a15683b396c2169909ef}
      \field{labelnamesource}{author}
      \field{labeltitlesource}{title}
      \field{booktitle}{Handbook of proof theory}
      \field{title}{The logic of provability}
      \field{year}{1998}
      \field{pages}{475\bibrangedash 546}
      \range{pages}{72}
    \endentry
    \entry{viss:inte11}{article}{}
      \name{author}{3}{}{%
        {{hash=af475dcf6f4bea11895ee84cab19a816}{%
           family={Jongh},
           familyi={J\bibinitperiod},
           given={D.},
           giveni={D\bibinitperiod},
           prefix={de},
           prefixi={d\bibinitperiod}}}%
        {{hash=2a797a2abeb800251638a6a46dffc88f}{%
           family={Verbrugge},
           familyi={V\bibinitperiod},
           given={R.},
           giveni={R\bibinitperiod}}}%
        {{hash=7c59a81f2d67e5dbdf77cc8061473f2d}{%
           family={Visser},
           familyi={V\bibinitperiod},
           given={A.},
           giveni={A\bibinitperiod}}}%
      }
      \list{language}{1}{%
        {English}%
      }
      \list{publisher}{1}{%
        {Springer-Verlag}%
      }
      \strng{namehash}{7efff0461107f0b9ba719e866bdcd7ee}
      \strng{fullhash}{7efff0461107f0b9ba719e866bdcd7ee}
      \strng{authornamehash}{7efff0461107f0b9ba719e866bdcd7ee}
      \strng{authorfullhash}{7efff0461107f0b9ba719e866bdcd7ee}
      \field{labelalpha}{JVV11}
      \field{sortinit}{J}
      \field{sortinithash}{c86bd6cced82a15683b396c2169909ef}
      \field{labelnamesource}{author}
      \field{labeltitlesource}{title}
      \field{issn}{0933-5846}
      \field{journaltitle}{Archive for Mathematical Logic}
      \field{number}{1--2}
      \field{title}{Intermediate Logics and the de {J}ongh property}
      \field{volume}{50}
      \field{year}{2011}
      \field{pages}{197\bibrangedash 213}
      \range{pages}{17}
      \verb{doi}
      \verb 10.1007/s00153-010-0209-4
      \endverb
      \verb{url}
      \verb http://dx.doi.org/10.1007/s00153-010-0209-4
      \endverb
      \keyw{Intuitionistic logic; Heyting's arithmetic; 03F25; 03F30; 03-02; 03B20; 03F50; 03F40}
    \endentry
    \entry{dejo:prov90}{inproceedings}{}
      \name{author}{2}{}{%
        {{hash=de5640ed7db048129ae5153bae66c64f}{%
           family={Jongh},
           familyi={J\bibinitperiod},
           given={D.H.J.},
           giveni={D\bibinitperiod},
           prefix={de},
           prefixi={d\bibinitperiod}}}%
        {{hash=40ad4fb4a0fd850f311cc3494f62e93e}{%
           family={Veltman},
           familyi={V\bibinitperiod},
           given={F.},
           giveni={F\bibinitperiod}}}%
      }
      \strng{namehash}{d5c13559d42878ef25e26171cf8c99c0}
      \strng{fullhash}{d5c13559d42878ef25e26171cf8c99c0}
      \strng{authornamehash}{d5c13559d42878ef25e26171cf8c99c0}
      \strng{authorfullhash}{d5c13559d42878ef25e26171cf8c99c0}
      \field{labelalpha}{JV90}
      \field{sortinit}{J}
      \field{sortinithash}{c86bd6cced82a15683b396c2169909ef}
      \field{labelnamesource}{author}
      \field{labeltitlesource}{title}
      \field{booktitle}{\cite{heyt:malo90}}
      \field{title}{Provability logics for relative interpretability}
      \field{year}{1990}
      \field{pages}{31\bibrangedash 42}
      \range{pages}{12}
    \endentry
    \entry{dejo:embe96}{inproceedings}{}
      \name{author}{2}{}{%
        {{hash=de5640ed7db048129ae5153bae66c64f}{%
           family={Jongh},
           familyi={J\bibinitperiod},
           given={D.H.J.},
           giveni={D\bibinitperiod},
           prefix={de},
           prefixi={d\bibinitperiod}}}%
        {{hash=7c59a81f2d67e5dbdf77cc8061473f2d}{%
           family={Visser},
           familyi={V\bibinitperiod},
           given={A.},
           giveni={A\bibinitperiod}}}%
      }
      \strng{namehash}{3c162806c0887e876a6095030630ac07}
      \strng{fullhash}{3c162806c0887e876a6095030630ac07}
      \strng{authornamehash}{3c162806c0887e876a6095030630ac07}
      \strng{authorfullhash}{3c162806c0887e876a6095030630ac07}
      \field{labelalpha}{JV96}
      \field{sortinit}{J}
      \field{sortinithash}{c86bd6cced82a15683b396c2169909ef}
      \field{labelnamesource}{author}
      \field{labeltitlesource}{title}
      \field{booktitle}{\cite{hodg:foun96}}
      \field{title}{Embeddings of {H}eyting algebras}
      \field{year}{1996}
      \field{pages}{187\bibrangedash 213}
      \range{pages}{27}
    \endentry
    \entry{JungSSSTBD15:popl}{inproceedings}{}
      \name{author}{7}{}{%
        {{hash=b941e8725b8fd9992bdee2980fefad07}{%
           family={Jung},
           familyi={J\bibinitperiod},
           given={Ralf},
           giveni={R\bibinitperiod}}}%
        {{hash=42f1bd9f4fa031c343b0daa0c61d5afe}{%
           family={Swasey},
           familyi={S\bibinitperiod},
           given={David},
           giveni={D\bibinitperiod}}}%
        {{hash=8e6c8d55ba05cecc763597cfdfa4c0c4}{%
           family={Sieczkowski},
           familyi={S\bibinitperiod},
           given={Filip},
           giveni={F\bibinitperiod}}}%
        {{hash=09bbde70c2e6081108b9a23332cb48e4}{%
           family={Svendsen},
           familyi={S\bibinitperiod},
           given={Kasper},
           giveni={K\bibinitperiod}}}%
        {{hash=621c444d9755cdfcc6dbe35d5b37df90}{%
           family={Turon},
           familyi={T\bibinitperiod},
           given={Aaron},
           giveni={A\bibinitperiod}}}%
        {{hash=f4714354745ea31cf6c543434102f444}{%
           family={Birkedal},
           familyi={B\bibinitperiod},
           given={Lars},
           giveni={L\bibinitperiod}}}%
        {{hash=3c61a075f82133ed00b1cf07a5b5c194}{%
           family={Dreyer},
           familyi={D\bibinitperiod},
           given={Derek},
           giveni={D\bibinitperiod}}}%
      }
      \name{editor}{2}{}{%
        {{hash=bb241265fa873a5164ccf321b5f5cc42}{%
           family={Rajamani},
           familyi={R\bibinitperiod},
           given={Sriram\bibnamedelima K.},
           giveni={S\bibinitperiod\bibinitdelim K\bibinitperiod}}}%
        {{hash=f6ebee077ea6f2c8f67f01d541cad379}{%
           family={Walker},
           familyi={W\bibinitperiod},
           given={David},
           giveni={D\bibinitperiod}}}%
      }
      \list{publisher}{1}{%
        {ACM}%
      }
      \strng{namehash}{b9517442050429a39f511630fbc3911f}
      \strng{fullhash}{43ca35cdef9a8a21258ff536808675b7}
      \strng{authornamehash}{b9517442050429a39f511630fbc3911f}
      \strng{authorfullhash}{43ca35cdef9a8a21258ff536808675b7}
      \strng{editornamehash}{dae6cf08a71a9c1056578e758d69641b}
      \strng{editorfullhash}{dae6cf08a71a9c1056578e758d69641b}
      \field{labelalpha}{Jun+15}
      \field{sortinit}{J}
      \field{sortinithash}{c86bd6cced82a15683b396c2169909ef}
      \field{labelnamesource}{author}
      \field{labeltitlesource}{title}
      \field{booktitle}{Proceedings of {POPL}}
      \field{isbn}{978-1-4503-3300-9}
      \field{title}{Iris: Monoids and Invariants as an Orthogonal Basis for Concurrent Reasoning}
      \field{year}{2015}
      \field{pages}{637\bibrangedash 650}
      \range{pages}{14}
      \verb{doi}
      \verb 10.1145/2676726.2676980
      \endverb
      \verb{url}
      \verb http://doi.acm.org/10.1145/2676726.2676980
      \endverb
    \endentry
    \entry{Kobayashi97:tcs}{article}{}
      \name{author}{1}{}{%
        {{hash=0f9f7256bb233dc40e82b740c99d1aa7}{%
           family={Kobayashi},
           familyi={K\bibinitperiod},
           given={Satoshi},
           giveni={S\bibinitperiod}}}%
      }
      \strng{namehash}{0f9f7256bb233dc40e82b740c99d1aa7}
      \strng{fullhash}{0f9f7256bb233dc40e82b740c99d1aa7}
      \strng{authornamehash}{0f9f7256bb233dc40e82b740c99d1aa7}
      \strng{authorfullhash}{0f9f7256bb233dc40e82b740c99d1aa7}
      \field{labelalpha}{Kob97}
      \field{sortinit}{K}
      \field{sortinithash}{4c244ceae61406cdc0cc2ce1cb1ff703}
      \field{labelnamesource}{author}
      \field{labeltitlesource}{title}
      \field{abstract}{In 1989, Eugenio Moggi proposed a categorical framework for program semantics based on the notion of a strong monad. He showed that various kinds of computation can be modeled in his framework. On the other hand, strong monads are not suited for the categorical semantics of traditional modal logics. According to these observations, Moggi thought that the CurryHoward correspondence would not hold between programs and constructive proofs in modal logics. However, contrary to his view, we can show that proofs in a certain kind of modal logics are actually considered as programs. In this paper, first we shall introduce the notion of an l-strong monad which is a generalization of strong monads. Using this new notion, we can generalize Moggi's semantics-preserving soundness and completeness with respect to his equational logic. Next we shall show that l-strong monads give a sound and complete semantics of a constructive version of S4 modal logic. Finally, we present a method to extract a monad-based imperative functional program from a proof in the modal logic. Interestingly, this method can also be understood in terms of l-strong monads.}
      \field{issn}{0304-3975}
      \field{journaltitle}{Theoretical Computer Science}
      \field{number}{1}
      \field{title}{Monad as modality}
      \field{volume}{175}
      \field{year}{1997}
      \field{pages}{29\bibrangedash 74}
      \range{pages}{46}
      \verb{doi}
      \verb DOI: 10.1016/S0304-3975(96)00169-7
      \endverb
      \verb{url}
      \verb http://www.sciencedirect.com/science/article/pii/S0304397596001697
      \endverb
    \endentry
    \entry{Kohler81:ams}{article}{}
      \name{author}{1}{}{%
        {{hash=90faa5b06fb99297b7785df81d4e5c5f}{%
           family={K\"{o}hler},
           familyi={K\bibinitperiod},
           given={Peter},
           giveni={P\bibinitperiod}}}%
      }
      \list{publisher}{1}{%
        {American Mathematical Society}%
      }
      \strng{namehash}{90faa5b06fb99297b7785df81d4e5c5f}
      \strng{fullhash}{90faa5b06fb99297b7785df81d4e5c5f}
      \strng{authornamehash}{90faa5b06fb99297b7785df81d4e5c5f}
      \strng{authorfullhash}{90faa5b06fb99297b7785df81d4e5c5f}
      \field{labelalpha}{K\"{o}h81}
      \field{sortinit}{K}
      \field{sortinithash}{4c244ceae61406cdc0cc2ce1cb1ff703}
      \field{labelnamesource}{author}
      \field{labeltitlesource}{title}
      \field{abstract}{Let $P$ be the category whose objects are posets and whose morphisms are partial mappings $\alpha: P \rightarrow Q$ satisfying (i) $\forall p, q \in \operatorname{dom} \alpha lbrack p < q \Rightarrow \alpha(p) < \alpha (q) \rbrack$ and (ii) \forall p \in \operatorname{dom} \alpha \forall q \in Q \lbrack q < \alpha(p) \Rightarrow \exists r \in \operatorname{dom} \alpha \lbrack r < p {\&} \alpha(r) = q\rbrack\rbrack$. The full subcategory $P_f$ of $P$ consisting of all finite posets is shown to be dually equivalent to the category of finite Brouwerian semilattices and homomorphisms. Under this duality a finite Brouwerian semilattice $\underline A$ corresponds with M (\underline A)$, the poset of all meet-irreducible elements of $\underline A$. The produce (in $P_f$) of $n$ copies $(n \in \Bbb{N})$ of a one-element poset is constructed; in view of the duality this product is isomorphic to the poset of meet-irreducible elements of the free Brouwerian semilattice on $n$ generators. If $V$ is a variety of Brouwerian semilattices and if $\underline A$ is a Brouwerian semilattice, then $\underline A$ is $V$-critical if all proper subalgebras of $\underline A$ belong to $V$ but not $A$. It is shown that a variety $V$ of Brouwerian semilattices has a finite equational base if and only if there are up to isomorphism only finitely many $V$-critical Brouwerian semilattices. This is used to show that a variety generated by a finite Brouwerian semilattice as well as the join of two finitely based varieties is finitely based. A new example of a variety without a finite equational base is exhibited.}
      \field{issn}{00029947}
      \field{journaltitle}{Trans. Amer. Math. Soc.}
      \field{number}{1}
      \field{title}{Brouwerian Semilattices}
      \field{volume}{268}
      \field{year}{1981}
      \field{pages}{103\bibrangedash 126}
      \range{pages}{24}
      \verb{url}
      \verb http://www.jstor.org/stable/1998339
      \endverb
    \endentry
    \entry{KrishnaswamiB11:icfp}{inproceedings}{}
      \name{author}{2}{}{%
        {{hash=c9f46b4f6114c3dd9a8372cd0f891d6e}{%
           family={Krishnaswami},
           familyi={K\bibinitperiod},
           given={Neelakantan\bibnamedelima R.},
           giveni={N\bibinitperiod\bibinitdelim R\bibinitperiod}}}%
        {{hash=914c421c88aaf72767dd1211d17cac95}{%
           family={Benton},
           familyi={B\bibinitperiod},
           given={Nick},
           giveni={N\bibinitperiod}}}%
      }
      \name{editor}{3}{}{%
        {{hash=00ee4874ccb0f2a9bf23e1a5b45ee481}{%
           family={Chakravarty},
           familyi={C\bibinitperiod},
           given={Manuel\bibnamedelimb M.\bibnamedelimi T.},
           giveni={M\bibinitperiod\bibinitdelim M\bibinitperiod\bibinitdelim T\bibinitperiod}}}%
        {{hash=aa123f617e3811bba0121021e0329f91}{%
           family={Hu},
           familyi={H\bibinitperiod},
           given={Zhenjiang},
           giveni={Z\bibinitperiod}}}%
        {{hash=79ed408858d4950e144853df16d9c5db}{%
           family={Danvy},
           familyi={D\bibinitperiod},
           given={Olivier},
           giveni={O\bibinitperiod}}}%
      }
      \list{organization}{1}{%
        {ACM SIGPLAN}%
      }
      \list{publisher}{1}{%
        {ACM}%
      }
      \strng{namehash}{46b4f02c207f72e40de20357646dc5d7}
      \strng{fullhash}{46b4f02c207f72e40de20357646dc5d7}
      \strng{authornamehash}{46b4f02c207f72e40de20357646dc5d7}
      \strng{authorfullhash}{46b4f02c207f72e40de20357646dc5d7}
      \strng{editornamehash}{0877ee71f553d97add28646afd9d6d5c}
      \strng{editorfullhash}{0877ee71f553d97add28646afd9d6d5c}
      \field{labelalpha}{KB11}
      \field{sortinit}{K}
      \field{sortinithash}{4c244ceae61406cdc0cc2ce1cb1ff703}
      \field{extraalpha}{1}
      \field{labelnamesource}{author}
      \field{labeltitlesource}{title}
      \field{booktitle}{Proceedings of {ICFP}}
      \field{isbn}{978-1-4503-0865-6}
      \field{title}{A semantic model for graphical user interfaces}
      \field{year}{2011}
      \field{pages}{45\bibrangedash 57}
      \range{pages}{13}
    \endentry
    \entry{KrishnaswamiB11:lics}{inproceedings}{}
      \name{author}{2}{}{%
        {{hash=c9f46b4f6114c3dd9a8372cd0f891d6e}{%
           family={Krishnaswami},
           familyi={K\bibinitperiod},
           given={Neelakantan\bibnamedelima R.},
           giveni={N\bibinitperiod\bibinitdelim R\bibinitperiod}}}%
        {{hash=914c421c88aaf72767dd1211d17cac95}{%
           family={Benton},
           familyi={B\bibinitperiod},
           given={Nick},
           giveni={N\bibinitperiod}}}%
      }
      \list{organization}{1}{%
        {IEEE}%
      }
      \strng{namehash}{46b4f02c207f72e40de20357646dc5d7}
      \strng{fullhash}{46b4f02c207f72e40de20357646dc5d7}
      \strng{authornamehash}{46b4f02c207f72e40de20357646dc5d7}
      \strng{authorfullhash}{46b4f02c207f72e40de20357646dc5d7}
      \field{labelalpha}{KB11}
      \field{sortinit}{K}
      \field{sortinithash}{4c244ceae61406cdc0cc2ce1cb1ff703}
      \field{extraalpha}{2}
      \field{labelnamesource}{author}
      \field{labeltitlesource}{title}
      \field{booktitle}{Proceedings of {LiCS}}
      \field{isbn}{978-0-7695-4412-0}
      \field{title}{Ultrametric Semantics of Reactive Programs}
      \field{year}{2011}
      \field{pages}{257\bibrangedash 266}
      \range{pages}{10}
    \endentry
    \entry{KrishnaswamiBH12:popl}{inproceedings}{}
      \name{author}{3}{}{%
        {{hash=c9f46b4f6114c3dd9a8372cd0f891d6e}{%
           family={Krishnaswami},
           familyi={K\bibinitperiod},
           given={Neelakantan\bibnamedelima R.},
           giveni={N\bibinitperiod\bibinitdelim R\bibinitperiod}}}%
        {{hash=914c421c88aaf72767dd1211d17cac95}{%
           family={Benton},
           familyi={B\bibinitperiod},
           given={Nick},
           giveni={N\bibinitperiod}}}%
        {{hash=9bb7505af6f56de3344e7ca74aa3fb3c}{%
           family={Hoffmann},
           familyi={H\bibinitperiod},
           given={Jan},
           giveni={J\bibinitperiod}}}%
      }
      \name{editor}{2}{}{%
        {{hash=595e2f55b8e3a82633f9da7acf83e1dd}{%
           family={Field},
           familyi={F\bibinitperiod},
           given={John},
           giveni={J\bibinitperiod}}}%
        {{hash=69cb741734a7429ce7f1dd4dd7eacade}{%
           family={Hicks},
           familyi={H\bibinitperiod},
           given={Michael},
           giveni={M\bibinitperiod}}}%
      }
      \list{publisher}{1}{%
        {ACM}%
      }
      \strng{namehash}{430ff2bdbe2d3a877536568060a29991}
      \strng{fullhash}{430ff2bdbe2d3a877536568060a29991}
      \strng{authornamehash}{430ff2bdbe2d3a877536568060a29991}
      \strng{authorfullhash}{430ff2bdbe2d3a877536568060a29991}
      \strng{editornamehash}{84e110cbc53537960cf9e4f22868ee71}
      \strng{editorfullhash}{84e110cbc53537960cf9e4f22868ee71}
      \field{labelalpha}{KBH12}
      \field{sortinit}{K}
      \field{sortinithash}{4c244ceae61406cdc0cc2ce1cb1ff703}
      \field{labelnamesource}{author}
      \field{labeltitlesource}{title}
      \field{booktitle}{Proceedings of {POPL}}
      \field{isbn}{978-1-4503-1083-3}
      \field{title}{Higher-order functional reactive programming in bounded space}
      \field{year}{2012}
      \field{pages}{45\bibrangedash 58}
      \range{pages}{14}
      \verb{doi}
      \verb 10.1145/2103656.2103665
      \endverb
      \verb{url}
      \verb http://doi.acm.org/10.1145/2103656.2103665
      \endverb
    \endentry
    \entry{KurzP13:lmcs}{article}{}
      \name{author}{2}{}{%
        {{hash=7be8d94e3431ca9d859e728041dec160}{%
           family={Kurz},
           familyi={K\bibinitperiod},
           given={Alexander},
           giveni={A\bibinitperiod}}}%
        {{hash=aeacb76593b669fffa2b9890e8305f1e}{%
           family={Palmigiano},
           familyi={P\bibinitperiod},
           given={Alessandra},
           giveni={A\bibinitperiod}}}%
      }
      \strng{namehash}{d4ee408837ad085ea48e4b1ef700bd16}
      \strng{fullhash}{d4ee408837ad085ea48e4b1ef700bd16}
      \strng{authornamehash}{d4ee408837ad085ea48e4b1ef700bd16}
      \strng{authorfullhash}{d4ee408837ad085ea48e4b1ef700bd16}
      \field{labelalpha}{KP13}
      \field{sortinit}{K}
      \field{sortinithash}{4c244ceae61406cdc0cc2ce1cb1ff703}
      \field{labelnamesource}{author}
      \field{labeltitlesource}{title}
      \field{journaltitle}{Logical Methods in Computer Science}
      \field{number}{4}
      \field{title}{Epistemic Updates on Algebras}
      \field{volume}{9}
      \field{year}{2013}
      \verb{doi}
      \verb 10.2168/LMCS-9(4:17)2013
      \endverb
      \verb{url}
      \verb https://doi.org/10.2168/LMCS-9(4:17)2013
      \endverb
    \endentry
    \entry{Lewis13:jppsm}{article}{}
      \name{author}{1}{}{%
        {{hash=df8cbeef654b467487c1296e7f00860f}{%
           family={Lewis},
           familyi={L\bibinitperiod},
           given={C.\bibnamedelimi I.},
           giveni={C\bibinitperiod\bibinitdelim I\bibinitperiod}}}%
      }
      \list{publisher}{1}{%
        {Journal of Philosophy, Inc.}%
      }
      \strng{namehash}{df8cbeef654b467487c1296e7f00860f}
      \strng{fullhash}{df8cbeef654b467487c1296e7f00860f}
      \strng{authornamehash}{df8cbeef654b467487c1296e7f00860f}
      \strng{authorfullhash}{df8cbeef654b467487c1296e7f00860f}
      \field{labelalpha}{Lew13}
      \field{sortinit}{L}
      \field{sortinithash}{7bba64db83423e3c29ad597f3b682cf3}
      \field{labelnamesource}{author}
      \field{labeltitlesource}{title}
      \field{issn}{01609335}
      \field{journaltitle}{The Journal of Philosophy, Psychology and Scientific Methods}
      \field{number}{16}
      \field{title}{A New Algebra of Implications and Some Consequences}
      \field{volume}{10}
      \field{year}{1913}
      \field{pages}{428\bibrangedash 438}
      \range{pages}{11}
      \verb{url}
      \verb http://www.jstor.org/stable/2012900
      \endverb
    \endentry
    \entry{Lewis15:jppsm}{article}{}
      \name{author}{1}{}{%
        {{hash=df8cbeef654b467487c1296e7f00860f}{%
           family={Lewis},
           familyi={L\bibinitperiod},
           given={C.\bibnamedelimi I.},
           giveni={C\bibinitperiod\bibinitdelim I\bibinitperiod}}}%
      }
      \list{publisher}{1}{%
        {Journal of Philosophy, Inc.}%
      }
      \strng{namehash}{df8cbeef654b467487c1296e7f00860f}
      \strng{fullhash}{df8cbeef654b467487c1296e7f00860f}
      \strng{authornamehash}{df8cbeef654b467487c1296e7f00860f}
      \strng{authorfullhash}{df8cbeef654b467487c1296e7f00860f}
      \field{labelalpha}{Lew15}
      \field{sortinit}{L}
      \field{sortinithash}{7bba64db83423e3c29ad597f3b682cf3}
      \field{labelnamesource}{author}
      \field{labeltitlesource}{title}
      \field{issn}{01609335}
      \field{journaltitle}{The Journal of Philosophy, Psychology and Scientific Methods}
      \field{number}{19}
      \field{title}{A Too Brief Set of Postulates for the Algebra of Logic}
      \field{volume}{12}
      \field{year}{1915}
      \field{pages}{523\bibrangedash 525}
      \range{pages}{3}
      \verb{url}
      \verb http://www.jstor.org/stable/2012996
      \endverb
    \endentry
    \entry{Lewis32:monist}{article}{}
      \name{author}{1}{}{%
        {{hash=df8cbeef654b467487c1296e7f00860f}{%
           family={Lewis},
           familyi={L\bibinitperiod},
           given={C.\bibnamedelimi I.},
           giveni={C\bibinitperiod\bibinitdelim I\bibinitperiod}}}%
      }
      \list{publisher}{1}{%
        {The Oxford University Press}%
      }
      \strng{namehash}{df8cbeef654b467487c1296e7f00860f}
      \strng{fullhash}{df8cbeef654b467487c1296e7f00860f}
      \strng{authornamehash}{df8cbeef654b467487c1296e7f00860f}
      \strng{authorfullhash}{df8cbeef654b467487c1296e7f00860f}
      \field{labelalpha}{Lew32}
      \field{sortinit}{L}
      \field{sortinithash}{7bba64db83423e3c29ad597f3b682cf3}
      \field{labelnamesource}{author}
      \field{labeltitlesource}{title}
      \field{issn}{0026-9662}
      \field{journaltitle}{The Monist}
      \field{number}{4}
      \field{title}{Alternative Systems of Logic}
      \field{volume}{42}
      \field{year}{1932}
      \field{pages}{481\bibrangedash 507}
      \range{pages}{27}
      \verb{doi}
      \verb 10.5840/monist19324241
      \endverb
      \verb{eprint}
      \verb http://monist.oxfordjournals.org/content/42/4/481.full.pdf
      \endverb
      \verb{url}
      \verb http://monist.oxfordjournals.org/content/42/4/481
      \endverb
    \endentry
    \entry{Lewis12}{article}{}
      \name{author}{1}{}{%
        {{hash=df8cbeef654b467487c1296e7f00860f}{%
           family={Lewis},
           familyi={L\bibinitperiod},
           given={C.\bibnamedelimi I.},
           giveni={C\bibinitperiod\bibinitdelim I\bibinitperiod}}}%
      }
      \list{publisher}{1}{%
        {[Oxford University Press, Mind Association]}%
      }
      \strng{namehash}{df8cbeef654b467487c1296e7f00860f}
      \strng{fullhash}{df8cbeef654b467487c1296e7f00860f}
      \strng{authornamehash}{df8cbeef654b467487c1296e7f00860f}
      \strng{authorfullhash}{df8cbeef654b467487c1296e7f00860f}
      \field{labelalpha}{Lew12}
      \field{sortinit}{L}
      \field{sortinithash}{7bba64db83423e3c29ad597f3b682cf3}
      \field{labelnamesource}{author}
      \field{labeltitlesource}{title}
      \field{issn}{00264423, 14602113}
      \field{journaltitle}{Mind}
      \field{number}{84}
      \field{title}{Implication and the Algebra of Logic}
      \field{volume}{21}
      \field{year}{1912}
      \field{pages}{522\bibrangedash 531}
      \range{pages}{10}
      \verb{url}
      \verb http://www.jstor.org/stable/2249157
      \endverb
    \endentry
    \entry{Lewis20:jppsm}{article}{}
      \name{author}{1}{}{%
        {{hash=df8cbeef654b467487c1296e7f00860f}{%
           family={Lewis},
           familyi={L\bibinitperiod},
           given={C.\bibnamedelimi I.},
           giveni={C\bibinitperiod\bibinitdelim I\bibinitperiod}}}%
      }
      \list{publisher}{1}{%
        {Journal of Philosophy, Inc.}%
      }
      \strng{namehash}{df8cbeef654b467487c1296e7f00860f}
      \strng{fullhash}{df8cbeef654b467487c1296e7f00860f}
      \strng{authornamehash}{df8cbeef654b467487c1296e7f00860f}
      \strng{authorfullhash}{df8cbeef654b467487c1296e7f00860f}
      \field{labelalpha}{Lew20}
      \field{sortinit}{L}
      \field{sortinithash}{7bba64db83423e3c29ad597f3b682cf3}
      \field{labelnamesource}{author}
      \field{labeltitlesource}{title}
      \field{issn}{01609335}
      \field{journaltitle}{The Journal of Philosophy, Psychology and Scientific Methods}
      \field{number}{11}
      \field{title}{Strict Implication--An Emendation}
      \field{volume}{17}
      \field{year}{1920}
      \field{pages}{300\bibrangedash 302}
      \range{pages}{3}
      \verb{url}
      \verb http://www.jstor.org/stable/2940598
      \endverb
    \endentry
    \entry{Lewis14:jppsm}{article}{}
      \name{author}{1}{}{%
        {{hash=df8cbeef654b467487c1296e7f00860f}{%
           family={Lewis},
           familyi={L\bibinitperiod},
           given={C.\bibnamedelimi I.},
           giveni={C\bibinitperiod\bibinitdelim I\bibinitperiod}}}%
      }
      \list{publisher}{1}{%
        {Journal of Philosophy, Inc.}%
      }
      \strng{namehash}{df8cbeef654b467487c1296e7f00860f}
      \strng{fullhash}{df8cbeef654b467487c1296e7f00860f}
      \strng{authornamehash}{df8cbeef654b467487c1296e7f00860f}
      \strng{authorfullhash}{df8cbeef654b467487c1296e7f00860f}
      \field{labelalpha}{Lew14}
      \field{sortinit}{L}
      \field{sortinithash}{7bba64db83423e3c29ad597f3b682cf3}
      \field{labelnamesource}{author}
      \field{labeltitlesource}{title}
      \field{issn}{01609335}
      \field{journaltitle}{The Journal of Philosophy, Psychology and Scientific Methods}
      \field{number}{22}
      \field{title}{The Matrix Algebra for Implications}
      \field{volume}{11}
      \field{year}{1914}
      \field{pages}{589\bibrangedash 600}
      \range{pages}{12}
      \verb{url}
      \verb http://www.jstor.org/stable/2012652
      \endverb
    \endentry
    \entry{Lewis35}{article}{}
      \name{author}{2}{}{%
        {{hash=df8cbeef654b467487c1296e7f00860f}{%
           family={Lewis},
           familyi={L\bibinitperiod},
           given={C.\bibnamedelimi I.},
           giveni={C\bibinitperiod\bibinitdelim I\bibinitperiod}}}%
        {{hash=5aa8ce662e08a2c5a24967393c526b85}{%
           family={Langford},
           familyi={L\bibinitperiod},
           given={C.\bibnamedelimi H.},
           giveni={C\bibinitperiod\bibinitdelim H\bibinitperiod}}}%
      }
      \strng{namehash}{29568583188869f2be173d6aafd1bf50}
      \strng{fullhash}{29568583188869f2be173d6aafd1bf50}
      \strng{authornamehash}{29568583188869f2be173d6aafd1bf50}
      \strng{authorfullhash}{29568583188869f2be173d6aafd1bf50}
      \field{labelalpha}{LL14}
      \field{sortinit}{L}
      \field{sortinithash}{7bba64db83423e3c29ad597f3b682cf3}
      \field{labelnamesource}{author}
      \field{labeltitlesource}{title}
      \field{journaltitle}{History and Philosophy of Logic}
      \field{number}{1}
      \field{title}{A Note on Strict Implication (1935)}
      \field{volume}{35}
      \field{year}{2014}
      \field{pages}{44\bibrangedash 49}
      \range{pages}{6}
    \endentry
    \entry{Lewis18}{book}{}
      \name{author}{1}{}{%
        {{hash=fa6471d4df3547a084cb3997978aabd2}{%
           family={Lewis},
           familyi={L\bibinitperiod},
           given={C.I.},
           giveni={C\bibinitperiod}}}%
      }
      \list{publisher}{1}{%
        {University of California Press}%
      }
      \strng{namehash}{fa6471d4df3547a084cb3997978aabd2}
      \strng{fullhash}{fa6471d4df3547a084cb3997978aabd2}
      \strng{authornamehash}{fa6471d4df3547a084cb3997978aabd2}
      \strng{authorfullhash}{fa6471d4df3547a084cb3997978aabd2}
      \field{labelalpha}{Lew18}
      \field{sortinit}{L}
      \field{sortinithash}{7bba64db83423e3c29ad597f3b682cf3}
      \field{labelnamesource}{author}
      \field{labeltitlesource}{title}
      \field{title}{A Survey of Symbolic Logic}
      \field{year}{1918}
    \endentry
    \entry{lewisincap}{incollection}{}
      \name{author}{1}{}{%
        {{hash=fa6471d4df3547a084cb3997978aabd2}{%
           family={Lewis},
           familyi={L\bibinitperiod},
           given={C.I.},
           giveni={C\bibinitperiod}}}%
      }
      \name{editor}{2}{}{%
        {{hash=883d6e99e73811a7d8a91e3cc19fc559}{%
           family={Adams},
           familyi={A\bibinitperiod},
           given={G.\bibnamedelimi P.},
           giveni={G\bibinitperiod\bibinitdelim P\bibinitperiod}}}%
        {{hash=b2920000b98705ba1d633f880bd12e75}{%
           family={Montague},
           familyi={M\bibinitperiod},
           given={W.\bibnamedelimi P.},
           giveni={W\bibinitperiod\bibinitdelim P\bibinitperiod}}}%
      }
      \list{publisher}{1}{%
        {G. Allen \& Unwin Limited}%
      }
      \strng{namehash}{fa6471d4df3547a084cb3997978aabd2}
      \strng{fullhash}{fa6471d4df3547a084cb3997978aabd2}
      \strng{authornamehash}{fa6471d4df3547a084cb3997978aabd2}
      \strng{authorfullhash}{fa6471d4df3547a084cb3997978aabd2}
      \strng{editornamehash}{e88537a5a53368e9264e323b9612d9e1}
      \strng{editorfullhash}{e88537a5a53368e9264e323b9612d9e1}
      \field{labelalpha}{Lew30}
      \field{sortinit}{L}
      \field{sortinithash}{7bba64db83423e3c29ad597f3b682cf3}
      \field{labelnamesource}{author}
      \field{labeltitlesource}{title}
      \field{booktitle}{Contemporary American Philosophy: Personal Statements}
      \field{series}{Library of philosophy, ed. by J. H. Muirhead}
      \field{title}{Logic and Pragmatism}
      \field{volume}{2}
      \field{year}{1930}
      \verb{url}
      \verb https://books.google.de/books?id=3VYNAAAAIAAJ
      \endverb
    \endentry
    \entry{Lewis32:book}{book}{}
      \name{author}{2}{}{%
        {{hash=fa6471d4df3547a084cb3997978aabd2}{%
           family={Lewis},
           familyi={L\bibinitperiod},
           given={C.I.},
           giveni={C\bibinitperiod}}}%
        {{hash=5690cb04a5adb7fac12fc8eeb7ba2a83}{%
           family={Langford},
           familyi={L\bibinitperiod},
           given={C.H.},
           giveni={C\bibinitperiod}}}%
      }
      \list{publisher}{1}{%
        {Dover}%
      }
      \strng{namehash}{129a720a2a5482cc46fa9bd67a00df73}
      \strng{fullhash}{129a720a2a5482cc46fa9bd67a00df73}
      \strng{authornamehash}{129a720a2a5482cc46fa9bd67a00df73}
      \strng{authorfullhash}{129a720a2a5482cc46fa9bd67a00df73}
      \field{labelalpha}{LL32}
      \field{sortinit}{L}
      \field{sortinithash}{7bba64db83423e3c29ad597f3b682cf3}
      \field{labelnamesource}{author}
      \field{labeltitlesource}{title}
      \field{title}{Symbolic Logic}
      \field{year}{1932}
    \endentry
    \entry{Lindley14:wgp}{inproceedings}{}
      \name{author}{1}{}{%
        {{hash=57b68e7d9ae666f817c2c0ec08bec7c9}{%
           family={Lindley},
           familyi={L\bibinitperiod},
           given={Sam},
           giveni={S\bibinitperiod}}}%
      }
      \name{editor}{2}{}{%
        {{hash=f5d63d029b9164b6dadf7ef60bb4b5da}{%
           family={Magalh\~{a}es},
           familyi={M\bibinitperiod},
           given={Jos\'{e}\bibnamedelima Pedro},
           giveni={J\bibinitperiod\bibinitdelim P\bibinitperiod}}}%
        {{hash=70137d0155741849a45d00480176b2ea}{%
           family={Rompf},
           familyi={R\bibinitperiod},
           given={Tiark},
           giveni={T\bibinitperiod}}}%
      }
      \list{publisher}{1}{%
        {ACM}%
      }
      \strng{namehash}{57b68e7d9ae666f817c2c0ec08bec7c9}
      \strng{fullhash}{57b68e7d9ae666f817c2c0ec08bec7c9}
      \strng{authornamehash}{57b68e7d9ae666f817c2c0ec08bec7c9}
      \strng{authorfullhash}{57b68e7d9ae666f817c2c0ec08bec7c9}
      \strng{editornamehash}{fc68acaf0c891b7fb71e6fd921e33fb8}
      \strng{editorfullhash}{fc68acaf0c891b7fb71e6fd921e33fb8}
      \field{labelalpha}{Lin14}
      \field{sortinit}{L}
      \field{sortinithash}{7bba64db83423e3c29ad597f3b682cf3}
      \field{labelnamesource}{author}
      \field{labeltitlesource}{title}
      \field{booktitle}{Proceedings of {WGP}}
      \field{title}{Algebraic effects and effect handlers for idioms and arrows}
      \field{year}{2014}
      \field{pages}{47\bibrangedash 58}
      \range{pages}{12}
    \endentry
    \entry{LindleyWY08:msfp}{article}{}
      \name{author}{3}{}{%
        {{hash=57b68e7d9ae666f817c2c0ec08bec7c9}{%
           family={Lindley},
           familyi={L\bibinitperiod},
           given={Sam},
           giveni={S\bibinitperiod}}}%
        {{hash=b4d861345944ddeeaf658a11c2f0f6a7}{%
           family={Wadler},
           familyi={W\bibinitperiod},
           given={Philip},
           giveni={P\bibinitperiod}}}%
        {{hash=80e36bf63bd686e5360f55b9ac2a1c36}{%
           family={Yallop},
           familyi={Y\bibinitperiod},
           given={Jeremy},
           giveni={J\bibinitperiod}}}%
      }
      \strng{namehash}{818c7042976fcd51a382697c998e7efe}
      \strng{fullhash}{818c7042976fcd51a382697c998e7efe}
      \strng{authornamehash}{818c7042976fcd51a382697c998e7efe}
      \strng{authorfullhash}{818c7042976fcd51a382697c998e7efe}
      \field{labelalpha}{LWY11}
      \field{sortinit}{L}
      \field{sortinithash}{7bba64db83423e3c29ad597f3b682cf3}
      \field{labelnamesource}{author}
      \field{labeltitlesource}{title}
      \field{issn}{1571-0661}
      \field{journaltitle}{{ENTCS}}
      \field{note}{Proceedings of {MSFP}}
      \field{number}{5}
      \field{title}{Idioms are Oblivious, Arrows are Meticulous, Monads are Promiscuous}
      \field{volume}{229}
      \field{year}{2011}
      \field{pages}{97\bibrangedash 117}
      \range{pages}{21}
      \verb{doi}
      \verb http://dx.doi.org/10.1016/j.entcs.2011.02.018
      \endverb
      \verb{url}
      \verb http://www.sciencedirect.com/science/article/pii/S1571066111000557
      \endverb
      \keyw{monads}
    \endentry
    \entry{lind:prov96}{article}{}
      \name{author}{1}{}{%
        {{hash=b8a76ccca732005ea36885d8c114a300}{%
           family={Lindstr\"{o}m},
           familyi={L\bibinitperiod},
           given={P.},
           giveni={P\bibinitperiod}}}%
      }
      \strng{namehash}{b8a76ccca732005ea36885d8c114a300}
      \strng{fullhash}{b8a76ccca732005ea36885d8c114a300}
      \strng{authornamehash}{b8a76ccca732005ea36885d8c114a300}
      \strng{authorfullhash}{b8a76ccca732005ea36885d8c114a300}
      \field{labelalpha}{Lin96}
      \field{sortinit}{L}
      \field{sortinithash}{7bba64db83423e3c29ad597f3b682cf3}
      \field{labelnamesource}{author}
      \field{labeltitlesource}{title}
      \field{journaltitle}{Theoria}
      \field{number}{1-2}
      \field{title}{Provability logic -- a short introduction}
      \field{volume}{62}
      \field{year}{1996}
      \field{pages}{19\bibrangedash 61}
      \range{pages}{43}
    \endentry
    \entry{Litak05:phd}{thesis}{}
      \name{author}{1}{}{%
        {{hash=8f5431086357f9f5a0eef79eb62f74f4}{%
           family={Litak},
           familyi={L\bibinitperiod},
           given={Tadeusz},
           giveni={T\bibinitperiod}}}%
      }
      \list{institution}{2}{%
        {Japan Advanced Institute of Science}%
        {Technology}%
      }
      \strng{namehash}{8f5431086357f9f5a0eef79eb62f74f4}
      \strng{fullhash}{8f5431086357f9f5a0eef79eb62f74f4}
      \strng{authornamehash}{8f5431086357f9f5a0eef79eb62f74f4}
      \strng{authorfullhash}{8f5431086357f9f5a0eef79eb62f74f4}
      \field{labelalpha}{Lit05}
      \field{sortinit}{L}
      \field{sortinithash}{7bba64db83423e3c29ad597f3b682cf3}
      \field{labelnamesource}{author}
      \field{labeltitlesource}{title}
      \field{title}{An algebraic approach to incompleteness in modal logic}
      \field{type}{phdthesis}
      \field{year}{2005}
    \endentry
    \entry{Litak14:trends}{incollection}{}
      \name{author}{1}{}{%
        {{hash=8f5431086357f9f5a0eef79eb62f74f4}{%
           family={Litak},
           familyi={L\bibinitperiod},
           given={Tadeusz},
           giveni={T\bibinitperiod}}}%
      }
      \name{editor}{1}{}{%
        {{hash=a7c042aaa1a3fa49ceb88ffd5870c794}{%
           family={Bezhanishvili},
           familyi={B\bibinitperiod},
           given={Guram},
           giveni={G\bibinitperiod}}}%
      }
      \list{publisher}{1}{%
        {Springer}%
      }
      \strng{namehash}{8f5431086357f9f5a0eef79eb62f74f4}
      \strng{fullhash}{8f5431086357f9f5a0eef79eb62f74f4}
      \strng{authornamehash}{8f5431086357f9f5a0eef79eb62f74f4}
      \strng{authorfullhash}{8f5431086357f9f5a0eef79eb62f74f4}
      \strng{editornamehash}{a7c042aaa1a3fa49ceb88ffd5870c794}
      \strng{editorfullhash}{a7c042aaa1a3fa49ceb88ffd5870c794}
      \field{labelalpha}{Lit14}
      \field{sortinit}{L}
      \field{sortinithash}{7bba64db83423e3c29ad597f3b682cf3}
      \field{labelnamesource}{author}
      \field{labeltitlesource}{title}
      \field{booktitle}{Leo Esakia on duality in modal and intuitionistic logics}
      \field{series}{Outstanding Contributions to Logic}
      \field{title}{Constructive modalities with provability smack}
      \field{volume}{4}
      \field{year}{2014}
      \field{pages}{179\bibrangedash 208}
      \range{pages}{30}
      \verb{doi}
      \verb 10.1007/978-94-017-8860-1_8
      \endverb
      \verb{url}
      \verb https://arxiv.org/abs/1708.05607
      \endverb
    \endentry
    \entry{Litak07:bsl}{article}{}
      \name{author}{1}{}{%
        {{hash=8f5431086357f9f5a0eef79eb62f74f4}{%
           family={Litak},
           familyi={L\bibinitperiod},
           given={Tadeusz},
           giveni={T\bibinitperiod}}}%
      }
      \strng{namehash}{8f5431086357f9f5a0eef79eb62f74f4}
      \strng{fullhash}{8f5431086357f9f5a0eef79eb62f74f4}
      \strng{authornamehash}{8f5431086357f9f5a0eef79eb62f74f4}
      \strng{authorfullhash}{8f5431086357f9f5a0eef79eb62f74f4}
      \field{labelalpha}{Lit07}
      \field{sortinit}{L}
      \field{sortinithash}{7bba64db83423e3c29ad597f3b682cf3}
      \field{labelnamesource}{author}
      \field{labeltitlesource}{title}
      \field{journaltitle}{Bulletin of the Section of Logic}
      \field{note}{A special issue In Honorem Hiroakira Ono edited by Piotr {\L{}}ukowski}
      \field{number}{3--4}
      \field{title}{The non-reflexive counterpart of {Grz}}
      \field{volume}{36}
      \field{year}{2007}
      \field{pages}{195\bibrangedash 208}
      \range{pages}{14}
      \verb{url}
      \verb http://www.filozof.uni.lodz.pl/bulletin/pdf/36_34_10.pdf
      \endverb
    \endentry
    \entry{LitakV:otw}{unpublished}{}
      \name{author}{2}{}{%
        {{hash=8f5431086357f9f5a0eef79eb62f74f4}{%
           family={Litak},
           familyi={L\bibinitperiod},
           given={Tadeusz},
           giveni={T\bibinitperiod}}}%
        {{hash=b18b77c95e3c304a3c6e1d7ff611dd7f}{%
           family={Visser},
           familyi={V\bibinitperiod},
           given={Albert},
           giveni={A\bibinitperiod}}}%
      }
      \strng{namehash}{e687f7a71bdbe42eeda01980c2dee89e}
      \strng{fullhash}{e687f7a71bdbe42eeda01980c2dee89e}
      \strng{authornamehash}{e687f7a71bdbe42eeda01980c2dee89e}
      \strng{authorfullhash}{e687f7a71bdbe42eeda01980c2dee89e}
      \field{labelalpha}{LV}
      \field{sortinit}{L}
      \field{sortinithash}{7bba64db83423e3c29ad597f3b682cf3}
      \field{labelnamesource}{author}
      \field{labeltitlesource}{title}
      \field{note}{in preparation}
      \field{title}{Lewis arrow fell off the wall: decompositions of constructive strict implication}
    \endentry
    \entry{loeb:solu55}{article}{}
      \name{author}{1}{}{%
        {{hash=c6b2a9a180d747afa75d1db527ef3bc0}{%
           family={L\"{o}b},
           familyi={L\bibinitperiod},
           given={M.H.},
           giveni={M\bibinitperiod}}}%
      }
      \strng{namehash}{c6b2a9a180d747afa75d1db527ef3bc0}
      \strng{fullhash}{c6b2a9a180d747afa75d1db527ef3bc0}
      \strng{authornamehash}{c6b2a9a180d747afa75d1db527ef3bc0}
      \strng{authorfullhash}{c6b2a9a180d747afa75d1db527ef3bc0}
      \field{labelalpha}{L\"{o}b55}
      \field{sortinit}{L}
      \field{sortinithash}{7bba64db83423e3c29ad597f3b682cf3}
      \field{labelnamesource}{author}
      \field{labeltitlesource}{title}
      \field{journaltitle}{{Journal of Symbolic Logic}}
      \field{title}{{Solution of a problem of Leon Henkin}}
      \field{volume}{20}
      \field{year}{1955}
      \field{pages}{115\bibrangedash 118}
      \range{pages}{4}
    \endentry
    \entry{Mares14:note}{article}{}
      \name{author}{1}{}{%
        {{hash=78deb26026f67db43620134d09f50903}{%
           family={Mares},
           familyi={M\bibinitperiod},
           given={Edwin},
           giveni={E\bibinitperiod}}}%
      }
      \strng{namehash}{78deb26026f67db43620134d09f50903}
      \strng{fullhash}{78deb26026f67db43620134d09f50903}
      \strng{authornamehash}{78deb26026f67db43620134d09f50903}
      \strng{authorfullhash}{78deb26026f67db43620134d09f50903}
      \field{labelalpha}{Mar14}
      \field{sortinit}{M}
      \field{sortinithash}{c26a05ef03e4429073ed5c825140fac3}
      \field{labelnamesource}{author}
      \field{labeltitlesource}{title}
      \field{journaltitle}{History and Philosophy of Logic}
      \field{number}{1}
      \field{title}{Editor's Introduction to {C.I. Lewis and C.H. Langford} ``{A Note on Strict Implication}''}
      \field{volume}{35}
      \field{year}{2014}
      \field{pages}{38\bibrangedash 43}
      \range{pages}{6}
    \endentry
    \entry{Mares04}{book}{}
      \name{author}{1}{}{%
        {{hash=d88635e5f2d48c0656ce52b6daa7e5d7}{%
           family={Mares},
           familyi={M\bibinitperiod},
           given={Edwin\bibnamedelima D.},
           giveni={E\bibinitperiod\bibinitdelim D\bibinitperiod}}}%
      }
      \list{publisher}{1}{%
        {Cambridge University Press}%
      }
      \strng{namehash}{d88635e5f2d48c0656ce52b6daa7e5d7}
      \strng{fullhash}{d88635e5f2d48c0656ce52b6daa7e5d7}
      \strng{authornamehash}{d88635e5f2d48c0656ce52b6daa7e5d7}
      \strng{authorfullhash}{d88635e5f2d48c0656ce52b6daa7e5d7}
      \field{labelalpha}{Mar04}
      \field{sortinit}{M}
      \field{sortinithash}{c26a05ef03e4429073ed5c825140fac3}
      \field{labelnamesource}{author}
      \field{labeltitlesource}{title}
      \field{title}{Relevant Logic: A Philosophical Interpretation}
      \field{year}{2004}
    \endentry
    \entry{McbrideP08:jfp}{article}{}
      \name{author}{2}{}{%
        {{hash=2c2c2406f3f7d52d846f02ca883d9b3d}{%
           family={McBride},
           familyi={M\bibinitperiod},
           given={Conor},
           giveni={C\bibinitperiod}}}%
        {{hash=1f3f8610e67d8b56b596848671113178}{%
           family={Paterson},
           familyi={P\bibinitperiod},
           given={Ross},
           giveni={R\bibinitperiod}}}%
      }
      \strng{namehash}{d11e3575d60bccaa56e3cba587aaae9f}
      \strng{fullhash}{d11e3575d60bccaa56e3cba587aaae9f}
      \strng{authornamehash}{d11e3575d60bccaa56e3cba587aaae9f}
      \strng{authorfullhash}{d11e3575d60bccaa56e3cba587aaae9f}
      \field{labelalpha}{MP08}
      \field{sortinit}{M}
      \field{sortinithash}{c26a05ef03e4429073ed5c825140fac3}
      \field{labelnamesource}{author}
      \field{labeltitlesource}{title}
      \field{journaltitle}{J. Funct. Programming}
      \field{number}{1}
      \field{title}{Applicative programming with effects}
      \field{volume}{18}
      \field{year}{2008}
      \field{pages}{1\bibrangedash 13}
      \range{pages}{13}
    \endentry
    \entry{mcca:cons88}{article}{}
      \name{author}{1}{}{%
        {{hash=8ade194e4fcf852a0e2548f395004641}{%
           family={McCarty},
           familyi={M\bibinitperiod},
           given={D.C.},
           giveni={D\bibinitperiod}}}%
      }
      \strng{namehash}{8ade194e4fcf852a0e2548f395004641}
      \strng{fullhash}{8ade194e4fcf852a0e2548f395004641}
      \strng{authornamehash}{8ade194e4fcf852a0e2548f395004641}
      \strng{authorfullhash}{8ade194e4fcf852a0e2548f395004641}
      \field{labelalpha}{McC88}
      \field{sortinit}{M}
      \field{sortinithash}{c26a05ef03e4429073ed5c825140fac3}
      \field{labelnamesource}{author}
      \field{labeltitlesource}{title}
      \field{journaltitle}{The Journal of Symbolic Logic}
      \field{title}{Constructive validity is nonarithmetic}
      \field{volume}{53}
      \field{year}{1988}
      \field{pages}{1036\bibrangedash 1041}
      \range{pages}{6}
    \endentry
    \entry{mcca:inco91}{article}{}
      \name{author}{1}{}{%
        {{hash=8ade194e4fcf852a0e2548f395004641}{%
           family={McCarty},
           familyi={M\bibinitperiod},
           given={D.C.},
           giveni={D\bibinitperiod}}}%
      }
      \strng{namehash}{8ade194e4fcf852a0e2548f395004641}
      \strng{fullhash}{8ade194e4fcf852a0e2548f395004641}
      \strng{authornamehash}{8ade194e4fcf852a0e2548f395004641}
      \strng{authorfullhash}{8ade194e4fcf852a0e2548f395004641}
      \field{labelalpha}{McC91}
      \field{sortinit}{M}
      \field{sortinithash}{c26a05ef03e4429073ed5c825140fac3}
      \field{labelnamesource}{author}
      \field{labeltitlesource}{title}
      \field{journaltitle}{Notre Dame Journal of Formal Logic}
      \field{title}{Incompleteness in intuitionistic metamathematics}
      \field{volume}{32}
      \field{year}{1991}
      \field{pages}{323\bibrangedash 358}
      \range{pages}{36}
    \endentry
    \entry{MiliusL17:fi}{article}{}
      \name{author}{2}{}{%
        {{hash=dd3908f80ceafcff3de8d0813d83b14a}{%
           family={Milius},
           familyi={M\bibinitperiod},
           given={Stefan},
           giveni={S\bibinitperiod}}}%
        {{hash=8f5431086357f9f5a0eef79eb62f74f4}{%
           family={Litak},
           familyi={L\bibinitperiod},
           given={Tadeusz},
           giveni={T\bibinitperiod}}}%
      }
      \list{publisher}{1}{%
        {IOS Press}%
      }
      \strng{namehash}{7ebacbae0065221b999f8a88b318a8dd}
      \strng{fullhash}{7ebacbae0065221b999f8a88b318a8dd}
      \strng{authornamehash}{7ebacbae0065221b999f8a88b318a8dd}
      \strng{authorfullhash}{7ebacbae0065221b999f8a88b318a8dd}
      \field{labelalpha}{ML17}
      \field{sortinit}{M}
      \field{sortinithash}{c26a05ef03e4429073ed5c825140fac3}
      \field{labelnamesource}{author}
      \field{labeltitlesource}{title}
      \field{journaltitle}{Fundamenta Informaticae}
      \field{note}{special issue FiCS'13 edited by David Baelde, Arnaud Carayol, Ralph Matthes and Igor Walukiewicz}
      \field{title}{Guard Your Daggers and Traces: Properties of Guarded (Co-)recursion}
      \field{volume}{150}
      \field{year}{2017}
      \field{pages}{407\bibrangedash 449}
      \range{pages}{43}
      \verb{doi}
      \verb 10.3233/FI-2017-1475
      \endverb
      \verb{url}
      \verb http://arxiv.org/abs/1603.05214
      \endverb
    \endentry
    \entry{moer:mode13}{book}{}
      \name{author}{2}{}{%
        {{hash=8e7d26d866770635f2c8240d90787545}{%
           family={Moerdijk},
           familyi={M\bibinitperiod},
           given={I.},
           giveni={I\bibinitperiod}}}%
        {{hash=51ae9bbcc327e5c12aa8aae45ce655a0}{%
           family={Reyes},
           familyi={R\bibinitperiod},
           given={G.E.},
           giveni={G\bibinitperiod}}}%
      }
      \list{location}{1}{%
        {New York}%
      }
      \list{publisher}{1}{%
        {Springer Science \& Business Media}%
      }
      \strng{namehash}{82d84e927b48a5719cd1773b2ea364bb}
      \strng{fullhash}{82d84e927b48a5719cd1773b2ea364bb}
      \strng{authornamehash}{82d84e927b48a5719cd1773b2ea364bb}
      \strng{authorfullhash}{82d84e927b48a5719cd1773b2ea364bb}
      \field{labelalpha}{MR13}
      \field{sortinit}{M}
      \field{sortinithash}{c26a05ef03e4429073ed5c825140fac3}
      \field{labelnamesource}{author}
      \field{labeltitlesource}{title}
      \field{title}{Models for smooth infinitesimal analysis}
      \field{year}{2013}
    \endentry
    \entry{Mogelberg14:lics}{inproceedings}{}
      \name{author}{1}{}{%
        {{hash=87f3fdf5910beec6d8dba13dfb21e1be}{%
           family={M{\o{}}gelberg},
           familyi={M\bibinitperiod},
           given={Rasmus\bibnamedelima Ejlers},
           giveni={R\bibinitperiod\bibinitdelim E\bibinitperiod}}}%
      }
      \list{location}{1}{%
        {Vienna, Austria}%
      }
      \list{publisher}{1}{%
        {ACM}%
      }
      \strng{namehash}{87f3fdf5910beec6d8dba13dfb21e1be}
      \strng{fullhash}{87f3fdf5910beec6d8dba13dfb21e1be}
      \strng{authornamehash}{87f3fdf5910beec6d8dba13dfb21e1be}
      \strng{authorfullhash}{87f3fdf5910beec6d8dba13dfb21e1be}
      \field{labelalpha}{M\o{}g14}
      \field{sortinit}{M}
      \field{sortinithash}{c26a05ef03e4429073ed5c825140fac3}
      \field{labelnamesource}{author}
      \field{labeltitlesource}{title}
      \field{booktitle}{Proceedings of CSL-LiCS}
      \field{isbn}{978-1-4503-2886-9}
      \field{title}{A Type Theory for Productive Coprogramming via Guarded Recursion}
      \field{year}{2014}
      \field{pages}{71:1\bibrangedash 71:10}
      \range{pages}{-1}
      \verb{doi}
      \verb 10.1145/2603088.2603132
      \endverb
      \verb{url}
      \verb http://doi.acm.org/10.1145/2603088.2603132
      \endverb
      \keyw{categorical semantics,corecursion,denotational semantics,dependent types,guarded recursion}
    \endentry
    \entry{Moggi91:ic}{article}{}
      \name{author}{1}{}{%
        {{hash=c039bd66d17f15b251307a409c68df7c}{%
           family={Moggi},
           familyi={M\bibinitperiod},
           given={Eugenio},
           giveni={E\bibinitperiod}}}%
      }
      \list{location}{1}{%
        {Duluth, MN, USA}%
      }
      \list{publisher}{1}{%
        {Academic Press, Inc.}%
      }
      \strng{namehash}{c039bd66d17f15b251307a409c68df7c}
      \strng{fullhash}{c039bd66d17f15b251307a409c68df7c}
      \strng{authornamehash}{c039bd66d17f15b251307a409c68df7c}
      \strng{authorfullhash}{c039bd66d17f15b251307a409c68df7c}
      \field{labelalpha}{Mog91}
      \field{sortinit}{M}
      \field{sortinithash}{c26a05ef03e4429073ed5c825140fac3}
      \field{labelnamesource}{author}
      \field{labeltitlesource}{title}
      \field{issn}{0890-5401}
      \field{issue}{1}
      \field{journaltitle}{Inform. and Comput.}
      \field{month}{7}
      \field{title}{Notions of computation and monads}
      \field{volume}{93}
      \field{year}{1991}
      \field{pages}{55\bibrangedash 92}
      \range{pages}{38}
      \verb{doi}
      \verb 10.1016/0890-5401(91)90052-4
      \endverb
      \verb{url}
      \verb http://portal.acm.org/citation.cfm?id=116981.116984
      \endverb
    \endentry
    \entry{Murphey05}{book}{}
      \name{author}{1}{}{%
        {{hash=6fbadd65c3cc676057475706694249be}{%
           family={Murphey},
           familyi={M\bibinitperiod},
           given={Murray\bibnamedelima G.},
           giveni={M\bibinitperiod\bibinitdelim G\bibinitperiod}}}%
      }
      \list{publisher}{1}{%
        {SUNY Press}%
      }
      \strng{namehash}{6fbadd65c3cc676057475706694249be}
      \strng{fullhash}{6fbadd65c3cc676057475706694249be}
      \strng{authornamehash}{6fbadd65c3cc676057475706694249be}
      \strng{authorfullhash}{6fbadd65c3cc676057475706694249be}
      \field{labelalpha}{Mur05}
      \field{sortinit}{M}
      \field{sortinithash}{c26a05ef03e4429073ed5c825140fac3}
      \field{labelnamesource}{author}
      \field{labeltitlesource}{title}
      \field{series}{SUNY Series in Philosophy}
      \field{title}{C. I. Lewis: The Last Great Pragmatist}
      \field{year}{2005}
    \endentry
    \entry{Nakano00:lics}{inproceedings}{}
      \name{author}{1}{}{%
        {{hash=92df934bc3496608c1580dd843703971}{%
           family={Nakano},
           familyi={N\bibinitperiod},
           given={Hiroshi},
           giveni={H\bibinitperiod}}}%
      }
      \list{organization}{1}{%
        {IEEE}%
      }
      \strng{namehash}{92df934bc3496608c1580dd843703971}
      \strng{fullhash}{92df934bc3496608c1580dd843703971}
      \strng{authornamehash}{92df934bc3496608c1580dd843703971}
      \strng{authorfullhash}{92df934bc3496608c1580dd843703971}
      \field{labelalpha}{Nak00}
      \field{sortinit}{N}
      \field{sortinithash}{1163c28585427c673ad5a010cbf82f52}
      \field{labelnamesource}{author}
      \field{labeltitlesource}{title}
      \field{booktitle}{Proceedings of {LiCS}}
      \field{isbn}{0-7695-0725-5}
      \field{title}{A Modality for Recursion}
      \field{year}{2000}
      \field{pages}{255\bibrangedash 266}
      \range{pages}{12}
    \endentry
    \entry{Nakano01:tacs}{inproceedings}{}
      \name{author}{1}{}{%
        {{hash=92df934bc3496608c1580dd843703971}{%
           family={Nakano},
           familyi={N\bibinitperiod},
           given={Hiroshi},
           giveni={H\bibinitperiod}}}%
      }
      \name{editor}{2}{}{%
        {{hash=f22f6e9ad2379a5e5de8a329e3c0e010}{%
           family={Kobayashi},
           familyi={K\bibinitperiod},
           given={Naoki},
           giveni={N\bibinitperiod}}}%
        {{hash=8eb1c0a7778a466aa612ab3259d489ec}{%
           family={Pierce},
           familyi={P\bibinitperiod},
           given={Benjamin\bibnamedelima C.},
           giveni={B\bibinitperiod\bibinitdelim C\bibinitperiod}}}%
      }
      \list{publisher}{1}{%
        {Springer}%
      }
      \strng{namehash}{92df934bc3496608c1580dd843703971}
      \strng{fullhash}{92df934bc3496608c1580dd843703971}
      \strng{authornamehash}{92df934bc3496608c1580dd843703971}
      \strng{authorfullhash}{92df934bc3496608c1580dd843703971}
      \strng{editornamehash}{55f789c7195e24e4093940f3e4206947}
      \strng{editorfullhash}{55f789c7195e24e4093940f3e4206947}
      \field{labelalpha}{Nak01}
      \field{sortinit}{N}
      \field{sortinithash}{1163c28585427c673ad5a010cbf82f52}
      \field{labelnamesource}{author}
      \field{labeltitlesource}{title}
      \field{booktitle}{Proceedings of {TACS}}
      \field{isbn}{3-540-42736-8}
      \field{series}{{LNCS}}
      \field{title}{Fixed-Point Logic with the Approximation Modality and Its {K}ripke Completeness}
      \field{volume}{2215}
      \field{year}{2001}
      \field{pages}{165\bibrangedash 182}
      \range{pages}{18}
    \endentry
    \entry{OHearnP99:jsl}{article}{}
      \name{author}{2}{}{%
        {{hash=5c1a8414aa0ced0c02d7aa7d269f53a8}{%
           family={O'Hearn},
           familyi={O\bibinitperiod},
           given={Peter\bibnamedelima W.},
           giveni={P\bibinitperiod\bibinitdelim W\bibinitperiod}}}%
        {{hash=4c353e5061ad3d7e3852fe7e65c578d4}{%
           family={Pym},
           familyi={P\bibinitperiod},
           given={David\bibnamedelima J.},
           giveni={D\bibinitperiod\bibinitdelim J\bibinitperiod}}}%
      }
      \list{language}{1}{%
        {English}%
      }
      \list{publisher}{1}{%
        {Association for Symbolic Logic}%
      }
      \strng{namehash}{3ae1a0ea7d2eafb4eca01a34adf1391f}
      \strng{fullhash}{3ae1a0ea7d2eafb4eca01a34adf1391f}
      \strng{authornamehash}{3ae1a0ea7d2eafb4eca01a34adf1391f}
      \strng{authorfullhash}{3ae1a0ea7d2eafb4eca01a34adf1391f}
      \field{labelalpha}{OP99}
      \field{sortinit}{O}
      \field{sortinithash}{f5d80ef71b5b58852663ecfd775b175a}
      \field{labelnamesource}{author}
      \field{labeltitlesource}{title}
      \field{abstract}{We introduce a logic BI in which a multiplicative (or linear) and an additive (or intuitionistic) implication live side-by-side. The propositional version of BI arises from an analysis of the proof-theoretic relationship between conjunction and implication; it can be viewed as a merging of intuitionistic logic and multiplicative intuitionistic linear logic. The naturality of BI can be seen categorically: models of propositional BI's proofs are given by bicartesian doubly closed categories, i.e., categories which freely combine the semantics of propositional intuitionistic logic and propositional multiplicative intuitionistic linear logic. The predicate version of BI includes, in addition to standard additive quantifiers, multiplicative (or intensional) quantifiers \x{fffd}\x{fffd}\x{fffd} <sub>new</sub> and \x{fffd}\x{fffd}\x{fffd} <sub>new</sub> which arise from observing restrictions on structural rules on the level of terms as well as propositions. We discuss computational interpretations, based on sharing, at both the propositional and predicate levels.}
      \field{issn}{10798986}
      \field{journaltitle}{B. Symb. Log.}
      \field{number}{2}
      \field{title}{The Logic of Bunched Implications}
      \field{volume}{5}
      \field{year}{1999}
      \field{pages}{215\bibrangedash 244}
      \range{pages}{30}
      \verb{url}
      \verb http://www.jstor.org/stable/421090
      \endverb
    \endentry
    \entry{Parry70}{article}{}
      \name{author}{1}{}{%
        {{hash=046a643790dc78d7eae9e65383fe1bab}{%
           family={Parry},
           familyi={P\bibinitperiod},
           given={William\bibnamedelima Tuthill},
           giveni={W\bibinitperiod\bibinitdelim T\bibinitperiod}}}%
      }
      \list{publisher}{1}{%
        {Duke University Press}%
      }
      \strng{namehash}{046a643790dc78d7eae9e65383fe1bab}
      \strng{fullhash}{046a643790dc78d7eae9e65383fe1bab}
      \strng{authornamehash}{046a643790dc78d7eae9e65383fe1bab}
      \strng{authorfullhash}{046a643790dc78d7eae9e65383fe1bab}
      \field{labelalpha}{Par70}
      \field{sortinit}{P}
      \field{sortinithash}{24100cef455d7974167575052c29146e}
      \field{labelnamesource}{author}
      \field{labeltitlesource}{title}
      \field{journaltitle}{Notre Dame J. Formal Logic}
      \field{month}{04}
      \field{number}{2}
      \field{title}{In memoriam: {C}larence {I}rving {L}ewis (1883--1964).}
      \field{volume}{11}
      \field{year}{1970}
      \field{pages}{129\bibrangedash 140}
      \range{pages}{12}
      \verb{doi}
      \verb 10.1305/ndjfl/1093893933
      \endverb
      \verb{url}
      \verb http://dx.doi.org/10.1305/ndjfl/1093893933
      \endverb
    \endentry
    \entry{Paterson03:fop}{incollection}{}
      \name{author}{1}{}{%
        {{hash=1f3f8610e67d8b56b596848671113178}{%
           family={Paterson},
           familyi={P\bibinitperiod},
           given={Ross},
           giveni={R\bibinitperiod}}}%
      }
      \name{editor}{2}{}{%
        {{hash=d7effacc4950f11aea627dc1bac84255}{%
           family={Gibbons},
           familyi={G\bibinitperiod},
           given={Jeremy},
           giveni={J\bibinitperiod}}}%
        {{hash=775d2bba69321780e9bbbaa86729733e}{%
           family={Moor},
           familyi={M\bibinitperiod},
           given={Oege},
           giveni={O\bibinitperiod},
           prefix={de},
           prefixi={d\bibinitperiod}}}%
      }
      \list{publisher}{1}{%
        {Palgrave}%
      }
      \strng{namehash}{1f3f8610e67d8b56b596848671113178}
      \strng{fullhash}{1f3f8610e67d8b56b596848671113178}
      \strng{authornamehash}{1f3f8610e67d8b56b596848671113178}
      \strng{authorfullhash}{1f3f8610e67d8b56b596848671113178}
      \strng{editornamehash}{fd88beba7b3a28859698958f899bf34e}
      \strng{editorfullhash}{fd88beba7b3a28859698958f899bf34e}
      \field{labelalpha}{Pat03}
      \field{sortinit}{P}
      \field{sortinithash}{24100cef455d7974167575052c29146e}
      \field{labelnamesource}{author}
      \field{labeltitlesource}{title}
      \field{booktitle}{The Fun of Programming}
      \field{title}{Arrows and Computation}
      \field{year}{2003}
      \field{pages}{201\bibrangedash 222}
      \range{pages}{22}
      \verb{url}
      \verb http://www.soi.city.ac.uk/~ross/papers/fop.html
      \endverb
    \endentry
    \entry{heyt:malo90}{proceedings}{}
      \name{editor}{1}{}{%
        {{hash=6481cde4d8cb9e1508c2c95dec87911a}{%
           family={Petkov},
           familyi={P\bibinitperiod},
           given={P.P.},
           giveni={P\bibinitperiod}}}%
      }
      \list{publisher}{1}{%
        {Plenum Press, Boston}%
      }
      \strng{namehash}{6481cde4d8cb9e1508c2c95dec87911a}
      \strng{fullhash}{6481cde4d8cb9e1508c2c95dec87911a}
      \strng{editornamehash}{6481cde4d8cb9e1508c2c95dec87911a}
      \strng{editorfullhash}{6481cde4d8cb9e1508c2c95dec87911a}
      \field{labelalpha}{Pet90}
      \field{sortinit}{P}
      \field{sortinithash}{24100cef455d7974167575052c29146e}
      \field{labelnamesource}{editor}
      \field{labeltitlesource}{title}
      \field{title}{Mathematical logic, {P}roceedings of the {H}eyting 1988 summer school in {V}arna, {B}ulgaria}
      \field{year}{1990}
    \endentry
    \entry{DBLP:conf/fossacs/2015}{proceedings}{}
      \name{editor}{1}{}{%
        {{hash=2e998c3985e6a8efab31950621ac04ff}{%
           family={Pitts},
           familyi={P\bibinitperiod},
           given={Andrew\bibnamedelima M.},
           giveni={A\bibinitperiod\bibinitdelim M\bibinitperiod}}}%
      }
      \list{publisher}{1}{%
        {Springer}%
      }
      \strng{namehash}{2e998c3985e6a8efab31950621ac04ff}
      \strng{fullhash}{2e998c3985e6a8efab31950621ac04ff}
      \strng{editornamehash}{2e998c3985e6a8efab31950621ac04ff}
      \strng{editorfullhash}{2e998c3985e6a8efab31950621ac04ff}
      \field{labelalpha}{Pit15}
      \field{sortinit}{P}
      \field{sortinithash}{24100cef455d7974167575052c29146e}
      \true{crossrefsource}
      \field{labelnamesource}{editor}
      \field{labeltitlesource}{title}
      \field{series}{{LNCS}}
      \field{title}{Proceedings of {FoSSaCS}}
      \field{volume}{9034}
      \field{year}{2015}
    \endentry
    \entry{plis:surv09}{article}{}
      \name{author}{1}{}{%
        {{hash=b6d0e191a5944c545ace8de9274c688f}{%
           family={Plisko},
           familyi={P\bibinitperiod},
           given={V.E.},
           giveni={V\bibinitperiod}}}%
      }
      \list{publisher}{1}{%
        {Cambridge Univ Press}%
      }
      \strng{namehash}{b6d0e191a5944c545ace8de9274c688f}
      \strng{fullhash}{b6d0e191a5944c545ace8de9274c688f}
      \strng{authornamehash}{b6d0e191a5944c545ace8de9274c688f}
      \strng{authorfullhash}{b6d0e191a5944c545ace8de9274c688f}
      \field{labelalpha}{Pli09}
      \field{sortinit}{P}
      \field{sortinithash}{24100cef455d7974167575052c29146e}
      \field{labelnamesource}{author}
      \field{labeltitlesource}{title}
      \field{journaltitle}{The Bulletin of Symbolic Logic}
      \field{number}{01}
      \field{title}{{A survey of propositional realizability logic}}
      \field{volume}{15}
      \field{year}{2009}
      \field{pages}{1\bibrangedash 42}
      \range{pages}{42}
    \endentry
    \entry{Proietti2012}{article}{}
      \name{author}{1}{}{%
        {{hash=80828a62878770b356584bc6567bc947}{%
           family={Proietti},
           familyi={P\bibinitperiod},
           given={Carlo},
           giveni={C\bibinitperiod}}}%
      }
      \strng{namehash}{80828a62878770b356584bc6567bc947}
      \strng{fullhash}{80828a62878770b356584bc6567bc947}
      \strng{authornamehash}{80828a62878770b356584bc6567bc947}
      \strng{authorfullhash}{80828a62878770b356584bc6567bc947}
      \field{labelalpha}{Pro12}
      \field{sortinit}{P}
      \field{sortinithash}{24100cef455d7974167575052c29146e}
      \field{labelnamesource}{author}
      \field{labeltitlesource}{title}
      \field{issn}{1573-0433}
      \field{journaltitle}{Journal of Philosophical Logic}
      \field{number}{5}
      \field{title}{Intuitionistic Epistemic Logic, Kripke Models and Fitch's Paradox}
      \field{volume}{41}
      \field{year}{2012}
      \field{pages}{877\bibrangedash 900}
      \range{pages}{24}
      \verb{doi}
      \verb 10.1007/s10992-011-9207-1
      \endverb
      \verb{url}
      \verb http://dx.doi.org/10.1007/s10992-011-9207-1
      \endverb
    \endentry
    \entry{PymOHY04:tcs}{article}{}
      \name{author}{3}{}{%
        {{hash=4c353e5061ad3d7e3852fe7e65c578d4}{%
           family={Pym},
           familyi={P\bibinitperiod},
           given={David\bibnamedelima J.},
           giveni={D\bibinitperiod\bibinitdelim J\bibinitperiod}}}%
        {{hash=5c1a8414aa0ced0c02d7aa7d269f53a8}{%
           family={O'Hearn},
           familyi={O\bibinitperiod},
           given={Peter\bibnamedelima W.},
           giveni={P\bibinitperiod\bibinitdelim W\bibinitperiod}}}%
        {{hash=f54dd89dfae376ce30704e5a2ea3b4b1}{%
           family={Yang},
           familyi={Y\bibinitperiod},
           given={Hongseok},
           giveni={H\bibinitperiod}}}%
      }
      \strng{namehash}{9de00bdb9567cee5379ca31fd8558d5a}
      \strng{fullhash}{9de00bdb9567cee5379ca31fd8558d5a}
      \strng{authornamehash}{9de00bdb9567cee5379ca31fd8558d5a}
      \strng{authorfullhash}{9de00bdb9567cee5379ca31fd8558d5a}
      \field{labelalpha}{POY04}
      \field{sortinit}{P}
      \field{sortinithash}{24100cef455d7974167575052c29146e}
      \field{labelnamesource}{author}
      \field{labeltitlesource}{title}
      \field{issn}{0304-3975}
      \field{journaltitle}{Theoretical Computer Science}
      \field{note}{Mathematical Foundations of Programming Semantics}
      \field{number}{1}
      \field{title}{Possible worlds and resources: the semantics of {BI}}
      \field{volume}{315}
      \field{year}{2004}
      \field{pages}{257\bibrangedash 305}
      \range{pages}{49}
      \verb{doi}
      \verb http://dx.doi.org/10.1016/j.tcs.2003.11.020
      \endverb
      \verb{url}
      \verb http://www.sciencedirect.com/science/article/pii/S0304397503006248
      \endverb
    \endentry
    \entry{Pym02:book}{book}{}
      \name{author}{1}{}{%
        {{hash=0119af8138d1516a727b57f52c825c59}{%
           family={Pym},
           familyi={P\bibinitperiod},
           given={D.J.},
           giveni={D\bibinitperiod}}}%
      }
      \list{publisher}{1}{%
        {Kluwer Academic Publishers}%
      }
      \strng{namehash}{0119af8138d1516a727b57f52c825c59}
      \strng{fullhash}{0119af8138d1516a727b57f52c825c59}
      \strng{authornamehash}{0119af8138d1516a727b57f52c825c59}
      \strng{authorfullhash}{0119af8138d1516a727b57f52c825c59}
      \field{labelalpha}{Pym02}
      \field{sortinit}{P}
      \field{sortinithash}{24100cef455d7974167575052c29146e}
      \field{labelnamesource}{author}
      \field{labeltitlesource}{title}
      \field{isbn}{9789048160723}
      \field{series}{Appl. Log. Ser.}
      \field{title}{The Semantics and Proof Theory of the Logic of Bunched Implications}
      \field{volume}{26}
      \field{year}{2002}
    \endentry
    \entry{rena:inte89}{article}{}
      \name{author}{1}{}{%
        {{hash=320b15a8f7721f24ab526f75f17ae162}{%
           family={Renardel\bibnamedelimb de\bibnamedelima Lavalette},
           familyi={R\bibinitperiod\bibinitdelim d\bibinitperiod\bibinitdelim L\bibinitperiod},
           given={G.R.},
           giveni={G\bibinitperiod}}}%
      }
      \strng{namehash}{320b15a8f7721f24ab526f75f17ae162}
      \strng{fullhash}{320b15a8f7721f24ab526f75f17ae162}
      \strng{authornamehash}{320b15a8f7721f24ab526f75f17ae162}
      \strng{authorfullhash}{320b15a8f7721f24ab526f75f17ae162}
      \field{labelalpha}{Ren89}
      \field{sortinit}{R}
      \field{sortinithash}{c15bc8eb6936bc6b3c8baa9e8575af53}
      \field{labelnamesource}{author}
      \field{labeltitlesource}{title}
      \field{journaltitle}{{Journal of Symbolic Logic}}
      \field{number}{04}
      \field{title}{{Interpolation in fragments of intuitionistic propositional logic}}
      \field{volume}{54}
      \field{year}{1989}
      \field{pages}{1419\bibrangedash 1430}
      \range{pages}{12}
    \endentry
    \entry{Severi2017}{incollection}{}
      \name{author}{1}{}{%
        {{hash=f08beafad423a2056fc2a458e32a8283}{%
           family={Severi},
           familyi={S\bibinitperiod},
           given={Paula},
           giveni={P\bibinitperiod}}}%
      }
      \name{editor}{2}{}{%
        {{hash=aae57fe3f3726fcbced8c287c38380ec}{%
           family={Esparza},
           familyi={E\bibinitperiod},
           given={Javier},
           giveni={J\bibinitperiod}}}%
        {{hash=1a19969a1159c2f53f8fa54a91ad4679}{%
           family={Murawski},
           familyi={M\bibinitperiod},
           given={Andrzej\bibnamedelima S.},
           giveni={A\bibinitperiod\bibinitdelim S\bibinitperiod}}}%
      }
      \list{location}{1}{%
        {Berlin, Heidelberg}%
      }
      \list{publisher}{1}{%
        {Springer Berlin Heidelberg}%
      }
      \strng{namehash}{f08beafad423a2056fc2a458e32a8283}
      \strng{fullhash}{f08beafad423a2056fc2a458e32a8283}
      \strng{authornamehash}{f08beafad423a2056fc2a458e32a8283}
      \strng{authorfullhash}{f08beafad423a2056fc2a458e32a8283}
      \strng{editornamehash}{0c7c942b9572336772010ab045e94580}
      \strng{editorfullhash}{0c7c942b9572336772010ab045e94580}
      \field{labelalpha}{Sev17}
      \field{sortinit}{S}
      \field{sortinithash}{3c1547c63380458f8ca90e40ed14b83e}
      \field{labelnamesource}{author}
      \field{labeltitlesource}{title}
      \field{booktitle}{Proceedings of FOSSACS 2017}
      \field{isbn}{978-3-662-54458-7}
      \field{title}{A Light Modality for Recursion}
      \field{year}{2017}
      \field{pages}{499\bibrangedash 516}
      \range{pages}{18}
      \verb{doi}
      \verb 10.1007/978-3-662-54458-7_29
      \endverb
      \verb{url}
      \verb http://dx.doi.org/10.1007/978-3-662-54458-7_29
      \endverb
    \endentry
    \entry{shav:smar94}{article}{}
      \name{author}{1}{}{%
        {{hash=2fe8ee496e2b1271729f9c44679bf57a}{%
           family={Shavrukov},
           familyi={S\bibinitperiod},
           given={V.Yu.},
           giveni={V\bibinitperiod}}}%
      }
      \strng{namehash}{2fe8ee496e2b1271729f9c44679bf57a}
      \strng{fullhash}{2fe8ee496e2b1271729f9c44679bf57a}
      \strng{authornamehash}{2fe8ee496e2b1271729f9c44679bf57a}
      \strng{authorfullhash}{2fe8ee496e2b1271729f9c44679bf57a}
      \field{labelalpha}{Sha94}
      \field{sortinit}{S}
      \field{sortinithash}{3c1547c63380458f8ca90e40ed14b83e}
      \field{labelnamesource}{author}
      \field{labeltitlesource}{title}
      \field{journaltitle}{Notre Dame Journal of Formal Logic}
      \field{title}{A smart child of {P}eano's}
      \field{volume}{35}
      \field{year}{1994}
      \field{pages}{161\bibrangedash 185}
      \range{pages}{25}
    \endentry
    \entry{shav:suba93}{article}{}
      \name{author}{1}{}{%
        {{hash=2fe8ee496e2b1271729f9c44679bf57a}{%
           family={Shavrukov},
           familyi={S\bibinitperiod},
           given={V.Yu.},
           giveni={V\bibinitperiod}}}%
      }
      \strng{namehash}{2fe8ee496e2b1271729f9c44679bf57a}
      \strng{fullhash}{2fe8ee496e2b1271729f9c44679bf57a}
      \strng{authornamehash}{2fe8ee496e2b1271729f9c44679bf57a}
      \strng{authorfullhash}{2fe8ee496e2b1271729f9c44679bf57a}
      \field{labelalpha}{Sha93}
      \field{sortinit}{S}
      \field{sortinithash}{3c1547c63380458f8ca90e40ed14b83e}
      \field{labelnamesource}{author}
      \field{labeltitlesource}{title}
      \field{journaltitle}{Dissertationes mathematicae (Rozprawy matematyczne)}
      \field{title}{Subalgebras of diagonalizable algebras of theories containing arithmetic}
      \field{volume}{CCCXXIII}
      \field{year}{1993}
    \endentry
    \entry{shav:rela88}{report}{}
      \name{author}{1}{}{%
        {{hash=2fe8ee496e2b1271729f9c44679bf57a}{%
           family={Shavrukov},
           familyi={S\bibinitperiod},
           given={V.Yu.},
           giveni={V\bibinitperiod}}}%
      }
      \list{institution}{1}{%
        {Stekhlov Mathematical Institute, Moscow}%
      }
      \strng{namehash}{2fe8ee496e2b1271729f9c44679bf57a}
      \strng{fullhash}{2fe8ee496e2b1271729f9c44679bf57a}
      \strng{authornamehash}{2fe8ee496e2b1271729f9c44679bf57a}
      \strng{authorfullhash}{2fe8ee496e2b1271729f9c44679bf57a}
      \field{labelalpha}{Sha88}
      \field{sortinit}{S}
      \field{sortinithash}{3c1547c63380458f8ca90e40ed14b83e}
      \field{labelnamesource}{author}
      \field{labeltitlesource}{title}
      \field{number}{Report No.5}
      \field{title}{The logic of relative interpretability over {P}eano arithmetic (in {R}ussian)}
      \field{type}{techreport}
      \field{year}{1988}
    \endentry
    \entry{SieczkowskiBB15:itp}{inproceedings}{}
      \name{author}{3}{}{%
        {{hash=8e6c8d55ba05cecc763597cfdfa4c0c4}{%
           family={Sieczkowski},
           familyi={S\bibinitperiod},
           given={Filip},
           giveni={F\bibinitperiod}}}%
        {{hash=c246861cdfff84700e5f8605b429159f}{%
           family={Bizjak},
           familyi={B\bibinitperiod},
           given={Ales},
           giveni={A\bibinitperiod}}}%
        {{hash=f4714354745ea31cf6c543434102f444}{%
           family={Birkedal},
           familyi={B\bibinitperiod},
           given={Lars},
           giveni={L\bibinitperiod}}}%
      }
      \name{editor}{2}{}{%
        {{hash=7e646f94ce1ca043d93239cb0cf55cdf}{%
           family={Urban},
           familyi={U\bibinitperiod},
           given={Christian},
           giveni={C\bibinitperiod}}}%
        {{hash=55f1e3e1f91df039fc0dcabef4248b76}{%
           family={Zhang},
           familyi={Z\bibinitperiod},
           given={Xingyuan},
           giveni={X\bibinitperiod}}}%
      }
      \list{publisher}{1}{%
        {Springer}%
      }
      \strng{namehash}{d60bd74677e45bd45315986f8bdd5380}
      \strng{fullhash}{d60bd74677e45bd45315986f8bdd5380}
      \strng{authornamehash}{d60bd74677e45bd45315986f8bdd5380}
      \strng{authorfullhash}{d60bd74677e45bd45315986f8bdd5380}
      \strng{editornamehash}{d6b6cf9420bc7e84da1c160034b4c725}
      \strng{editorfullhash}{d6b6cf9420bc7e84da1c160034b4c725}
      \field{labelalpha}{SBB15}
      \field{sortinit}{S}
      \field{sortinithash}{3c1547c63380458f8ca90e40ed14b83e}
      \field{labelnamesource}{author}
      \field{labeltitlesource}{title}
      \field{booktitle}{Proceedings of {ITP}}
      \field{isbn}{978-3-319-22101-4}
      \field{series}{{LNCS}}
      \field{title}{Modu{R}es: {A} {C}oq Library for Modular Reasoning About Concurrent Higher-Order Imperative Programming Languages}
      \field{volume}{9236}
      \field{year}{2015}
      \field{pages}{375\bibrangedash 390}
      \range{pages}{16}
      \verb{doi}
      \verb 10.1007/978-3-319-22102-1_25
      \endverb
      \verb{url}
      \verb http://dx.doi.org/10.1007/978-3-319-22102-1_25
      \endverb
    \endentry
    \entry{Simpson94:phd}{thesis}{}
      \name{author}{1}{}{%
        {{hash=e7a2336b1deb4382df8d055c7b4db2ee}{%
           family={Simpson},
           familyi={S\bibinitperiod},
           given={Alex\bibnamedelima K.},
           giveni={A\bibinitperiod\bibinitdelim K\bibinitperiod}}}%
      }
      \list{institution}{1}{%
        {University of Edinburgh}%
      }
      \strng{namehash}{e7a2336b1deb4382df8d055c7b4db2ee}
      \strng{fullhash}{e7a2336b1deb4382df8d055c7b4db2ee}
      \strng{authornamehash}{e7a2336b1deb4382df8d055c7b4db2ee}
      \strng{authorfullhash}{e7a2336b1deb4382df8d055c7b4db2ee}
      \field{labelalpha}{Sim94}
      \field{sortinit}{S}
      \field{sortinithash}{3c1547c63380458f8ca90e40ed14b83e}
      \field{labelnamesource}{author}
      \field{labeltitlesource}{title}
      \field{title}{{The Proof Theory and Semantics of Intuitionistic Modal Logic}}
      \field{type}{phdthesis}
      \field{year}{1994}
      \verb{url}
      \verb http://homepages.inf.ed.ac.uk/als/Research/thesis.ps.gz
      \endverb
      \keyw{kripke-models,modal-logic,natural-deduction}
    \endentry
    \entry{smor:appl73}{incollection}{}
      \name{author}{1}{}{%
        {{hash=b14df81bb21719285d6b80b03eaf1644}{%
           family={Smory\'{n}ski},
           familyi={S\bibinitperiod},
           given={C.},
           giveni={C\bibinitperiod}}}%
      }
      \name{editor}{1}{}{%
        {{hash=2bd43d4291c46b9a1f9e75cf54cccae3}{%
           family={Troelstra},
           familyi={T\bibinitperiod},
           given={A.S.},
           giveni={A\bibinitperiod}}}%
      }
      \list{location}{1}{%
        {Berlin}%
      }
      \list{publisher}{1}{%
        {Springer}%
      }
      \strng{namehash}{b14df81bb21719285d6b80b03eaf1644}
      \strng{fullhash}{b14df81bb21719285d6b80b03eaf1644}
      \strng{authornamehash}{b14df81bb21719285d6b80b03eaf1644}
      \strng{authorfullhash}{b14df81bb21719285d6b80b03eaf1644}
      \strng{editornamehash}{2bd43d4291c46b9a1f9e75cf54cccae3}
      \strng{editorfullhash}{2bd43d4291c46b9a1f9e75cf54cccae3}
      \field{labelalpha}{Smo73}
      \field{sortinit}{S}
      \field{sortinithash}{3c1547c63380458f8ca90e40ed14b83e}
      \field{labelnamesource}{author}
      \field{labeltitlesource}{title}
      \field{booktitle}{{Metamathematical Investigations of Intuitionistic Arithmetic and Analysis}}
      \field{series}{Springer Lecture Notes 344}
      \field{title}{{Applications of {K}ripke Models}}
      \field{year}{1973}
      \field{pages}{324\bibrangedash 391}
      \range{pages}{68}
    \endentry
    \entry{smor:self85}{book}{}
      \name{author}{1}{}{%
        {{hash=b14df81bb21719285d6b80b03eaf1644}{%
           family={Smory\'{n}ski},
           familyi={S\bibinitperiod},
           given={C.},
           giveni={C\bibinitperiod}}}%
      }
      \list{location}{1}{%
        {New York}%
      }
      \list{publisher}{1}{%
        {Springer}%
      }
      \strng{namehash}{b14df81bb21719285d6b80b03eaf1644}
      \strng{fullhash}{b14df81bb21719285d6b80b03eaf1644}
      \strng{authornamehash}{b14df81bb21719285d6b80b03eaf1644}
      \strng{authorfullhash}{b14df81bb21719285d6b80b03eaf1644}
      \field{labelalpha}{Smo85}
      \field{sortinit}{S}
      \field{sortinithash}{3c1547c63380458f8ca90e40ed14b83e}
      \field{labelnamesource}{author}
      \field{labeltitlesource}{title}
      \field{series}{Universitext}
      \field{title}{{Self-Reference and Modal Logic}}
      \field{year}{1985}
    \endentry
    \entry{solo:prov76}{article}{}
      \name{author}{1}{}{%
        {{hash=7c986aba9e7af9692ee741d15d69dea3}{%
           family={Solovay},
           familyi={S\bibinitperiod},
           given={R.M.},
           giveni={R\bibinitperiod}}}%
      }
      \strng{namehash}{7c986aba9e7af9692ee741d15d69dea3}
      \strng{fullhash}{7c986aba9e7af9692ee741d15d69dea3}
      \strng{authornamehash}{7c986aba9e7af9692ee741d15d69dea3}
      \strng{authorfullhash}{7c986aba9e7af9692ee741d15d69dea3}
      \field{labelalpha}{Sol76}
      \field{sortinit}{S}
      \field{sortinithash}{3c1547c63380458f8ca90e40ed14b83e}
      \field{labelnamesource}{author}
      \field{labeltitlesource}{title}
      \field{journaltitle}{Israel Journal of Mathematics}
      \field{title}{Provability interpretations of modal logic}
      \field{volume}{25}
      \field{year}{1976}
      \field{pages}{287\bibrangedash 304}
      \range{pages}{18}
    \endentry
    \entry{SorensenU06:book}{book}{}
      \name{author}{2}{}{%
        {{hash=83e8b113c0bac5c881ae6d7405764255}{%
           family={S{\o{}}rensen},
           familyi={S\bibinitperiod},
           given={Morten\bibnamedelima Heine},
           giveni={M\bibinitperiod\bibinitdelim H\bibinitperiod}}}%
        {{hash=e457a88e05f324fb1406eb3ba848e68f}{%
           family={Urzyczyn},
           familyi={U\bibinitperiod},
           given={Pawel},
           giveni={P\bibinitperiod}}}%
      }
      \list{location}{1}{%
        {New York, NY, USA}%
      }
      \list{publisher}{1}{%
        {Elsevier Science Inc.}%
      }
      \strng{namehash}{bc8616851362c944eb4cb9494c3ba70d}
      \strng{fullhash}{bc8616851362c944eb4cb9494c3ba70d}
      \strng{authornamehash}{bc8616851362c944eb4cb9494c3ba70d}
      \strng{authorfullhash}{bc8616851362c944eb4cb9494c3ba70d}
      \field{labelalpha}{SU06}
      \field{sortinit}{S}
      \field{sortinithash}{3c1547c63380458f8ca90e40ed14b83e}
      \field{labelnamesource}{author}
      \field{labeltitlesource}{title}
      \field{isbn}{0444520775}
      \field{series}{Stud. Logic Found. Math.}
      \field{title}{Lectures on the Curry-Howard Isomorphism}
      \field{volume}{149}
      \field{year}{2006}
    \endentry
    \entry{Sotirov84:ml}{inproceedings}{}
      \name{author}{1}{}{%
        {{hash=befaf23d639f24c5172ace268973740a}{%
           family={Sotirov},
           familyi={S\bibinitperiod},
           given={Vladimir},
           giveni={V\bibinitperiod}}}%
      }
      \strng{namehash}{befaf23d639f24c5172ace268973740a}
      \strng{fullhash}{befaf23d639f24c5172ace268973740a}
      \strng{authornamehash}{befaf23d639f24c5172ace268973740a}
      \strng{authorfullhash}{befaf23d639f24c5172ace268973740a}
      \field{labelalpha}{Sot84}
      \field{sortinit}{S}
      \field{sortinithash}{3c1547c63380458f8ca90e40ed14b83e}
      \field{labelnamesource}{author}
      \field{labeltitlesource}{title}
      \field{booktitle}{Mathematical Logic, Proc. Conf. Math. Logic Dedicated to the Memory of A. A. Markov (1903 - 1979), Sofia, September 22 - 23, 1980}
      \field{title}{Modal theories with intuitionistic logic}
      \field{year}{1984}
      \field{pages}{139\bibrangedash 171}
      \range{pages}{33}
    \endentry
    \entry{svej:prov00}{article}{}
      \name{author}{1}{}{%
        {{hash=8090f81c8f13182bdab842a5cb537ead}{%
           family={\v{S}vejdar},
           familyi={\v{S}\bibinitperiod},
           given={V.},
           giveni={V\bibinitperiod}}}%
      }
      \strng{namehash}{8090f81c8f13182bdab842a5cb537ead}
      \strng{fullhash}{8090f81c8f13182bdab842a5cb537ead}
      \strng{authornamehash}{8090f81c8f13182bdab842a5cb537ead}
      \strng{authorfullhash}{8090f81c8f13182bdab842a5cb537ead}
      \field{labelalpha}{\v{S}ve00}
      \field{sortinit}{\v{S}}
      \field{sortinithash}{3c1547c63380458f8ca90e40ed14b83e}
      \field{labelnamesource}{author}
      \field{labeltitlesource}{title}
      \field{journaltitle}{Nordic Journal of Philosophical Logic}
      \field{number}{2}
      \field{title}{On Provability Logic}
      \field{volume}{4}
      \field{year}{2000}
      \field{pages}{95\bibrangedash 116}
      \range{pages}{22}
    \endentry
    \entry{SvendsenB14:esop}{inproceedings}{}
      \name{author}{2}{}{%
        {{hash=09bbde70c2e6081108b9a23332cb48e4}{%
           family={Svendsen},
           familyi={S\bibinitperiod},
           given={Kasper},
           giveni={K\bibinitperiod}}}%
        {{hash=f4714354745ea31cf6c543434102f444}{%
           family={Birkedal},
           familyi={B\bibinitperiod},
           given={Lars},
           giveni={L\bibinitperiod}}}%
      }
      \name{editor}{1}{}{%
        {{hash=47a4917313d2b2f846d6a9d80cedd291}{%
           family={Shao},
           familyi={S\bibinitperiod},
           given={Zhong},
           giveni={Z\bibinitperiod}}}%
      }
      \list{publisher}{1}{%
        {Springer}%
      }
      \strng{namehash}{0e726d39250456f13f1c80d101b16adb}
      \strng{fullhash}{0e726d39250456f13f1c80d101b16adb}
      \strng{authornamehash}{0e726d39250456f13f1c80d101b16adb}
      \strng{authorfullhash}{0e726d39250456f13f1c80d101b16adb}
      \strng{editornamehash}{47a4917313d2b2f846d6a9d80cedd291}
      \strng{editorfullhash}{47a4917313d2b2f846d6a9d80cedd291}
      \field{labelalpha}{SB14}
      \field{sortinit}{S}
      \field{sortinithash}{3c1547c63380458f8ca90e40ed14b83e}
      \field{labelnamesource}{author}
      \field{labeltitlesource}{title}
      \field{booktitle}{Proceedings of {ESOP}}
      \field{isbn}{978-3-642-54832-1}
      \field{series}{{LNCS}}
      \field{title}{Impredicative Concurrent Abstract Predicates}
      \field{volume}{8410}
      \field{year}{2014}
      \field{pages}{149\bibrangedash 168}
      \range{pages}{20}
      \verb{doi}
      \verb 10.1007/978-3-642-54833-8_9
      \endverb
      \verb{url}
      \verb http://dx.doi.org/10.1007/978-3-642-54833-8_9
      \endverb
    \endentry
    \entry{Troelstra92}{book}{}
      \name{author}{1}{}{%
        {{hash=a2f37cf4fb34653d3006015bc02615ed}{%
           family={Troelstra},
           familyi={T\bibinitperiod},
           given={Anne\bibnamedelima S.},
           giveni={A\bibinitperiod\bibinitdelim S\bibinitperiod}}}%
      }
      \list{location}{1}{%
        {Stanford, California}%
      }
      \list{publisher}{2}{%
        {CSLI Lecture Notes 29, Center for the Study of Language}%
        {Information}%
      }
      \strng{namehash}{a2f37cf4fb34653d3006015bc02615ed}
      \strng{fullhash}{a2f37cf4fb34653d3006015bc02615ed}
      \strng{authornamehash}{a2f37cf4fb34653d3006015bc02615ed}
      \strng{authorfullhash}{a2f37cf4fb34653d3006015bc02615ed}
      \field{labelalpha}{Tro92}
      \field{sortinit}{T}
      \field{sortinithash}{2e5c2f51f7fa2d957f3206819bf86dc3}
      \field{labelnamesource}{author}
      \field{labeltitlesource}{title}
      \field{title}{Lectures on Linear Logic}
      \field{year}{1992}
    \endentry
    \entry{troe:meta73}{book}{}
      \name{author}{1}{}{%
        {{hash=2bd43d4291c46b9a1f9e75cf54cccae3}{%
           family={Troelstra},
           familyi={T\bibinitperiod},
           given={A.S.},
           giveni={A\bibinitperiod}}}%
      }
      \list{location}{1}{%
        {Berlin}%
      }
      \list{publisher}{1}{%
        {Springer Verlag}%
      }
      \strng{namehash}{2bd43d4291c46b9a1f9e75cf54cccae3}
      \strng{fullhash}{2bd43d4291c46b9a1f9e75cf54cccae3}
      \strng{authornamehash}{2bd43d4291c46b9a1f9e75cf54cccae3}
      \strng{authorfullhash}{2bd43d4291c46b9a1f9e75cf54cccae3}
      \field{labelalpha}{Tro73}
      \field{sortinit}{T}
      \field{sortinithash}{2e5c2f51f7fa2d957f3206819bf86dc3}
      \field{labelnamesource}{author}
      \field{labeltitlesource}{title}
      \field{series}{Springer Lecture Notes 344}
      \field{title}{{Metamathematical investigations of intuitionistic arithmetic and analysis}}
      \field{year}{1973}
    \endentry
    \entry{troe:cons88vol1}{book}{}
      \name{author}{2}{}{%
        {{hash=2bd43d4291c46b9a1f9e75cf54cccae3}{%
           family={Troelstra},
           familyi={T\bibinitperiod},
           given={A.S.},
           giveni={A\bibinitperiod}}}%
        {{hash=69326f92073fb7fd996a20a997aed6b3}{%
           family={Dalen},
           familyi={D\bibinitperiod},
           given={D.},
           giveni={D\bibinitperiod},
           prefix={van},
           prefixi={v\bibinitperiod}}}%
      }
      \list{location}{1}{%
        {Amsterdam}%
      }
      \list{publisher}{1}{%
        {North Holland}%
      }
      \strng{namehash}{e3dd00d3e886231f2c2dbe97915b89ab}
      \strng{fullhash}{e3dd00d3e886231f2c2dbe97915b89ab}
      \strng{authornamehash}{e3dd00d3e886231f2c2dbe97915b89ab}
      \strng{authorfullhash}{e3dd00d3e886231f2c2dbe97915b89ab}
      \field{labelalpha}{TD88}
      \field{sortinit}{T}
      \field{sortinithash}{2e5c2f51f7fa2d957f3206819bf86dc3}
      \field{labelnamesource}{author}
      \field{labeltitlesource}{title}
      \field{series}{Studies in Logic and the Foundations of Mathematics}
      \field{title}{{Constructivism in Mathematics, vol 1}}
      \field{volume}{121}
      \field{year}{1988}
    \endentry
    \entry{vanatten:hypo}{inproceedings}{}
      \name{author}{1}{}{%
        {{hash=7a9755293ca2522eeb7f170b7a43f072}{%
           family={{van}\bibnamedelima Atten},
           familyi={v\bibinitperiod\bibinitdelim A\bibinitperiod},
           given={Mark},
           giveni={M\bibinitperiod}}}%
      }
      \name{editor}{1}{}{%
        {{hash=ead23300cf4a8ee4fbf5ecd6e3ee91cb}{%
           family={Clark\bibnamedelima Glymour},
           familyi={C\bibinitperiod\bibinitdelim G\bibinitperiod},
           given={Dag\bibnamedelima Westerstahl},
           giveni={D\bibinitperiod\bibinitdelim W\bibinitperiod},
           suffix={Wei\bibnamedelima Wang},
           suffixi={W\bibinitperiod\bibinitdelim W\bibinitperiod}}}%
      }
      \list{location}{1}{%
        {P\'{e}kin, China}%
      }
      \list{publisher}{1}{%
        {College Publications}%
      }
      \strng{namehash}{7a9755293ca2522eeb7f170b7a43f072}
      \strng{fullhash}{7a9755293ca2522eeb7f170b7a43f072}
      \strng{authornamehash}{7a9755293ca2522eeb7f170b7a43f072}
      \strng{authorfullhash}{7a9755293ca2522eeb7f170b7a43f072}
      \strng{editornamehash}{ead23300cf4a8ee4fbf5ecd6e3ee91cb}
      \strng{editorfullhash}{ead23300cf4a8ee4fbf5ecd6e3ee91cb}
      \field{labelalpha}{van07}
      \field{sortinit}{v}
      \field{sortinithash}{555737dafdcf1396ebfeae5822e5bde2}
      \field{labelnamesource}{author}
      \field{labeltitlesource}{title}
      \field{booktitle}{{Logic, Methodology, and Philosophy of Science XIII (LMPS XIII)}}
      \field{month}{8}
      \field{title}{{The hypothetical judgement in the history of intuitionistic logic}}
      \field{year}{2007}
      \field{pages}{662}
      \range{pages}{1}
      \verb{url}
      \verb https://halshs.archives-ouvertes.fr/halshs-00791548
      \endverb
      \keyw{philosophy of mathematics}
    \endentry
    \entry{oost:lifs90}{article}{}
      \name{author}{1}{}{%
        {{hash=af8417d1b1ccc60a4eab513861cd6547}{%
           family={{van}\bibnamedelima Oosten},
           familyi={v\bibinitperiod\bibinitdelim O\bibinitperiod},
           given={J.},
           giveni={J\bibinitperiod}}}%
      }
      \list{publisher}{1}{%
        {Cambridge University Press}%
      }
      \strng{namehash}{af8417d1b1ccc60a4eab513861cd6547}
      \strng{fullhash}{af8417d1b1ccc60a4eab513861cd6547}
      \strng{authornamehash}{af8417d1b1ccc60a4eab513861cd6547}
      \strng{authorfullhash}{af8417d1b1ccc60a4eab513861cd6547}
      \field{labelalpha}{van90}
      \field{sortinit}{v}
      \field{sortinithash}{555737dafdcf1396ebfeae5822e5bde2}
      \field{extraalpha}{1}
      \field{labelnamesource}{author}
      \field{labeltitlesource}{title}
      \field{journaltitle}{The Journal of Symbolic Logic}
      \field{number}{02}
      \field{title}{Lifschitz' Realizability}
      \field{volume}{55}
      \field{year}{1990}
      \field{pages}{805\bibrangedash 821}
      \range{pages}{17}
    \endentry
    \entry{vanStigt90}{book}{}
      \name{author}{1}{}{%
        {{hash=99dc8df7336ccd79ca735a4ba669dc5f}{%
           family={{van}\bibnamedelima Stigt},
           familyi={v\bibinitperiod\bibinitdelim S\bibinitperiod},
           given={Walter\bibnamedelima P.},
           giveni={W\bibinitperiod\bibinitdelim P\bibinitperiod}}}%
      }
      \list{publisher}{1}{%
        {North-Holland}%
      }
      \strng{namehash}{99dc8df7336ccd79ca735a4ba669dc5f}
      \strng{fullhash}{99dc8df7336ccd79ca735a4ba669dc5f}
      \strng{authornamehash}{99dc8df7336ccd79ca735a4ba669dc5f}
      \strng{authorfullhash}{99dc8df7336ccd79ca735a4ba669dc5f}
      \field{labelalpha}{van90}
      \field{sortinit}{v}
      \field{sortinithash}{555737dafdcf1396ebfeae5822e5bde2}
      \field{extraalpha}{2}
      \field{labelnamesource}{author}
      \field{labeltitlesource}{title}
      \field{title}{Brouwer's intuitionism}
      \field{year}{1990}
    \endentry
    \entry{viss:over98}{incollection}{}
      \name{author}{1}{}{%
        {{hash=7c59a81f2d67e5dbdf77cc8061473f2d}{%
           family={Visser},
           familyi={V\bibinitperiod},
           given={A.},
           giveni={A\bibinitperiod}}}%
      }
      \name{editor}{4}{}{%
        {{hash=4e054d081e86571478a1f8501899e5b9}{%
           family={Kracht},
           familyi={K\bibinitperiod},
           given={M.},
           giveni={M\bibinitperiod}}}%
        {{hash=6e150f85499ca2287cd779371757d3d9}{%
           family={Rij\-ke},
           familyi={R\bibinithyphendelim k\bibinitperiod},
           given={M.},
           giveni={M\bibinitperiod},
           prefix={de},
           prefixi={d\bibinitperiod}}}%
        {{hash=d5374c6e7930632867115a1caf5465de}{%
           family={Wansing},
           familyi={W\bibinitperiod},
           given={H.},
           giveni={H\bibinitperiod}}}%
        {{hash=190a42b73bdb11cbdaae6e480a85b179}{%
           family={Zakharyaschev},
           familyi={Z\bibinitperiod},
           given={M.},
           giveni={M\bibinitperiod}}}%
      }
      \list{location}{1}{%
        {Stanford}%
      }
      \list{publisher}{2}{%
        {Center for the Study of Language}%
        {Information}%
      }
      \strng{namehash}{7c59a81f2d67e5dbdf77cc8061473f2d}
      \strng{fullhash}{7c59a81f2d67e5dbdf77cc8061473f2d}
      \strng{authornamehash}{7c59a81f2d67e5dbdf77cc8061473f2d}
      \strng{authorfullhash}{7c59a81f2d67e5dbdf77cc8061473f2d}
      \strng{editornamehash}{c1af90830b0f11e43925218bd0665252}
      \strng{editorfullhash}{c1af90830b0f11e43925218bd0665252}
      \field{labelalpha}{Vis98}
      \field{sortinit}{V}
      \field{sortinithash}{555737dafdcf1396ebfeae5822e5bde2}
      \field{labelnamesource}{author}
      \field{labeltitlesource}{title}
      \field{booktitle}{{Advances in Modal Logic}}
      \field{series}{CSLI Lecture Notes}
      \field{title}{{An Overview of Interpretability Logic}}
      \field{volume}{1, 87}
      \field{year}{1998}
      \field{pages}{307\bibrangedash 359}
      \range{pages}{53}
    \endentry
    \entry{viss:aspe81}{book}{}
      \name{author}{1}{}{%
        {{hash=7c59a81f2d67e5dbdf77cc8061473f2d}{%
           family={Visser},
           familyi={V\bibinitperiod},
           given={A.},
           giveni={A\bibinitperiod}}}%
      }
      \list{publisher}{1}{%
        {Ph.D. Thesis, Department of Philosophy, Utrecht University}%
      }
      \strng{namehash}{7c59a81f2d67e5dbdf77cc8061473f2d}
      \strng{fullhash}{7c59a81f2d67e5dbdf77cc8061473f2d}
      \strng{authornamehash}{7c59a81f2d67e5dbdf77cc8061473f2d}
      \strng{authorfullhash}{7c59a81f2d67e5dbdf77cc8061473f2d}
      \field{labelalpha}{Vis81}
      \field{sortinit}{V}
      \field{sortinithash}{555737dafdcf1396ebfeae5822e5bde2}
      \field{labelnamesource}{author}
      \field{labeltitlesource}{title}
      \field{title}{Aspects of diagonalization and provability}
      \field{year}{1981}
    \endentry
    \entry{viss:cate06}{incollection}{}
      \name{author}{1}{}{%
        {{hash=7c59a81f2d67e5dbdf77cc8061473f2d}{%
           family={Visser},
           familyi={V\bibinitperiod},
           given={A.},
           giveni={A\bibinitperiod}}}%
      }
      \name{editor}{3}{}{%
        {{hash=1088c34076d86d95d0a3f236afcc4e42}{%
           family={Enayat},
           familyi={E\bibinitperiod},
           given={Ali},
           giveni={A\bibinitperiod}}}%
        {{hash=372133081413c1af5dd38dd553bf7710}{%
           family={Kalantari},
           familyi={K\bibinitperiod},
           given={Iraj},
           giveni={I\bibinitperiod}}}%
        {{hash=fe679e1f2c8329b249b6081e7956e3c0}{%
           family={Moniri},
           familyi={M\bibinitperiod},
           given={Mojtaba},
           giveni={M\bibinitperiod}}}%
      }
      \list{location}{1}{%
        {Wellesley, Mass.}%
      }
      \list{publisher}{1}{%
        {ASL, A.K. Peters, Ltd.}%
      }
      \strng{namehash}{7c59a81f2d67e5dbdf77cc8061473f2d}
      \strng{fullhash}{7c59a81f2d67e5dbdf77cc8061473f2d}
      \strng{authornamehash}{7c59a81f2d67e5dbdf77cc8061473f2d}
      \strng{authorfullhash}{7c59a81f2d67e5dbdf77cc8061473f2d}
      \strng{editornamehash}{ba80a8abd588033a302f27a0c717c917}
      \strng{editorfullhash}{ba80a8abd588033a302f27a0c717c917}
      \field{labelalpha}{Vis06}
      \field{sortinit}{V}
      \field{sortinithash}{555737dafdcf1396ebfeae5822e5bde2}
      \field{extraalpha}{1}
      \field{labelnamesource}{author}
      \field{labeltitlesource}{title}
      \field{booktitle}{{Logic in {T}ehran. {P}roceedings of the workshop and conference on {L}ogic, {A}lgebra and {A}rithmetic, held {O}ctober 18--22, 2003}}
      \field{series}{Lecture {N}otes in {L}ogic}
      \field{title}{Categories of {T}heories and {I}nterpretations}
      \field{volume}{26}
      \field{year}{2006}
      \field{pages}{284\bibrangedash 341}
      \range{pages}{58}
    \endentry
    \entry{viss:close08}{article}{}
      \name{author}{1}{}{%
        {{hash=7c59a81f2d67e5dbdf77cc8061473f2d}{%
           family={Visser},
           familyi={V\bibinitperiod},
           given={A.},
           giveni={A\bibinitperiod}}}%
      }
      \strng{namehash}{7c59a81f2d67e5dbdf77cc8061473f2d}
      \strng{fullhash}{7c59a81f2d67e5dbdf77cc8061473f2d}
      \strng{authornamehash}{7c59a81f2d67e5dbdf77cc8061473f2d}
      \strng{authorfullhash}{7c59a81f2d67e5dbdf77cc8061473f2d}
      \field{labelalpha}{Vis08}
      \field{sortinit}{V}
      \field{sortinithash}{555737dafdcf1396ebfeae5822e5bde2}
      \field{labelnamesource}{author}
      \field{labeltitlesource}{title}
      \field{journaltitle}{Journal of Symbolic Logic}
      \field{number}{3}
      \field{title}{Closed Fragments of Provability Logics of Constructive Theories}
      \field{volume}{73}
      \field{year}{2008}
      \field{pages}{1081\bibrangedash 1096}
      \range{pages}{16}
    \endentry
    \entry{viss:eval85}{report}{}
      \name{author}{1}{}{%
        {{hash=7c59a81f2d67e5dbdf77cc8061473f2d}{%
           family={Visser},
           familyi={V\bibinitperiod},
           given={A.},
           giveni={A\bibinitperiod}}}%
      }
      \list{institution}{1}{%
        {Faculty of Humanities, Philosophy, Utrecht University}%
      }
      \list{location}{1}{%
        {Janskerkhof 13, 3512 BL Utrecht}%
      }
      \strng{namehash}{7c59a81f2d67e5dbdf77cc8061473f2d}
      \strng{fullhash}{7c59a81f2d67e5dbdf77cc8061473f2d}
      \strng{authornamehash}{7c59a81f2d67e5dbdf77cc8061473f2d}
      \strng{authorfullhash}{7c59a81f2d67e5dbdf77cc8061473f2d}
      \field{labelalpha}{Vis85}
      \field{sortinit}{V}
      \field{sortinithash}{555737dafdcf1396ebfeae5822e5bde2}
      \field{labelnamesource}{author}
      \field{labeltitlesource}{title}
      \field{number}{4}
      \field{title}{Evaluation, provably deductive equivalence in {H}eyting's {A}rithmetic of substitution instances of propositional formulas}
      \field{type}{Logic Group Preprint Series}
      \field{year}{1985}
    \endentry
    \entry{viss:inte90}{inproceedings}{}
      \name{author}{1}{}{%
        {{hash=7c59a81f2d67e5dbdf77cc8061473f2d}{%
           family={Visser},
           familyi={V\bibinitperiod},
           given={A.},
           giveni={A\bibinitperiod}}}%
      }
      \name{editor}{1}{}{%
        {{hash=6481cde4d8cb9e1508c2c95dec87911a}{%
           family={Petkov},
           familyi={P\bibinitperiod},
           given={P.P.},
           giveni={P\bibinitperiod}}}%
      }
      \list{publisher}{1}{%
        {Plenum Press, Boston}%
      }
      \strng{namehash}{7c59a81f2d67e5dbdf77cc8061473f2d}
      \strng{fullhash}{7c59a81f2d67e5dbdf77cc8061473f2d}
      \strng{authornamehash}{7c59a81f2d67e5dbdf77cc8061473f2d}
      \strng{authorfullhash}{7c59a81f2d67e5dbdf77cc8061473f2d}
      \strng{editornamehash}{6481cde4d8cb9e1508c2c95dec87911a}
      \strng{editorfullhash}{6481cde4d8cb9e1508c2c95dec87911a}
      \field{labelalpha}{Vis90}
      \field{sortinit}{V}
      \field{sortinithash}{555737dafdcf1396ebfeae5822e5bde2}
      \field{labelnamesource}{author}
      \field{labeltitlesource}{title}
      \field{booktitle}{Mathematical logic, {P}roceedings of the {H}eyting 1988 summer school in {V}arna, {B}ulgaria}
      \field{title}{Interpretability logic}
      \field{year}{1990}
      \field{pages}{175\bibrangedash 209}
      \range{pages}{35}
    \endentry
    \entry{viss:comp82}{article}{}
      \name{author}{1}{}{%
        {{hash=7c59a81f2d67e5dbdf77cc8061473f2d}{%
           family={Visser},
           familyi={V\bibinitperiod},
           given={A.},
           giveni={A\bibinitperiod}}}%
      }
      \list{publisher}{1}{%
        {Elsevier}%
      }
      \strng{namehash}{7c59a81f2d67e5dbdf77cc8061473f2d}
      \strng{fullhash}{7c59a81f2d67e5dbdf77cc8061473f2d}
      \strng{authornamehash}{7c59a81f2d67e5dbdf77cc8061473f2d}
      \strng{authorfullhash}{7c59a81f2d67e5dbdf77cc8061473f2d}
      \field{labelalpha}{Vis82}
      \field{sortinit}{V}
      \field{sortinithash}{555737dafdcf1396ebfeae5822e5bde2}
      \field{labelnamesource}{author}
      \field{labeltitlesource}{title}
      \field{journaltitle}{Annals of Mathematical Logic}
      \field{number}{3}
      \field{title}{On the completeness principle: A study of provability in {H}eyting's arithmetic and extensions}
      \field{volume}{22}
      \field{year}{1982}
      \field{pages}{263\bibrangedash 295}
      \range{pages}{33}
    \endentry
    \entry{viss:pred06}{article}{}
      \name{author}{1}{}{%
        {{hash=7c59a81f2d67e5dbdf77cc8061473f2d}{%
           family={Visser},
           familyi={V\bibinitperiod},
           given={A.},
           giveni={A\bibinitperiod}}}%
      }
      \strng{namehash}{7c59a81f2d67e5dbdf77cc8061473f2d}
      \strng{fullhash}{7c59a81f2d67e5dbdf77cc8061473f2d}
      \strng{authornamehash}{7c59a81f2d67e5dbdf77cc8061473f2d}
      \strng{authorfullhash}{7c59a81f2d67e5dbdf77cc8061473f2d}
      \field{labelalpha}{Vis06}
      \field{sortinit}{V}
      \field{sortinithash}{555737dafdcf1396ebfeae5822e5bde2}
      \field{extraalpha}{2}
      \field{labelnamesource}{author}
      \field{labeltitlesource}{title}
      \field{journaltitle}{Journal of Symbolic Logic}
      \field{number}{4}
      \field{title}{Predicate logics of constructive arithmetical theories}
      \field{volume}{71}
      \field{year}{2006}
      \field{pages}{1311\bibrangedash 1326}
      \range{pages}{16}
    \endentry
    \entry{viss:prop94}{book}{}
      \name{author}{1}{}{%
        {{hash=7c59a81f2d67e5dbdf77cc8061473f2d}{%
           family={Visser},
           familyi={V\bibinitperiod},
           given={A.},
           giveni={A\bibinitperiod}}}%
      }
      \list{location}{1}{%
        {Janskerkhof 13, 3512 BL Utrecht}%
      }
      \list{publisher}{1}{%
        {Faculty of Humanities, Philosophy, Utrecht University}%
      }
      \strng{namehash}{7c59a81f2d67e5dbdf77cc8061473f2d}
      \strng{fullhash}{7c59a81f2d67e5dbdf77cc8061473f2d}
      \strng{authornamehash}{7c59a81f2d67e5dbdf77cc8061473f2d}
      \strng{authorfullhash}{7c59a81f2d67e5dbdf77cc8061473f2d}
      \field{labelalpha}{Vis94}
      \field{sortinit}{V}
      \field{sortinithash}{555737dafdcf1396ebfeae5822e5bde2}
      \field{labelnamesource}{author}
      \field{labeltitlesource}{title}
      \field{series}{Logic Group Preprint Series 117}
      \field{title}{{Propositional combinations of $\Sigma$-sentences in Heyting's Arithmetic}}
      \field{year}{1994}
    \endentry
    \entry{viss:rule99}{article}{}
      \name{author}{1}{}{%
        {{hash=7c59a81f2d67e5dbdf77cc8061473f2d}{%
           family={Visser},
           familyi={V\bibinitperiod},
           given={A.},
           giveni={A\bibinitperiod}}}%
      }
      \strng{namehash}{7c59a81f2d67e5dbdf77cc8061473f2d}
      \strng{fullhash}{7c59a81f2d67e5dbdf77cc8061473f2d}
      \strng{authornamehash}{7c59a81f2d67e5dbdf77cc8061473f2d}
      \strng{authorfullhash}{7c59a81f2d67e5dbdf77cc8061473f2d}
      \field{labelalpha}{Vis99}
      \field{sortinit}{V}
      \field{sortinithash}{555737dafdcf1396ebfeae5822e5bde2}
      \field{labelnamesource}{author}
      \field{labeltitlesource}{title}
      \field{journaltitle}{Notre Dame Journal of Formal Logic}
      \field{number}{1}
      \field{title}{{R}ules and {A}rithmetics}
      \field{volume}{40}
      \field{year}{1999}
      \field{pages}{116\bibrangedash 140}
      \range{pages}{25}
    \endentry
    \entry{viss:subs02}{article}{}
      \name{author}{1}{}{%
        {{hash=7c59a81f2d67e5dbdf77cc8061473f2d}{%
           family={Visser},
           familyi={V\bibinitperiod},
           given={A.},
           giveni={A\bibinitperiod}}}%
      }
      \strng{namehash}{7c59a81f2d67e5dbdf77cc8061473f2d}
      \strng{fullhash}{7c59a81f2d67e5dbdf77cc8061473f2d}
      \strng{authornamehash}{7c59a81f2d67e5dbdf77cc8061473f2d}
      \strng{authorfullhash}{7c59a81f2d67e5dbdf77cc8061473f2d}
      \field{labelalpha}{Vis02}
      \field{sortinit}{V}
      \field{sortinithash}{555737dafdcf1396ebfeae5822e5bde2}
      \field{labelnamesource}{author}
      \field{labeltitlesource}{title}
      \field{journaltitle}{Annals of Pure and Applied Logic}
      \field{title}{Substitutions of ${\Sigma}^0_1$-sentences: explorations between intuitionistic propositional logic and intuitionistic arithmetic}
      \field{volume}{114}
      \field{year}{2002}
      \field{pages}{227\bibrangedash 271}
      \range{pages}{45}
    \endentry
    \entry{viss:whyR14}{incollection}{}
      \name{author}{1}{}{%
        {{hash=7c59a81f2d67e5dbdf77cc8061473f2d}{%
           family={Visser},
           familyi={V\bibinitperiod},
           given={A.},
           giveni={A\bibinitperiod}}}%
      }
      \name{editor}{1}{}{%
        {{hash=d888e2237d42048d285cf82c7791008a}{%
           family={Tennant},
           familyi={T\bibinitperiod},
           given={Neil},
           giveni={N\bibinitperiod}}}%
      }
      \list{location}{1}{%
        {UK}%
      }
      \list{publisher}{1}{%
        {College Publications}%
      }
      \strng{namehash}{7c59a81f2d67e5dbdf77cc8061473f2d}
      \strng{fullhash}{7c59a81f2d67e5dbdf77cc8061473f2d}
      \strng{authornamehash}{7c59a81f2d67e5dbdf77cc8061473f2d}
      \strng{authorfullhash}{7c59a81f2d67e5dbdf77cc8061473f2d}
      \strng{editornamehash}{d888e2237d42048d285cf82c7791008a}
      \strng{editorfullhash}{d888e2237d42048d285cf82c7791008a}
      \field{labelalpha}{Vis14}
      \field{sortinit}{V}
      \field{sortinithash}{555737dafdcf1396ebfeae5822e5bde2}
      \field{labelnamesource}{author}
      \field{labeltitlesource}{title}
      \field{booktitle}{{Foundational Adventures. Essays in honour of Harvey Friedman}}
      \field{note}{Originally published online by Templeton Press in 2012. See {\tt http://foundationaladventures.com/}}
      \field{title}{Why the theory {{\sf R}} is special}
      \field{year}{2014}
      \field{pages}{7\bibrangedash 23}
      \range{pages}{17}
    \endentry
    \entry{viss:nnil95}{incollection}{}
      \name{author}{4}{}{%
        {{hash=7c59a81f2d67e5dbdf77cc8061473f2d}{%
           family={Visser},
           familyi={V\bibinitperiod},
           given={A.},
           giveni={A\bibinitperiod}}}%
        {{hash=ad67e3d60630c2c29ee93f7aea1ca4dc}{%
           family={Benthem},
           familyi={B\bibinitperiod},
           given={J.},
           giveni={J\bibinitperiod},
           prefix={van},
           prefixi={v\bibinitperiod}}}%
        {{hash=af475dcf6f4bea11895ee84cab19a816}{%
           family={Jongh},
           familyi={J\bibinitperiod},
           given={D.},
           giveni={D\bibinitperiod},
           prefix={de},
           prefixi={d\bibinitperiod}}}%
        {{hash=6d570020bf61a3d038387be81d4d78c6}{%
           family={Lavalette},
           familyi={L\bibinitperiod},
           given={G.\bibnamedelimi Renardel},
           giveni={G\bibinitperiod\bibinitdelim R\bibinitperiod},
           prefix={de},
           prefixi={d\bibinitperiod}}}%
      }
      \name{editor}{3}{}{%
        {{hash=1c0f457d2bde22cad3898a3f6629c24b}{%
           family={Ponse},
           familyi={P\bibinitperiod},
           given={A.},
           giveni={A\bibinitperiod}}}%
        {{hash=6e150f85499ca2287cd779371757d3d9}{%
           family={Rij\-ke},
           familyi={R\bibinithyphendelim k\bibinitperiod},
           given={M.},
           giveni={M\bibinitperiod},
           prefix={de},
           prefixi={d\bibinitperiod}}}%
        {{hash=1330d46499ee334061b28a3cce7e0225}{%
           family={Venema},
           familyi={V\bibinitperiod},
           given={Y.},
           giveni={Y\bibinitperiod}}}%
      }
      \list{location}{1}{%
        {Stanford}%
      }
      \list{publisher}{2}{%
        {Center for the Study of Language}%
        {Information}%
      }
      \strng{namehash}{0394f2a10b33906961a49fe3df52fc26}
      \strng{fullhash}{0394f2a10b33906961a49fe3df52fc26}
      \strng{authornamehash}{0394f2a10b33906961a49fe3df52fc26}
      \strng{authorfullhash}{0394f2a10b33906961a49fe3df52fc26}
      \strng{editornamehash}{7e433ebff0a675fcf050815e5178430b}
      \strng{editorfullhash}{7e433ebff0a675fcf050815e5178430b}
      \field{labelalpha}{VBJL95}
      \field{sortinit}{V}
      \field{sortinithash}{555737dafdcf1396ebfeae5822e5bde2}
      \field{labelnamesource}{author}
      \field{labeltitlesource}{title}
      \field{booktitle}{{Modal Logic and Process Algebra, a Bisimulation Perspective}}
      \field{series}{CSLI Lecture Notes, no. 53}
      \field{title}{{{N}{N}{I}{L}, a Study in Intuitionistic Propositional Logic}}
      \field{year}{1995}
      \field{pages}{289\bibrangedash 326}
      \range{pages}{38}
    \endentry
    \entry{Williamson92:jpl}{article}{}
      \name{author}{1}{}{%
        {{hash=3c5a6d2d74c469ca17ec2d2e79d66431}{%
           family={Williamson},
           familyi={W\bibinitperiod},
           given={Timothy},
           giveni={T\bibinitperiod}}}%
      }
      \list{publisher}{1}{%
        {Springer}%
      }
      \strng{namehash}{3c5a6d2d74c469ca17ec2d2e79d66431}
      \strng{fullhash}{3c5a6d2d74c469ca17ec2d2e79d66431}
      \strng{authornamehash}{3c5a6d2d74c469ca17ec2d2e79d66431}
      \strng{authorfullhash}{3c5a6d2d74c469ca17ec2d2e79d66431}
      \field{labelalpha}{Wil92}
      \field{sortinit}{W}
      \field{sortinithash}{6d25b3eefe5aa2147d1f339686808918}
      \field{labelnamesource}{author}
      \field{labeltitlesource}{title}
      \field{issn}{00223611, 15730433}
      \field{journaltitle}{Journal of Philosophical Logic}
      \field{number}{1}
      \field{title}{On Intuitionistic Modal Epistemic Logic}
      \field{volume}{21}
      \field{year}{1992}
      \field{pages}{63\bibrangedash 89}
      \range{pages}{27}
      \verb{url}
      \verb http://www.jstor.org/stable/30226465
      \endverb
    \endentry
    \entry{Wolter1993}{thesis}{}
      \name{author}{1}{}{%
        {{hash=5f6ef4ecaac2c898bee8ef9429ca6990}{%
           family={Wolter},
           familyi={W\bibinitperiod},
           given={F.},
           giveni={F\bibinitperiod}}}%
      }
      \list{institution}{1}{%
        {Fachbereich Mathematik, Freien Universit\"{a}t Berlin}%
      }
      \strng{namehash}{5f6ef4ecaac2c898bee8ef9429ca6990}
      \strng{fullhash}{5f6ef4ecaac2c898bee8ef9429ca6990}
      \strng{authornamehash}{5f6ef4ecaac2c898bee8ef9429ca6990}
      \strng{authorfullhash}{5f6ef4ecaac2c898bee8ef9429ca6990}
      \field{labelalpha}{Wol93}
      \field{sortinit}{W}
      \field{sortinithash}{6d25b3eefe5aa2147d1f339686808918}
      \field{labelnamesource}{author}
      \field{labeltitlesource}{title}
      \field{title}{Lattices of Modal Logics}
      \field{type}{phdthesis}
      \field{year}{1993}
    \endentry
    \entry{WolterZ98:lw}{inproceedings}{}
      \name{author}{2}{}{%
        {{hash=feac505bad30d5bd2b8e65b72656d06e}{%
           family={Wolter},
           familyi={W\bibinitperiod},
           given={Frank},
           giveni={F\bibinitperiod}}}%
        {{hash=76b868a2302301df0809621d25537a21}{%
           family={Zakharyaschev},
           familyi={Z\bibinitperiod},
           given={Michael},
           giveni={M\bibinitperiod}}}%
      }
      \name{editor}{1}{}{%
        {{hash=7a1c8ae374b655439d5c37470ded31f1}{%
           family={Orlowska},
           familyi={O\bibinitperiod},
           given={Ewa},
           giveni={E\bibinitperiod}}}%
      }
      \list{publisher}{1}{%
        {Springer--Verlag}%
      }
      \strng{namehash}{a574e472866842abeb4b08cc39909096}
      \strng{fullhash}{a574e472866842abeb4b08cc39909096}
      \strng{authornamehash}{a574e472866842abeb4b08cc39909096}
      \strng{authorfullhash}{a574e472866842abeb4b08cc39909096}
      \strng{editornamehash}{7a1c8ae374b655439d5c37470ded31f1}
      \strng{editorfullhash}{7a1c8ae374b655439d5c37470ded31f1}
      \field{labelalpha}{WZ98}
      \field{sortinit}{W}
      \field{sortinithash}{6d25b3eefe5aa2147d1f339686808918}
      \field{labelnamesource}{author}
      \field{labeltitlesource}{title}
      \field{booktitle}{Logic at Work, Essays in honour of Helena Rasiowa}
      \field{title}{Intuitionistic Modal Logics as fragments of Classical Bimodal Logics}
      \field{year}{1998}
      \field{pages}{168\bibrangedash 186}
      \range{pages}{19}
    \endentry
    \entry{WolterZ97:al}{article}{}
      \name{author}{2}{}{%
        {{hash=feac505bad30d5bd2b8e65b72656d06e}{%
           family={Wolter},
           familyi={W\bibinitperiod},
           given={Frank},
           giveni={F\bibinitperiod}}}%
        {{hash=76b868a2302301df0809621d25537a21}{%
           family={Zakharyaschev},
           familyi={Z\bibinitperiod},
           given={Michael},
           giveni={M\bibinitperiod}}}%
      }
      \strng{namehash}{a574e472866842abeb4b08cc39909096}
      \strng{fullhash}{a574e472866842abeb4b08cc39909096}
      \strng{authornamehash}{a574e472866842abeb4b08cc39909096}
      \strng{authorfullhash}{a574e472866842abeb4b08cc39909096}
      \field{labelalpha}{WZ97}
      \field{sortinit}{W}
      \field{sortinithash}{6d25b3eefe5aa2147d1f339686808918}
      \field{labelnamesource}{author}
      \field{labeltitlesource}{title}
      \field{journaltitle}{Algebra and Logic}
      \field{title}{On the relation between intuitionistic and classical modal logics}
      \field{volume}{36}
      \field{year}{1997}
      \field{pages}{121\bibrangedash 125}
      \range{pages}{5}
    \endentry
    \entry{yang:intu08}{thesis}{}
      \name{author}{1}{}{%
        {{hash=eb5d5ae09823dbd21f6e0bbbf8a1b5ba}{%
           family={Yang},
           familyi={Y\bibinitperiod},
           given={F.},
           giveni={F\bibinitperiod}}}%
      }
      \list{institution}{1}{%
        {Universiteit van Amsterdam}%
      }
      \strng{namehash}{eb5d5ae09823dbd21f6e0bbbf8a1b5ba}
      \strng{fullhash}{eb5d5ae09823dbd21f6e0bbbf8a1b5ba}
      \strng{authornamehash}{eb5d5ae09823dbd21f6e0bbbf8a1b5ba}
      \strng{authorfullhash}{eb5d5ae09823dbd21f6e0bbbf8a1b5ba}
      \field{labelalpha}{Yan08}
      \field{sortinit}{Y}
      \field{sortinithash}{f86c10978eb0c0879609ce6a4f43da7b}
      \field{labelnamesource}{author}
      \field{labeltitlesource}{title}
      \field{title}{Intuitionistic subframe formulas, {N}{N}{I}{L}-formulas and $n$-universal models}
      \field{type}{mathesis}
      \field{year}{2008}
    \endentry
    \entry{zamb:shav94}{article}{}
      \name{author}{1}{}{%
        {{hash=e90f5421a60a2db4f68903aade76d9ee}{%
           family={Zambella},
           familyi={Z\bibinitperiod},
           given={D.},
           giveni={D\bibinitperiod}}}%
      }
      \strng{namehash}{e90f5421a60a2db4f68903aade76d9ee}
      \strng{fullhash}{e90f5421a60a2db4f68903aade76d9ee}
      \strng{authornamehash}{e90f5421a60a2db4f68903aade76d9ee}
      \strng{authorfullhash}{e90f5421a60a2db4f68903aade76d9ee}
      \field{labelalpha}{Zam94}
      \field{sortinit}{Z}
      \field{sortinithash}{35589aa085e881766b72503e53fd4c97}
      \field{labelnamesource}{author}
      \field{labeltitlesource}{title}
      \field{journaltitle}{Notre Dame Journal of Formal Logic}
      \field{title}{Shavrukov's theorem on the subalgebras of diagonalizable algebras for theories containing {$I\Delta_0+{\sf EXP}$}}
      \field{volume}{35}
      \field{year}{1994}
      \field{pages}{147\bibrangedash 157}
      \range{pages}{11}
    \endentry
    \entry{Zhou03}{thesis}{}
      \name{author}{1}{}{%
        {{hash=4b035c791adad7a574cf1ce874b4d72a}{%
           family={Zhou},
           familyi={Z\bibinitperiod},
           given={Chunlai},
           giveni={C\bibinitperiod}}}%
      }
      \list{institution}{1}{%
        {ILLC, University of Amsterdam}%
      }
      \strng{namehash}{4b035c791adad7a574cf1ce874b4d72a}
      \strng{fullhash}{4b035c791adad7a574cf1ce874b4d72a}
      \strng{authornamehash}{4b035c791adad7a574cf1ce874b4d72a}
      \strng{authorfullhash}{4b035c791adad7a574cf1ce874b4d72a}
      \field{labelalpha}{Zho03}
      \field{sortinit}{Z}
      \field{sortinithash}{35589aa085e881766b72503e53fd4c97}
      \field{labelnamesource}{author}
      \field{labeltitlesource}{title}
      \field{title}{Some Intuitionistic Provability and Preservativity Logics (and their interrelations)}
      \field{type}{mathesis}
      \field{year}{2003}
    \endentry
  \endsortlist
\endrefsection

  \blx@bblend
  \endgroup
  \csnumgdef{blx@labelnumber@\the\c@refsection}{0}}
\makeatother
\else\fi

\usepackage[all]{xy}
\xyoption{2cell}
\xyoption{curve}
\UseTwocells
\SelectTips{cm}{}

\newcommand{\lna}[1]{\ensuremath{\mathsf{#1}}} %logic name
\newcommand{\lva}[1]{\ensuremath{\mathcal{#1}}} %logic variable
\newcommand{\ha}{\lna{HA}} %HA
\newcommand{\cpc}{\lna{CPC}}
\newcommand{\ipc}{\lna{IPC}}

\def\tbskip{2mm}

\usepackage[usenames]{xcolor}

\DeclareSymbolFont{symbolsC}{U}{txsyc}{m}{n}
\DeclareMathSymbol{\strictif}{\mathrel}{symbolsC}{74}
\DeclareMathSymbol{\strictfi}{\mathrel}{symbolsC}{75}
\DeclareMathSymbol{\strictiff}{\mathrel}{symbolsC}{76}

\definecolor{airforceblue}{rgb}{0.36, 0.54, 0.66}
\definecolor{brickred}{rgb}{0.8, 0.25, 0.33}
\definecolor{ao}{rgb}{0.0, 0.0, 1.0}
\definecolor{cobalt}{rgb}{0.0, 0.28, 0.67}

\newcommand{\tto}{\strictif}
\newcommand{\ifff}{\strictiff}

\newcommand{\kmodels}{\Vdash}
\newcommand{\arbop}{\circledcirc}
\newcommand{\gnumm}[1]{{\ulcorner #1 \urcorner}}
\newcommand{\gnum}[1]{\underline{\ulcorner #1 \urcorner}} 
\newcommand{\jump}{\mathrel{\mbox{\textcolor{gray}{$\blacktriangleright$}}}}

\newcommand{\iea}{\ensuremath{{\mathrm i}\hyph{\sf EA}}}
\newcommand{\InF}[2]{{\sf F}_{#1,#2}}
\DeclareMathSymbol{\boxright}{\mathrel}{symbolsC}{128}

\newcommand{\nosmurf}{\noindent}
\newcommand{\nosmurfduo}{\medskip\noindent}
\usepackage{amssymb}
\usepackage{latexsym}
\usepackage{amsmath}
\usepackage{enumerate}
\usepackage{stmaryrd}
\usepackage{amsthm}
\newcommand{\shapipe}{\,
                    \setlength{\unitlength}{1ex}
                    \begin{picture}(2,2)
                    \put(.25,.15){\line(0,1){1.22}}
                    \put(.25,1.37){\line(1,0){0.61}}
                    \put(.45,-.05){\line(1,0){0.61}}
                    \put(.4,0){\line(1,0){0.61}}
                    \put(.35,.05){\line(1,0){0.61}}
                    \put(.3,.1){\line(1,0){0.61}}
                    \put(.25,.15){\line(1,0){0.61}}
                    \put(1.04,-.05){\line(0,1){1.22}}
                    \put(0.99,0){\line(0,1){1.22}}
                    \put(0.94,.05){\line(0,1){1.22}}
                    \put(0.89,.1){\line(0,1){1.22}}
                    \put(0.84,.15){\line(0,1){1.22}}
                    \end{picture}
                    \!}

\newcommand{\qee} {\hspace*{2mm}\hfill $\shapipe$}

\newtheorem{theorem}{Theorem}[section]
\newtheorem{define}[theorem]{Definition}
\newenvironment{definition}{\begin{define} \rm}{\qee\end{define}}
\newtheorem{exa}[theorem]{Example}
\newenvironment{example}{\begin{exa} \rm}{\qee\end{exa}}
\newtheorem{propo}[theorem]{Proposition}

\newtheorem{exerc}[theorem]{Exercise}

\newtheorem{conj}[theorem]{Conjecture}

\newtheorem{ques}[theorem]{Open Question}
\newenvironment{question}{\begin{ques} \rm}{\qee\end{ques}}
\newtheorem{lem}[theorem]{Lemma}
\newenvironment{lemma}{\begin{lem} \it}{\end{lem}}
\newtheorem{cor}[theorem]{Corollary}
\newenvironment{corollary}{\begin{cor} \it}{\end{cor}}
\newtheorem{factief}[theorem]{Fact}
\newenvironment{fact}{\begin{factief} \it}{\end{factief}}
\newtheorem{rem}[theorem]{Remark}
\newenvironment{remark}{\begin{rem} \rm}{\qee\end{rem}}
\newtheorem{dis}[theorem]{Discussion}

\newcommand{\medent}{\medskip\noindent}

  \newcommand{\tupel}[1]{{\langle #1 \rangle}}

\newcommand{\hyph}{\mbox{-}}

\newcommand{\To}{\Rightarrow}

\newcommand{\vvdash}{\mathrel{{\setlength{\unitlength}{1ex}
                     {\begin{picture}(2.6,2)(-.1,0)
                     \thinlines
                     \put(0.9,-.05){\line(0,1){1.65}}
                     \put(0.75,0.20){{\small ${\sim}$}}
                     \end{picture}}}}}

\renewcommand{\iff}{\leftrightarrow}

\newcommand\Item[1][]{%
  \ifx\relax#1\relax  \item \else \item[#1] \fi
  \abovedisplayskip=0pt\abovedisplayshortskip=0pt~\vspace*{-\baselineskip}}

\newcommand{\takeout}[1]{}

\usepackage[draft]{fixme} %inline
\FXRegisterAuthor{av}{aav}{AV}%Albert
\FXRegisterAuthor{tl}{atl}{TL}%Tadeusz
\newcommand{\tadeusz}[1]{\tlnote[inline,marginclue]{\textcolor{purple}{#1}}}

\definecolor{darkgreen}{rgb}{0,0.4,0}
\definecolor{darkergreen}{rgb}{0.1,0.28,0}
\newcommand{\tdiscussed}{\tadeusz{Already discussed}\textcolor{darkgreen}}
\newcommand{\tmoved}{}%{\tadeusz{To be moved earlier?}\textcolor{darkergreen}}

\usepackage{environ}

\NewEnviron{discussed}{
\tdiscussed
\BODY
}
\NewEnviron{tomove}{
\tmoved
\BODY
}

\NewEnviron{deriproof}        
{
\ifderiv
\begin{proof}
\BODY
\end{proof}
\else\fi
}

\usepackage{mdframed}

\NewEnviron{smquote}
{
\small
\begin{quote}
\small
\BODY
\end{quote}
}

\NewEnviron{foots}
{
\begin{mdframed}[roundcorner=10pt,backgroundcolor=gray!32,linecolor=blue]
\small
\BODY
\end{mdframed}
}

\NewEnviron{quest}
{
\begin{mdframed}[roundcorner=10pt,backgroundcolor=yellow!15,innertopmargin=0pt,innerbottommargin=10pt,linecolor=blue]
\begin{question}
\BODY
\end{question}
\end{mdframed}
}

\newcommand{\Section}{\S}
\newcommand{\Subsection}{\S}
\newcommand{\Subsubsection}{\S}

\newcommand{\deq}{:=}

\def\refeq#1{(\ref{#1})}

\newcommand{\ded}{\vdash}
\newcommand{\dedp}[1]{\vdash_{#1}}
\newcommand{\dedm}{\vdash_{-}}
\newcommand{\dedz}{\vdash_{0}}
\newcommand{\dedv}{\vdash_\vee}
\newcommand{\dedb}{\dedp{\Box}}

\newcommand{\inff}{\mid\!\strictif}

\newcommand{\dedlp}[1]{\inff_{#1}}
\newcommand{\dedlm}{\dedlp{-}}
\newcommand{\dedlz}{\dedlp{0}}
\newcommand{\dedlv}{\dedlp{\vee}}

\newcommand{\eqd}{\dashv\vdash}
\newcommand{\eqdp}[1]{\dashv\vdash_{#1}}
\newcommand{\eqdm}{\eqdp{-}}
\newcommand{\eqdz}{\eqdp{0}}
\newcommand{\eqdv}{\eqdp{\vee}}
\newcommand{\eqdb}{\eqdp{\Box}}

\newcommand{\eqv}{\leftrightarrow}

\newcommand{\ands}{\,\&\,}
\newcommand{\upclo}[2]{#1\!\uparrow_{#2}\,}

\newcommand{\la}{\langle}
\newcommand{\ra}{\rangle}

\newcommand{\bro}{\textup{(}}
\newcommand{\brc}{\textup{)\,}}

\newcommand{\ma}[1]{\mathcal{#1}}

\newcommand{\Mod}{\mathsf{Fram}}

\newcommand{\tuc}[2]{\textup{\cite[#1]{#2}}}

\newcommand{\bnref}[1]{\hyperref[fig:boxax]{\ensuremath{#1}}}

\newcommand{\bi}{\bnref{{\sf N}_\Box}} %  b for box and i the number in roman numerals
\newcommand{\bii}{\bnref{{\sf K}_\Box}}
\newcommand{\biii}{\bnref{4_\Box}}
\newcommand{\biv}{\bnref{{\sf L}_\Box}}
\newcommand{\bv}{\bnref{{\sf Lei}}}
\newcommand{\bvi}{\lS}%\lnref{{\sf S}}
\newcommand{\bvii}{\bnref{{\sf SL}_\Box}}

\newcommand{\logba}{\bnref{{\mathrm i}\hyph{\sf GL}_\Box}}
\newcommand{\logbb}{\bnref{{\mathrm c}\hyph{\sf GL}_\Box}}

\newcommand{\lnref}[1]{\hyperref[tab:mainax]{\ensuremath{#1}}}

\newcommand{\subl}{{\sf a}}
\newcommand{\li}{\lnref{{\sf N}_\subl}} % gl for global, l for lewis and i the number in roman numerals
\newcommand{\lii}{\lnref{{\sf Tr}}}
\newcommand{\liii}{\lnref{{\sf K}_{\subl}}}
\newcommand{\liiialt}{\lnref{{\sf K}''_{\subl}}}
\newcommand{\liiialtalt}{\lnref{{\sf K}'_{\subl}}}
\newcommand{\liiialtaltalt}{\lnref{{\sf K}'''_{\subl}}}
\newcommand{\liv}{\lnref{{\sf Di}}}
\newcommand{\livalt}{\lnref{{\sf Di}'}}
\newcommand{\lv}{\lnref{4_{\subl}}}
\newcommand{\lvi}{\lnref{{\sf L}_{\subl}}}
\newcommand{\lvii}{\lnref{{\sf W}_{\subl}}}
\newcommand{\lviialt}{\lnref{{\sf W}'_{\subl}}}
\newcommand{\lviii}{\lnref{{\sf M}_{\subl}}}
\newcommand{\lviiialt}{\lnref{{\sf M}'_{\subl}}}
\newcommand{\lix}{\lnref{{\sf Box'}}}
\newcommand{\lx}{\lnref{\sf BL}}
\newcommand{\lxi}{\lnref{\sf LB}}
\newcommand{\lxii}{\lnref{{\sf P}_{\subl}}}

\newcommand{\loglzero}{\lnref{{\sf {\mathrm i}A}_0}}
\newcommand{\loglzeromin}{\lnref{{\sf {\mathrm i}A}^-_0}}
\newcommand{\logla}{\lnref{{\sf {\mathrm i}A}^{-}}}
\newcommand{\loglb}{\lnref{{\sf {\mathrm i}A}}}
\newcommand{\loglc}{{\lnref{{\sf {\mathrm i}\hyph PreL}}}}
\newcommand{\logld}{\lnref{{\mathrm i}\hyph{\sf PreL}^{-}}}
\newcommand{\logle}{\lnref{{\mathrm i}\hyph{\sf GL}^{-}_{\subl}}}
\newcommand{\loglf}{\lnref{{\mathrm i}\hyph{\sf GW}^{-}_{\subl}}}
\newcommand{\loglg}{\lnref{{\mathrm i}\hyph{\sf GL}_{\subl}}}
\newcommand{\loglh}{\lnref{{\mathrm i}\hyph{\sf GW}_{\subl}}}

\newcommand{\iK}{\bnref{\isys{K}_\Box}}

\newcommand{\iS}{\bnref{\isys{S}_\Box}}

\newcommand{\CF}{\bnref{\lna{C4}_\Box}}

\newcommand{\iPLL}{\bnref{\isys{PLL}_\Box}}

\newcommand{\Di}{\liv}

\newcommand{\iA}{\iP}

\newcommand{\mHCb}{\bnref{\lna{CB}_\Box}}
\newcommand{\mHCbalt}{\bnref{\lna{CB'_\Box}}}
\newcommand{\imHCb}{\bnref{\isys{mHC}_\Box}}

\newcommand{\mHCl}{\lnref{\lna{CB}_\subl}}
\newcommand{\imHCl}{\lnref{\isys{mHC}_\subl}}

\newcommand{\aLin}{\lnref{\lna{Lin}_\subl}}
\newcommand{\bLin}{\bnref{\lna{Lin}_\Box}}

\newcommand{\aApp}{\lnref{\lna{App}_\subl}}

\newcommand{\aCF}{\lnref{\lna{C4}_\subl}}

\newcommand{\iPLLa}{\lnref{\isys{PLAA}}}

\newcommand{\Hug}{\lnref{\lna{Hug}}}

\newcommand{\iSLb}{\bnref{\isys{SL}_\Box}}

\newcommand{\iKMb}{\bnref{\isys{KM}_\Box}}

\newcommand{\iKMl}{\bnref{\isys{KM}_\subl}}

\newcommand{\iKMlin}{\bnref{\isys{KM.lin}_\subl}}
\newcommand{\iKMlinb}{\bnref{\isys{KM.lin}_\Box}}

\newcommand{\mix}{mix\,}
\newcommand{\strictp}{\ensuremath{\strictif}-p\,}
\newcommand{\boxp}{\ensuremath{\Box}-p\,}
\newcommand{\boxcol}{brilliancy\,}%{$\Box$-collapse\,}

\newcommand{\lb}{\lnref{\lna{Box}}}%_{\strictif}
\newcommand{\lbalt}{\lix}
\newcommand{\lbaltalt}{\lnref{\lna{Box''}}}%_{\strictif}
\newcommand{\lS}{\bnref{\lna{S_\Box}}}%_{\strictif}
\newcommand{\lSalt}{\lnref{\lna{S_\subl}}}%_{\strictif}
\newcommand{\lSaltalt}{\lnref{\lna{S'_\subl}}}
\newcommand{\wkdisf}{\lnref{\lna{{\mathrm i}A^-}}}%{\lna{wKInt^{\wedge,\to,\bot}_{\strictif}}

\newcommand{\wk}{\loglb}%{\lna{iP}}%wKInt_{\strictif}}
\newcommand{\iP}{\wk}
\newcommand{\bk}{\lnref{\lna{{\mathrm i}\hyph BoxA}}}%{\lna{ib}}%\lna{bKInt_{\strictif}}

\newcommand{\lnam}[1]{\lnref{\lna{#1}}}

\newcommand{\bs}{\lnam{{\mathrm i}\hyph BoxSA}}%bSInt_{\strictif}
\newcommand{\isys}[1]{{\ensuremath{{\sf {\mathrm i}\hyph#1}}}}

\newcommand{\ivsys}[1]{{\ensuremath{{ {\mathrm i}\hyph\mathcal{#1}}}}}
\newcommand{\icol}[1]{\lnam{{\mathrm i}\hyph Box#1}}%{\lnam{ib}}%
\newcommand{\ivcol}[1]{\lnam{{\mathrm i}\hyph Box\mathcal{#1}}}%{\lnam{ib}}%
\newcommand{\ivstr}[1]{\lnam{{\mathrm i}\hyph S\mathcal{#1}}}

\newcommand{\ws}{\lnam{{\mathrm i}\hyph SA}}%wSInt_{\strictif}

\newcommand{\labox}{\ensuremath{\mathcal{L}_\Box}}
\newcommand{\latto}{\ensuremath{\mathcal{L}_\tto}}
\newcommand{\laipc}{\ensuremath{\mathcal{L}}}

\newcommand{\pathct}[1]{\textcolor{darkergreen}{#1}}

\usepackage{linearb}
\newcommand{\CSp}{\pathct{\linbfamily\protect\Bswi}}%\MVAt}}
\newcommand{\KRp}{ \pathct{\linbfamily\protect\Bdwo}}%\pathct{\LooseBearing\ensuremath{\kmodels}}}
\newcommand{\AIp}{\pathct{\linbfamily\protect\BNc}}
\newcommand{\AIIp}{\pathct{\linbfamily\protect\BNcc}}%\pathcm{\mathbb{+2}}}%%\MVTwo}}%\shortcastling}}%\StrokeTwo}}%\shortcastling}}

\newcommand{\AIIIp}{\pathct{\linbfamily\protect\BNccc}}%\pathcm{\mathbb{+3}}}%%\longcastling}}%\StrokeThree}}%\longcastling}}

\usepackage{textcomp}
\newcommand{\HIp}{\pathct{\linbfamily\protect\Bau}}%\textleaf}}%
\newcommand{\comp}{\!\cdot\!}

\begin{document}

\begin{abstract}
C. I. Lewis invented modern modal logic as a theory of  ``strict implication'' $\strictif$. Over the classical propositional calculus one can as well work with the unary box connective. %which became the modal primitive ever since. 
 Intuitionistically, however, the strict implication has greater expressive power than $\Box$ and allows to make distinctions invisible in the ordinary syntax. In particular, the logic determined by the most popular semantics of intuitionistic \lna{K} becomes a proper extension of the minimal normal logic of the binary connective. Even an extension of this minimal logic with the ``strength'' axiom, classically near-trivial, preserves the distinction between the binary and the unary setting. In fact, this distinction  %and the strong constructive strict implication itself
  has been discovered by the functional programming community in their study of ``arrows'' as contrasted with ``idioms''. %There is already a body of work on proof theory and categorical semantics of the arrow. 
Our particular focus is on arithmetical interpretations of intuitionistic $\tto$ in terms of \emph{preservativity} in extensions of \ha, i.e., Heyting's Arithmetic.
\end{abstract}

\maketitle

\tableofcontents

\section{Introduction}

%\section{}
%\subsection{}
\nosmurfduo
More is possible in the constructive realm than
is dreamt of in classical philosophy. For example, we have nilpotent infinitesimals (\cite{moer:mode13}) and the 
categoricity of weak first-order theories of arithmetic
(\cite{mcca:cons88}, \cite{mcca:inco91}, this paper Appendix~\ref{ctzero}).
We zoom in on one such possibility: the original modal connective of ``strict implication'' $\tto$  proposed by C. I. 
Lewis \cite{Lewis18,Lewis32:book}, and hence called here the \emph{Lewis arrow},  does not reduce to the unary box 
$\Box$ over constructive logic. This simple insight  opens 
the doors for a plethora of new intuitionistic modal logics that cannot be understood solely in terms of the box. 
%We have discovered that the most natural intuitionistic semantics of \emph{strict implication} $\strictif$, a connective dating back to C. I. Lewis, does \emph{not} make it definable in terms of unary $\Box$. In particular, the logic determined by the most popular semantics of intuitionistic \lna{K} becomes a proper extension of the minimal normal logic of the binary connective.  
% invented modern modal logic as a theory of  ``strict implication'' $\strictif$. Over the classical propositional calculus, however, one can as well work with the unary box connective, which soon became the modal primitive. Our starting point is the observation 
To the best of our knowledge, this observation was originally made in the area of \emph{preservativity logic} \cite{viss:eval85,viss:prop94,iemh:pres03,iemh:prop05}
 and metatheory of arithmetic provides perhaps the most interesting applications of intuitionistic $\tto$. However, one can claim that a similar discovery has 
 been independently made in the study of \emph{functional programming} in computer science (cf. \S~\ref{sec:comparr}).

 %can be found among arithmetically inspired logics. 
% that %over the intuitionistic propositional calculus, the strict implication is in no way equivalent to the unary box:
 %the most natural intuitionistic semantics of  $\strictif$ does \emph{not} make it definable in terms of the unary $\Box$. The connective
%$\strictif$ has greater expressive power and allows us to make distinctions invisible in the ordinary syntax. We study its minimal logic, 
%provide completeness results for several other extensions relevant for this paper and prove suitable completeness 
%results in \S~\ref{sec:basic}. In \S~\ref{sec:residual}, we discuss the connection between this system and substructural logic.
%\tadeusz{Reflect the new structure here}

We begin in \S~\ref{sec:lewisfall} by recalling  Lewis' invention of strict implication,  %and elements of his research program.  
 %His original connective is  
 mostly remembered by historians; these days, modal logic is almost by default taken to be the 
theory of boxes and diamonds. After sketching how $\tto$ fell into disuse and neglect, we speculate  whether removing the law of excluded middle could have saved Lewis' vision of modal logic. This is also a good opportunity 
to highlight some unexpected analogies between the fates of Brouwer's and Lewis' projects.

In \S~\ref{sec:basic}, we clarify how the intuitionistic distinction between $\phi \tto \psi$ and \mbox{$\Box(\phi\to\psi)$} is reflected in 
Kripke semantics. This may well prove the most natural way of introducing this connective for many readers.

In \S~\ref{sec:axiomatizations}, we  present the minimal deduction system\footnote{It was baptised ``\lna{iP}'' by  Iemhoff and 
coauthors    \cite{Iemhoff01:phd,iemh:pres03,iemh:prop05}, but this acronym ties $\tto$ too tightly to preservativity.} $\iP$ and  
numerous additional principles used in the remainder in the paper. In \S~\ref{sec:deriv}, we clarify  %mutual 
connections between them, i.e., the inclusion relation between corresponding logics. %and provide ample material on syntactic deductions involving $\tto$

With the syntactic apparatus ready, we turn in \S~\ref{sec:arith} to a major  motivation for the study of $\tto$: logics of 
$\Sigma^0_1$-preservativity of arithmetical theories as contrasted with more standard logics of provability. %These logics are defined via the notion
%of arithmetical interpretation.
 In order to provide an umbrella 
notion for the study of arithmetical interpretations of modal connectives, we begin this section by setting up a 
general framework for \emph{schematic logics}, which may prove of interest in its own right. 

In \S~\ref{sec:completeness}, we are finally tying together the semantic setup of \S~\ref{sec:basic} and the syntactic infrastructure of 
\S~\ref{sec:axiomatizations} by providing a discussion of completeness and correspondence results. Some of them are well-known, 
others are new. Having a complete semantics for the logics under consideration allows us  in \S~\ref{sec:nonderivations} to complement 
earlier syntactic derivations (given in \S~\ref{sec:deriv}) with examples of \emph{non}-derivations.

In \S~\ref{sec:arrows}, we are presenting other applications of \emph{strong} arrows and \emph{strong} boxes. 
In fact, what we call here ``strong arrows'' turns out to correspond directly to ``arrows'' in functional programming. We are also briefly 
discussing connections with logics of guarded (co)recursion and intuitionistic logics of knowledge. 

But while intuitionistic $\tto$ can be (re)discovered in areas ranging from computer science to philosophy, 
in our view arithmetical interpretations are most developed and interesting. 
Thus, in \S~\ref{sec:apppre} we return to the theme of \S~\ref{sec:arith} presenting some applications of the 
logic of preservativity. In \S~\ref{nnil} we discuss the 
application of preservativity to the study of the provability logic of Heyting Arithmetic {\sf HA}.
In \S~\ref{sec:falsity}, we show that preservativity allows a more satisfactory expression of the failure of \emph{Tertium non Datur}. 
%Finally, in \S~\ref{dejo}, we demonstrate in detail how to use
%preservativity to prove De Jongh's Theorem of the completeness of Intuitionistic Propositonal Logic for interpretations in {\sf HA}.

%\tadeusz{Shall we briefly introduce appendices too?}\albert{An excellent idea. I add the brief descriptions for A,B,C.}
The paper has several appendices that offer some supporting material. Appendix \ref{sec:real} collects basic facts about realizability needed in other sections. In Appendices~\ref{sec:picon} and \ref{sec:inter}, we provide
some basic insights in $\Pi_1^0$-conservativity logics and interpretability logics. These insights strengthen our understanding of preservativity
logic both by extending this understanding and by offering a contrast to this understanding. 
% In Appendix~\ref{luiesmurf}, we discuss, in some detail, a further principle that is valid in the $\Sigma^0_1$-preservativity logic
%of {\sf HA}. 
Finally, Appendix \ref{sec:survprob} discusses the collapse of $\tto$ in Lewis' first monograph, i.e.,  \emph{A Survey of Symbolic Logic} \cite{Lewis18} from the perspective of our deductive systems.

Of course, we are of the opinion that the reader should carefully study everything we put in the paper. However, we realize that this expectation is not realistic.
For this reason, we present several roadmaps through the paper.\footnote{Note also that reading the electronic version may sometimes prove easier due to omnipresent hyperlinks: apart from all the usually clickable entities (citations or numbers of (sub)sections, footnotes and table- or theorem-like environments \dots), even most names of logical systems can be clicked upon to retrieve their definition in Tables \ref{fig:boxax} and \ref{tab:mainax}. When reading a hardcopy, we advise keeping these Tables handy, perhaps jointly with Figure \ref{fig:compl}.}

\begin{foots}
\begin{description}
\item
The basic option is to read \S\S~\ref{sec:lewisfall}--\ref{sec:axiomatizations}  %Sections \ref{sec:lewisfall}, \ref{sec:basic} and \ref{sec:axiomatizations}. 
%This gives the reader
 to get the basics of motivational background, the Kripke semantics and an impression of possible reasoning systems.
 \item[\KRp]
 The reader who wants more solid treatment of Kripke semantics can extend the basic option with %consider
 %\ref{sec:lewisfall}, \ref{sec:basic}, \ref{sec:axiomatizations} and 
 \S~\ref{sec:completeness}. 
  \item[\CSp]
 The computer science package consists of the basic option %Sections~\ref{sec:lewisfall}, \ref{sec:basic}, \ref{sec:axiomatizations} 
  and \S~\ref{sec:arrows}.
  \item[\HIp]
The reader who wants to go somewhat more deeply into the history of the subject can extend the basic option with %study Sections \ref{sec:lewisfall}, \ref{sec:basic},
 %\ref{sec:axiomatizations} and
   Appendix~\ref{sec:survprob}.
 \item[\AIp]
 The reader who wants to understand the basics of arithmetical interpretations can %study 
 %Sections \ref{sec:lewisfall}, \ref{sec:basic}, \ref{sec:axiomatizations} and 
 extend the basic option with \S~\ref{sec:arith}.
 \item[\AIIp]%\label{paar}
 An extended package for arithmetical interpretations combines \AIp\ with %consists of Sections \ref{sec:lewisfall}, \ref{sec:basic}, \ref{sec:axiomatizations}, \ref{sec:arith} and 
 \S~\ref{sec:apppre}.
 \item[\AIIIp]
 The full arithmetical package extends  \AIIp\ %~(\ref{paar}) 
 with Appendices~\ref{sec:real}, \ref{sec:picon} and \ref{sec:inter}.
\end{description}
\end{foots}

%\tadeusz{Paths also mentioned in text below titles of sections?}

\section{The rise and fall of the house of Lewis}  \label{sec:lewisfall}

%\tadeusz{Move below Kripke section?}

\subsection{``The error of philosophers''}
\nosmurfduo
 We are reflecting on L.E.J.  Brouwer's heritage half a century after his passing. Given his negative views on the r\^ole of logic and formalisms in mathematics, it seems somewhat paradoxical that these days the name of intuitionism survives mostly in the context of \emph{intuitionistic logic}.\footnote{A related and better-known paradox is that Brouwer's own name survives in mainstream mathematics mostly in connection with  his work on topology, which is confirmed by several contributions in this collection. This despite the fact that he rejected these results on philosophical grounds and was actively involved in topological research only for the period necessary to secure academic recognition and international status. Moreover, it seems a myth that the non-constructive character of his most famous topological publications \emph{turned} Brouwer into an intuitionist. There is ample evidence that while the exact form of his intuitionism evolved somewhat, his philosophical beliefs predate these results. Cf. van Stigt \cite{vanStigt90} for a detailed discussion of all these points.}
 %He was never  philosophically convinced of these results. The main 
 %he just needed them to secure academic recognition. They did not mean to him more than that and he did not see a way to make them philosophically acceptable. As described by Van der Waerden, having obtained his appointment at the University of Amsterdam, \emph{%[e]ven though his most important research contributions were in topology, 
 %Brouwer never gave courses in topology, but always on---and only on---the foundations of his intuitionism.} %Van der Waerden goes on to say that \emph{[i]t seemed that he was no longer convinced of his results in topology because they were not correct from the point of view of intuitionism} and the use of \emph{no longer} in this sentence %, and he judged everything he had done before, his greatest output, false according to his philosophy
% reveals precisely the common historical misunderstanding that we are talking about. Brouwer had never seemed philosophically \emph{convinced} of these results, he just needed them to secure academic recognition.
 %Even Brouwer's own name, apart from his topological contributions (which are also not exactly in line with his philosophy of mathematics), seems perhaps most often mentioned in the context of Brouwer-Heyting-Kolmogorov interpretation of intuitionistic connectives.  
  %But Clio is an ironic goddess. %and even several contributions to this special issue...
One is reminded in this context of what Nietzsche called \emph{the error of philosophers}:
  
\begin{smquote}
The philosopher believes that the value of his philosophy lies in the whole, in the structure. 
Posterity finds it in the stone with which he built and with which, from that time forth, men will build oftener and better---in other words, 
in the fact that the structure may be destroyed and yet have value as 
material.\footnote{\emph{Human, All-Too-Human, Part II}, translated by Paul V. Cohn.} 
\end{smquote}
 
%Thus, it seems appropriate on this occasion

\noindent
We feel thus excused to focus on propositional logics based on the intuitionistic propositional calculus (\ipc). More specifically, our interest lies in  an intuitionistic take on a formal language developed by an author nearly perfectly contemporary with Brouwer: Clarence Irving Lewis\footnote{He was born two years later than Brouwer and died two years earlier.}, the father of modern modal logic. And this time, the reason for this does \emph{not} come from the well-known G\"odel(-McKinsey-Tarski) translation of \ipc\ into the system Lewis denoted as \lna{S4}, which is discussed elsewhere in this collection. 

One can also see a certain irony in the fate of Lewis' systems. They were explicitly designed to give an account of ``strict implication'' $\strictif$. The unary  $\Box$  %considered standard nowadays 
 can be introduced using
\begin{equation} \label{boxdef}
 \Box \phi \leftrightarrow (\top \strictif \phi).
 \end{equation}
 In fact, %the fact that %over the classical propositional calculus (\cpc), 
 Lewis designed $\tto$ and $\Box$ as mutually definable,\footnote{\label{ft:sdiam}To be  precise, in his books Lewis did not use $\Box$ as a primitive. His exact formulation of $\phi \strictif \psi$ was $\neg\Diamond(\phi \wedge \neg\psi)$. However, in the classical setting, this one is obviously equivalent to the one given by \refeq{eq:lewbox}, and the reliance of Lewis' formulation on involutive negation would be a major problem over \ipc. See Appendix \ref{sec:survprob} for a more detailed examination of the r\^ole of involutive/classical negation in Lewis' original system.}  setting 
 \begin{equation} \label{eq:lewbox}
\phi \strictif \psi := \Box(\phi \to \psi)
\end{equation}
%gradually turned 
 and over subsequent decades, modal logic in a narrow sense turned into the theory of unary $\Box$ and/or $\Diamond$. In a broader sense, pretty much any \emph{intensional} operator extending the usual supply of connectives can be called a modality. Modalities came to represent not only \emph{necessity}, but also \emph{arithmetical provability}, \emph{knowledge}, \emph{belief}, \emph{obligation}, %, abilities of coalitions of agents, 
 and various forms of \emph{guarded quantification}: \emph{validity after all possible program executions}, \emph{in all accessible states}, \emph{in all future time instants} or \emph{at every point in an open neighbourhood} (the list, of course, is far from being exhaustive). Just like in the case of intuitionistic logic, a wide range of semantics for modalities have been investigated, the most prominent being the Kripke semantics (relational structures), %and possible worlds, 
 but also topologies, coalgebras, monoidal endofunctors on categories or more recent ``possibility semantics''.
 
 Thus, Lewis' dissatisfaction with \emph{material} or \emph{extensional}  implication and disjunction, expressed first in a short 1912 article \cite{Lewis12}, has ultimately led to the spectacular success story of modal logic, much like Brouwer's\footnote{Speaking of Brouwer, note again the parallelism of dates: 1912, the year when Lewis fired his first shots for intensional connectives by publishing \emph{Implication and the Algebra of Logic} \cite{Lewis12}, is also the year when Brouwer obtained his position at the University of Amsterdam,  was elected to the Royal Netherlands Academy of Arts and Sciences, delivered his famous inaugural address \emph{Intuitionism and Formalism} and became liberated to pursue his own program. We refrain here from investigating further analogies, such as the fact that Lewis wrote his 1910 PhD on \emph{The Place of Intuition in Knowledge} (cf. Murphey \cite[Ch. 1]{Murphey05} for an extended discussion), that he had a solid background in idealism and Kant and that he remained under strong influence of these philosophical  positions throughout his  career.} dissatisfaction with non-constructive usage of implication and disjunction has ultimately led to the spectacular success story of intuitionistic logic. 
  And yet, while  Lewis did not write much on formal logic after \emph{Symbolic Logic}\footnote{\emph{Symbolic Logic} was a collaboration between C. I. Lewis and C. H. Langford. The authors, however, made it clear in the preface who wrote and is ``ultimately responsible'' for which chapter, a practice rather uncommon today. All the passages quoted in this paper %and, indeed, all the passages relevant for this paper 
  come from chapters written by Lewis. %This is reflected even in the way some of contemporary authors were referring to that material. 
   As Murray G. Murphey says in his monograph on C. I. Lewis: ``\emph{Symbolic Logic} was less a cooperative venue than a coauthored book \dots To what extent each advised the other on their separate chapters is left unclear, but probably there was not much of an attempt to harmonize \dots Langford's theory of propositions, for example, in Chapter IX is clearly not Lewis's theory.''  \cite[p.183]{Murphey05}.} published in 1932 \cite{Lewis32:book}, his occasional remarks do not suggest he would approve of the scattering of his Strict Implication systems into a bewildering galaxy of unimodal calculi. Indeed, he was not only opposed to the very name \emph{modal logic}, but believed that his formalisms is the exact opposite of real ``modal'' logic, which in his view was \dots the extensional system of \emph{Principia Mathematica}:
 
 \begin{smquote}
There \emph{is} a logic restricted to indicatives; the truth-value logic most impressively developed in \emph{Principia Mathematica}. But those who adhere to it usually have thought of it---so far as they understood what they were doing---as being the universal logic of propositions which is independent of mode. And when that universal logic was first formulated in exact terms, they failed to recognize it as the only logic which is \emph{independent} of the mode in which propositions are entertained and dubbed it ``modal logic''. (Cf. \cite[p. 203]{Murphey05}) %from ``a paper on the logic of imperatives in 1960''. It does not that the paper was published. Murphey appears to refer to a posthumous collection of Lewis manuscripts here.
 \end{smquote}
 
 \noindent
 His own belief was that 
\begin{smquote}
the relation of strict implication expresses \emph{precisely that relation which holds when valid deduction is possible} [emphasis ours]. It fails to hold when valid deduction is not possible. In that sense, the system of Strict Implication may be said to provide that canon and critique of deductive inference which is the desideratum of logical investigation \cite[p. 247]{Lewis32:book}
\end{smquote}

\noindent
and that

\begin{smquote}
Strict Implication explains the paradoxes incident to truth-implication. \cite[p. 247]{Lewis32:book}
\end{smquote}

%\tadeusz{Wes Holliday's comment: I think it would be good to say more about why Lewis was not satisfied with material implication as the only implication to be studied by logic. You donÕt say anything about the Òparadoxes of material implication.Ó IsnÕt that an important part of motivating strict implication?}

\nosmurf
%We will  not go here into the details of issue of the exact axioms to be satisfied  an implication encoding a genuine entailment relation---or into real or supposed paradoxes of material implication (see Footnote \ref{ft:relevance} though).
%Such passages help to understand why soon after the publication \emph{Symbolic Logic} Lewis largely ceased to pay attention to  modal logic; again, the analogy with Brouwer's ambivalent attitude towards the formalization of \ipc\ by Arend Heyting (or related work of Kolmogorov and Glivenko) and subsequent developments is rather striking.\footnote{Cf. Brouwer's more than disputable description of \ipc\ as an \emph{interesting but sterile exercise}. Incidentally, even the dates of Lewis' and Brouwer's respective withdrawals from participation in the development of systems they inspired are very close to each other.} Lewis' systems were not used to \emph{provide [a] canon and critique of deductive inference} free of  \emph{paradoxes incident to truth-implication}. 
%Let us simply state a rather indisputable fact that none of Lewis' systems was used to ``provide [the] canon and critique of deductive inference'' free of  ``paradoxes incident to truth-implication'' in decades to follow, despite  Lewis' own conviction he found the ultimate solution in \lna{S2}. In all fairness, however, 
While the  failure of Lewis' systems to conquer this intended territory had to do with philosophical prejudices of the following decades, %was partially caused  by rather dogmatic extensionalism dominating in Anglophone philosophy departments, %among followers of Russell, neopositivists and later Quine, %who came to dominate Anglo-Saxon philosophy departments, 
 they were also simply  less suited for these purposes than Lewis thought.   The original system of \emph{A Survey of Symbolic Logic} in 1918 \cite{Lewis18}---stemming back to a 1914 paper \cite{Lewis14:jppsm}---was plagued by a number of issues, the most famous one pointed out by Post: the combination of an axiom equivalent to (in an updated notation) %both
%$$
%(p \strictif q) \strictif (\Diamond p \strictif \Diamond q)
%$$
%and 
$$
(\Box \phi \strictif \Box \psi) \strictif (\neg\psi \strictif \neg\phi)
$$
with other axioms %, in particular $\Box \phi \strictif \phi$
 and classical negation laws trivialized the  modality and collapsed strict implication to material implication \cite{Lewis20:jppsm}. We provide an extended analysis of Lewis' SSL problem in Appendix \ref{sec:survprob}; we believe it is an interesting application of the intuitionistic theory of $\tto$ discussed in this paper.\footnote{Cf. also the discussion by Murphey \cite[pp. 101--102]{Murphey05} or Parry \cite{Parry70}.} %Lewis learned from his mistakes and in the development 
 In \emph{Symbolic Logic} \cite{Lewis32:book}---more precisely, in its famous Appendix II---Lewis was more cautious,  %used axioms proposed by Mordchaj Wajsberg and William T.  Parry to
 creating several ``lines of retreat'' (as Parry \cite{Parry70} described it)  in the form of \lna{S3}, \lna{S2} and \lna{S1}. At least on the technical front, this time things went better. %Most of immediate criticism that followed the 1932 book strikes a contemporary reader as misguided.
  Immediate polemics focused on possibility of definability of intensional connectives in extensional systems, %like \emph{Principia Mathematica},
    but none of the authors involved proposed anything resembling what we much later came to know as the \emph{Standard Translation} of modal logic into predicate logic.\footnote{Cf., e.g., the attempts of Bronstein\&Tarter or Abraham addressed, respectively, by McKinsey and Fitch; see Murphey \cite[Ch. 6]{Murphey05} for references. It is worth pointing out that Lewis himself \cite{Lewis35} dealt with this question in a paper published only posthumously (with Langford as a ``nominal'' coauthor, see editor's note \cite{Mares14:note} for a contemporary perspective).}  %had appeared were blatantly inadequate from technical and conceptual point of view (unlike critiques of his earlier \emph{Survey}),  
There were, however, subtler problems, %if not immediately realized. Already in the post-war period, % Most of the blame lied in the systems themselves: 
 pointed out in in the post-war period by Ruth Barcan Marcus:\footnote{Her earliest papers \cite{Barcan46:jsl} are signed by her maiden surname, Ruth Barcan, which survives until today in the name of the \emph{Barcan formula}.} %\cite{Barcan46:jsl,Barcan53:jsl} astutely pointed out that
\begin{smquote}
 it is plausible to maintain that if strict implication is intended to systematize the familiar concept of deducibility or entailment, then some form of the deduction theorem should hold for it. \cite{Barcan53:jsl}
\end{smquote}
%, one of the few authors at the time to publish a genuinely penetrating criticism of \emph{Symbolic Logic}, 
%She astutely pointed out how badly these calculi (especially \lna{S1} to \lna{S3}, that is, Lewis' favourite systems) behave with respect to the Deduction Theorem. 
She showed  \cite{Barcan46:jsl,Barcan53:jsl}  that \lna{S1} to \lna{S3} fail this criterion, for several conceivable formulations of the Deduction Theorem. 
And those which behave somewhat better in this respect, i.e., from \lna{S4} upwards are too strong to capture a general 
notion of strict implication which Lewis would approve of. %with the properties that Lewis desired. 

In fact, \lna{S4} and \lna{S5}, %which for good reasons soon gained more recognition, 
 which we came to count among \emph{normal} systems (unlike \lna{S1}--\lna{S3}) and for which the advantage of switching to the unary setting is most obvious, for Lewis himself were foster children he was forced to adopt. As is well-known, it was  Oskar Becker\footnote{\label{ft:orlov}Although   many developments discussed in this subsection---in particular proposing and justifying  \lna{S4} axioms with an explicit Brouwerian motivation---had their forerunner in a neglected 1928 paper by Ivan E. Orlov, cf. \cite{Dosen1992:orlov,Bazhanov03}.} \cite{Becker30} who proposed these axioms, even calling one of them the \emph{Brouwersche Axiom}; let us not discuss the adequacy of this name here, but not only does it provide us with another excuse to mention Brouwer in this paper, it has also survived until today in names of systems like \lna{KB} or \lna{KTB}. Becker intended  to cut the number of non-equivalent modalities in the calculus, a goal which seems rather orthogonal to Lewis' plans: 
\begin{smquote}
Those interested in the merely mathematical properties of such systems of symbolic logic tend to prefer more comprehensive and less ÔstrictÕ systems such as \lna{S5} and material implication. The interests of logical study would probably be best served by an exactly opposite tendency. \cite[p. 502]{Lewis32:book}
\end{smquote}

\takeout{Mark van Atten's comments: Becker en Gšdel hadden een voorganger in Orlov. Zie
Bazhanov, V.A., 2003, ÒThe scholar and the ÔWolfhound EraÕ: The fate
of Ivan E. Orlov's ideas in logic, philosophy, and scienceÓ, Science
in context, 16(4): 535Ð550.
Do?en, K., 1992, ÒThe first axiomatization of relevant logicÓ, Journal
of Philosophical Logic, 21: 339Ð356.
3. Ik denk dat Brouwer Ex Falso niet zou accepteren. Zie bijgaand
paper. (Een aantal dingen daarin zou ik nu anders zeggen, maar ik
bedoel maar.) Referentie:
van Atten, M., 2009, ÒThe hypothetical judgement in the history of
intuitionistic logicÓ,  in Logic, Methodology, and Philosophy of
Science XIII: Proceedings of the 2007 International Congress in
Beijing, C. Glymour, W. Wang, and D. WesterstŒhl, eds., London: King's
College Publications, 122--136.
}

\nosmurf
 Kurt G\"odel did review Becker's work \cite[p. 216--217]{Goedelv1} and was familiar with William T. Parry's early analysis of the notion of \emph{analytic} implication based on  $\strictif$ \cite[p. 266--267]{Goedelv1}.\footnote{As another small example how modal and intuitionistic inspirations tended to work hand-in-hand for G\"odel: his proof that \ipc\ is not characterized by any finite algebra \cite[p. 268--271]{Goedelv1} is presented as an answer to a question posed by Otto Hahn during a discussion following Parry's presentation.} This apparently led\footnote{His short review of Becker points out that Becker's attempts to relate modal logic to ``the intuitionistic logic of Brouwer and Heyting'' and claims that steps taken by Becker to ``deal with this problem on a formal plane'' are unlikely to succeed; Orlov (cf. Footnote \ref{ft:orlov}) was more insightful, but it does not appear that G\"odel was familiar with his paper.} to his landmark 1933 paper \cite[p. 296--303]{Goedelv1} translating the nascent intuitionistic calculus into what turns out to be a notational variant of \lna{S4} formulated with %(in contemporary notation) $\Box$ rather than $\Diamond$ (or $\strictif$) 
 unary box as a primitive. 
 Thus,  immediately after \emph{Symbolic Logic} was published, G\"odel pretty much doomed the fate of $\strictif$ and condemned non-normal systems to at most secondary status: his paper not only provided an independent motivation (in terms of ``the intuitionistic logic of Brouwer and Heyting'' \dots) for the study of extensions of \lna{S4} rather than subsystems of \lna{S3}, %that the weaker systems were lacking, 
 but also  highlighted the elegance and conciseness of $\Box$-based axiomatizations for these logics.
 %showed how focusing on $\Box$ leads to elegant, concise axiomatizations based on rule which we came to know as the G\"odel Rule, alternatively also known as the Rule of Necessitation.

In short, it appears that regardless of the fact that historical circumstances did not favour %developed quite unfortunately for 
 Lewis, none of his systems was destined to success or genuinely free of design or conceptual issues. %---and we have not even touched here, e.g., on the question of ``paradoxes of strict implication''. 
 %But part of the problem was that the entire
  Nevertheless, the idea of providing an implication connective yielding tautologies only when the antecedent is genuinely \emph{relevant} for the consequent %was decades ahead of its time.
  proved prescient.\footnote{\label{ft:relevance}The connection between modal logics and relevance logics has been always actively debated, see, e.g., Mares \cite[Ch. 6]{Mares04} for an extended presentation, including a reminder that Ackermann's 1956 paper which ``began the study of relevant entailment'' took issue with some tautologies valid for Lewis' $\strictif$, in particular  \emph{ex falso quodlibet}. But in fact the relationship can be traced back at least  to  1933, when Parry in his work on analytic implication based on  $\strictif$  proposed what relevance logicians came to know as the \emph{variable sharing criterion}: much later, Dunn \cite{Dunn72} noted that Parry's system is contained in \lna{S4} and proposed a ``demodalization'' of Parry's original system still preserving that criterion. As another connection with G\"odel, let us note that his discussion   \cite[p. 266--267]{Goedelv1} of the work of Parry suggested a completeness result that was only proved in 1986 by Fine \cite{Fine86}. Moreover, one can push the clock back even beyond Parry and G\"odel, to the paper of Orlov (cf. Footnote \ref{ft:orlov}), which  seems the first attempt to relate relevance, intuitionistic, and modal principles, including the first axiomatization of  what came to be known as the implicative-negative fragment of the relevance logic R \cite{Dosen1992:orlov}. Let us note here the view of van Atten \cite{vanatten:hypo} that ``logic as  Brouwer sees it is a relevance logic'', rejecting in particular \emph{ex falso} (absent also in earliest versions of formalizations of intuitionistic logic by Kolmogorov and Glivenko), which subverts the standard understanding of the \emph{BHK interpretation} (cf \S~\ref{sec:comparr} below).}
  In fact, one can easily argue that even the later enterprise of relevance logic would not satisfy Lewis' expectations: he wanted to \emph{supplement} material implication with a strict one, not \emph{replace} it altogether. In this sense,  still more recent \emph{resource-aware} formalisms with computer-science motivation where \emph{both} a substructural \emph{and} an intuitionistic/classical implication are present (either as an abbreviation or directly in the signature) like linear logic \cite{Girard87,Troelstra92,Abramsky93:tcs,Bierman94} or the logic of bunched implications \lna{BI} \cite{OHearnP99:jsl,Pym02:book,PymOHY04:tcs} seem closer to Lewis' original idea.

\subsection{Could Brouwerian inspiration help Lewis' systems?} \label{sec:couldbrouw}
\nosmurfduo

%We  are not committing to claims about whether or not 
 %Lewis' philosophical goals would have been better served by switching to the intuitionistic propositional base. 
\nosmurf
  %It is worth noting, however, that 
  At the time of publication of \emph{Symbolic Logic}, Lewis was both familiar with and open to non-boolean extensional connectives. The chapters he wrote for that monograph deal in detail with  $n$-valued systems of \L ukasiewicz.\footnote{At the time, Lewis still attributed it to a collaboration between \L ukasiewicz and Tarski.}   At the same time, he published a paper  on \emph{Alternative systems of logic} \cite{Lewis32:monist}.   In both these references, he discusses possible definitions of ``truth-implications'' \cite{Lewis32:book} or ``implication-relations'' \cite{Lewis32:monist} one can entertain in finite, but not necessarily binary matrices. %The paper in \emph{The Monist} also contains  %pointing out that no such definition satisfies the criteria of strict implication. %Still, he was willing to give serious consideration to such alternative proposals for material implication. 
 %Apart from relevant parts of \emph{Symbolic Logic}, While most of its contents focuses again on sympathetic analysis of the system of \L ukasiewicz (which at the time he still attributed to a collaboration between \L ukasiewicz and Tarski), it 
 The latter paper  also contains a rare (perhaps the only one) reference to Brouwer in his writings: %contains one of the few mentions of Brouwer in Lewis' writings we were able to find:
 
\begin{smquote}
[T]he mathematical logician Brouwer has maintained that the law of the Excluded Middle is not a valid
principle at all. The issues of so difficult a question could not be discussed here; but let us suggest a point of view at least something like his. \dots The law of the Excluded Middle is not writ in the heavens: it but reflects our rather stubborn adherence to the simplest of all possible modes of division, and our predominant interest in concrete objects as opposed to abstract concepts. The reasons for the choice of our logical categories are not themselves reasons of logic any more than the reasons for choosing Cartesian, as against polar
or Gaussian co\"ordinates, are themselves principles of mathematics, or the reason for the radix 10 is of the essence of number. \cite[p. 505]{Lewis32:monist}
\end{smquote}

\nosmurf

Of course, the question of Lewis' own potential take on combining  \ipc\ and $\strictif$ remains speculative: it does not seem he 
was familiar with the work of Kolmogorov, Glivenko and Heyting, turning Brouwer's philosophical insights into a propositional calculus. %\footnote{It is worth pointing out that \emph{ex falso quodlibet} was not a theorem in the systems of Kolmogorov and Orlov, and at least according to some authors \cite{vanatten:hypo}, was not accepted by Brouwer either. Glivenko was only convinced to incorporate it by Heyting. Cf.  \cite{Troelstra1990,Dosen1992:orlov}. }
%And the question whether such a system would best meet the criteria that Lewis imposed on strict implication  has to be postponed for further study. %(possibly also debating the issue raised in Footnote \ref{ft:sdiam}). 
% We do believe it is of independent interest, as is the above suggested question how linear logic (either in its intuitionistic, or in the classical variant) and the logic of Bunched Implications (\lna{BI}) would fare in this light. %it does not seem such questions have been thoroughly investigated.
   Nevertheless, let us note two points:
\begin{itemize}
\item even the collapse of Lewis' original system  \cite{Lewis14:jppsm,Lewis18}  was caused by classical laws combined with a misguided boolean inspiration, namely the insistence on involutivity of the \emph{strict} negation (cf. Appendix \ref{sec:survprob}); %would not collapse so easily without involutive, classical modality;
\item even when considering classical Kripke frames, the negation-free logic obtained by replacing $\to$ with $\strictif$ is a sublogic of the intuitionistic logic %, a point investigated in several references 
 \cite{Corsi87:mlq,Dosen93,CelaniJ01:ndjfl,CelaniJ05:mlq} (see also Question \ref{que:subi}). %We will discuss this in more detail in \S~\ref{sec:axiomatizations}. %logic in more detail later; it is going to play an important technical r\^ole here. 
%From the point of view of this introduction, it is simply worth pointing out 
%For the time being, we simply observe that the logic where, as expected by Lewis, the deductive tasks traditionally played by $\to$ are performed by $\strictif$, is naturally intuitionistic, at least in the absence of involutive negation.
\end{itemize}

%But, as we said, the starting point of our paper is simpler. 
\nosmurf
Our paper, however, focuses on an even more fundamental advantage of studying the theory of $\tto$ over \ipc. Whatever is there to be said about \emph{the universal logic of propositions which is independent of mode} and %whether or not Lewis would approve of using \ipc\ as its propositional basis,  
 its extensional basis, defining $\strictif$ using \refeq{eq:lewbox} is premature in the constructive setting. Furthermore, instances of such a ``constructive strict implication'' can be seen in areas ranging from metatheory of intuitionistic arithmetic to functional programming, often  satisfying very different laws to those strict implication was supposed to obey; indeed, sometimes rather meaningless classically. For example,   %In particular, in those discussed in \S~\ref{sec:arrows}, but also ,
\begin{itemize}
\item[\lSalt] $(\phi \to \psi) \to (\phi \strictif \psi)$
\end{itemize}
holds in numerous logics justified from a computational/Curry-Howard (\S~\ref{sec:arrows}), arithmetical 
(\S~\ref{sec:hastar}) or even philosophical (\S~\ref{sec:intepi}) point of view.\footnote{From a Lewisian point of view, 
would  intuitionistic $\to$  be the ``strict'' implication and $\strictif$ be the ``material'' implication in such systems?} 
%\avblue{I am in favor of not doing this. It would overload the presentation.}

%Let us roll up sleeves and, finally, start having fun.
%Let us start turning these observations into mathematical flesh.

%After this extended introduction, we are ready to begin

\section{Strict implication in intuitionistic Kripke semantics} \label{sec:basic}

\nosmurfduo
It is time to begin a more systematic discussion, starting with the %question of %We begin in the next section with the question of
 relational interpretation of $\tto$. In this paper, we are concerned with the following propositional languages: $\latto$ (with Lewis' arrow), $\labox$ (the unimodal one, identified with a fragment of $\latto$) and $\laipc$ (the propositional language of \ipc): 

\begin{align*}
\latto \quad \phi & ::= \bot \mid \top \mid p \mid (\phi \wedge \phi) \mid (\phi \vee \phi) \mid (\phi \to \phi) \mid (\phi \tto \phi), \\
\labox \quad \phi &::= \bot \mid \top \mid p \mid (\phi \wedge \phi) \mid (\phi \vee \phi) \mid (\phi \to \phi) \mid (\Box\phi), \\
\laipc \quad \phi &  ::= \bot \mid \top \mid p \mid (\phi \wedge \phi) \mid (\phi \vee \phi) \mid (\phi \to \phi).
\end{align*}

As usual, $\neg\phi$ abbreviates $\phi \to \bot$.

\medent

\begin{foots}
\emph{For the sake of clarity, the binding priorities are as follows:
%\begin{itemize}
%\item 
unary connectives $\neg$ and $\Box$ bind strongest,
%\item 
 next comes $\tto$,
%\item
 then $\wedge$ and $\vee$,
%\item
 and finally $\to$. 
%\end{itemize}
}

\emph{Regarding associativity, it is used tacitly for $\wedge$ and $\vee$, just like commutativity. 
Regarding $\to$ and $\tto$, they are commonly assumed to associate to the right, but we will be careful not to overuse this convention, as it can be confusing.}
\end{foots}

%\tadeusz{recall from the intro:}

%\begin{smquote}
%\emph{We observe that
% the most natural intuitionistic semantics of  $\strictif$ does \emph{not} make it definable in terms of unary $\Box$:
%$\strictif$ has greater expressive power and allows to make distinctions invisible in the ordinary syntax. We study its basic theory and prove suitable completeness results}
%\end{smquote}

\nosmurfduo
We begin with recalling the basic setup of intuitionistic Kripke frames for $\labox$.\footnote{As far as $\labox$ is concerned, %the unary unimodal language is concerned, 
 our discussion largely follows Litak \cite{Litak14:trends}. The reader is referred there for more details and references.} They come equipped with 
two accessibility relations. %$\preceq$ and $\sqsubset$. 
  One of them, which we will denote by $\preceq$, is a partial ordering\footnote{In fact, it is essential only that the relation is a preorder (i.e., a reflexive and transitive relation), 
  but such a generalization brings no tangible benefits from the point of view of expressivity, definability and completeness of propositional logics.} 
  interpreting  intuitionistic implication:
  \begin{equation} \label{eq:forcimp}
 k \kmodels \phi \to \psi \text{ if, for all } \ell \succeq k \text{, if } \ell \kmodels \phi \text{, then } \ell \kmodels \psi.
\end{equation}
  
   This forces the denotation of $\to$ to be \emph{$\preceq$-persistent} or, 
  as some authors say, ``monotone'' or ``upward-closed''. %The denotation of other formulas has to satisfy the same condition, but as far as  \ipc\ is concerned, 
  It is enough to impose \refeq{eq:forcimp} and require $\preceq$-persistence  of atoms %, interpreting all the remaining connectives  locally,
   to ensure persistence for all $\laipc$-formulas.  
  The other accessibility relation $\sqsubset$ is the modal one. %Of course, the interpretation of modal formulas also has to be $\preceq$-persistent.  %Here is what a comprehensive overview of 
%  Simpson \cite[\Section 3.3]{Simpson94:phd} says about 
  There are two choices one can make to ensure $\preceq$-persistence for $\Box$:

\begin{smquote}
%One could take also
%the usual satisfaction clauses for the modalities in modal models \dots %(see page 33).
%However, an essential feature of intuitionistic models is the monotonicity lemma \dots
%(Lemma 2.2.1). 
%If the standard satisfaction clauses for the modalities are used
%then the monotonicity lemma does not hold. There are two possible remedies.
 One is to modify the satisfaction clauses. This might be a reasonable thing to do,
for one might wish to use the partial order to give a more intuitionistic reading of
the modalities.  The other remedy
  is to impose conditions on models that ensure
that the monotonicity lemma does hold. \cite[\Section 3.3]{Simpson94:phd}
\end{smquote}

\noindent
In fact,  in a unimodal language the difference between 
these two strategies is not essential; it becomes more consequential when a single accessibility relation is used to 
interpret, for example, both $\Box$ and $\Diamond$ (see  \cite[\Section 3.3]{Simpson94:phd} for a discussion and more references). 
Still, most references choose the latter one, i.e., keeping the same reading 
of $\Box$ as in the classical case and imposing conditions on the interaction of $\preceq$ and $\sqsubset$ to ensure persistence.  

\newcommand{\fgskip}{\qquad\qquad}

\begin{figure}
\footnotesize
\begin{tabular}{@{\quad}c@{\fgskip}c@{\fgskip}c@{\fgskip}c} 
$\vcenter{
    \xymatrix@-1.1pc{
            & \ell \ar@{~>}[r]  & m\\
           k \ar@{~~>}[rr] 
           \ar[ur] & & \ell' \ar@{-->}[u]  
      } }$ &
      $\vcenter{
    \xymatrix@-1.1pc{
            \ell \ar@{~>}[r]  & m\\
           k \ar@{~~>}[ur] 
           \ar[u] &
      } }$    &    
 $\vcenter{
    \xymatrix@-1.1pc{
    	%& n \\
           & \ell \ar@{~>}[dr]  & & n \\
           k \ar@{~~>}[urrr] \ar[ur] & & m \ar[ur] &  
      } }$     &
$\vcenter{
    \xymatrix@-1.1pc{
             & m\\
           k \ar@{~~>}[ur]   \ar@{~>}[r]
           & \ell \ar[u]
      } }$       
      \\
\boxp & \strictp & \mix & \boxcol
\end{tabular}

\caption{\label{fig:conditions} Minimal conditions one can impose on $\Box$-frames and $\tto$-frames. See Figure \ref{fig:compl} for a visual representation of other conditions corresponding to additional axioms.}
\end{figure}
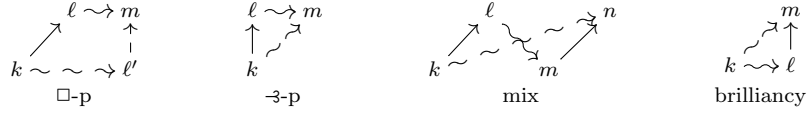

Bo\u{z}i\'{c} and Do\u{s}en \cite{BozicD84:sl} have established that in the presence of unary $\Box$ with semantics defined by 
\begin{center}
$ k \kmodels \Box\phi$ if, for all $\ell \sqsupset k$, $\ell \kmodels \psi$
\end{center}
   persistence is equivalent to the condition 
\begin{description}
\item[\boxp] if $k \preceq \ell \sqsubset m$, then, for some $\ell'$, we have $k \sqsubset \ell' \preceq m$ %\hfill 
\end{description}

\noindent
(i.e., $\preceq \comp \sqsubset \; \subseteq \; \sqsubset \comp \preceq$, where ``$\;\comp\;$'' denotes relational composition). However, most references require tighter interaction. 
On certain occasions, like in Goldblatt \cite{Goldblatt81:mlq}, one sees a strengthening to
\begin{description}
\item[\strictp] if $k \preceq \ell \sqsubset m$, then $k \sqsubset m$ \qquad (i.e., $\preceq \comp \;\sqsubset\; \subseteq\;  \sqsubset$).   %\hfill 
\end{description}

\noindent
But the most common one (see, e.g., \cite{Sotirov84:ml,WolterZ97:al,WolterZ98:lw}) is the still stronger

\begin{description}
\item[\mix] if $k \preceq \ell \sqsubset m \preceq n$, then $k \sqsubset n$  \qquad (i.e., $\preceq \comp \sqsubset \comp \preceq \; \subseteq \; \sqsubset$). %\hfill 
\end{description}

\noindent
%One justification for adopting \mix
%is that it naturally 
This condition naturally obtains in a canonical model construction \`a la Stone and J\'onsson-Tarski for prime filters of (reducts of) Heyting algebras with normal $\Box$ \cite{BozicD84:sl,Sotirov84:ml,Kohler81:ams,BezhanishviliJ12:acs}. Moreover, \mix is ``mostly harmless'' for $\Box$:  
 it can be obtained from $\Box$-p by adding the requirement that for  any $\ell$, the set of its $\sqsubset$-successors  is $\preceq$-upward closed, that is,

\begin{description}
\item[\boxcol] if $k \sqsubset \ell \preceq m$, then $k \sqsubset m$ \qquad (i.e., $\sqsubset\comp \preceq \; \subseteq \; \sqsubset$).%\hfill 
\end{description}

\noindent
The name, to the best of our knowledge, has been proposed by 
Iemhoff \cite{Iemhoff01:phd,iemh:moda01,iemh:pres03,iemh:prop05}, another one being \emph{strongly condensed} 
\cite{BozicD84:sl}. As noted in standard references \cite{BozicD84:sl,Goldblatt81:mlq}, not only \boxcol can\emph{not} be defined using 
$\Box$, but any model satisfying  \boxp can be made brilliant %turned into a one satisfying \boxcol and hence \mix
  \emph{without changing 
the satisfaction relation for $\Box$-formulas} in a straightforward way: by replacing $\sqsubset$ by its composition with $\preceq$.

Consider now the Lewisian strict implication $ \phi \tto \psi$. Here is the natural satisfaction clause in this semantics, directly transferring the classical one:
\begin{equation} \label{eq:forctto}
 k \kmodels \phi \tto \psi \text{ if, for all } \ell \sqsupset k \text{, if } \ell \kmodels \phi \text{, then } \ell \kmodels \psi.
\end{equation}

%The persistence of $\tto$ corresponds precisely with the property: 

\nosmurf
The first consequence of such an enrichment of the language is that \boxp becomes too weak to ensure persistence. 
Let us state this formally, defining for this purpose  a somewhat too general notion:

\begin{definition}
A \emph{preframe} is a triple $\ma F \deq \la W, \preceq, \sqsubset \ra$, where $\preceq$ is a partial order,  and $\sqsubset$ is a binary relation. A \emph{premodel} based on $\ma F$ is  $\ma K \deq \la \ma F, V \ra$, where $V$ is a \emph{valuation} mapping propositional variables to $\preceq$-upward closed sets. The \emph{forcing relation}  $\ma K, k \kmodels \phi$ is defined in the standard way for the intuitionistic connectives and using equation \refeq{eq:forctto} for $\tto$.
\end{definition}

\noindent
It can be easily shown (see, e.g., \cite{Zhou03,iemh:prop05}) that \emph{the condition equivalent to persistence becomes  precisely} \strictp, that is:

\begin{fact} \label{fact:strictp}
For a preframe $\ma K \deq \la W, \preceq, \sqsubset \ra$,  \strictp above corresponds to the following condition: 
%is equivalent to the condition: \avred{perhaps use the word `correspondence' here?}\tlred{Absolutely a good idea, though this will require changes elsewhere too}
%\begin{center}
 for any two sets $U, V$ upward closed wrt $\preceq$, the set %\newline
$$
U \strictif V := \{ k \in W \mid \forall \ell \sqsupset k, \text{ if } \ell \in U, \text{ then } \ell \in V \}
$$ %\newline
is upward closed wrt $\preceq$.
%\end{center}
\end{fact}

\noindent
We will thus take \strictp\ to be the minimal condition in what follows. 

\begin{definition}
A ($\strictif$-)\emph{frame} is a preframe satisfying \strictp.
\end{definition}

\noindent
%Most of the time, we will simply call these \emph{frames} and models are premodels based on frames. 
%An immediate consequence of Fact~\ref{fact:strictp} is that the denotations of all formulas are upward closed in all models. 
We can define in a standard way what it means for a formula to be \emph{valid} or \emph{refuted} in a class of models.

 %. Valuations, just like one would expect in an intuitionistic Kripke frame, send propositional variables to upward closed sets. They are extended to arbitrary formulas in the usual way; we have already seen how it is done for $\strictif$. All the usual notions of satisfaction, validity etc. generalize in an obvious manner.

%Now, how about the \boxcol condition? Does it remain ``mostly harmless'' in the sense described above for $\labox$? 
As we have already suggested,  for $\latto$ the \boxcol condition does not remain ``mostly harmless'' in the sense described above for $\labox$: %, the answer is an emphatic ``no'': %---and this observation is at the heart of our claim that over \ipc, $\strictif$ is not definable in terms of $\Box$ anymore.

\begin{fact}\textup{\cite{Zhou03}} \label{fact:boxcol}
%For a frame  $(W, \preceq, \sqsubset)$,
 The following conditions are equivalent for a $\tto$-frame:
\begin{itemize}
\item validity of $(\phi \wedge \psi)  \strictif \chi \to \phi \strictif (\psi  \to \chi)$; 
\item validity of $\psi  \strictif \chi \to \top \strictif (\psi  \to \chi)$;
\item validity of \boxcol. 
\end{itemize}
\end{fact}

\noindent One easily sees the converse implication $$\phi \strictif (\psi  \to \chi) \to (\phi \wedge \psi)  \strictif \chi$$ and, consequently, its special instance (where $\phi$ is equal to $\top$)
$$\Box (\psi  \to \chi) \to \psi \strictif \chi$$
 to be valid on any $\strictif$-%intuitionistic Kripke 
frame; see Lemma \ref{th:basicia} for a syntactic derivation.
 
 %Let us return now to the definition of $\Box\phi$ suggested by 
 %\refeq{boxdef} above, that is, as
 Let us take stock. %While $\Box \phi$ is definable as $\top \strictif \phi$,
  %Fact \ref{fact:boxcol} states that 
  In order to restore definability of $\strictif$ in terms of $\Box$, i.e., validity of \refeq{eq:lewbox} above, one needs to 
  impose the \boxcol condition. In general, %over our %$\strictif$-intuitionistic Kripke 
   %frames %, i.e., those
    %satisfying only \strictp, %but not necessarily brilliancy 
  $\Box(\phi \to \psi)$ implies $\phi \strictif \psi$, %and ; 
  %the latter implies the former, 
  but not necessarily the other way around. %Put still differently,  the  procedure described in standard references
%   references \cite{BozicD84:sl,Goldblatt81:mlq} of making the set of $\sqsubset$-successors  of each point $\preceq$-upward closed 
 %  by replacing $\sqsubset$ by its composition with $\preceq$ is not innocent from the point of view of $\latto$. %validity in the language with $\strictif$
 Of course, %such a distinction between $\Box(\phi \to \psi)$ and $\phi \tto \psi$ is only meaningful intuitionistically. 
 in classical Kripke frames, 
$\preceq$ is a discrete order, which trivializes all conditions discussed above and all distinctions between them.  
As we will see in Corollary \ref{cor:emderbox}, the %fact that  classical logic forces the Lewisian deconstruction of $\strictif$ %as in \refeq{eq:lewbox} 
  boolean deconstruction of $\tto$ can be also derived syntactically. %in a purely syntactic way. 
  We will return to Kripke semantics in \S~\ref{sec:completeness} below. %For now, let us turn to a %%an old-fashioned, purely syntactic
  %Hilbert-style study of logics %%potentially important axioms and derivations
    %in our extended language.
  
  %%%%%%%%%%%%%%%
  
\section{Axiomatizations}\label{lolsmurf}\label{sec:axiomatizations}

\subsection{A fistful of logics}

\nosmurfduo
%We have finally arrived at the point where we can give 
In this section, we present a Hilbert-style study %large, structured collection
 of %axioms and
 $\latto$-logics. %appearing throughout this paper. %principles and systems appearing in later sections. %this pape, 
%. Moreover, we discuss 
 %jointly with some derivabilities and inclusions. %and non-derivabilities
 %among principles. 
 %It has to be stressed that %at least as far as arithmetically oriented principles are concerned, not much material regarding axiomatization and completeness is new.
  Discussion of arithmetically oriented  principles was originated by Visser \cite{viss:aspe81,viss:comp82,viss:eval85,viss:prop94} 
 and  developed further by Iemhoff and coauthors  \cite{Iemhoff01:phd,iemh:pres03,iemh:prop05}, who also studied the basic theory of $\tto$-frames. %over arbitrary Kripke frames satisfying \strictp. %We are freely building on these foundations.
%\nosmurf
 \ipc\ and \cpc\ denote, respectively, the intuitionistic propositional calculus and its classical counterpart.

\subsubsection{Logics in $\labox$} \label{sec:axbox}

\nosmurfduo
Before we start discussing $\tto$-logics in \S~\ref{sec:axtto}, Table \ref{fig:boxax} presents some axioms involving only $\Box$, which is a definable connective in $\latto$.

\begin{table}[h!]

\footnotesize

\caption{\label{fig:boxax}List of principles for $\Box$. Here, the names of systems in the right column refer 
to languages restricted to connectives appearing in the axiomatization, i.e., not involving $\tto$. Later in the text, 
we will also use some of these principles  as axioms over $\iP$, i.e., the minimal ``normal'' system 
for $\tto$ (cf. Table \ref{tab:mainax}), where $\Box$ is  a defined connective. }
\begin{multicols}{2}
\begin{itemize}
\item[\bi]
$ \vdash \phi \;\; \To \;\;  \vdash \Box \phi$
\item[\bii]
$ \Box (\phi \to \psi) \to \Box \phi \to \Box \psi$ \medskip
\item[\biii]
$ \Box \phi \to \Box\Box\phi$ \medskip
\item[\CF]
$ \Box\Box \phi \to \Box\phi$ \medskip
\item[\biv]
$ \Box(\Box \phi \to \phi) \to \Box \phi$ \medskip
\item[\bvi]
$ \phi \to \Box \phi$ \medskip
\item[\bvii]
$ (\Box \phi \to \phi) \to \phi$ \medskip
\item[\bv]
$ \Box(\phi \vee \psi) \to \Box(\phi \vee \Box \psi)$ \medskip
\item[\mHCb]
$\Box\phi \to (\psi \to \phi) \vee \psi$ 
\item[\mHCbalt]
$\Box(\psi \to \phi) \to (\psi \to \phi) \vee \psi$ \medskip
\item[\bLin]
$\Box(\phi \to \psi) \vee \Box(\psi \to \phi)$ \medskip
\item[\lna{peirce}] $((\phi \to \psi) \to \phi) \to \phi$
\item[\lna{em}] $\phi \vee \neg\phi$
\end{itemize}
\columnbreak
\noindent
\begin{align*}
\cpc & \deq \ipc + \lna{peirce} \\
\iK & \deq \ipc + \bi + \bii \\
%\cK & \deq \cpc + \iK \\
\logba & \deq  \iK + \biv \\
\logbb & \deq \cpc + \logba \\
\iS & \deq \iK + \bvi \\
\iSLb & \deq  \iK + \bvii \\
\iPLL & \deq \iS + \CF \\
\imHCb & \deq \iS + \mHCb \\ 
\iKMb & \deq \iSLb + \mHCb \\
\iKMlinb & \deq \iKMb + \bLin 
\end{align*}
%The logic {\logba} is the intuitionistic version of L\"ob's Logic. It is given by the principles 
%{\bi,\bii,\biv}. It is known that  
$\biii$  is known to be derivable in $\logba$.  

In the provability logic literature: 
\begin{itemize}
\item \bi\ is known as {\sf L}1, 
\item \bii\ is known as  {\sf L}2,
\item \biii\ is known as  {\sf L}3, 
\item \biv\ is known as {\sf L}4.
\end{itemize}
%these principles are also known under the respective names
%{\sf L}1, {\sf L}2,  {\sf L}4 and {\sf L}3. The logic {\logbb}, often denoted simply as {\sf GL}, is {\logba} plus classical logic.
\end{multicols}
\end{table}

\begin{itemize}
\item The axioms of $\logba$ (intuitionistic L\"ob logic) and $\logbb$ (classical L\"ob logic) are well-known.  
The logic $\logbb$ is arithmetically complete for all classical $\Sigma_1^0$-sound theories extending Elementary Arithmetic {\sf EA}.
The logic $\logba$ is arithmetically valid in all arithmetical theories extending {\iea}. We discuss these matters further in \S~\ref{sec:provo}.
\item The principle \bv\ is known  as \emph{Leivant's Principle}. The principle is, in a sense, a shadow of the disjunction property. The disjunction property of an arithmetical
 theory $T$ cannot be verified in $T$ itself. Leivant's Principle is arithmetically valid in a substantial class of arithmetical theories that includes Heyting Arithmetic {\sf HA}.
 We discuss Leivant's Principle in \S~\ref{sec:provo}.  
\item   \bvi\  axiomatizes \emph{strong modalities} (cf. \S~\ref{sec:arrows}), but arises also in some arithmetically motivated logics %. It is arithmetically
%valid in theories
 %(see ${\sf HA}^\ast$ and ${\sf PA}^\ast$ in 
 (\S~\ref{sec:hastar}). 
 \emph{Strong L\"ob logic} is obtained by adding $\bvi$ to $\logba$---or, alternatively, by using $\bvii$ instead of $\biv$ as an axiom. %\tadeusz{needs to be stated as a theorem}
\item The principle \CF\ classically corresponds to a semantic condition known as \emph{density} %on classical Kripke frames 
 (cf. Figure \ref{fig:compl}). 
 %in \S~\ref{sec:completeness}). %below for semantic correspondents of formulas on $\tto$-frames). 
 From another point of view, this axiom arises 
naturally in the Curry-Howard logic of \emph{monads} (\S~\ref{sec:arrows}). It is a typical ``non-L\"ob-like'' axiom: 
in combination with \biv, we could derive $\Box\bot$.
\item $\mHCb$ comes from the intuitionistic system $\iKMb$ of Kuznetsov and Muravitsky and its later weakening to $\imHCb$ by Esakia and the Tbilisi group  (see \cite{Litak14:trends} for more information and references); its equivalent variant $\mHCbalt$ (see Lemma \ref{th:mhcderiv}\ref{mhcbeq}) was discussed  \cite{viss:comp82} in connection with ${\sf PA}^\ast$ (\S~\ref{sec:hastar}). In our setting, it is interesting
 to contrast it with $\mHCl$ in Table \ref{tab:mainax} (Figure \ref{fig:compl} %in \S~\ref{sec:completeness} %n semantic correspondents of both 
 %axioms (indistinguishable over brilliant frames, but as discussed in 
  and Example \ref{ex:mhcnonequiv}). %, in general $\Box$-variants and $\tto$-variants of this axiom differ). 
 See also \S~\ref{sec:falsity} for the arithmetical perspective on the contrast between $\imHCb$ and $\imHCl$.
 
 The name \lna{CB} used here comes from Litak \cite{Litak14:trends}, where it was used to suggest the Cantor-Bendixson derivative.
\item \bLin\ is a typical axiom valid on total orders. %used %in modal logic over transitive Kripke 
 %to ensure linearity in transitive relations. 
 In Fact \ref{fact:linnon} and Example \ref{ex:linavsb}, we compare and contrast this axiom with its $\tto$-counterpart.

\end{itemize}

\subsubsection{Logics in $\latto$} \label{sec:axtto}

\nosmurfduo
Table \ref{tab:minax} displays potential axioms for $\tto$  central for this paper. Most of them come with an explicit arithmetical interpretation. 
 %as discussed in \S~\ref{prelo} and thereafter. %\tadeusz{or wherever it's going to end up in}.
 Typically, the ``primed'' variants of axioms 
will be their equivalent reformulations (\S\ \ref{sec:deriv}).  %and the equivalence will be established in \S~\ref{sec:deriv} below.

\begin{table}
\caption{\label{tab:mainax} \label{tab:minax} List of principles for $\tto$ and logics considered in this paper.}
\footnotesize
\begin{multicols}{2}
\begin{itemize}
\item[\li]	$ \vdash \phi\to \psi \;\Rightarrow\; \vdash  \phi \tto \psi $
\item[\lii]	$ \phi \tto \psi \to \psi \tto \chi \to \phi \tto \chi$ \medskip 
\item[\liii]    $ \phi \tto \psi \to \phi \tto \chi \to \phi\tto (\psi \wedge \chi) $
\item[\liiialtalt] $ \phi \tto \psi \to (\phi \wedge \chi) \tto (\psi \wedge \chi)$
\item[\liiialt] $ \phi \tto \psi \to \phi\tto(\psi\to \chi) \to \phi \tto \chi$
\item[\liiialtaltalt] $ \phi \tto (\psi\to \chi) \to (\phi\wedge\psi) \tto \chi$ %\medskip
%\item[\liiiaaaa] $ \phi \tto \psi \to (\phi \wedge \chi) \tto (\psi \wedge \chi) $
\item[\lx] $ \Box (\phi \to \psi) \to \phi \tto \psi$ %\medskip
\item[\lxi] $ \phi \tto \psi \to \Box \phi \to \Box \psi$ \medskip
\item[\liv]	$ \phi \tto \chi \to \psi \tto \chi  \to ( \phi \vee \psi )\tto \chi $
\item[\livalt] $ \phi \tto \psi \to (\phi \vee \chi) \tto (\psi\vee \chi)$ \medskip
\item[\lxii] $ \phi \tto \psi \to \Box (\phi \tto  \psi)$ 
\item[\lv] $ \phi \tto \Box \phi$ \medskip
\item[\lvi] $ (\Box \phi \to \phi) \tto \phi$ \medskip
\item[\lvii] $ (\phi \wedge \Box \psi) \tto \psi \to  \phi \tto \psi$
\item[\lviialt] $ \phi \tto \psi \to (\Box \psi \to \phi) \tto \psi$ \medskip
\item[\lviii] $ \phi \tto \psi \to (\Box \chi \to \phi) \tto (\Box\chi \to \psi)$
\item[\lviiialt] $ (\phi \wedge \Box \chi) \tto \psi \to \phi \tto (\Box\chi \to \psi)$ \medskip
%\item[\lS] $\phi \to \Box\phi$ 
\item[\lSalt] $(\phi \to \psi) \to \phi \strictif \psi$
\item[\lSaltalt] $\phi \tto \psi \to \phi \to \Box\psi$ \medskip
\item[\lb] $\phi  \strictif \psi \to \Box(\phi  \to \psi)$ 
\item[\lix] $ (\chi \wedge \phi) \tto \psi \to \chi \tto (\phi \to \psi) $
\item[\lbaltalt] $\phi \tto \psi \to (\chi \to \phi) \tto (\chi \to \psi)$  \medskip
\item[\mHCl] $\phi \tto \psi \to (\phi \to \psi) \vee \phi$ \medskip
\item[\aLin] $\phi \tto \psi \vee \psi \tto \phi$ \medskip
\item[\aApp]
 $(\phi \wedge (\phi \tto \psi)) \tto \psi$ \medskip
 \item[\aCF]
 $\Box\phi \tto \phi$ \medskip
 \item[\Hug]
 $(\phi \to \Box\psi) \to (\phi \tto \psi)$
\end{itemize} 
\columnbreak
%\noindent 
Everywhere below, when we write $\ivsys{X}^{?}$, the superscript ``$?$'' can be either ``$-$'' or nothing, depending whether or 
not $\liv$ is used. 
%the former denoting the disjunction-free fragment, the latter denoting the full syntax. %Moreover, 
 %In particular, $\ipc^-$ denotes the disjunction-free fragment of intuitionistic propositional calculus.
\begin{align*}
\loglzero & \deq \ipc + \li + \lii, \\
 \logla & \deq \loglzero + \liii, \\%[\tbskip]
 \loglb & :=  \logla + \liv, \\[\tbskip]
 %\bkdisf & \deq  \wkdisf + \lb, \\
 %\end{align*}
%\begin{align*} 
%\ws^? & := \loglb^? + \lS,  \\ [\tbskip]
 %\logle & \deq \logla + \lvi, \\
 \loglg^? & := \loglb^? + \lvi, \\[\tbskip]
% & = \logle + \liv, \\ [\tbskip]
 %\loglf & := \logla + \lvii, \\
 \loglh^? & := \loglb^? + \lvii,  \\[\tbskip] 
% & =  \loglf + \liv, \\ [\tbskip]
 \loglc^? & := \loglh^? + \lviii.% \\ [\tbskip]
% \iKMl & := \iPC + 
  \end{align*}
% \end{itemize}

For each logic $\ivsys{X}^?$,  $\ivstr{X}^?$ denotes its extension with $\lS$, in particular

\smallskip
%\begin{center}
\hspace{1cm} $\ws  \deq \loglb + \lS$.
%\end{center}

Set also:

%\medskip

\hspace{1cm} $\iPLLa \deq \ws + \aCF$, \medskip

\hspace{1cm} $\imHCl  \deq \ws + \mHCl$, \medskip

\hspace{1cm} $\iKMl  \deq \imHCl + \biv$, \medskip

\hspace{1cm} $\iKMlin  \deq \iKMl + \aLin$, \medskip

\medskip

%For each logic $\ivsys{X}^?$,  
For each logic $\ivsys{X}$, $\ivcol{X}$ denotes its extension with $\lb$, e.g., \smallskip

\hspace{1cm}  $\bk \deq  \wk + \lb$, \medskip

\hspace{1cm}  $\icol{GL_a} := \loglg + \lb$.

\smallskip

Note that $\icol{GL_a}$ is just a notational variant of \logba. Note also that notation $\ivcol{X}^-$ would be redundant, see 
Lemma \ref{th:bkax}\ref{clausequiv}. A fortiori, the same applies to extensions of \imHCl\ by Lemma \ref{th:mhcderiv}\ref{mhcdi}. 
Similarly, $\iPLLa^-$ would be redundant by Lemma \ref{lem:derpll}\ref{plaadi}. In all these systems, $\Di$ can be derived from the remaining axioms.  Furthermore, as we will show in Lemma \ref{lem:kmlin}, $\iKMlin$ and $\iKMlinb$ are notational variants of the same system.

\end{multicols}
\end{table}

\begin{itemize}
\item Iemhoff \cite{Iemhoff01:phd,iemh:pres03,iemh:prop05} identified system $\iP$ as the logic of all (finite) frames satisfying 
the \strictp condition; this and other completeness results are discussed in \S~\ref{sec:completeness}. However, $\liv$ is an axiom which is not exactly trivial from an arithmetical point of view.
It does hold in the preservativity logic of Heyting Arithmetic but it fails in the preservativity logic of Peano Arithmetic 
 (\S~\ref{prelo} and Appendix \ref{sec:picon}).  The non-triviality of {\liv}, 
 and the potential interest in a disjunction-free system (\S~\ref{sec:arrows}) are the reasons why we isolated $\iP^-$ as a subsystem.
 \item
The principles  \lv,  \lvii, \lviii\/ are arithmetically valid for the preservativity interpretation of $\tto$. This means that they are in
the logic {\logld} which is arithmetically valid in all arithmetical theories we consider in this paper  (\S~\ref{prelo}).
The principle {\lvi}, weaker than {\lvii}, is mainly of technical interest.
  \item
  If we interpret $\phi \tto \psi$ as  $\neg\,\psi \rhd \neg\,\phi$, then
  the principle { \lxii} is the distinctive principle  of the interpretability logic of finitely axiomatized extensions
  of ${\sf EA}^+$ aka $\mathrm{I}\Delta_0+{\sf Supexp}$. The modality $\rhd$ stands for interpretability over a theory. This modality is
  explained in  \S~\ref{vrolijkesmurf}.\footnote{On a side note, some CS readers may be familiar with the use of triangle-like notation like $\rhd$ for \emph{unary} modalities in the context of guarded (co)recursion discussed in \S~\ref{sec:guarded}. The tradition of using such notation for \emph{binary} operators and connectives such as arithmetical interpretability is much longer and we believe this convention to be more natural.} The specific result mentioned here is discussed in detail in \S~\ref{kalsbeek}.
\item \lSalt\ and \lSaltalt\ are $\tto$-variants axiomatize the same logic  as \lS (Lemma \ref{lem:lsder}). %above. However, as made clear in Lemma \ref{lem:lsder}, they axiomatize exactly the same logic. 
 In general, this is rarely the case with $\tto$-generalizations of  $\Box$-axioms; often the $\tto$-version is stronger, but %we have already mentioned that 
  $\aLin$ illustrates such a rule is not universal. 
\item We have already seen \lb\ in \S~\ref{sec:basic} above; its equivalence with \lbalt\ and \lbaltalt\ is established in Lemma \ref{th:bkax}. The conjunction of this axiom with \lx, derivable in \iP\ (Lemma \ref{th:basicia}\ref{blderiv}), collapses $\tto$.  Note that $\mHCl$ makes \lb\ derivable (Lemma \ref{th:mhcderiv}), unlike $\mHCb$ (Example \ref{ex:mhcnonequiv}).
\item The last group of $\tto$-principles---i.e., $\aApp$, $\aCF$ and $\Hug$---which should be contrasted with $\CF$, will play a prominent r\^ole in \S~\ref{sec:comparr} on monads, idioms and arrows in functional programming. For similar reasons as \CF, they are of drastically ``anti-L\"ob'' character, a fact made explicit by their semantic correspondents  displayed in Figure \ref{fig:compl} in \S~\ref{sec:completeness}. It is worth mentioning that $\aApp$ was in fact adopted by Lewis as an axiom even in his weakest system \lna{S1}, cf. Remark \ref{rem:lewapp}.
\end{itemize}

%Instead of providing more detailed discussion here especially that

% the subscript.
%When writing $\lva{X} + \lva{Y}$, we mean 

%\subsection{Principles for the $\Box$}

%\subsection{Principles for $\tto$}

\subsection{An armful of derivations} \label{sec:deriv}

%\tadeusz{A lot of this material can be also moved to our forthcoming paper on decompositions of $\tto$; many derivations and remarks in this subsection are dealing with these issues}

\nosmurfduo
In this subsection we put the Hilbert-systems proposed above to actual use.  We begin with a discussion of minimal axiom systems, with and without \liii\ or \Di. %\ or \lb.
  Later on, we move to those inspired by concrete applications. We are not giving the details of these derivations here; some are available in existing references (and we give references in several cases), some are left for the reader as an exercise, and some will be published in future work  
  \cite{LitakV:otw}.
   %However, the reader is encouraged to 
  
  \medskip
  
\begin{foots}  
  %\nosmurf
For a calculus $\lva{X}$ defined by a list of axioms and rules,  write  $\lva{X}\vdash\phi$  to denote deducibility from all substitution instances of axioms/rules in $\lva{X}$ plus Modus Ponens. Whenever we have 
that for any $\phi$, $\lva{Y} \vdash \phi$ implies $\lva{X} \vdash \phi$, we write $\lva{X} \vdash \lva{Y}$. %When $\phi$ is of the form $\phi' \to \psi$, we also write $\phi' \dedp{\lva{X}} \psi$.
 For $\lva{X} \ded \phi \to \psi$,  $\loglzero \ded \phi \to \psi$, $\logla \ded  \phi \to \psi$  or $\iP \ded \phi \to \psi$ 
 (see Table \ref{tab:minax} below), write, respectively,  
 $\phi \dedp{\lva{X}} \psi$, $\phi \dedz \psi$, $\phi \dedm \psi$ and $\phi \dedv \psi$. In other words, we use $\dedm$ (and $\dedp{\lva{X}^-}$) 
  %to stress that $\phi \to \psi$ can be derived 
  for derivability without  instances of non-\ipc\ schemes involving 
 disjunction ($\liv$ or equivalently $\livalt$) and $\dedz$ for a still more restrictive case when deduction relies on monotonicity only. 
 %and $\dedv$ in the rare cases these rules are actually required.
%in contrast, the still more minimal deduction system $\dedz$ is only used to establish equivalence of several ways of defining $\logla$. 
Correspondingly, interderivability (equivalence) is denoted using, respectively $\eqd$, $\eqdp{\lva{X}}$, $\eqdz$, $\eqdm$ and $\eqdv$. 
Also, let us abbreviate $\lva{X} \ded \phi \tto \psi$,  $\loglzero \ded \phi \tto \psi$, $\logla \ded  \phi \tto \psi$  or $\iP \ded \phi \tto \psi$ as, 
respectively,  $\phi \dedlp{\lva{X}} \psi$, $\phi \dedlz \psi$, $\phi \dedlm \psi$ and $\phi \dedlv \psi$. Note that 
even the weakest of these relations, i.e., $\dedlz$ is transitive and contains $\dedz$; in fact, this is 
precisely essence of the minimal deduction system $\loglzero$.
 Finally, for deductions in $\Box$-only language, using $\iK$ as the minimal system, one can use similar conventions as above with $\Box$ as subscript (e.g, ``$\dedb$'' and ``$\eqdb$'').
\end{foots}

\subsubsection{Axiomatizations for minimal systems}

%Presenting full motivation for some principles will require further discussion; this task will have only been completed in \tadeusz{somewhere} %\S~\ref{prelo}.
% Instead of doing so here, let us  give some examples of deductions, starting from the minimal system and its fragments:

\nosmurfduo

%Let us begin with a discussion of the basic set of axioms:

\begin{lemma} \ \label{th:basicia}
\begin{enumerate}[a.] 
\item
The principles {\liii}, {\liiialtalt}, {\liiialt} and {\liiialtaltalt}  are equivalent over \loglzero.
\item
The principles {\liv} and {\livalt} are equivalent over \loglzero. 
%$\ipc + \loglzero + \liv \eqd \ipc + \loglzero + \livalt.$ 
%The principles {\lvii} and {\lviialt} are equivalent over \loglzero.
\item \label{blderiv} $\loglzero \vdash \lx$ and
$\loglzero \vdash \lxi$.
\end{enumerate}

%\medent
%The proofs of \textup{(}a\textup{)}, \textup{(}c\textup{)} and \textup{(}d\textup{)} work also in the disjunction-free fragment.
\end{lemma}
 
\begin{deriproof}
(a)  We just show some of these equivalences, leaving the rest to the reader:
%We prove the equivalence of {\liii} and {\liiialtalt}.

\medent
%First we assume {\liii}.% Suppose $\phi \tto \psi$.
%It follows that $(\phi \wedge \chi) \tto \psi$ and $(\phi \wedge \chi) \tto \chi$.
%Ergo, by \liii, we find $(\phi \wedge \chi) \tto (\psi \wedge \chi)$. 
\begin{tabular}{>{$}r<{$}>{$}l<{$}>{$}l<{$}}
\multicolumn{3}{l}{For $\loglzero + \liii \ded \liiialtalt$:} \\ [\tbskip]
\hspace{1cm} \phi \tto \psi   & \dedz & (\phi \wedge \chi) \tto \psi \wedge (\phi \wedge \chi) \tto \chi \\
 & \dedp{\loglzero + \liii} & (\phi \wedge \chi) \tto (\psi \wedge \chi). \\ [\tbskip]  
 \end{tabular}
 
 \medent
 \begin{tabular}{>{$}r<{$}>{$}l<{$}>{$}l<{$}}
 \multicolumn{3}{l}{For $\loglzero + \liiialtalt \ded \liii$:} \\ [\tbskip]   
\hspace{1cm} \phi \tto \psi \wedge \phi \tto \chi & \dedp{\loglzero +  \liiialtalt} & (\phi \wedge \phi ) \tto (\psi \wedge \phi) \wedge (\phi\wedge \psi) \tto (\chi \wedge \psi) \\
 & \dedz & \phi \tto (\phi \wedge \psi) \wedge (\phi\wedge \psi) \tto (\psi \wedge \chi) \\
& \dedz & \phi \tto (\psi \wedge \chi). 
\end{tabular}

\medent
\begin{tabular}{>{$}r<{$}>{$}l<{$}>{$}l<{$}}
\multicolumn{3}{l}{For $\loglzero + \liii \ded \liiialtaltalt$:} \\ [\tbskip] 
\hspace{1cm}\phi \strictif (\psi  \to \chi)   & \dedm & (\phi \wedge \psi)  \strictif (\psi \wedge (\psi  \to \chi)) \\ %\text{ by } \lnec \text { and } \lnorm\\
 & \dedz  & (\phi \wedge \psi)  \strictif \chi. %\\ [\tbskip]  %& \text{ by monotonicity of $\tto$}
 \end{tabular}

 \medent
 \begin{tabular}{>{$}r<{$}>{$}l<{$}>{$}l<{$}}
 \multicolumn{3}{l}{For $\loglzero + \liiialtaltalt \vdash \liiialtalt$:}\\ [\tbskip] 
\hspace{1cm} \phi \tto \psi   & \dedz & \phi \tto \psi  \wedge \psi \tto (\chi \to (\psi \wedge \chi)) \\
 & \dedz & \phi \tto (\chi \to (\psi \wedge \chi)) \\
 & \dedp{\loglzero + \liiialtaltalt} & (\phi \wedge \chi) \tto (\psi \wedge \chi).
 \end{tabular}

%Next we assume \liiialtalt. Suppose $\phi \tto \psi$ and $\phi \tto \chi$. If follows using {\liiialtalt}, that
%$(\phi \wedge \phi ) \tto (\psi \wedge \phi)$ and $(\phi\wedge \psi) \tto (\chi \wedge \psi)$.
%From these we find: $\phi \tto (\phi \wedge \psi)$ and $(\phi\wedge \psi) \tto (\psi \wedge \chi)$.
%Ergo, $\phi \tto (\psi \wedge \chi)$.

%The equivalence of {\liii} and {\liiialt} is easy:

%Moreover, the implication from {\liiialtaltalt} to {\liiialtalt} is also easy.
%We prove that
%{\liiialtaltalt} implies {\liiialtalt}. Suppose $\phi \tto \psi$.
%We also have $\psi \tto (\chi \to (\psi \wedge \chi))$.
%Hence $\phi \tto (\chi \to (\psi \wedge \chi))$.
%We may conclude $(\phi \wedge \chi) \tto (\psi \wedge \chi)$.

\medent
(b) The proof of the equivalence of {\liv} and {\livalt} is similar to the proof of the equivalence of
{\liii} and {\liiialtalt}.

\medent
Both claims of (c) are straightforward. Numerous such facts were listed by Iemhoff and coauthors 	\cite{Iemhoff01:phd,iemh:pres03,iemh:prop05}.
\end{deriproof}

\noindent
 As noted in existing references (cf., e.g., \cite[Chapter 3]{Iemhoff01:phd} or \cite[Theorem 2.5]{iemh:prop05}), there is some freedom in the choice of minimal rules:
 
\begin{fact} \label{fact:necunnec}
${\logla} \eqd \ipc + \bi + \bii + \lii + \liii.$ %This result also holds for the disjunction-free fragment.
\end{fact}

%\begin{remark} 
\begin{quest}
\label{que:subi}
%We have already mentioned that 
 Even in the absence of intuitionistic $\preceq$, the negation-free logic obtained by replacing $\to$ with $\strictif$ is a subintuitionistic logic \cite{Corsi87:mlq,Dosen93,CelaniJ01:ndjfl,CelaniJ05:mlq}.  %It is worth recalling that standard applications of the logic of bunched implications \lna{BI} \cite{OHearnP99:jsl,Pym02:book,PymOHY04:tcs}: by making \ipc\ and the minimal subintuitionistic logic in question share the same (distributive) lattice reduct. There are various ways of describing such combinations of logics, e.g., in terms 
  Is there a good a presentation of the minimal logic of $\tto$ in terms of  \emph{fusion} or \emph{dovetailing}/\emph{fibring} of \ipc\ and this minimal subintuitionistic logic, rather analogous to the logic of bunched implications \lna{BI} \cite{OHearnP99:jsl,Pym02:book,PymOHY04:tcs}? %in question share the same (distributive) lattice reduct. 
  %However, it is not obvious how well this would work, %whether a good axiomatization  can be obtained this way, %by such a generic procedure would  use the interaction between $\to$ and $\strictif$ in an optimal way, 
  %especially that
  Note that the analogy with  \lna{BI} is limited, e.g., both local and global consequence relations of $\strictif$ in the absence of $\to$ cease to be \emph{protoalgebraic} \cite{CelaniJ01:ndjfl,CelaniJ05:mlq}.
 \end{quest} 
%\end{remark}

\subsubsection{Collapsing and decomposing $\tto$}

\nosmurfduo
In Fact \ref{fact:boxcol}, we have observed that there are two syntactically similar conditions one can use to enforce \boxcol\!. 
%We also noted that one of them clearly implies the other. 
 Now we can prove syntactically their equivalence, which explains why we used the seemingly weaker one as \lb:

\begin{lemma} \ \label{th:bkax}
\begin{enumerate}[a.] 
\item 
$\bk^{-} \eqd \wk + \lix$,
%\avblue{Why not say that {\lb}, {\lix} and {\lbaltalt} are equivalent over {\logla}?}
 i.e., {\lb} and  {\lix} are equivalent over {\logla}.
\item $\bk^- \eqd \wk + \lbaltalt$,
 i.e., {\lb} and  {\lbaltalt} are equivalent over {\logla}.
\item \label{clausequiv} $\bk^- \ded \liv$ and consequently $\ivcol{X} \eqd \ivcol{X}^-$ for any $\lva{X}$. 
\end{enumerate}
\end{lemma}

\begin{deriproof}
(a)  $\wkdisf + \lix \ded \bk$ is obvious. For the opposite direction, we have:
%\medskip
%\begin{tabular}{>{$}r<{$}>{$}l<{$}>{$}l<{$}@{\hspace{1.5cm}}r}
\begin{alignat*}{3}
(\phi \wedge \psi)  \strictif \chi & \dedp{\bk} &\quad& \phi  \strictif \phi \wedge  \phi \strictif \top  \wedge \top \strictif (\phi \wedge \psi  \to \chi) & \\ %\hspace{3cm} \\%& \text{ by  \lb\ and \lnec}\\
& \dedz && \phi  \strictif \phi \wedge   \phi  \strictif (\phi \to \psi  \to \chi)  & \\ %& \text{ by  \lmono} \\ 
 & \dedm && \phi  \strictif (\psi  \to \chi)  & %\hfill \qedhere %\hfill \qedhere %& \text{ by \lnorm\ and \lmono.}
%\end{tabular} 
\end{alignat*}

\noindent
(b) We have:
\begin{alignat*}{3}
\phi \tto \psi & \dedz &\quad& ((\chi \to \phi) \wedge \chi) \tto \psi \\
& \dedp{\wkdisf + \lix} && (\chi \to \phi) \tto (\chi \to \psi) .  
\end{alignat*}

\noindent
In order to deduce $\lb$ from  $\lbaltalt$, simply substitute $\chi = \phi$ in $\lbaltalt$.

\medent
(c) is straightforward, by the properties of intuitionistic implication and normal modal operators.
\end{deriproof}

\begin{remark}\label{zakensmurf}
We presented one possible way to translate a $\tto$-logic $\ivsys{X}$ into a $\Box$-logic, to wit to take $\phi \tto \psi$ as an abbreviation for $\Box(\phi \to \psi)$. 
This translation relates $\ivsys{X}$ to its extension $\ivcol{X}$, which is term-equivalent to a $\Box$-logic. 
 Another way, studied in detail by Iemhoff and coauthors   \cite{Iemhoff01:phd,iemh:pres03,iemh:prop05}, 
takes the validity of $\lxi$ as a starting point and translates $\phi \tto \psi$ as $\Box\phi \to \Box\psi$. %; see  	for more on this
 A third interpretation of $\tto$ in terms of $\Box$ relating $\iPLL$ and $\iPLLa$ is discussed in Remark \ref{rem:laxembed};  
it builds on a $\iPLLa$-decomposition of $\tto$ in terms of $\to$ provided by Lemma \ref{lem:derpll}\ref{plaacollapse}. 
%grmbl
%We can view this collapse of $\tto$ as a special case of reduction of $\tto$ to unary modalities. 
 For more on reductions of $\tto$ to unary modalities see \cite{LitakV:otw}.
\end{remark}

%\paragraph{Classical trivialization of $\strictif$} 
\nosmurf
We have suggested that the degeneration of $\strictif$ in the presence of classical laws can be derived syntactically. 
In fact, this can be obtained as a consequence of an equivalence derivable over the intuitionistic base but, atypically, using disjunction with its $\liv$ axiom in an essential way:
% And indeed, there is more than one way to go about it. Here is one example derivation. Assuming the law of excluded middle (\lna{em}) $\phi \vee \neg\phi$, we get:

\begin{lemma} \label{th:emdi}
We have:
$\psi \tto \chi \eqdv (\psi \vee \neg\psi) \tto (\psi \to \chi)$.
\end{lemma}

\begin{deriproof}
These derivations are fairly straightforward. In one direction:

\[
\begin{tabular}{>{$}r<{$}>{$}l<{$}>{$}l<{$}}
\psi  \strictif \chi & \dedz & \psi  \strictif (\psi \to \chi) \wedge \neg \psi  \strictif (\psi \to \chi) \\% & \text{by \lmono\ and \lnec} \\
& \dedv & (\psi \vee \neg\psi) \strictif (\psi \to \chi)  \\ %&  \text{by \ldis} \\
 %& \ded \top \strictif (\psi  \to \chi) & \text{by $\lna{em}$}.
\end{tabular}
\]

\medent
In the opposite direction:
%
%\medskip
%
%\begin{tabular}{>{$}r<{$}>{$}l<{$}>{$}l<{$}}
\begin{alignat*}{3}
(\psi \vee \neg\psi) \strictif (\psi \to \chi) & \dedm &\quad& ((\psi \vee \neg\psi) \wedge \psi) \strictif \chi  & \\
& \dedz && \psi  \strictif \chi  &  \qedhere
\end{alignat*}
%\end{tabular}
\end{deriproof}

%Let us note that there is an equivalence provable without the law of excluded middle: namely, even in \wk\ we have that %is equivalent to 
%$$\psi \tto \chi \dashv\vdash (\psi \vee \neg\psi) \tto (\psi \to \chi).$$ 
\nosmurf
Nevertheless, as $(\psi \vee \neg\psi) \tto (\psi \to \chi)$ is parametric in the antecedent of strict implication, %one cannot see it as 
 it does not seem a satisfying reduction of $\tto$ to $\to$.  Let us also note in passing that if one 
adds $\tto$ to Johansson's minimal logic instead of \ipc, even this transformation does not work anymore. Moreover, there is no one-variable formula 
$\phi(p)$ in the disjunction-free fragment of the intuitionistic signature s.t. \mbox{$p  \strictif q \eqdm \phi(p) \strictif (p \to q)$} and $\cpc \ded \phi(p)$, cf. Example \ref{ex:nondnegbox}.

\begin{quest} \label{que:disfree}
In general, we stick to extensions of $\laipc$, but let us make a digression concerning a language without all standard connectives. Suppose we define
 $[\phi]\psi$  as $(\phi\vee \neg\,\phi)\tto \psi$. As we saw above, $\phi \tto \psi$ is equivalent with $[\phi](\phi \to\psi)$.  
 Is there an elegant axiomatization for the minimal fragment of the language with $[\cdot](\cdot)$? It seems richer than the disjunction-free fragment of $\latto$. %the
%language with just $\tto$.
 %Is there an elegant axiomatization for the minimal disjunction-free system with  $[\cdot](\cdot)$? %What does the disjunction-free fragment over $\loglb$ with the connective  $[\cdot](\cdot)$
%look like?
\end{quest}

\begin{corollary} \label{cor:emderbox}
$\iP + \lna{em} \vdash \bk$. % \Box (\phi \to \psi) \eqv \phi \tto \psi$.
\end{corollary}

\begin{deriproof}
%One half of this equivalence is $\lx$, which as we noted before is valid 
Use Lemma \ref{th:emdi} and the fact that $\cpc \vdash \lna{em} \eqv \top$.
\end{deriproof}

\begin{remark}
This is one of very few places in this section where we need full $\iA$ rather than $\iP^-$, i.e., where $\Di$ is used in an essential way. 
This happens for a very good reason: 
it is not possible to derive $\bk$ from $\iP^-+{\sf em}$. One can see this, e.g., by considering the interpretation of $\phi\tto \psi$ as $\Box\phi \to \Box \psi$.
\end{remark}

%\medent
\begin{foots}
In Appendix~\ref{sec:picon} we will explain that the logic {\sf ILM} of $\Pi_1^0$-conservativity and interpretability
 corresponds to ${\mathrm c}\hyph{\sf PreL} := {\mathrm i}\hyph{\sf PreL}^{-} + \lna{em}$.
This provides a proof that even ${\mathrm c}\hyph{\sf PreL}$ does not extend $\bk$. The proof may use either the arithmetical interpretation or the Veltman semantics used for
{\sf ILM}.
\end{foots}

\nosmurf
We will discuss collapsing and decomposing further in a later paper \cite{LitakV:otw}; see also remarks preceding Theorem \ref{th:mainco} below.

%\tadeusz{Can one replace $\lna{em}$ by, say, $\lna{dneg}$ and $\iP$ by $\iP^-$ in the above Corollary??} \tadeusz{NOPE. Things go wrong in the absence of $\liv$.}

\subsubsection{Derivations between arithmetical principles}
\nosmurfduo
We turn our attention to derivations between principles of central importance, especially from the point of view of arithmetical 
interpretations.
 %(\S~\ref{sec:arith}, \S~\ref{sec:apppre}, Appendix \ref{sec:picon} and \ref{sec:inter}), but also other applications and interpretations  (\S~\ref{sec:arrows}). Here is the first example:

\begin{lemma} \label{lem:lsder}
We have:
\begin{enumerate}[a.]
\item $\ws^- \eqd \loglb^- + \lSalt$, i.e., over $\loglb^-$, the principles $\lS$ and $\lSalt$ are equivalent.
\item \label{hugconv} $\ws^- \eqd \loglb^- + \lSaltalt$, i.e., over $\loglb^-$, the principles $\lS$ and $\lSaltalt$ are equivalent.
\item In the presence of $\lSalt$, $\li$ is derivable using just Modus Ponens.
\end{enumerate}
\end{lemma}

\begin{deriproof}
(a) In order to derive $\lS$ from $\lSalt$, just substitute $\top$ for $\phi$. In the converse direction, use $\lx$. 
\medent
(b) In one direction:
\begin{alignat*}{3}
\phi \tto \psi & \dedp{\ws^-} &\quad& \phi \to (\Box\phi \wedge \phi \tto \psi) \\
& \dedz && \phi \to \Box\psi
\end{alignat*}

\noindent
In the reverse direction, substitute $\phi = \psi$ in \lSaltalt.

\medent
(c) is trivial.
\end{deriproof}

\nosmurf
Hence, %in the case of the  axioms, its
 axiomatizations of ``strength'' in terms of $\Box$ and in terms of $\tto$ 
yield the same logic over $\logla$. %\avblue{\logla?}
  As we  are going to see below, this is a relatively rare phenomenon. %Such a nice symmetry, as we will see below, breaks down with L\"ob-related 
%and arithmetical principles: formulations in terms of $\tto$ are typically stronger than those in terms of $\Box$.  %\tadeusz{We need to link to counterexamples to this effect}
Still, many well-known modal derivations %and dependencies known in their $\Box$-variants 
  can be easily translated into the $\tto$-setting, e.g., %i.e., a $\tto$-analogue of 
  a derivation  of $\lna{4}$ from the L\"ob axiom:

%\noindent
%The system {\logle} is given by $\logla + \lvi$.

\begin{lemma} \label{lem:gla}
%The system {\logle} is equivalent (over $\logla$) to $\logla + \biv + \lv$.
${\logle} \eqd \logla + \biv + \lv$. It follows that, over $\logla  + \lv$, the
principles $\biv$ and $\lvi$ are interderivable.
%The result also works for the disjunction-free fragment.
\end{lemma}

\begin{deriproof}
%We work in \logla. 
%We assume \lvi. L\"ob's principle {\biv} is an immediate consequence using the fact that $\lxi$ is derivable in \logla.
As ${\logla} \ded \lxi$, we get immediately that ${\logle} \ded \biv$. For $\lv$:

%\medskip
%\begin{tabular}{>{$}r<{$}>{$}l<{$}>{$}l<{$}}
\begin{alignat*}{2}
\phi & \dedlz &\quad& \Box(\phi \wedge \Box \phi) \to (\phi \wedge \Box \phi) \\
 & \dedlp{\loglzero + \lvi}  && \phi \wedge \Box \phi \\
& \dedlz && \Box \phi.
\end{alignat*}
%\end{tabular}

%\medskip
\noindent
For the converse direction,
%\medskip

%We assume {\biv} and \lv. We have 
%\begin{tabular}{>{$}r<{$}>{$}l<{$}>{$}l<{$}}
\begin{alignat*}{2}
(\Box\phi \to \phi) & \dedlp{\loglzero + \lv} &\quad& \Box(\Box\phi \to \phi)  \\
& \dedlp{\loglzero + \biv} && \Box \phi.  
\end{alignat*}
%\end{tabular}

%\medskip
 %Hence, $\logla + \biv + \lv \ded (\Box\phi \to \phi) \tto \Box \phi$.
\noindent
Moreover, $\logla  \ded (\Box\phi \to \phi) \tto (\Box\phi \to \phi).$ Ergo, by \liiialt, we derive  
$(\Box\phi \to \phi) \tto \phi$ (note this is the only use of normality in the entire proof).
\end{deriproof}

%With the above le

\begin{lemma}\ \label{lem:sa}
\begin{multicols}{2}
\begin{enumerate}[a.]
\item $\loglb^- + \lv \ded \biii$.
\item $\ws^- \ded \lv$. \columnbreak
\item $\ws^- \ded \lxii$.
\item 
$\ws^- + \biv \ded \logle$.
%$\ws^- + \biv \ded \loglg^-$.
%\avred{The minus should be above the {\sf a}. One binary macro or two unary ones?}
\end{enumerate}
\end{multicols}
\end{lemma}

\begin{deriproof}
Straightforward, using Lemma \ref{lem:gla} for the last claim.
\end{deriproof}

\begin{lemma}{\textup{\cite[Cor. 2.6 and 2.7]{iemh:prop05}}} \
%\begin{itemize}
%\item \textup{\cite[Cor. 2.7]{iemh:prop05}} 
%%$\loglzero + \lvii \eqd \loglzero + \lviialt.$ \avblue{Formulation with `over {\ldots} are equivalent?'}
{\lvii} and {\lviialt} are equivalent over {\loglzeromin}. Similarly, 
%\item \textup{\cite[Cor. 2.6]{iemh:prop05}} 
%%$\loglzero + \lviii \eqd \loglzero + \lviiialt.$ 
{\lviii} and {\lviiialt} are equivalent over {\loglzeromin}.
%\end{itemize}
\end{lemma}

\begin{deriproof}
We just give a proof of the first of these equivalences, leaving the second to the reader (or to the references; cf. also Lemma \ref{th:bkax}). In one direction:
\begin{alignat*}{3}
\phi \tto \psi & \dedz &\quad& ((\Box \psi \to \phi) \wedge \Box \psi) \tto \psi &\\
& \dedp{\loglzeromin + \lvii} &&  (\Box \psi \to \phi) \tto \psi.
\end{alignat*}
For the converse,

\begin{alignat*}{3}
 (\phi \wedge \Box\psi) \tto \psi & \dedp{\loglzeromin + \lviialt} &\quad& (\Box \psi \to (\phi \wedge \Box \psi) )\tto \psi & \\
 & \dedz && (\Box \psi \to \phi) \tto \psi & \\
 & \dedz && \phi \tto \psi. & \qedhere
 \end{alignat*}
\end{deriproof}

\begin{lemma} \label{lem:gwader}
We have:
%\begin{multicols}{2}
\begin{enumerate}[a.]
\item \label{wtol}
${\loglf} \vdash \logle.$ 
\item \label{matowa}
$\logle + \lviii \vdash \lvii$. In other words, $\logld \eqd \logle + \lviii$.  %\columnbreak
\item 
$\logle+\lxii \vdash \lvii.$
\item \label{matobox} $\iP^- + \lviii \ded (\Box\phi \tto \psi) \to \Box(\Box\phi \to \psi)$. 
\end{enumerate}
%\end{multicols}
\end{lemma}

\begin{deriproof}

\begin{enumerate}[(a)]
\item {\lvi} is immediate from \lviialt.
\Item
 %We work in $\logle + \lviii$.
%We  prove {\lvii}. Suppose $
%\tadeusz{We have not proved yet equivalence between $ \lviii$ and $\lviiialt$. Is it because it's analogous to derivations above? This should be stated.}
\begin{alignat*}{2}
(\phi \wedge \Box \psi) \tto \psi & \dedp{\loglzero + \lviiialt} &\quad&
\phi \tto (\Box\psi \to \psi) \\
&  \dedp{\loglzero + \lvi} && \phi \tto \psi.
\end{alignat*}
\Item
 %We work in $\logle + \lxii$.
%We  prove {\lviialt}. Suppose $\phi \tto \psi$. It follows that $
\begin{alignat*}{3}
\phi \tto \psi & \dedp{\loglzero + \lxii} &\quad& \Box (\phi \tto \psi)  & \\ 
 & \dedz && \Box(\Box \phi \to \Box\psi) & \\ 
 & \dedz &&  \Box ((\Box\psi \to \phi) \to (\Box \phi \to \phi)) & \\
& \dedm && (\Box \psi \to \phi) \tto (\Box \phi \to \phi) &  \\
& \dedp{\logle} && (\Box \psi \to \phi) \tto \phi. & 
\end{alignat*} 
 %It follows that $(\Box\psi \to \phi) \tto \phi$.
%Combining this with $\phi \tto \psi$ gives us the desired $(\Box\psi \to \phi) \tto %\psi$.
%\end{alignat*}
\item
We substitute $\Box\phi$ for $\phi$ and $\phi$ for $\chi$ in $\lviii$. \qedhere
 \end{enumerate}
\end{deriproof}

%\begin{tomove}
\nosmurf
%The item \ref{matowa} can be reformulated as 
%\begin{remark}
Examples \ref{notltow}, \ref{eerlijkesmurf} and \ref{eerlijkesmurfalt}  illustrate that clauses \refeq{wtol}  and \refeq{matowa} cannot be reversed. 
%Examples \ref{eerlijkesmurf} and \ref{eerlijkesmurfalt} illustrate that neither can \refeq{matowa}. 
%\end{remark}

\begin{lemma}
$ \loglb + \lv \vdash \bv$. 
\end{lemma}

\begin{deriproof}
We reason in \loglb + \lv. We have $\phi \tto (\phi\vee \Box \psi)$. We also have $\psi \tto ( \phi \vee \Box\psi)$, by
\lv. So, $(\phi\vee\psi) \tto (\phi \vee \Box \psi)$. We may conclude $\Box(\phi \vee \psi) \to
\Box(\phi \vee \Box \psi)$.
\end{deriproof}

\nosmurf
This implies that the logics  {\loglg}, {\loglh} and  {\loglc}  are not conservative over
\logba. %We can see this by noting that Leivant's Principle {\bv} is derivable in $\loglb + \lv$,
%which is contained in \loglg. 
 Both  {\loglg} and {\loglh}  are conservative over $\logba+{\bv}$ \cite{iemh:prop05}.
%\end{tomove}

\subsubsection{More derivations}
\nosmurfduo
Derivations discussed in the remainder of this section are mostly of importance in \S~\ref{sec:arrows}, although, e.g., Lemma \ref{th:mhcderiv}\ref{mhcbeq} will be also relevant in \S~\ref{sec:hastar}:

\begin{lemma} \label{th:mhcderiv}
We have:

\begin{enumerate}[a.]
\item \label{mhcbeq} $\iK^- + \mHCb \eqdb \iK^- + \mHCbalt$
\item \label{mhcttob} $\iP^- + \mHCl \ded \mHCb$.
\item \label{mhcdi} $\imHCl^- \ded \lb$ and consequently $\imHCl^- \ded \Di$. 
\item $\imHCl \ded \lviii$. %\avblue{Also a minus here?}
\item $\iKMl \ded \lvii$.
%\item $\iKMl \ded \lvii$.
\end{enumerate}
\end{lemma}

%\begin{remark}
\nosmurf
Clause \refeq{mhcdi} implies that the notation ``$\imHCl^-$'' is redundant.
 Example \ref{ex:mhcnonequiv} below illustrates that clause \refeq{mhcttob} cannot be reversed.
%\end{remark}

\begin{deriproof}
(a) To deduce $\mHCbalt$, substitute $\psi \to \phi$ for $\phi$. In the converse direction, use the fact that $\Box\phi \dedb \Box(\psi \to \phi)$.

\medent
(b) By  the (a) clause, it is enough  to derive $\mHCbalt$ with $\mHCb$  and this follows via $\lx$.

\medent
(c) Thanks to \mHCl, it is enough to prove those two claims: %\avred{Why does this give so much white space?} @@@
\begin{alignat*}{3}
%(\phi \to \psi) & \dedp{\iKMl} &\quad& (\Box\psi \to \phi) \tto (\Box\psi \to \psi) & \text{\quad and}\\ 
 \phi \to \psi  & \;\; \dedp{\iS} &\; & \Box(\phi \to \psi) & \text{\quad and}\\ 
 \phi  \wedge \phi \tto \psi  & \;\; \dedp{\iS} && \Box(\phi \to \psi). & 
%\phi & \dedp{\iKMl} &\quad& (\Box\psi \to \phi) \tto (\Box\psi \to \psi) & 
\end{alignat*}

\noindent
The first of them is immediate. The second uses $\lS$ and $\lii$.

\medent
(d) Note that thanks to \mHCl, it is enough to prove those two claims:
\begin{alignat*}{3}
\phi \to \psi & \;\; \dedp{\imHCl} & & (\Box\chi \to \phi) \tto (\Box\chi \to \psi) & \text{\quad and}\\ 
\phi \tto \psi \wedge \phi & \;\; \dedp{\imHCl} &\quad& (\Box\chi \to \phi) \tto (\Box\chi \to \psi) & 
\end{alignat*}

\noindent
For the first of them, just note that  
\[
\phi \to \psi \dedp{\ipc}  (\Box\chi \to \phi) \to (\Box\chi \to \psi)
\]
and use $\lSalt$. For the second, observe that 
\[
\phi \tto \psi \wedge \phi  \dedp{\ws} \top \tto (\alpha \to \psi)
\]

\noindent
for any $\alpha$ (including $\Box\chi$) and use \bii.

\medent
(e) First, observe that $\iKMl \ded \lvi$ by Lemma \ref{lem:sa}. Second, note that thanks to \mHCl, it is enough to prove those two claims:
\begin{alignat*}{3}
%(\phi \to \psi) & \dedp{\iKMl} &\quad& (\Box\psi \to \phi) \tto (\Box\psi \to \psi) & \text{\quad and}\\ 
 (\phi \wedge \Box \psi) \to \psi &\;\;  \dedp{\iKMl} &&\phi \tto \psi & \text{\quad and}\\ 
 \phi \wedge \Box \psi &\;\; \dedp{\iKMl} &\quad&\phi \tto \psi. & 
%\phi & \dedp{\iKMl} &\quad& (\Box\psi \to \phi) \tto (\Box\psi \to \psi) & 
\end{alignat*}

\noindent
The first of these is a straightforward application of $\bvii$ and $\lSalt$. The second is valid even in $\iP^-$.
\end{deriproof}

\begin{lemma} \label{lem:derpll}
We have:
\begin{multicols}{2}
\raggedcolumns
\begin{enumerate}[a.]
\item $\iP^- + \aCF \ded \CF$.
\item $\iP^- + \aApp \ded \aCF$.
\item $\iP^- + \Hug \ded \aCF$.
\item $\iPLLa^- \ded \aApp$. %\columnbreak \sloppy %\noindent
\item \label{plaahug} $\iPLLa^- \ded \Hug$.
\item \label{plaacollapse} $\phi \tto \psi \eqdp{\iPLLa^-} \phi \to \Box \psi$
\item \label{plaadi} $\iPLLa^- \ded \Di$.
\end{enumerate}
%\medskip
\end{multicols}
\end{lemma}

%\begin{remark}

%\end{remark}

\begin{deriproof}
For (a), use $\lxi$. For (b), substitute $\top$ for $\phi$. For (c), substitute $\Box\psi$ for $\phi$ in $\Hug$. 
For (d), note that $$\phi \wedge (\phi \tto \psi) \;\dedlp{\lna{S}} \;\Box\psi$$
and use \aCF. For (e), reason as follows:
\begin{alignat*}{3}
\phi \to \Box\psi & \;\;\dedp{\iPLLa^-}\; &&  \phi \tto \Box\psi \wedge \Box\psi\tto\psi \\
& \;\; \dedz\; && \phi \tto \psi.
\end{alignat*}
 (f) is a corollary of the previous clause and Lemma \ref{lem:lsder}\ref{hugconv}. This deconstruction of 
 $\tto$ in terms of $\to$ yields $\Di$ and thus (g) using the laws of intuitionistic logic.
\end{deriproof}

%\nosmurf
\nosmurf
Example \ref{ex:cfsep} illustrates irreversibility of several clauses in this lemma.
%The above examples may suggest 
So far,  principles involving $\tto$ tended to be  stronger than their 
relatives formulated in $\labox$. %using $\Box$ only. 
 It is indeed quite often but not always the case. For example, in the case of 
``semi-linearity'' axioms, the situation is reversed: %$\Box$-variant $\bLin$ is obviously the stronger one: %and equivalent to the weaker 
%$\aLin$ only in rather pathological setups, in particular of course in the presence of $\lb$:

\begin{fact} \label{fact:linnon}
 We have:
\begin{itemize}
\item $\iP^- + \bLin \ded \aLin$.
\item $\bk + \aLin \ded \bLin$, in particular $\imHCl + \aLin \ded \bLin$.
\end{itemize}
\end{fact}

\nosmurf
%Even quite powerful axioms, do no
 It is hard to make $\aLin$ and $\bLin$ coincide in the absence of $\lb$, cf. Example \ref{ex:linavsb}. Nevertheless, here is an important exception,
 %Unless, that is, these additional axioms are very strong indeed---and this is an 
 used in \S~\ref{sec:guarded} below:

\begin{lemma} \label{lem:kmlin}
We have:
\begin{enumerate}[a.]
\item $(\Box\phi \to \Box\psi) \to \Box(\phi \to \psi) \dedz \lb$.
\item $\iKMlinb \ded (\Box\phi \to \Box\psi) \to \Box(\phi \to \psi)$.
\item $\iKMlinb \eqd \iKMlin$. That is, not only both systems are notational variants of each other, 
but the $\lb$ axiom can be derived from $\iKMlinb$.
\end{enumerate}
\end{lemma}

\begin{deriproof}
(a) is straightforward using $\lx$ and $\lxi$.

(b) is the interesting one \tadeusz{I will add a proof}

(c) is a corollary of the two preceding clauses.
\end{deriproof}

%%%%%%%%%%%%%%

\section{Arithmetical interpretations: provability and preservativity \AIp} \label{sec:arith} \label{sec:preservativity}
\nosmurfduo
In \S~\ref{sec:completeness}  we continue the discussion of the modal side of our calculi.
 But now, we cannot postpone any further the presentation of our original motivation for studying
 constructive $\tto$ and a number of its axiom systems %the arithmetical %É
%Some of the axioms we discussed above may seem rather exotic. Why do we focus on them?
 %As said above, a major motivation for studying constructive $\tto$ and several axioms introduced in  
 %the preceding section is provided by the arithmetical
  %interpretation of $\Box$ and $\tto$. 
 %More specifically, the original interest in constructive $\tto$ came from its arithmetical  %interpretation 
 in terms of \emph{$\Sigma^0_1$-preservativity} for an arithmetical theory $T$: 
%This notion is defined as follows (cf. \S~\ref{prelo}):
\begin{itemize}
\item
 $A \tto_T B$ if, for all $\Sigma^0_1$-sentences $S$, if
$T \vdash S \to A$, then $T \vdash S \to B$.
\end{itemize}

\nosmurf 
In order to  provide a framework for such interpretations of modal connectives, we introduce the notion of a \emph{schematic logic}. This notion can be given a very general treatment.
However, for the purposes of this paper, we will restrict ourselves to the case of arithmetical theories, studying propositional logics of
theories (\S~\ref{sec:proptho}), provability logics (\S~\ref{sec:provo}) and our true target: preservativity logics (\S~\ref{prelo}). 
%Our true target in this section is preservativity and its connection with the Lewis arrow.
%\medent
  For an instructive contrast, we provide some extra information 
 about logics for $\Pi^0_1$-conservativity and  interpretability in Appendices~\ref{sec:picon} and \ref{vrolijkesmurf}.   %These logics provide an instructive contrast
%to the logic of preservativity. 

%%%%%%%%%%%%%%%%%%%%
\subsection{Schematic logics} \label{sec:schema}
\nosmurfduo
\emph{An arithmetical theory $T$} is, for the purposes of this paper,
 an extension of i-{\sf EA}, the intuitionistic version of Elementary Arithmetic, in the arithmetical language.\footnote{The classical
 theory {\sf EA} is $\mathrm{I}\Delta_0+{\sf Exp}$. This theory consists of the basic axioms for zero, successor, addition, multiplication
 plus $\Delta_0$-induction plus the axiom that states that exponentiation is total. The theory i-{\sf EA} is the same theory only
 with constructive logic as underlying logic. The theory proves the decidability of $\Delta_0({\sf exp})$-formulas. Some
  basic information about constructive arithmetic can be found in \cite{troe:cons88vol1,troe:meta73,Dragalin88:trams}.}
  We demand that the axiom set of $T$ is given by a $\Delta_0({\sf exp})$-formula. 
 
Let $\mathcal L_{\arbop_0, \ldots, \arbop_{k-1}}$ be the language extending $\mathcal L$ %$\mathcal L$ of propositional logic
with operators $\arbop_0$, \ldots, $\arbop_{k-1}$, where $\arbop_i$ has arity $n_i$. 
%We call the resulting language .
Let a function $F$ be given that assigns to every $\arbop_i$ an arithmetical formula $A(v_0,\ldots,v_{n_i-1})$, where all free variables are among the 
variables shown. We write $\arbop_{i,F}(B_0,\ldots, B_{n_i-1})$ for $F(\arbop_i)(\gnum{B_0},\ldots,\gnum{B_{n_i-1}})$. Here $\gnum{C}$ is the numeral of the
G\"odel number of $C$.
Suppose $f$ is a mapping from the propositional atoms to
arithmetical sentences. 
We define  $(\phi)_F^f$ as follows:
\begin{itemize}
\item
$(p)_F^f := f(p)$
\item
$(\cdot)_F^f$ commutes with the propositional connectives
\item
$ (\arbop_i(\phi_0,\ldots,\phi_{n_i-1}))_F^f := \arbop_{i,F}((\phi_0)_F^f,\ldots, (\phi_{n_i-1})_F^f)$
\end{itemize}

\nosmurf 
Let $T$ be an arithmetical theory.
We say that a modal formula in our given signature is \emph{$T$-valid} w.r.t. $F$ if, for all assignments $f$ of arithmetical sentences to
the propositional atoms, we have $T \vdash (\phi)_F^f$. We write $\Lambda_{T,F}$ for the set of modal formulas that are $T$-valid w.r.t. $F$. 
 Of course,   we will focus exclusively on ``natural'' $F$ yielding well-behaved  $\Lambda_{T,F}$ with interesting properties. 
 %will  mostly fail to have any good or interesting properties. We will be interested in $F$ that
%assign very good arithmetical formulas to the operators. 

%%%%%%%%%%%%%%%%%%%%%%%%
\subsection{Propositional logics of a theory} \label{sec:proptho}
\nosmurfduo
Let us first consider the case where our finite set of modal operators is empty. If $T$ is consistent and classical, then $\Lambda_T := \Lambda_{T,\emptyset}$
is, trivially, precisely  {\sf CPC} and if $T$ is Heyting
Arithmetic ({\sf HA}), then $\Lambda_T$  has the \emph{de Jongh property}: $\Lambda_T = {\sf IPC}$. %is  precisely {\sf IPC}, %Intuitionistic Propositional Logic.
%The property of theories that  is called 

There are theories for which $\Lambda_T$
 is an intermediate logic strictly between {\sf IPC} and {\sf CPC}. 
 De Jongh, Verbrugge {\&} Visser \cite{viss:inte11} show that whenever $\Theta$  is an intermediate logic with the finite frame property (cf. \S~\ref{sec:completeness}) and $U$ is the result of extending ${\sf HA}$ with  all axioms of $\Theta$ as schemes, $\Lambda_{U} = \Theta$.
 
 For some theories like \emph{Markov's Arithmetic} ${\sf MA} = {\sf HA}+{\sf MP} + {\sf ECT}_0$, where
%\begin{itemize}
%\item 
${\sf MP}$ is the \emph{Markov's Principle}  (\cite[4.5, p.203]{troe:cons88vol1}, \cite[1.11.5, p.93]{troe:meta73}): %\newline
%\qquad\qquad 
\[
(\forall x (Ax \vee \neg Ax) \wedge \neg\neg\exists x Ax) \to \exists x Ax.
\]
and  ${\sf ECT}_0$ is the \emph{Extended Church's Thesis} (cf. Appendix \ref{sec:real}), 
 %\end{itemize}
  the
characterization of the set of valid principles is an open problem connected to the question of the propositional logic of realizability.
See e.g. \cite[\S~13]{plis:surv09}.
 For more on intuitionistic schematic logics see  \cite{smor:appl73,viss:rule99,plis:surv09,viss:inte11,arde:sigm14}.
%In \Subsection~\ref{dejo}, we will sketch a proof that Heyting Arithmetic has the de Jongh Property.

%%%%%%%%%%%%%%%%%%%%%%%%%
\subsection{Provability Logic} \label{sec:provo}
\nosmurfduo
Next we consider the extension of propositional logic with a unary modal operator $\Box$. 
 It allows numerous interesting arithmetical interpretations, but %of a unary modal operator. 
 at this point we focus on the interpretation of $\Box$ as provability. Consider any arithmetical
theory $T$. We assume that $T$ comes equipped with a $\Delta_0({\sf exp})$-predicate $\alpha_T$ that gives
the codes of its axiom set. 
Let provability in $T$ be arithmetized by ${\sf prov}_T$. We note that  $T$ really occurs in the guise of $\alpha_T$. 
We set $\InF{0}{T}(\Box) := {\sf prov}_{\sf T}(v_0)$. Let $\Lambda^\ast_T := \Lambda_{T,\InF{0}{T}}$. 
%\medent
 Intuitionistic L\"ob's Logic {\logba} is contained in all $\Lambda^\ast_T$, where $T$
is an arithmetical theory in the sense of this paper.  This insight is
due to L\"ob \cite{loeb:solu55}.\footnote{Three remarks are in order. The
fact that L\"ob's Principle follows from L\"ob's work was noted by Leon
Henkin who was the referee of L\"ob's paper. Secondly, L\"ob's proof of
L\"ob's Principle is fully constructive and goes through even in constructive
versions of ${\sf S}^1_2$. Thirdly, Kripke's proof of L\"ob's Principle from the Second
Incompleteness Theorem is not constructive ---even if we give the Second Incompleteness
Theorem the form: if a theory proves its own consistency then it is inconsistent.}

\begin{remark}
In many treatments of intuitionistic modal logic the interdefinability of $\Box$ and $\Diamond$ fails and
$\Box$ and $\Diamond$ both are treated as primitive operations.  This is not so in 
 the context of provability logic for constructive theories and its extensions. Here the connective $\Diamond$ is always  defined as
 $\neg\Box \neg$. Thus, $\Box$, which signals the existence of a proof, is the positive notion and 
$\Diamond$ is the negative less informative notion. We note that $\neg \Diamond \neg$ is equivalent to 
$\neg\neg\Box \neg\neg$ which, in the context of theories like {\sf HA}, is certainly weaker than $\Box$.
\end{remark}

\begin{remark}
One of the first global insights into schematic theories is due to George Gargov \cite{garg:prov84}:  
 they inherit the disjunction property from the underlying arithmetic theory. Thus, if an extension
of i-{\sf EA} has the disjunction property, then so has its provability logic.\footnote{Interestingly, Gargov's argument itself
uses classical logic.}
\end{remark}

\nosmurf
The theory {\logbb} is obtained by extending {\logba} with classical logic.
If $T$ is a $\Sigma_0^1$-sound classical theory, then  $\Lambda_T^\ast = \logbb$. This insight is due to Solovay \cite{solo:prov76}.
In contrast, the logic {\logba} is not complete for {\sf HA}. 
 The system for preservativity logic $\loglc+{\sf V}$ discussed in \Subsection~\ref{prelo} derives many more arithmetically valid principles
for the provability logic of {\sf HA} underivable in \logba, e.g., %Here we just give a few examples.
%The following principles are valid for {\sf HA} but not derivable in \logba.
\begin{itemize}
\item
$  \Box \neg\neg\,\Box \phi \to  \Box\Box\phi$.
\item
$ \Box(\neg\neg\, \Box \phi \to \Box\phi)\to \Box\Box\phi$
\item
$ \Box(\phi\vee \psi) \to \Box(\phi \vee \Box \psi)$. ({\bv})
\end{itemize}

\noindent
We note that the first principle is a consequence of classical \logbb, but the second and third are not.
This illustrates that  $\Lambda^\ast_{(\cdot)}$ is not monotonic. To make this understandable, the reader may note that
 we both change the theory and the interpretation of the modal operator.

%We note that if we
 Note that substituting $\Box\bot$ for $\phi$ and $\neg\, \Box\bot$ for $\psi$ in %Leivant's principle 
 {\bv} yields %, we obtain: 
\begin{eqnarray*}
\vdash \Box(\Box \bot \vee \neg \, \Box \bot) & \to &  \Box(\Box \bot \vee \Box \neg \, \Box \bot) \\
& \to & \Box (\Box \bot \vee \Box \bot) \\
& \to & \Box\Box \bot
\end{eqnarray*}

\noindent  Hence, adding {\bv} to {\logbb} yields $\Box\Box \bot$, i.e., Leivant's principle is `weakly inconsistent' with classical logic over
\logba.

We write $\Lambda \boxplus \Lambda'$ for the closure of $\Lambda \cup \Lambda'$ under modus ponens.
Our insight above yields: $\logbb + \Box\Box \bot \subseteq \Lambda^\ast_{\sf HA}\boxplus \Lambda^\ast_{\sf PA}$. 
In fact, by Theorem~\ref{sillysmurf}, we have:   $\Lambda^\ast_{\sf HA}\boxplus \Lambda^\ast_{\sf PA} = \logbb + \Box\Box \bot$. 

\begin{theorem}[Silly Upperbound] \label{sillysmurf}
We have: \[(\Box\Box\bot \to \neg\neg\,\Box\bot) \not\in \Lambda^\ast_T \;\text{ iff }\; \Lambda^\ast_T \subseteq \logbb + \Box\Box\bot.\]
\end{theorem}

\noindent Our proof presupposes knowledge of the proof of Solovay's Theorem. The proof can be skipped since nothing but the Silly Upperbound rests on it.

\begin{proof}
``$\to$'' Suppose $\phi\in  \Lambda^\ast_T$ and $\logbb + \Box\Box\bot \nvdash \phi$.
Then there is a counter Kripke model of depth 2 to $\phi$, say with nodes 0,\ldots, $n-1$ and root 0.
We have $i \sqsubset j$ iff $i=0$ and $j>0$.
 Let $T^+$ be $T$ plus $\Box_T\Box_T \bot$ plus \emph{sentential excluded third}.
We work in $T^+$. We define a Solovay function in the usual way: 
\begin{itemize}
\item
$h0 : = 0$ 
\item
$ h(p+1) := 
\begin{cases} i & \text{if $h(p) \sqsubset i$ and ${\sf proof}_T(p,\gnumm{\ell \neq i}$)} \\
h(p) & \text{otherwise} 
\end{cases}
$ 

Here $\ell$ is the limit of $h$.
\end{itemize}

\noindent
We note that, since $\Box_T\Box_T \bot$  we have $\Box_T \bigvee_{0< j < n}\exists x\, hx=j$. 
This tells us that inside the box, we can indeed prove that the limit exists. Moreover we
have excluded third for sentences of the form $\ell =i$.
Outside the box we can also prove the existence of the limit by sentential excluded third.
Using these two observations we can execute the usual Solovay argument. This
gives us $\Box_T\bot$ and 
 %\smadent
 we may conclude that $T \vdash \Box_T\Box_T \bot \to \neg\neg\Box_T\bot$. 

\medent
``$\leftarrow$''
Suppose $(\Box\Box\bot \to \neg\neg\,\Box\bot) \in \Lambda^\ast_T$ and $\Lambda^\ast_T \subseteq \logbb + \Box\Box\bot$.
Then it would follow that $\logbb \vdash \Box\Box\bot \to \Box \bot$. Quod non.
\end{proof}

\noindent 
Note that %the condition that 
$(\Box\Box\bot \to \neg\neg\,\Box\bot) \not\in \Lambda^\ast_T$ %is satisfied 
 if
$T$ is one of {\sf HA}, ${\sf HA}+{\sf MP}$, ${\sf HA}+{\sf ECT}_0$. %.
%On the other hand, it does not hold 
 The situation is different for $T = {\sf HA}^\ast$ (cf. \Subsubsection~\ref{sec:hastar}).

%\medent
We formulate the main question of constructive provability logic.

\begin{quest} \label{que:main}
What is the provability logic of {\sf HA}? Is it decidable? We note that  the logic is \emph{prima facie} $\Pi^0_2$.
\end{quest}

\noindent
The basic information about classical provability logic can be found in \cite{smor:self85,bool:emer91,Boolos1993,lind:prov96,japa:logi98,svej:prov00,arte:prov04,halb:henk14}. For information about intuitionistic
provability logic, see e.g. \cite{viss:prop94,Iemhoff01:phd,iemh:moda01,viss:close08,arde:sigm14}.

%%%%%%%%%%%%%%%%%%%%%

\subsection{Preservativity Logic}\label{prelo}
%\nosmurfduo
%The notion of $\Sigma^0_1$-preservativity was introduced in \cite{viss:eval85}.
%It was further studied in \cite{viss:prop94}, \cite{viss:subs02},
%\cite{iemh:pres03} and  \cite{iemh:prop05}.

%\subsubsection{Basics}
\nosmurfduo
As stated above,  $\Sigma^0_1$-preservativity \cite{viss:eval85,viss:prop94,viss:subs02,iemh:pres03,iemh:prop05} for a theory $T$ is defined as follows:
\begin{itemize}
\item
 $A \tto_T B$ if, for all $\Sigma^0_1$-sentences $S$, if
$T \vdash S \to A$, then $T \vdash S \to B$.
\end{itemize}

\noindent In contrast to $\Pi_1^0$-conservativity and interpretability (see Appendices~\ref{sec:picon} and \ref{vrolijkesmurf}),  
defining $\Sigma^0_1$-preservativity does not require an inter-theory notion %of the form 
$T \tto U$.
 
 We give a characterization of $\Sigma_1^0$-preservativity that is analogous to the Orey-H\'ajek characterization for
 interpretability over {\sf PA}.
  Suppose $T$ extends {\sf HA}. We write $\Box_{T,n}$ for the arithmetization of provability from the axioms of $T$ with G\"odel number
 $\leq n$. The theory $T$ is, {\sf HA}-verifiably,  essentially reflexive: for all $n$ and $A$, we have $T \vdash \Box_{T,n}A \to A$.
 Here we allow parameters in the formulation of the reflection principle.\footnote{This result is folklore. We could not locate a
 fully worked-out proof in the literature. Some ingredients can be found
 in  \cite[Part I, \Section 5]{troe:meta73}, but the treatment of these ingredients contains some gaps.
 The proof looks as follows. The theory {\sf HA} verifies cut-elimination for predicate logic.
 Consider any $n$.
 Reason in $T$. Suppose $\Box_{T,n}A$. Let $p$ be a cut-free witness of $\Box_{T,n} A$. All formulas occurring in $p$ will have complexity $\leq m$, for some
 standard $m$. Here our complexity measure is \emph{depth of logical connectives}. We can develop a partial satisfaction predicate for  formulas of
 complexity $\leq m$ that {\sf HA}-verifiably satisfies the commutation conditions. The standard axioms of $T$ that have
 G\"odel number $\leq n$ are true (in the sense of our satisfaction predicate), since the Tarski bi-conditionals are
 derivable. By induction, we can show that all $m$-derivations from true axioms yield
 true conclusions. So, \emph{a fortiori}, we have $A$.}
 
 \begin{theorem}\label{tuinsmurf} 
 Suppose $T$ is an extension of {\sf HA}. Then, $A \tto_T B$ iff, for all $n$, $T \vdash \Box_{T,n} A \to B$. 
 This result is verifiable in i-{\sf EA}.
 \end{theorem}
 
 \begin{proof}
 ``$\to$'' Suppose $A \tto_T B$. It follows that (a)
 if $T \vdash \Box_{T,n} A \to A$, then  (b) $T \vdash \Box_{T,n} A \to B$.  Now note that (a) follows from
 essential reflexivity.
 
 \medent
``$\leftarrow$'' Suppose (c) for all $n$, $T \vdash \Box_{T,n} A \to B$ and (d)  $T \vdash S \to A$.
 From (d), we have, for some $m$, that $T \vdash \Box_{T,m}(S \to A)$. We choose $m$ so large that the finite axiomatization
 of i-{\sf EA} has G\"odel number $\leq m$. %It follows, 
 By  i-{\sf EA} verifiable $\Sigma_1^0$-completeness of extensions of i-{\sf EA}, %that
  $T \vdash S \to \Box_{T,m} A$. Hence, by (c),  $T \vdash S \to B$. 
  
  \medent
  The verifiability in i-{\sf EA} can be seen by inspection of the above proof.
   \end{proof}
   
  %\grmbl
 
\subsubsection{$\logld$ and $\loglc$} \label{sec:iprel}
\nosmurfduo
We note that $\top \tto_{\sf HA} A$ is $\iea$-provably equivalent to $\Box_{\sf HA}A$. This means that %we have two design options
%for treating the
 in our study of $\Sigma^0_1$-preservativity logic of arithmetical theories, we can treat $\Box$ as a defined connective and focus on $\latto$. %or we can treat $\Box$ as defined by
%$\top \tto (\cdot)$. We opt for treating $\Box$ as defined. Thus, we extend the language of propositional logic with a binary operator $\tto$. 
Let $\InF{1}{T}(\tto):= {\sf p}_T(v_0,v_1)$, where ${\sf p}_T$ is a good arithmetization of the relation $A\tto_T B$.
We define $\Lambda_T^\circ:= \Lambda_{T,\InF{1}{T}}$.

\medent
The principles of {\logld} are arithmetically valid for all arithmetical theories in our sense,
\takeout{
%The principles of $\loglc = \logld + \liv$ are valid for a more restricted group.
We will give a sufficient condition for the satisfaction of {\liv} below.
For the convenience of the reader, we repeat the relevant principles here.

\begin{itemize}
\item[\li]	$ \vdash \phi\to \psi \;\Rightarrow\; \vdash  \phi \tto \psi $
\item[\lii]	$ (\phi \tto \psi \wedge \psi \tto \chi ) \to \phi \tto \chi$
\item[ \liii]    $ (\phi \tto \psi \wedge \phi \tto \chi ) \to \phi\tto (\psi \wedge \chi) $
\item[\liv]	$ (\phi \tto \chi \wedge \psi \tto \chi ) \to ( \phi \vee \psi )\tto \chi $
\item[ \lvii]   $ (\phi \wedge \Box \psi) \tto \psi \to \phi \tto \psi$
\item[\lviii]   $ (\phi \tto \psi) \to ( \Box \chi \to \phi) \tto ( \Box\chi \to \psi)$
\end{itemize}

\noindent
 The arithmetical validity of all \logld-principles is immediate for all
arithmetical theories in our sense,} to wit all $\Delta_0({\sf exp})$-axiomatized theories in the arithmetical language extending i-{\sf EA}.  However, $\loglc = \logld + \liv$ is valid for a more restricted group.
  %For a sufficient condition for the
%validity of {\liv},  say that %is as follows. %: \emph{verifiable closure under q-realizability} (cf. \cite[\S~3.2.3]{troe:meta73} and Appendix \ref{sec:real}).
%We refer the reader to
 %\cite{troe:meta73} or \cite{troe:cons88vol1} for an explanation of q-realizability.
% \medent
 We need the notion of \emph{closure under q-realizability}: %(cf. \cite[\S~3.2.3]{troe:meta73} and Appendix \ref{sec:real}) 
 %if %we have:
  $T\vdash A$ implies there is an $n$ such that $T\vdash \underline n\cdot \varepsilon{\downarrow}\, \wedge \, {\underline n\cdot \varepsilon \mathrel{\widetilde q} A}$ (cf. Appendix \ref{sec:real} for notation). A sufficient condition for $\liv$ to hold is that $T$ is not only closed under q-realizability, but it also verifies this fact. %where---just like in Appendix \ref{sec:real}---the dot stands for Kleene application and $\varepsilon$ stands for the empty sequence. 

\begin{theorem}\label{borrelsmurf}
Suppose $T$ is $T$-provably closed under q-realizability. Then {\liv} is arithmetically valid
in $T$, i.e.,
{\liv} is in $\Lambda^\circ_T$.
\end{theorem}

\begin{proof} 
We write $\tto$ for $\tto_T$ and $\Box$ for $\Box_T$.

\medent
Suppose $T$ is $T$-provably closed under q-realizability. 
We reason in $T$. Suppose (a) $A \tto C$ and (b) $B\tto C$. Suppose
 $\Box (S\to (A \vee B))$, where $S$ is a $\Sigma^0_1$-sentence. By q-realizability and the fact that $\Sigma^0_1$-sentences are self-realizing,
  we can find a recursive index $e$ of a 0-ary,  0,1-valued function, such that
(c) $\Box (S \to e\cdot \varepsilon \downarrow)$, (d) $\Box ((e\cdot \varepsilon =0 \wedge S) \to A)$ and (e) $\Box ((e\cdot \varepsilon =1 \wedge S) \to B)$.
From (a) and (d), we get: (f) $\Box ((e\cdot \varepsilon =0 \wedge S) \to C)$. From (b) and (e) we get (g) $\Box ((e\cdot \varepsilon =1 \wedge S) \to C)$.
From (c), (f) and (g), we obtain $\Box(S \to C)$. 
\end{proof}

\nosmurf
The following salient theories $T$ are $T$-provably (even i-{\sf EA}-provably) closed under q-realizability: 
{\sf HA},  ${\sf HA}+{\sf MP}$, ${\sf HA}+{\sf ECT}_0$, ${\sf MA} = {\sf HA}+{\sf MP}+{\sf ECT}_0$
and ${\sf HA}^\ast$ (see \Subsubsection~\ref{sec:hastar}), hence 
 {\loglc} is  arithmetically valid in them. %{\sf HA},  ${\sf HA}+{\sf MP}$, ${\sf HA}+{\sf ECT}_0$, ${\sf MA}$. %= {\sf HA}+{\sf MP}+{\sf ECT}_0$.

\begin{quest} \label{que:dichar}
It would be interesting to have a more perspicuous condition for the satisfaction of {\liv} than closure under
q-realizability. 

\medent
Moreover, in many cases we can also prove {\liv} using the de Jongh translation. Are there separating examples where either q-realizability works
and the de Jongh translation does not or where the de Jongh translation works but q-realizability does not?
\end{quest}

\subsubsection{The Preservativity Logic of {\sf HA}}
\nosmurfduo
The logic {\loglc} is incomplete for {\sf HA} % %We specify an extra principle that was proved in 
\cite{viss:prop94,iemh:pres03}. Define: %the translation $(\cdot)(\cdot)$ as follows:
\begin{itemize}
\item
$(\chi)(\sigma) := \sigma$ for $\sigma ::= \top \mid \bot \mid (\top\tto \phi) \mid (\sigma \vee \sigma)$, where $\phi$ ranges over the full language.
\item
$(\chi)(\phi \wedge \psi) := ((\chi)(\phi) \wedge (\chi)(\psi))$,
\item
$(\chi)(\phi) := (\chi \to \phi)$ in all other cases.
\end{itemize}
The following principle is arithmetically valid over {\sf HA}. 
\begin{itemize}
\item[{\sf V}]
For $\chi := \bigwedge_{i<n} (\phi_i \to \psi_i)$, we have:\\
$\vdash (\chi \to (\phi_n \vee \phi_{n+1})) \tto \bigvee_{j < n+2} (\chi)(\phi_j)$.
\end{itemize}

\nosmurf
An example of a consequence of {\sf V} is as follows. Consider any  non-modal propositional formula 
$\phi(p)$ with at most $p$ free. Suppose that $\phi(p)$ is not constructively valid. Then,
the principle $\phi(\Box\psi) \tto (\Box\psi\vee \neg\, \Box\psi)$ is arithmetically valid over {\sf HA}.

\begin{remark}
We have the following salient result about the admissible rules of {\sf HA}. Suppose $\phi$ and $\psi$ are non-modal propositional formulas.
 Define:
\begin{itemize}
\item
$\phi\vvdash_{\sf HA} \psi$ if for all arithmetical substitutions $\sigma$ we have:\\
${\sf HA} \vdash \sigma(\phi) \;\; \To \;\; {\sf HA}\vdash \sigma(\psi)$.
\end{itemize}
The following are equivalent: \\
\hspace*{1cm} (i) $\phi \vvdash_{\sf HA} \psi$, (ii) $\loglc +{\sf V}\vdash \phi \tto \psi$, (iii)  $\loglc +{\sf V}\vdash \Box\phi \to \Box\psi$.\\
See \cite{iemh:pres03} in combination with \cite{viss:prop94}.
\end{remark}

\nosmurf
Is $\loglc +{\sf V}$ the preservativity logic of {\sf HA}? We do not think so.
The second author has discovered a valid scheme that does not appear to be derivable from  $\loglc +{\sf V}$. To save space, we postpone a detailed discussion to future work.
%We present an as yet unnoticed strengthening of {\liv} that is valid over {\sf HA} in Appendix~\ref{luiesmurf}. 

%\tadeusz{Present auto-q explicitly here? If motivation is neglected, can be done in a few lines, it seems}

\begin{quest} \label{que:hapre}
Here is a list of more open problems.
\begin{enumerate}[I.]
\item
Is {\logld} the preservativity logic of all extensions of i-{\sf EA}?\\
In other words, is {\logld} the intersection of all $\Lambda^\circ_T$, where $T$ is an
arithmetical extension of i-{\sf EA}?
\item
Is {\logld} the preservativity logic of all extensions of {\sf HA}?
\item
Is there an extension $T$ of i-{\sf EA} such that $\Lambda_T^\circ = \logld$?
\item
Is there an extension $T$ of {\sf HA} such that $\Lambda_T^\circ = \logld$?
\item
Is there an extension $T$ of i-{\sf EA} such that $\Lambda_T^\circ = \loglc$?
\item
Is there an extension $T$ of {\sf HA} such that $\Lambda_T^\circ = \loglc$?
\item
What is the preservativity logic of {\sf HA}?
\item
What is the preservativity logic of ${\sf HA}+{\sf MP}$?
\item
What is the preservativity logic of ${\sf HA}+{\sf ECT}_0$?
\end{enumerate}
The questions VII, VIII, IX are obviously quite difficult.\footnote{They could be easier than the
question what the provability logic of {\sf HA}, ${\sf HA}+{\sf MP}$ or  ${\sf HA}+{\sf ECT}_0$ is.
Sometimes theories in a richer language are easier to manage.} As far as we know nobody has seriously worked on
questions I--VI.
\end{quest}

\subsubsection{The Preservativity Logic of classical theories}
\nosmurfduo
We know a lot  about the preservativity logic of classical theories, 
since $\tto_T$ can be intertranslated with $\Pi^0_1$-conservativity $\jump_T$ in the classical case.
As a consequence we can translate what we know about the logic of $\Pi^0_1$-conservativity to a result about
preservativity logic. Let 
${\mathrm c}\hyph{\sf PreL} := {\mathrm i}\hyph{\sf PreL}^{-} + \lna{em}$.

\begin{theorem}\label{partijsmurf}
Suppose that $T$ is $\Sigma^0_1$-sound classical theory that extends $\mathrm{I}\Pi_1^{-}+{\sf Exp}$. 
Then, the preservativity logic of $T$ is precisely  ${\mathrm c}\hyph{\sf PreL} $.
\end{theorem}

\noindent
This result is a translation of Theorem 12 of \cite{bekl:limi05}, which  
%We note that  Theorem 12 of \cite{bekl:limi05} 
 is a strengthening of the
main result of \cite{haje:cons90}, the latter in turn being an adaptation of
 \cite{shav:rela88} and \cite{bera:inte90}.
For details %on how Theorem~\ref{partijsmurf}  is obtained from  Theorem 12 of \cite{bekl:limi05},
 see 
Appendix~\ref{sec:picon}. 

\subsubsection{${\sf HA}^\ast$ and ${\sf PA}^\ast$}\label{sec:hastar}
\nosmurfduo
The Completeness Principle  for a theory $T$ is defined as 
\begin{description} 
\item[${\sf CP}_T$] $ A \to \Box_T A$.
\end{description}
Here $A$ is allowed to contain parameters. 
Consider any theory $T$ such that $T$ is ${\sf HA}+{\sf CP}_T$.
Such a theory is easily constructed by the Fixed Point Lemma.
 One can show that, if {\sf HA}
verifies that  $T$ is ${\sf HA}+{\sf CP}_T$, then $T$ is unique modulo provable equivalence.
Thus, the following definition is justified:
${\sf HA}^\ast$ is the unique theory such that, {\sf HA}-verifiably, ${\sf HA}^\ast$ is ${\sf HA}+{\sf CP}_{{\sf HA}^\ast}$.
The theory ${\sf HA}^\ast$ was introduced and studied in \cite{viss:comp82}.

We have a second way of access to ${\sf HA}^\ast$ via a variant of G\"odel's translation of {\sf IPC} in {\sf S}4.
We define:
\begin{itemize}
\item
$A^{\sf g} := A$ if $A$ is atomic.
\item
$(\cdot)^{\sf g}$ commutes with $\wedge$, $\vee$ and $\exists$.
\item
$(B\to C)^{\sf g} := ((B^{\sf g} \to C^{\sf g}) \wedge \Box_{\sf HA}(B^{\sf g} \to C^{\sf g}))$.
\item
$(\forall x\,B)^{\sf g} := (\forall x\, B^{\sf g} \wedge \Box_{\sf HA}\forall x\, B^{\sf g})$.
\end{itemize}   

\noindent We have ${\sf HA}^\ast \vdash A$ iff ${\sf HA} \vdash A^{\sf g}$.
Using the translation $(\cdot)^{\sf g}$ on can show that ${\sf HA}^\ast$is conservative over {\sf HA}
with respect to formulas that have only $\Sigma_1$-formulas as antecedents of implications.

\medent
The theory ${\sf HA}^\ast$ is the theory in which the incompleteness phenomena lie most closely
to the logical surface. We have the strong form of L\"ob's Principle $ {\sf HA}^\ast\vdash (\Box_{{\sf HA}^\ast} A \to A )\to A$.
Note that ${\sf HA}^\ast\vdash \neg\neg\,\Box_{{\sf HA}^\ast}\bot$ is a special case.
We are inclined to read this principle as: inconsistency can never be excluded.

If we  extend {\sf PA} to $U = {\sf PA}+{\sf CP}_U$, we end up
with the inglorious $U \vdash \Box_U\bot$. However, ${\sf HA}^\ast$ is conservative over
{\sf HA} for a wide class of formulas. So, the Completeness Principle is an example of
a kind of extension  that makes no real sense in the classical case. 

\medent
The theory ${\sf HA}^\ast$ can be used to provide easy proofs of the independence of {\sf KLS} (Kreisel-Lacombe-Schoenfield) and {\sf MS}
(Myhill-Shepherdson) from {\sf HA} \cite{viss:comp82}, simplifying the original ones by Beeson \cite{bees:nond75} while preserving
their basic idea.

De Jongh and Visser showed that every prime recursively enumerable Heyting algebra on finitely many generators can be embedded in the
Heyting algebra of ${\sf HA}^\ast$. See \cite{dejo:embe96}. Their proof is an adaptation of a proof by Shavrukov \cite{shav:suba93} in the simplified
form due to Zambella \cite{zamb:shav94} concerning the embeddability of Magari algebras in the Magari algebra of Peano Arithmetic.
 
 A consequence of the De Jongh-Visser result is the fact that the admissible propositional rules for ${\sf HA}^\ast$ are precisely the
derivable rules. In contrast, the admissible propositional rules for {\sf HA} are the same as the admissible rules for {\sf IPC}: this is the maximal
set of admissible rules that is possible for a theory with the de Jongh property. Thus among theories with the de Jongh property both the minimal possible
set of admissible rules and the maximal one are exemplified. See also \cite{viss:rule99}.

\medent
We want to show that ${\sf HA}^\ast$  is {\sf HA}-verifiably closed under q-realizability. 
The easiest route %way to see this,
 is via the notion of self-q-realizability.
A formula $A(\vec x)$ (with all free variables shown) is \emph{self-q-realizing}
if there is a number $s^A$ such that \newline ${\sf HA}\vdash A(\vec x) \to (\underline{s^A} \cdot (\vec x)) \mathrel{\widetilde q}A(\vec x)$, cf. Appendix \ref{sec:real} for notation. 
%Recall that the dot means \emph{Kleene application}

A substantial class of i-{\sf EA}-verifiably self-q-realizing formulas is the class of \emph{auto-q formulas} %which is
  given as follows.
Let $S$ range over all $\Sigma^0_1$-formulas, let $A$ range over all formulas and let $v$ range over all variables: 
\begin{itemize}
\item
$B ::= S \mid (B\wedge B) \mid \forall v\, B \mid (A \to B)$
\end{itemize}

\noindent We note that the class of auto-q formulas substantially extends the almost negative formulas that are self-r-realising.

The instances of the completeness scheme have the form $\forall \vec x\, (A (\vec x) \to S(\vec x))$, where $S$ is
$\Sigma_1^0$. Thus, these instances are auto-q. 
%(See \Subsection~\ref{luiesmurf} for the definition of \emph{auto-q}.) 
It follows that ${\sf HA}^\ast$  is {\sf HA}-verifiably closed under q-realizability.
Thus, $\loglc^\ast \deq \iA + \bvii + \lviii$ is 
 contained in the preservativity logic of ${\sf HA}^\ast$, to wit
 $\Lambda^\circ_{{\sf HA}^\ast}$. 
\takeout{
\begin{itemize}
\item[\li]	 $\vdash \phi\to \psi\;\; \To \;\; \vdash \phi \tto \psi $
\item[\lii]	$(\phi \tto \psi \wedge \psi \tto \chi ) \to \phi \tto \chi$
\item[\liii]    $(\phi \tto \psi \wedge \phi \tto \chi ) \to \phi\tto (\psi \wedge \chi) $
\item[\liv]	$(\phi \tto \chi \wedge \psi \tto \chi ) \to ( \phi \vee \psi )\tto \chi $
\item[\bvii]   $(\Box \phi \to \phi) \to \phi$
\item[\lviii]   $(\phi \tto \psi) \to ( \Box \chi \to \phi) \tto ( \Box\chi \to \psi)$
\end{itemize}

\noindent The above principles form the theory $\loglc^\ast$.}
There are examples of valid principles %extending $\liv$ 
 that are most probably not in $\loglc^\ast$. 
We do not know whether this has any traces in the provability logic of ${\sf HA}^\ast$. 
%\medent
 As will be explained in Remark~\ref{smulsmurf}, there is a certain analogy between ${\sf HA}+{\sf CT}_0!$ and ${\sf HA}^\ast$.

\medent
We turn to the theory ${\sf PA}^\ast$, axiomatized by the set $\alpha$ of all sentences $A$ such that ${\sf PA}\vdash A^{\sf g}$.
One can easily show that $\alpha$ is closed under deduction and that  ${\sf PA}^\ast$ satisfies ${\sf CP}_{{\sf HA}^\ast}$.\footnote{We have demanded
that the axiom set of a theory is $\Delta_0({\sf exp})$. The axioms of ${\sf PA}^\ast$ do not satisfy this demand. So, the official axiom set should be
a suitable $\Delta_0({\sf exp})$-set manufactured from $\alpha$ using a version of Craig's trick.} 
The theory ${\sf PA}^\ast$ verifies the Trace Principle :
\begin{description}
\item[{\sf TP}] $\Box_{{\sf PA}^\ast} \forall x\,(Ax \to Bx) \to (\exists x\, Ax \vee \forall x\,(Ax \to Bx))$.
\end{description}
This principle is equivalent to 
\[ \Box_{{\sf PA}^\ast} \forall x\,Bx \to (\exists x\, Ax \vee \forall x\,(Ax \to Bx)).\]
The presence of the trace principle has as a modal consequence the principle %\mHCb, to wit:
\begin{description}
\item[\mHCb]
$\Box\phi \to ((\psi \to \phi) \vee \psi)$
\end{description}

\noindent
In \cite{viss:comp82}, it is shown that the logic {\iKMb} is precisely the provability logic of ${\sf PA}^\ast$.\footnote{In \cite{viss:comp82}
 the equivalent form \mHCbalt\ is used, cf. Lemma \ref{th:mhcderiv}.} %$\Box(\psi \to \phi) \to ((\psi \to \phi) \vee \psi)$ of {\mHCb} is used.
We remind the reader that:
\[ \iKMb = \iSLb + \mHCb = \logba + \bvi  + \mHCb.\] 
 
\nosmurf
%What about 
 The preservativity logic of ${\sf PA}^\ast$ %? This logic 
 contains $\logld$ and $\lSalt$, but neither {\liv} nor  \mHCl\ (\Subsection~\ref{sec:falsity}). 
 %However, as we will show in \Subsection~\ref{sec:falsity}, the preservativity logic of ${\sf PA}^\ast$
%contains neither {\liv} nor  \mHCl. 

%%%%%%%%%%%%%%%%%

\section{Kripke completeness and correspondence \qquad \KRp} \label{sec:completeness} %\label{sec:axiomatizations} 

\nosmurfduo
Apart from being our original motivation to study $\tto$, the arithmetical interpretation can occasionally complement the deductive systems proposed in \S~\ref{sec:axiomatizations} by providing   a route to \emph{disprove} certain judgements of the form $\lva{X}\vdash\phi$, i.e., to show \emph{non}-deriv\-abil\-ity  from suitable sets of axioms (namely, those valid in some arithmetical interpretations):
 %as we have already witnessed. Nevertheless, it is both too unwieldy and too specific to perform this task for arbitrary non-theorems and arbitrary sets of axioms. 

%The arithmetical interpretation sketched above can be used 
\begin{example} \label{exaarith}
 Interpreting 
$\phi \tto \psi$ as  $\Box (\phi \to \psi)$ over {\sf HA} yields  $\loglb + \biv + \lviii$. 
This interpretation refutes {\lv}, {\lvi} and \emph{a fortiori} {\lvii}. %\tadeusz{We should give here more details why, but this requires a proper introduction of the idea of arithmetical interpretation.  Either we do this here, which I think is the right thing to do, or link forward to an explanation much further down in the text.}
It follows that {\lv} is really needed in Lemma \ref{lem:gwader}\ref{matowa} above to derive {\lvii}. 
\end{example}
%\tadeusz{Strictly speaking, we have not introduced the idea of arithmetical interpretation in this section. A link forward?}

\nosmurf
%Nevertheless, the arithmetical semantics is too unwieldy and too specific to perform the r\^ole
%  the \emph{only} semantics to show failure of derivability  neither
%much to do with 
%Nevertheless, in order to obtain a more flexible and general semantics, 
To disprove more such judgements, we need to return to relational insights of \S~\ref{sec:basic} and  
%a viable route,  . 
% For such a general purpose, we need to %connect the themes of \S~\ref{sec:axiomatizations} and  %, i.e., syntax and semantics. In other words, we are going to
 %and 
  provide Kripke completeness and correspondence results. Most of this section is based on work we will discuss in a parallel publication  \cite{LitakV:otw}. %are going to publish separately \cite{LitakV:otw}. %\dots \tadeusz{Wherever it is going to be submitted}

\subsection{Notions of completeness}

%\tadeusz{Most of this subsection can be also potentially moved to our forthcoming paper}

%  In order to discuss more \emph{non}-derivations, 
%we need to return to the Kripke semantics. %\tadeusz{\dots but perhaps also HA already here too. Here's another proposed section}

%\avblue{I
%think the items in the list below should have a little more whitespace between them. As it is the
%pictures are a bit close to each other.}
%
\begin{figure}
%\hrule
\vspace{0.2cm}
\footnotesize

\caption{\label{fig:compl} Correspondence conditions. In this figure, and elsewhere in this paper, $\leadsto$ stands 
for $\sqsubset$ and $\to$ stands for $\preceq$. %Furthermore, in the text below, ``$\circ$'' will denote $\sqsubset$-reflexive points, 
%whereas  ``$\bullet$'' will be used for $\sqsubset$-irreflexive ones.
  Some names of principles are taken from  Iemhoff 
and coauthors \cite{Iemhoff01:phd,iemh:pres03,iemh:prop05,Zhou03}, others come from our work to be published separately \cite{LitakV:otw}. %\textcolor{red}{our forthcoming paper on decompositions of $\tto$}. For any relation $R \times R \subseteq W^2$ 
and subset $X \subseteq W$, set $\upclo{X}{R} \deq \{y \in  W \mid \exists x \in X. xRy\}$; in particular, write $\upclo{x}{R}$ for $\upclo{\{x\}}{R}$.}
\vspace{0.2cm}

\newcommand{\sepo}{\vspace{\tbskip}}

\begin{tabular}{lL{2.3cm}L{4cm}c} \sepo
\lb & brilliant &  $k \sqsubset \ell \preceq m$ $\To$ $k \sqsubset m$ & 
$   \vcenter{
    \xymatrix@-1pc{
            & m\\
           k \ar@{~~>}[ur] 
           \ar@{~>}[r] & \ell \ar[u]      
      } }$ 
\\ [\tbskip]\sepo 
 \biii & semi-transitive & $k \sqsubset \ell \sqsubset m$ $\To$ \newline
 $\exists x. k \sqsubset x \preceq m$ %and \newline
%$\sqsubset$ is upwards well-founded \newline
%(Noetherian) 
&
$    \vcenter{
        \xymatrix@-0.5pc{
           x \ar@{-->}[r] & m\\
           k  \ar@{~~>}[u]
           \ar@{~>}[r] & \ell  \ar@{~>}[u]      
        }
      } $ %+ Noeth.
 \\ [\tbskip]\sepo 
%By Lemma~3.10 of \cite{iemh:prop05} the logic 
\lv & gathering & $k \sqsubset \ell \sqsubset m$ $\To$ $\ell \preceq m$ %and \newline
%$\sqsubset$ is upwards well-founded \newline
%(Noetherian) 
&
$    \vcenter{
        \xymatrix{
           k  
           \ar@{~>}[r] & \ell \ar@/_/@{-->}[r] \ar@/^/@{~>}[r]  & m
        }
      } $ %+ Noeth.
%{\loglh} 
\\ [\tbskip]\sepo 
\biv & \multicolumn{3}{c}{$\sqsubset$-Noetherian (conversely well-founded) and semi-transitive} \\ [\tbskip]\sepo 
\lvii & supergathering & on finite frames: \newline
$k\sqsubset \ell \sqsubset m$ $\To$ \newline $\exists x \sqsupset k. (\ell \prec x\preceq m)$  &
$    \vcenter{
        \xymatrix@-1pc{
           %& & \\
           & x \ar@{-->}[r] & m\\
           k \ar@{~~>}[ur] 
           \ar@{~>}[rr] & & \ell \ar@{-->}|{\neq}[ul]  \ar@{~>}[u]      
        }
      } $ 
 \\ [\tbskip]\sepo      
\lviii & Montagna & \multicolumn{2}{c}{$k \sqsubset \ell \preceq m$ $\To$ 
$\exists x \sqsupset k. (\ell \preceq x \preceq m$ $\ands$ $\upclo{x}{\sqsubset\comp\preceq} 
\subseteq \upclo{m}{\sqsubset\comp\preceq})$} \\ [\tbskip]\sepo
\lS & strong &  $k \sqsubset \ell$ $\To$ $k \preceq \ell$ &
$      \vcenter{
        \xymatrix{
	   %\ell	 \\
          k  \ar@/_/@{-->}[r] \ar@/^/@{~>}[r] & \ell
        }
      } $ \\ [\tbskip]\sepo      
\mHCl & $\sqsubset$-dominated & $k \prec \ell$ $\To$ $k \sqsubset \ell$   &
$      \vcenter{
        \xymatrix{
	   %\ell	 \\
          k  \ar@/_/@{->}|{\neq}[r] \ar@/^/@{~~>}[r] & \ell
        }
      } $ \\ [\tbskip]\sepo      
\mHCb & weakly $\sqsubset$-dominated & $k \prec \ell$ $\To$ $\exists m \sqsupset k. m \preceq \ell$   &
$    \vcenter{
        \xymatrix@-0.5pc{
            & \ell\\
           k \ar@{->}|{\neq}[ur] 
           \ar@{~~>}[r] & m \ar@{-->}[u]      
        }
      } $ 
\\ [\tbskip]\sepo 
\aLin & weakly semi-linear &  \multicolumn{2}{c}{$k \sqsubset \ell$ \& $k \sqsubset m$ $\To$ ($m \preceq \ell$ OR   $\ell \preceq m$)}      \\ [\tbskip]
\bLin & strongly semi-linear &  \multicolumn{2}{c}{$k \sqsubset \ell \preceq \ell' $ \& $k \sqsubset m \preceq m'$ $\To$ ($m' \preceq \ell'$ OR   $\ell' \preceq m'$)}    \\ [\tbskip]\sepo 
\CF & semi-dense & $k \sqsubset \ell$ $\To$ $\exists x \preceq \ell. \exists y \sqsupset k.y \sqsubset x$ & 
$    \vcenter{
        \xymatrix@-0.5pc{
           \ell  & \ar@{-->}[l] x\\
           k  \ar@{~>}[u]
           \ar@{~~>}[r] & y  \ar@{~~>}[u]      
        }
      } $     
\\ [\tbskip]\sepo
\aCF & pre-reflexive &  $k \sqsubset \ell$ $\To$ $\exists x \sqsupset \ell. x \preceq \ell$ & 
$    \vcenter{
        \xymatrix{
           k  
           \ar@{~>}[r] & \ell \ar@/_/@{<--}[r] \ar@/^/@{~~>}[r]     
            & x
        }
      } $ 
\\  [\tbskip]\sepo
\Hug & semi-nucleic &  $k \sqsubset \ell$ $\To$ $\exists m \succeq k.$ \newline $\exists m' \sqsupset m.$ $\ell \preceq m. m' \preceq \ell$ & 
$    \vcenter{
        \xymatrix@-0.5pc{
        & m \ar@{~~>}[r] & \ar@{-->}[ld]    \\
         k   \ar@{-->}[ru]
           \ar@{~>}[r] & \ell \ar@{-->}[u]   
            &  }
      } $ 
\\ [\tbskip]\sepo
\aApp & almost reflexive & $k \sqsubset \ell$ $\To$ $\ell \sqsubset \ell$ &  
$    \vcenter{
        \xymatrix{
           k  
           \ar@{~>}[r] & \ell   
 \ar@(dr,ur)@{~~>}[]
}} $ %\POS!R(-.5)
\end{tabular}

%\vspace{0.3cm}

%\hrule
\end{figure} 

%Table \ref{tab:minax} folllows existing references \cite{iemh:moda01,iemh:prov01,iemh:pres03,iemh:prop05} in proposing direct axiomatizations which are hardly more complex than those of ordinary modal systems. Most of the time, we will be working in the full language involving disjunction. However, in order to  facilitate comparisons with some algebraic and categorical setups like the one discussed in \S~\ref{sec:arrows}, let us make clear that $\vee$ is entirely optional; it is just a question of adding corresponding axioms of \ipc\ and a single axiom involving $\strictif$.

%\tadeusz{Something about algebraizability being restored in this signature, I guess}

\nosmurf
%Following the usual notation, 
Given a logic $\ivsys{X}$, set %\avred{Not my usual notation. I call the thing without
%$V$ a \emph{frame}. I would use {\sf Fram} iso {\sf Mod}.}
$\Mod(\ivsys{X}) \deq \{ \ma F \mid \text{ for any } V, \la \ma F, V \ra \kmodels \ivsys{X} \}.$
%\nosmurf
 %As usual, we 
 Say that $\ivsys{X}$ is  \emph{\bro weakly\brc complete for} (or \emph{with respect to})  a class of frames $K$ if it is 
\begin{itemize}
\item \emph{sound} wrt $K$, i.e., $K \subseteq \Mod(\ivsys{X})$ and
\item any $\alpha$ s.t $\ivsys{X} \not\vdash \alpha$ can be refuted in a model based on a frame from $K$.
\end{itemize}

%It is well-known that not all logics are complete in this sense. 
\nosmurf
We say that a condition (which may be expressed in a natural language or in a formalized metalanguage like first- or second-order logic) \emph{corresponds} to a given $\tto$-logic $\ivsys{X}$ if it defines precisely $\Mod(\ivsys{X})$. In particular, when a condition is a correspondent of $\iP + \phi$, we say it \emph{corresponds} to $\phi$ and correspondingly (pun unintended) use notation $\Mod(\phi)$. A logic $\ivsys{X}$ can be complete for much smaller a class than $\Mod(\ivsys{X})$ but  if it is complete for \emph{some} class of frames, it is  also complete for  $\Mod(\ivsys{X})$; we can thus take this as a definition what it means to be (weakly) complete without additional qualifications. 
%, on the other hand, it can fail to be complete wrt  to the entire class of frames it defines; one speaks of an \emph{incomplete} logic then.
 Incomplete logics, i.e., those which have some non-theorems which cannot be refuted in $\Mod(\ivsys{X})$, are sometimes even encountered among those with an arithmetical interpretation, c.f. systems known as \lna{GLB} and \lna{GLP} \cite{Japaridze1988,Boolos1993,HollidayL16}, though most ``naturally'' defined logics tend to be complete.
 
 \begin{remark}
 Let us recall an important difference between completeness and correspondence when it comes to combinations (conjunctions) of axioms. %(which should hopefully be straightforward for most readers).
  Clearly, $\Mod(\bigwedge \Gamma) =  \bigcap\limits_{\gamma\in\Gamma} \Mod(\gamma)$, so whenever $\alpha$ is a correspondent of $\phi$ and $\beta$ is a correspondent of $\psi$, $\alpha \wedge \beta$ is a correspondent of $\phi \wedge \psi$. Nothing like this needs to hold for completeness, even for a finitely axiomatizable logic. Completeness of $\iP + \phi$ and  $\iP + \psi$ for  %is complete wrt 
  frames defined, respectively, by $\alpha$ and $\beta$ does not automatically imply that  $\iP + \phi + \psi$ is complete for $\alpha \wedge \beta$---or, indeed, for any class of frames whatsoever.  This is why in Figure \ref{fig:compl}, Theorems \ref{th:minco} and \ref{th:mainco} below we do not mention  \emph{correspondence} conditions for logics axiomatized by conjunctions/combinations of axioms, but \emph{completeness} results for such logics need to be stated explicitly. %require explicit statements  in Theorems \ref{th:minco} and \ref{th:mainco}.
 \end{remark}
 
 \nosmurf
The notion of completeness can be refined further in two orthogonal directions. One of them is the \emph{finite model property} (\emph{fmp}, also known as the \emph{finite frame property}) which simply means completeness wrt a class of \emph{finite} frames. While the fmp is a much stronger property  than weak completeness, it is still rather standard among most ``natural'' logics. %; as we are going to see, this includes those of interest for us.
  It is not quite the case, however, with another refinement of interest: the notion of \emph{strong} completeness, i.e., completeness for deductions from infinite sets of premises. This notion can be defined in two different ways using either 
  \begin{itemize}
  \item the relation $\Gamma \dedp{\ivsys{X}} \phi$ defined as ``$\phi$ is deducible from $\Gamma$ using all theorems of $\ivsys{X}$ and Modus Ponens'' or
  \item  the relation $\Gamma \dedp{\ivsys{X}}^\tto \phi$ defined as ``$\phi$ is deducible from $\Gamma$ using all theorems of $\ivsys{X}$, Modus Ponens \emph{and the rule $\li$}''.
  \end{itemize}
  
  %The corresponding definitions for 
  A given $\tto$-logic $\ivsys{X}$  is then %look as follows:

\begin{itemize}
\item  \emph{strongly locally complete} if whenever $\Gamma \not \dedp{\ivsys{X}} \phi$, there exists $\ma F \in \Mod(\ivsys{X})$, a valuation $V$ and a point $k$ in $\ma F$ s.t.  $\ma F, V, k \kmodels \Gamma$ and $\ma F, V, k \not\kmodels \phi$.
\item \emph{strongly globally complete} if whenever $\Gamma \not \dedp{\ivsys{X}}^\tto \phi$, there exists $\ma F \in \Mod(\ivsys{X})$, a valuation $V$  s.t.  $\ma F, V \kmodels \Gamma$ yet for some point $k$ in $\ma F$, $\ma F, V, k \not\kmodels \phi$.
\end{itemize}

\nosmurf
As discovered by Frank Wolter \cite{Wolter1993}, these two notions coincide for Kripke semantics of ordinary modal logics. While Wolter was not working with extensions of $\iP$, his reasoning extends to our setting:

\begin{theorem} \label{th:wolgloloc}
A $\tto$-logic $\ivsys{X}$ is strongly locally complete iff it is strongly globally complete.
\end{theorem}

\begin{deriproof}%[Proof sketch]
For the ``if'' direction, assume $\Gamma \not \dedp{\ivsys{X}} \phi$ for a strongly globally complete $\ivsys{X}$. This, in particular, means that for any finite $\Delta \subseteq_{\sf fin} \Gamma$, there exist $\ma F_\Delta \in \Mod(\ivsys{X}), V_\Delta, k_\Delta$ s.t. $\ma F_\Delta, V_\Delta, k_\Delta \kmodels \bigwedge\Delta$ and  $\ma F_\Delta, V_\Delta, k_\Delta \not\kmodels \phi$. Pick a fresh $p$, not occurring in either $\Gamma$ or $\phi$. By setting $V'_\Delta(p) = \upclo{k_\Delta}{\preceq}$ and keeping the valuation of other variables unchanged wrt $V_\Delta$, we get $\ma F_\Delta, V'_\Delta \kmodels p \to \bigwedge \Delta$ and $\ma F_\Delta, V'_\Delta \not\kmodels p \to \phi$. This implies that $\{p \to \bigwedge \Delta\} \not  \dedp{\ivsys{X}}^\tto p \to \phi$ for any finite $\Delta \subseteq_{\sf fin} \Gamma$. Using the fact that $\{p \to \bigwedge \Delta_1, \dots, p \to \bigwedge \Delta_n\}$ is equivalent to $\{p \to \bigwedge (\Delta_1 \cup \dots \cup \Delta_n)\}$ and compactness of $\dedp{\ivsys{X}}^\tto$, we get that  
\[
\{p \to \bigwedge \Delta \mid \Delta \subseteq_{\sf fin} \Gamma\} \not  \dedp{\ivsys{X}}^\tto p \to \phi.
\]
 Using strong global completeness, we find $\ma F \in \Mod(\ivsys{X}), V$ and $k$ s.t. $\ma F, V, k \kmodels \Gamma$ and $\ma F, V, k \not\kmodels \phi$. 

%The converse reasoning uses more specific features of both deducibility relation.
For the ``only if'' direction, first note that using reasoning such as the one leading to Fact \ref{fact:necunnec}, we can replace $\li$ in the definition of $\dedp{\ivsys{X}}^\tto$ by the G\"odel rule. This allows us to use a standard modal trick and reduce $\Gamma \dedp{\ivsys{X}}^\tto \phi$ to $\{ \Box^{\leq n} \gamma \mid \gamma \in \Gamma, n \in \omega \} \dedp{\ivsys{X}} \phi$. Now assume $\Gamma \not \dedp{\ivsys{X}}^\tto \phi$ for a strongly locally complete $\ivsys{X}$ and pick  $\ma F \in \Mod(\ivsys{X}), V$ and $k$ s.t. \[ \ma F, V, k \kmodels \{ \Box^{\leq n} \gamma \mid \gamma \in \Gamma, n \in \omega \} \] and $\ma F, V, k \not\kmodels \phi$. Taking the submodel generated by $k$, i.e., all the points accessible from it by iterated combinations of $\preceq$ and $\sqsubset$ yields a model where $\Gamma$ is globally satisfied, yet $\phi$ fails.
\end{deriproof}

\nosmurf
Strong completeness is typically achieved as a corollary of stronger results, such as \emph{canonicity}, which in turn, as first observed by Fine \cite{Fine75:connections} (see also Gehrke et al. \cite{GHV06} for a general treatment), can be obtained as a corollary of \emph{elementarity}: that is, being complete wrt a \emph{first-order definable} class of frames. It is not hard to see intuitively the reason for this connection: for a (weakly) complete logic at least, a failure of strong completeness implies a failure of \emph{compactness} of the Kripke consequence relation, whereas being elementarily definable guarantees    compactness of this relation.  A suitable notion of canonicity for $\tto$-logics has been proposed and studied in the literature \cite{Iemhoff01:phd,iemh:pres03,Zhou03}; in fact, clauses regarding strong completeness in Theorem \ref{th:minco} below are corollaries of such canonicity results.

\subsection{Completeness and correspondence results}
\nosmurfduo
%In 
Figure \ref{fig:compl} lists various %several 
 completeness/correspondence conditions for %various (mostly arithmetically-oriented) 
  $\latto$-principles. % mostly found by Iemhoff and coauthors \cite{Iemhoff01:phd,iemh:pres03,iemh:prop05,Zhou03}. 
   %As can be seen, 
    L\"ob-like axioms tend to have counterparts which are not of first-order character, but numerous others can in fact be expressed in first-order logic. 

Let us turn these claims into proper theorems. First, let us summarize results which are available in the existing literature, or can be relatively easily derived:

%\begin{theorem}[Correspondence]
%\end{theorem}

\begin{theorem}\ \label{th:minco} 
\begin{enumerate}[a.]
\item 
$\wk$ is strongly complete \bro wrt the class of all frames\brc and enjoys the finite model property \tuc{Prop. 4.1.1}{Iemhoff01:phd}, \tuc{Prop. 7}{iemh:pres03}, \tuc{Th. 2.1.10}{Zhou03}.  %\cite{Iemhoff01:phd,iemh:pres03,iemh:prop05}.
\item $\bk = \wk + \lb$ corresponds to the class of brilliant frames,  is strongly complete and enjoys the finite model property.
\item $\ws  = \loglb + \lS$ corresponds to the class of strong frames, is strongly complete and enjoys the finite model property.
\item $\loglb + \biii$ corresponds to the class of semi-transitive frames, is strongly complete and enjoys the finite model property.
\item $\loglb + \lv$ corresponds to the class of gathering frames, is strongly complete and enjoys the finite model property \tuc{Prop. 4.2.1}{Iemhoff01:phd}, \tuc{Prop. 8}{iemh:pres03}.
\item $\loglb + \biv$ corresponds to the class of Noetherian semi-transitive frames and enjoys the finite model property \tuc{Prop. 4.3.2}{Iemhoff01:phd}, \tuc{Th. 2.2.7}{Zhou03}.
%\ws^? & := \loglb^? + \lS,  \\ [\tbskip]
 %\logle & \deq \logla + \lvi, \\
\item \label{compgla} $\loglg  = \loglb + \biv + \lv$ \bro cf. Lemma \ref{lem:gla}\brc corresponds to the class of Noetherian gathering frames \tuc{Lem. 9}{iemh:pres03}, \tuc{Lem. 3.10}{iemh:prop05} and enjoys the finite model property.
% & = \logle + \liv, \\ [\tbskip] 
 %\loglf & := \logla + \lvii, \\
\item \label{compgwa} $\loglh = \loglb + \lvii$ corresponds to the ``supergathering'' property of Figure \ref{fig:compl} on the class of finite frames \tuc{Lem. 3.5.1}{Zhou03}, \tuc{Th. 3.31}{iemh:prop05}.
\item $\loglb + \lviii$ corresponds to the class of Montagna frames of Figure \ref{fig:compl}, is strongly complete \tuc{Prop. 11}{iemh:pres03} and enjoys the finite model property \tuc{Lem. 3.21}{iemh:prop05}, \tuc{Th. 3.3.5}{Zhou03}.
\end{enumerate}
\end{theorem}

% \tadeusz{Actually, where's the completeness result?}\albert{I will ask Roos.}

%\tadeusz{Claims given w/o references are proved in our forthcoming manuscript \dots}

%Let us recall that a logic is 

%\tadeusz{I got L\"ob wrong here, \iW\ is stronger. This needs to be corrected.}

%\begin{theorem}\
%\wsdisf, \ws, \bsdisf\, and \bs\ enjoy analogues of completeness results in Theorem \ref{th:minco} wrt suitable classes of strong frames.
%\item \wldisf, \wl, \bldisf\, and \bl\ enjoy analogues of the finite model property results in Theorem \ref{th:minco}  wrt suitable classes of L\"ob frames.
%\item \istr{L^-}, \istr{L}, \icol{SL^-} and \icol{SL} enjoy analogues of the finite model property results in Theorem \ref{th:minco}  wrt suitable classes of strong L\"ob frames.
%\end{itemize}
%\end{theorem}

\nosmurf
We could not find in the literature an explicit statement of  the finite model property of $\loglc = \loglh + \lviii$. Moreover, an astute reader probably noticed that we do not claim strong completeness for all logics appearing in the statement of this theorem.  The reason is obvious: it is very well-known that variants of 
 the L\"ob axiom clash with strong completeness and, a fortiori, with canonicity. Boolos and Sambin \cite{bool:emer91} credit Fine and Rautenberg with this observation, which can be now found in any standard monograph on modal logic. This can be extended in several directions, e.g., to logics with weaker axioms (cf. Amerbauer \cite{Amerbauer96:sl}) or to failure of broader notions of strong completeness \cite{Litak05:phd}; see Litak \cite[\S~3]{Litak07:bsl} for more on both counts. In the context of logics for (relative) interpretability (cf. Appendix \ref{sec:inter}), problems with canonicity and strong completeness have been pointed out, e.g., by de Jongh and Veltman \cite{dejo:prov90}. Let us adapt such arguments to our setting:
 
 \begin{theorem}
 $\ivsys{X}$ is not strongly complete whenever  
 \begin{itemize}
 \item it is contained between $\iP + \biv$ and $\logbb + \bLin$ or
 \item it is contained between $\iP + \biv$ and $\iKMlin$. 
 \end{itemize}
 In particular, $\loglg$, $\loglh$, $\loglc$  or $\iKMl$ fail to be strongly complete. 
 \end{theorem}

\begin{proof}[Proof sketch]
We can work in the standard modal language containing just $\Box$ rather $\tto$ (in fact, $\Box$ and $\to$ are the only connectives really used). We can also use the freedom offered by Theorem \ref{th:wolgloloc} and choose to disprove global completeness.  Consider now
$
\Gamma \deq \{ \Box p_{i+1} \to p_i \mid i \in \omega \}
$
and note that that in any model where $\Gamma$ is globally satisfied but $p_0$ fails, there must exists an infinite $\sqsubset$-ascending chain, which allows us to refute Noetherianity, hence refuting $\biv$. 

However, taking $\sqsubset$ to be an ordered sum of $\omega$ with its copy with reverse ordering $\omega^*$, $\preceq$ to be either (for the first clause) discrete or (for the second clause) the reflexive version of $\sqsubset$ and setting $V(p_i) := \upclo{(i+1)}{\preceq}$ produces a model where $\Gamma$ is globally valid, $p_0$ fails and all theorems of $\ivsys{X}$ hold under $V$ (despite being refutable in the underlying frame).
\end{proof}

\nosmurf
%Returning to Figure \ref{fig:compl}, 
 Theorem \ref{th:minco} above does not cover  correspondence and completeness claims  for all axioms and frame conditions displayed in Figure \ref{fig:compl}, especially those not directly related to preservativity and provability principles. As it turns out, there is a  technique of transferring  generic results %of this kind  
  available for (bi)modal logics over \cpc\ into the intuitionistic setting. % using a suitable generalization of the G\"odel-McKinsey-Tarski translation. 
  For $\Box$-logics, it has been developed in a series of papers by Wolter and Zakharyaschev \cite{WolterZ97:al,WolterZ98:lw}. We are going to present details of generalization of this technique to $\tto$-logics in a separate paper \cite{LitakV:otw}. For now, let us just list some consequences regarding strong completeness and canonicity (we leave the finite model property out of the picture here):

 \begin{theorem}[\cite{LitakV:otw}]\  \label{th:mainco}
\begin{enumerate}[a.] 
% & =  \loglf + \liv, \\ [\tbskip] 
\item $\iP + \mHCl$ correspond to the class of $\sqsubset$-dominated frames of Figure   \ref{fig:compl} and is strongly complete.
\item $\iP + \mHCb$ correspond to the class of weakly $\sqsubset$-dominated frames of Figure   \ref{fig:compl} and is strongly complete.
\item $\imHCl$ is strongly complete \bro wrt the class of strong $\sqsubset$-dominated frames\brc.
\item $\iP + \aLin$ correspond to the class of weakly semilinear frames of Figure   \ref{fig:compl} and is strongly complete.
\item $ \iP + \bLin$ correspond to the class of strongly semilinear frames of Figure   \ref{fig:compl} and is strongly complete.
\item $\iP + \CF$ correspond to the class of semi-dense frames of Figure   \ref{fig:compl} and is strongly complete.
\item $\iP + \aCF$ correspond to the class of pre-reflexive frames of Figure   \ref{fig:compl} and is strongly complete.
\item $\iP + \aApp$ correspond to the class of almost reflexive frames of Figure   \ref{fig:compl} and is strongly complete.
\item $\iPLLa$ is strongly complete \bro wrt the class of strong almost reflexive frames\brc.
%\item  $\bk$ corresponds to is strongly complete and enjoys the finite model property 
\end{enumerate}
%\begin{itemize}
%\item In their corresponding languages, $\wkdisf$ and $\wk$ are both strongly complete and enjoy the finite model property.
%\item In their corresponding languages, $\bkdisf$ and $\bk$ are both strongly complete and enjoy the finite model property wrt frames satisfying \boxcol.
%\end{itemize}
\end{theorem}

%\subsection{Example derivations}

%The completeness results above ensure that nothing is missing from proposed axiomatizations. Nevertheless, it is instructive to see some direct syntactic reasoning. Define a \emph{$\tto$-logic} to be any set of formulas containing $\wk$ and closed under MP, \lnec\ and substitution. Observe that \lnec\ and \lmono\ imply that $\tto$ is monotone in the second argument and antimonotone in the first.

%\paragraph{The relationship between $\Box$ and $\strictif$} We recall again that in the present setting, $\Box\phi$ is defined as $\top \tto \phi$. It is trivial to observe that the $\Box$-fragment of any $\tto$-logic is a \lna{IntK_\Box}-logic as defined in standard references \cite{Vakarelov81:sl,Dosen85:sl,Sotirov84:ml,Wolter97:sl,WolterZ97:al,WolterZ98:lw}. Moreover, we have already observed that  
%$$\phi \strictif (\psi  \to \chi) \to (\phi \wedge \psi)  \strictif \chi$$ 
%and consequently $\Box (\psi  \to \chi) \to (\psi \strictif \chi)$ are examples of valid laws. Let us see a syntactic proof: % with $\phi \ded \psi$ denoting derivability of $\phi \to \psi$ in the basic system: \avred{vdash?}\tlred{That's an option}

\subsection{Non-derivations} \label{sec:nonderivations}

%\tadeusz{I still need to edit this one}

\nosmurfduo
Having a developed semantics, we are now in a position %to complement derivations given in \S~\ref{sec:axiomatizations} with
 to provide more examples of \emph{non}-derivations between formulas and \emph{non}-containments between logics.

\begin{example} \label{notltow}
Consider the formula $\phi_0 := \Box \bot \tto \bot \to \Box \bot$. It is easy to see that this formula is in
the closed fragment of {\loglh}. This means that $\phi_0$ is variable-free and  provable in
{\loglh}. We show that  $\phi_0$ is not in the closed fragment of {\loglg}.

%By Lemma~3.10 of \cite{iemh:prop05} the logic {\loglg} corresponds to all frames that 
%are gathering, i.e. $u \sqsubset v \sqsubset w \To v \preceq w$, and where $\sqsubset$ is upwards well-founded.

By Theorem \ref{th:minco}\ref{compgla},  {\loglg} is determined by Noetherian gathering frames. Consider the following (Noetherian gathering) model: %$\mathcal K$:

\medskip

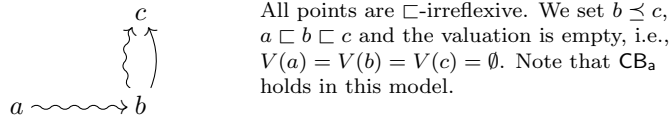
\begin{tabular}{C{4cm}L{6cm}}
$    \vcenter{
        \xymatrix@C+1pc{
        & c \\
            a \ar@{~>}[r]   & b \ar@/^/@{~>}[u]  \ar@/_/@{->}[u] \\
        }
      } $  & \footnotesize  All points are $\sqsubset$-irreflexive.  We set $b \preceq c$, $a \sqsubset b \sqsubset c$ and the valuation is empty, i.e.,  $V(a)=V(b) = V(c) = \emptyset$. Note that $\mHCl$ holds in this model.

\end{tabular}      

\medent
 %The worlds of $\mathcal K$ are $a,b,c$.
%The relation $\preceq$ is $\verz{\tupel{a,a},\tupel{b,b},\tupel{b,c},\tupel{c,c}}$. The relation
%$\sqsubset$ is $\verz{\tupel{a,b},\tupel{b,c}}$. We have $V(a)=V(b) = V(c) = \emptyset$.
Clearly, $a \Vdash \Box \bot \tto \bot$, but $a \nVdash \Box \bot$. 
%On the other hand, our model satisfies the above  frame conditions for {\loglg}.
\end{example}
 
\begin{example}\label{eerlijkesmurf}
Consider the formula $\phi_1 := \Box \bot \tto p \to \Box (\Box \bot \to p)$.
Lemma \ref{lem:gwader}\ref{matobox} implies that  this formula is provable in {\loglc}. We prove that ${\loglh} \nvdash \phi_1$ 
%By Theorem~{3.31} of  \cite{iemh:prop05} the logic {\loglh} corresponds on finite frames to the following property:
%$ u\sqsubset v \sqsubset w \To \exists x\, (u \sqsubset x \text{ and } v \prec x\preceq w)$.
 by  considering the following model  satisfying the condition for finite frames for {\loglh} as stated in Theorem \ref{th:minco}\ref{compgwa} (and Figure \ref{fig:compl}):
 
 \medskip

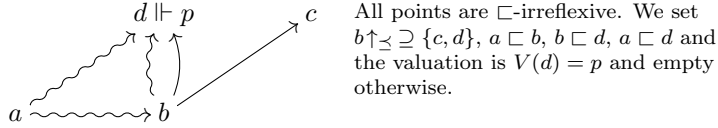
\begin{tabular}{C{4cm}@{\qquad}L{5cm}}
$    \vcenter{
        \xymatrix@C+1pc{
        & d \kmodels p& c \\
            a \ar@{~>}[r] \ar@{~>}[ur]  & b \ar@/^/@{~>}[u]  \ar@/_/@{->}[u]  \ar@{->}[ur] \\
        }
      } $  & \footnotesize  All points are $\sqsubset$-irreflexive.  We set $\upclo{b}{\preceq} \supseteq \{c,d\}$, $a \sqsubset b$,  $b \sqsubset d$, $a \sqsubset d$ and the valuation is $V(d) = p$ and empty otherwise. 
      \end{tabular}    
      
  \medent
%\begin{itemize}
%\item
%The domain of $\mathcal M$ is $\verz{a,b,c,d}$.
%\item
%The relation $\preceq$ is $\verz{\tupel{a,a},\tupel{b,b},\tupel{b,c}, \tupel{b,d}, \tupel{c,c},\tupel{d,d}}$.
%\item
%The relation $\sqsubset$ is $\verz{\tupel{a,b}, \tupel{a,d},\tupel{b,d}}$. 
%\item
%We take
%$V(a) = V(b) = V(c) = \emptyset$,  $V(d) = \verz{p}$.
%\end{itemize}
 It is now easy to see 
that $a\Vdash \Box \bot \tto p$, but $a\nVdash \Box (\Box\bot \to p)$.
 %On the other hand,
%$\mathcal M$ does satisfy the condition for finite frames for {\loglh} as stated in Theorem \ref{th:mainco}\ref{compgwa}, i.e., $ u\sqsubset v \sqsubset w \To \exists x\, (u \sqsubset x \text{ and } v \prec x\preceq w)$.
\end{example}

\begin{example} \label{eerlijkesmurfalt}
We can improve Example~\ref{eerlijkesmurf} by providing a separating closed formula. Consider the formula
$$\phi_2 := \Box \Box \bot \tto \neg\neg\, \Box \bot \to \Box(\Box \Box \bot \to \neg\neg\,\Box\bot).$$  Again, Lemma \ref{lem:gwader}\ref{matobox} implies that  this formula is provable in {\loglc}. %It is easy to see \tadeusz{Is it?}
%that $\phi_2$ is provable in {\loglc}. 
 We prove that $\loglh\nvdash \phi_2$ by considering the following model satisfying the condition for finite frames for {\loglh} as stated in Theorem \ref{th:minco}\ref{compgwa} (and Figure \ref{fig:compl}):
 
 \medskip

\begin{tabular}{C{6cm}@{\qquad}L{4cm}}
$    \vcenter
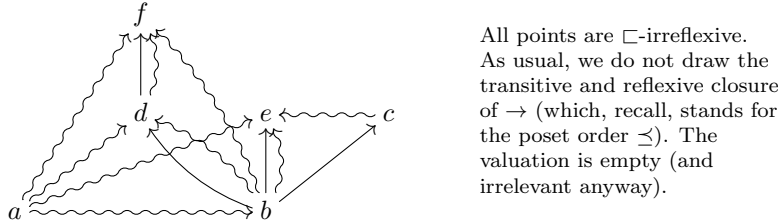
{
        \xymatrix@C+1pc{
        & f & & \\
        & d \ar@/_/@{~>}[u]  \ar@{->}[u]  & e & c \ar@{~>}[l]  \\
            a \ar@{~>}[urr] \ar@{~>}[rr] \ar@{~>}[ur] \ar@{~>}[uur]  & & b \ar@/_/@{~>}[lu]  \ar@/^/@{->}[lu] \ar@/_/@{~>}[luu] \ar@{->}[ur]   \ar@/_/@{~>}[u]  \ar@{->}[u]  \\
        }
      } $  & \footnotesize All points are $\sqsubset$-irreflexive. As usual, we do not draw the transitive and reflexive closure of $\to$ (which, recall, stands for the poset order $\preceq$). The valuation is  empty (and irrelevant anyway). 
      \end{tabular}    
      
  \medent
%\begin{itemize}
%\item
%The domain of $\mathcal M$ is $\verz{a,b,c,d,e,f}$.
%\item
%$\preceq := \verz{\tupel{a,a},\tupel{b,b},\tupel{b,c},\tupel{b,d},\tupel{b,e},\tupel{b,f}, \tupel{c,c}, \tupel{d,d},\tupel{d,f},\tupel{e,e},\tupel{f,f}}$.
%\item
%$\sqsubset := \verz{\tupel{a,b},\tupel{a,d},\tupel{a,f},\tupel{b,d},\tupel{b,e},\tupel{b,f},\tupel{c,e},\tupel{d,f}}$. 
%\item
%We take
%$V(a) = V(b) = V(c) = V(d) = V(e) = V(f) = \emptyset$.
%\end{itemize}
 It is now easy to see
that $a\Vdash \Box\Box \bot \tto \neg\neg\,\Box \bot$, but $a\nVdash \Box (\Box\Box\bot \to \neg\neg\,\Box \bot)$.
% On the other hand,
%$\mathcal N$ does satisfy the frame condition for {\loglh}. %In fact the frame of $\mathcal N$ is gathering and $\sqsubset$ is 
%transitive.
\end{example}

%\tadeusz{Btw, despite the commented out claim, $\sqsubset$ wasn't transitive in your description (no mention of $a \sqsubset e$). I added that.}

%\paragraph{A non-derivation} 
\begin{example} \label{ex:nondnegbox}
Recall that following Lemma \ref{th:emdi}, we noted that  in the disjunction-free setting, there is no one-variable formula $\phi(p)$  s.t. 
$p  \strictif q \eqdm \phi(p) \strictif (p \to q)$ and $\cpc \ded \phi(p)$. This follows from the fact that $$\phi_3 \deq  p \tto q \to (\neg\neg p \to p) \tto (p \to q)$$
is \emph{not} a theorem of \wk: 

\medskip

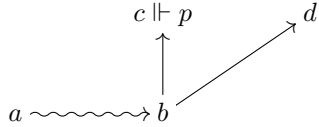
\begin{tabular}{C{4.5cm}@{\quad\qquad}L{5.5cm}}
$    \vcenter{
        \xymatrix@C+1pc{
        & c \kmodels p& d \\
            a \ar@{~>}[r] & b  \ar@{->}[u]  \ar@{->}[ur] \\
        }
      } $  & \footnotesize  All points are $\sqsubset$-irreflexive.  We set $q$ to be false everywhere. It is easy to check the antecendent of $\phi_3$ holds at $a$, but the consequent fails. It is worth noting that $\loglc$ holds in this model.
      %We set $\upclo{b}{\preceq} \supseteq \{c,d\}$, $a \sqsubset b$,  $b \sqsubset d$, $a \sqsubset d$ and the valuation is $V(d) = p$ and empty otherwise. 
      \end{tabular}    

\medskip\nosmurf
We can complement this observation by another one: it is not possible to improve the \wk-equivalence of Theorem \ref{th:emdi} by taking a one-variable intuitionistic formula  stronger that $p \vee \neg p$ as the antecedent of $\strictif$ replacing \lna{em}, as $$\phi_4 \deq  p \tto q \to (\neg\neg p \vee \neg p) \tto (p \to q)$$ is not a theorem of $\iP$ either:

\medskip

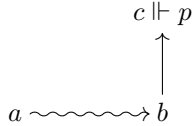
\begin{tabular}{C{4cm}@{\quad\qquad}L{5cm}}
$    \vcenter{
        \xymatrix@C+1pc{
        & c \kmodels p \\
            a \ar@{~>}[r] & b  \ar@{->}[u]   \\
        }
      } $  & \footnotesize  All points are $\sqsubset$-irreflexive.  We set $q$ to be false everywhere. Again the antecendent of $\phi_4$ holds at $a$, but the consequent fails; moreover, $\loglc$ holds in this model.
      %We set $\upclo{b}{\preceq} \supseteq \{c,d\}$, $a \sqsubset b$,  $b \sqsubset d$, $a \sqsubset d$ and the valuation is $V(d) = p$ and empty otherwise. 
      \end{tabular}         
%This time, we invite the reader to use semantics and find a finite $\tto$-intuitionistic Kripke model refuting this implication. \tadeusz{this to be developed now}
\end{example}
%\medskip

\begin{example} \label{ex:mhcnonequiv}

Here are  diagrams illustrating that $\mHCl$ is not a theorem of $\imHCb$; that is, strong frames which are only weakly $\sqsubset$-dominated. %In the left and the right one, 
 We use the convention that $\circ$ stand for a $\sqsubset$-reflexive loop and $\bullet$ for lack thereof. %For the middle diagram, recall that $\to$ (i.e., $\preceq$) is by default transitive, whereas $\leadsto$ (i.e., $\sqsubset$) is not:

\medskip

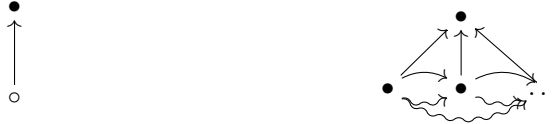
\begin{tabular}{C{5cm}@{\quad\qquad}C{5cm}}
$    \vcenter{
        \xymatrix{
        \bullet \\
            \circ \ar@{->}[u]  \\
        }
      } $  &
%$    \vcenter{
%        \xymatrix{
%        c \\
 %          %\bullet \ar@{->}[ru] \ar@/^/@{->}[r] \ar@/_/@{~>}[r] & 
%          b \ar@/^/@{->}[u] \ar@/_/@{~>}[u] \\
%           a \ar@/^/@{->}[u] \ar@/_/@{~>}[u] \\
 %       }
 %     } $  &
$    \vcenter{
        \xymatrix@-0.6pc{
            %{\vdots}\POS!U(-.5) \ar[r]  
            &  \bullet & \\         
        %\bullet \ar@/^/@{->}[u] \ar@/_/@{~>}[u]  \ar[r] &  \\
            %\bullet \ar@{->}[ru] \ar@/^/@{->}[r] \ar@/_/@{~>}[r] & 
          \bullet \ar@/^/@{->}[r] \ar@/_/@{~>}[r]  \ar[ur]  \ar@/_1pc/@{~>}[rr] & 
           \bullet \ar@/^/@{->}[r] \ar@/_/@{~>}[r] \ar[u] & \ar[lu]  {\hdots}\POS!L(+.5) 
        }
      } $             
\end{tabular}      

\medskip

\nosmurf
Arithmetical interpretation provides another interesting way of distinguishing between $\mHCb$ and $\mHCl$: \S~\ref{sec:hastar} noted that \mHCb\ is in the preservativity logic of ${\sf PA}^\ast$, whereas as stated in Theorem \ref{th:nomhcpast}, \mHCl\ does not belong to this system (the only problem is that neither does $\Di$).
\end{example}

%By now, the reader should feel comfortable enough to experiment with further derivations. And, in the presence of completeness results we have seen, one can easily travel back and forth between syntactic and semantic reasoning, just like in the well-known world of $\Box$.

%We already have some  infrastructure in place. But have we used to unleash a fantastic bestiary, a Jurassic park of species bread in a laboratory? Or do these creatures roam freely in nature? 

%We will soon see some full-bloodied and very much alive specimens. %Nevertheless, just to avoid confusion, 
%But let us begin with a somewhat more complicated link, which needs to be discussed here.

\begin{example} \label{ex:linavsb}
So far, we were seeing examples showing that principles for $\tto$ are often  properly stronger than their relatives formulated in terms of $\Box$ only. Recall that when introducing Fact \ref{fact:linnon}, we indicated it is not always the case, as witnessed by semi-linearity axioms. Here is a simple frame for $\loglh + \aLin$ where $\bLin$ fails (for both claims one can use Theorem \ref{th:mainco} and Figure \ref{fig:compl}, but they are straightforward to verify anyway). We are following the same conventions regarding $\sqsubset$-reflexive and $\sqsubset$-irreflexive points as in the preceding example: 
\medskip

%\begin{tabular}{C{4.5cm}@{\quad\qquad}L{5.5cm}}
%\[   
\qquad\qquad $\vcenter{
        \xymatrix{
        & \bullet & \bullet \\
            \bullet \ar@{~>}[r] \ar@{~>}[ur]  & \bullet  \ar@/^/@{~>}[u] \ar@/_/[u] \ar@{->}[ur] \\
        }}$
      %}$ %\]  %& \footnotesize  All points are $\sqsubset$-irreflexive.  %We set $q$ to be false everywhere. It is easy to check the antecendent of $\phi_3$ holds at $a$, but the consequent fails. It is worth noting that $\loglc$ holds in this model.
      %We set $\upclo{b}{\preceq} \supseteq \{c,d\}$, $a \sqsubset b$,  $b \sqsubset d$, $a \sqsubset d$ and the valuation is $V(d) = p$ and empty otherwise. 
%      \end{tabular}    
\end{example}

\begin{example} \label{ex:cfsep}
In order to separate $\CF$, $\aCF$ and $\aApp$, we provide an  example of a semi-dense frame which is not pre-reflexive (on the left) and a pre-reflexive one which are not almost reflexive (on the right):

\medskip

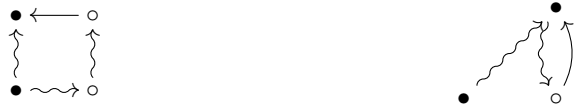
\begin{tabular}{C{5cm}@{\quad\qquad}C{5cm}}
$    \vcenter{
        \xymatrix@-0.5pc{
           \bullet  & \ar@{->}[l] \circ\\
           \ar@{~>}[u]
           \bullet \ar@{~>}[r] & \circ  \ar@{~>}[u]      
        }
      } $    &
$   \vcenter{
    \xymatrix{
            & \bullet  \ar@/_/@{~>}[d]\\
           \bullet \ar@{~>}[ur] 
            & \circ \ar@/_/[u]      
      } }$       
\end{tabular}      

\medskip \noindent
%Again, just as before, 
 Again, even without using completeness results of  Theorem \ref{th:mainco}, one can easily verify everything by hand (including finding suitable valuations).
%Corresponding valuations can be found easily; it is a straightforward exercise. 
\end{example}

\begin{example} \label{ex:plaanonb}
In \S~\ref{sec:comparr}, we will use the fact that \iPLLa\ does not contain \lb. As made clear by Theorem  \ref{th:mainco} and Figure \ref{fig:compl}, for this purpose we need a strong almost reflexive frame which is not brilliant: %Here is a straightforward example:

\medskip

\qquad\qquad$\vcenter{
        \xymatrix{
           \bullet  \ar@/_/@{->}[r] \ar@/^/@{~>}[r] 
            & \circ \ar@{->}[r]    
            & \bullet
        }
      }$
\end{example}

\section{Strength: arrows, monads, idioms and guards \qquad \CSp} \label{sec:arrows}
\nosmurfduo
We have already seen that arithmetical interpretation of modalities provides good motivation for studying intuitionistic logics with strict 
implication, including those with the strength axiom. This is a very good motivation indeed, but by no means the only one. Such formalisms 
have continuously reappeared in several recent lines of research, especially in theoretical computer science.
 %but also that of philosophy. In this section, we are going to quickly overview some of these  %connections and applications.

\subsection{Notions of computation and arrows} \label{sec:comparr}
\nosmurfduo
Surprisingly, the functional programming community discovered a variant of constructive strict implication at roughly the 
same time as it appeared in the context of preservativity. More specifically, ``(classical) arrows'' in the terminology of 
John Hughes \cite{Hughes00:scp} (see also \cite{LindleyWY08:msfp}) are in our terminology \emph{strong} Lewis arrows.  
Interestingly enough, their \emph{unary} cousins knows as %, which McBride and Paterson  baptised
 ``idioms'' or 
``applicative functors'' \cite{McbrideP08:jfp} were discovered  \emph{later} in this community, though a special subclass of applicative functors---to wit, 
\emph{monads} corresponding to $\iPLL$ modalities \cite{BentonBP98:jfp,FairtloughM97:ic,Kobayashi97:tcs}---has been 
enjoying continuous attention since the seminal paper of Moggi \cite{Moggi91:ic}. 
%as contrasted with ``idioms'' \cite{McbrideP08:jfp}. 
 %There is already a body of work on proof theory and categorical semantics of arrows \cite{Atkey08:msfp,JacobsHH09:jfp,LindleyWY08:msfp,LindleyWY10:jfp,Lindley14:wgp}. 
  A particularly convenient basis for our discussion contrasting arrows, idioms and monads is provided by Lindley et al. 
  \cite{LindleyWY08:msfp}, which we take as the main reference for this subsection. 
 
The connection between intuitionistic logics and functional programming is provided by the \emph{Curry-Howard correspondence}, also known as the \emph{Curry-Howard isomorphism} or \emph{proposition-as-types paradigm} (cf. \cite{SorensenU06:book}).\footnote{As pointed out by S{\o}rensen and Urzyczyn,  ``The \emph{Brouwer -
Heyting - Kolmogorov - Sch\"onfinkel - Curry - Meredith - Kleene - Feys -
G\"odel - L\"auchli - Kreisel - Tait - Lawvere - Howard - de Bruijn - Scott -
Martin-L\"of - Girard - Reynolds - Stenlund - Constable - Coquand - Huet - \dots - isomorphism} might be a more appropriate name, still not including
all the contributors.'' \cite[p. viii]{SorensenU06:book} Indeed, the Curry-Howard isomorphism provides the most commonly accepted specification of the Brouwer-Heyting-Kolmogorov interpretation of intuitionistic connectives. We could thus only half-jokingly argue that this subsection is yet another place in our paper where Lewis meets Brouwer.} While the details are outside of the scope of this paper, the shortest outline is that
\begin{itemize}
\item (intuitionistic) formulas correspond to \emph{types}, 
\item logical connectives correspond to type operators/constructors,
\item logical axioms correspond to \emph{inhabited} types and hence deciding theoremhood corresponds to the type inhabitation problem,
\item logical proofs---e.g., in a variant of a natural deduction system or in a Hilbert-style system---are encoded by \emph{proof terms}---in a variant of lambda calculus or of combinatory logic---understood as a (functional) programming language and hence
\item proof \emph{normalization} corresponds to reduction of these terms, understood as representing \emph{computation}. 
\end{itemize}

\nosmurf
In particular, ordinary intuitionistic implication $\phi \to \psi$ corresponds to forming the  \emph{function space} of programs (proofs) 
which take data from (proofs for) $\phi$ as their input and produce members of (proofs for) $\psi$ as their output. 
The introduction rule for $\to$ corresponds to $\lambda$-abstraction and its elimination rule 
(i.e., ordinary Modus Ponens) corresponds to function application. 

Nevertheless, one may ask: are ``computations'' exactly co-extensional with ``members of function space''? 
In the words of Ross Paterson
\begin{smquote}
Many programs and libraries involve components that are Ôfunction-likeÕ, in that
they take inputs and produce outputs, but are not simple functions from inputs
to outputs\dots [S]uch ``notions of computation''
defin[e] a common interface, called ``arrows''. \cite[p. 201]{Paterson03:fop}
\end{smquote}

\nosmurf
What are the laws such a notion of computation is supposed to satisfy? The inhabitation laws of the calculus of ``classic arrows''  
 %as presented by Lindley et al.
  \cite[Fig. 4]{LindleyWY08:msfp} in a disjunction-free language are given by the following axioms:\footnote{Lindley et al. \cite{LindleyWY08:msfp} call these axioms \lna{arr}, \lna{>\!>\!>} and \lna{first}, respectively. They also use $\leadsto$ in place of $\tto$. }

\begin{itemize}
\item[\lSalt] $(\phi \to \psi) \to \phi \strictif \psi$,
\item[\lii]	$ \phi \tto \psi \to \psi \tto \chi \to \phi \tto \chi$,  
\item[\liiialtalt]  $ \phi \tto \psi \to (\phi \wedge \chi) \tto (\psi \wedge \chi)$.
\end{itemize}

\noindent
Thanks to Lemmas \ref{th:basicia} and \ref{lem:lsder}, we know it is just an axiomatization for $\ws^-$!

\begin{quest} \label{que:diarrow}
As Lindley et al. \cite{LindleyWY08:msfp}  work in a type theory without the co-product operator (i.e., the Curry-Howard counterpart of disjunction), the issue of validity of $\Di$ simply does not arise. Nevertheless, given the problematic status of $\Di$ in preservativity logics of some theories (cf. Open Question \ref{que:dichar}), it seems a valid question whether $\Di$ should be a law imposed on all notions of computation---and if not, how to characterize those where it holds. It is an inhabited type for both \emph{arrows with apply} (\emph{monads}) and \emph{static arrows} (\emph{idioms}), as follows from the discussion below and, correspondingly, Lemmas \ref{lem:derpll}\ref{plaadi} and \ref{th:mhcderiv}\ref{mhcdi}.
\end{quest}

\nosmurf
What is the status of the \lb\ law then (or any of its equivalent forms)? As it turns out, the Curry-Howard interpretation provides another rationale for considering (strong) Lewis arrows not determined by an unary $\Box$. Lindley et al. \cite{LindleyWY08:msfp} call arrows satisfying \lb\ \emph{static arrows} and show that such arrows correspond to the
 ``idioms'' or ``applicative functors'' of McBride and Paterson  \cite{McbrideP08:jfp}. Indeed, the inhabitation laws of the calculus for idioms \cite[Fig. 3]{LindleyWY08:msfp} are exactly those of $\iS$. %But, as the very word ``static'' implies, 
 This is, however, only a special subclass of computations encoded by arrows: namely those computations ``in
which commands are \emph{oblivious} to input'' \cite{LindleyWY08:msfp}. Lindley \cite{Lindley14:wgp} rephrases this claim to the effect that idioms are distinguished by their static approach to \emph{data flow}.

%\tadeusz{An important class of applications The study of \emph{Curry-Howard correspondence}, analysis of computational content of proofs and program extraction.}

%\avblue{This sentence sounds strange. Name the subclass.}
 However, as said above, just a special subclass of applicative functors is by far the most important from a programming point of view: that of
 \emph{\textup{(}strong\textup{)} monads}. This subclass of 
 idioms whose type system satisfies in addition  the inhabitation law corresponding to $\CF$ (and, obviously, a number of  equalities between 
 proof terms, which are not of concern to us here) provides the most popular framework for \emph{effectful computations}. In other words, the Curry-Howard counterpart (the logic of type inhabitation) of the 
 calculus for (strong) monads proposed by Moggi under the name of \emph{computational metalanguage} \cite{Moggi91:ic} is $\iPLL$:  
 \emph{propositional lax logic} \cite{BentonBP98:jfp,FairtloughM97:ic,Kobayashi97:tcs}.

Monads can be shown \cite{Hughes00:scp,LindleyWY08:msfp} to be in 1-1 correspondence with \emph{higher-order arrows} or \emph{classical arrows with apply}. To wit, these are arrows satisfying the law:

\begin{itemize}
\item[\aApp]
 $(\phi \wedge (\phi \tto \psi)) \tto \psi$.
\end{itemize}

\nosmurf
Thus, by Lemma \ref{lem:derpll}, the logic of type inhabitation for this subclass of arrows is precisely $\iPLLa$ (\emph{propositional logic of arrows with apply}). Lindley et al. present a two-context natural deduction system for both $\ws^-$ and $\iPLLa$, whose proof-term assignment is based on a distinction between \emph{terms} and \emph{commands} and argue that higher-order arrows are ``\emph{promiscuous} (in the broader sense of \emph{undiscriminating})'', as the ``apply'' construct  corresponding to $\aApp$ bridges this distinction carefully maintained in the calculus for $\ws^-$  (which can be thus called \emph{meticulous}). Another perspective is offered by Lindley \cite{Lindley14:wgp}: higher-order arrows are distinguished by their dynamic approach not only to \emph{data flow}, but also to \emph{control flow}.

\begin{remark} \label{rem:laxembed}
%The logical aspect of 
 The correspondence between monads and arrows with apply should \emph{not} be %understood properly. %It works somewhat differently than in the case of the correspondence 
 %It works differently
  conflated with the one between idioms (whose logic of type inhabitation is $\ws^-$) and  static arrows, whose logic of type inhabitation is $\bs$: i.e., a system where $\phi \tto \psi$ is definable as $\Box(\phi \to \psi)$. In contrast, $\lb$ is obviously not valid in $\iPLLa$ (cf. Example \ref{ex:plaanonb}) and the $\Box$-only fragment of $\iPLLa + \lb$ is a $\Box$-logic stronger than $\iPLL$; e.g., we have that $$\iPLLa + \lb \ded \Box(\Box\phi \to \phi)$$  and one can easily check that $\iPLL \not\ded \Box(\Box\phi \to \phi)$. Instead, $\iPLLa$ is embedded into $\iPLL$ by interpreting $\phi \tto \psi$ as $\phi \to \Box\psi$, cf. \cite[\S~6]{LindleyWY08:msfp}.  In fact,  we can derive this fact syntactically from Lemma \ref{lem:derpll}\ref{plaacollapse} above!
\end{remark}

\begin{remark} \label{rem:lewapp}
To finish this subsection on another theme from Lewis, note that  \emph{Symbolic Logic} \cite{Lewis32:book} had this to say about $\aApp$ (appearing therein as postulate \lna{11.7} in the main text and in the famous Appendix II as \lna{B7}):
\begin{smquote}
It might be supposed that this principle would be implicit in any set of assumptions for a calculus of deductive inference. As a matter of fact, \lna{11.7} cannot be deduced from other postulates. \cite[p. 125]{Lewis32:book}
\end{smquote}
The last sentence\footnote{It is proved later on p. 495 of  \cite{Lewis32:book} using a matrix proposed by Parry.} is pertinent indeed: $\aApp$ is the only axiom  of the smallest system Lewis was interested in, i.e., \lna{S1}, which is not a theorem of $\iP^-$!
\end{remark}

\subsection{Modalities for guarded \textup{(}co\textup{)}recursion} \label{sec:guarded}
\nosmurfduo
Another area of recent computer science where strong intuitionistic modalities have found numerous  applications  is the study of  guarded (co)recursion: %In type systems of \emph{programming languages}, i.e., on the \emph{object level}
 as an important tool to
ensure \emph{productivity} in (co)programming with \emph{coinductive types} \cite{KrishnaswamiB11:icfp,KrishnaswamiB11:lics,KrishnaswamiBH12:popl,AtkeyMB13:icfp,Mogelberg14:lics,BizjakM15:entcs,CloustonBGB15:fossacs} and, on the metalevel,  in semantic reasoning about programs involving \emph{higher-order store} or a combination of \emph{impredicative
quantification} with \emph{recursive types} \cite{DreyerAB11:lmcs,BirkedalMSS12:lmcs,BizjakBM14:rtlc,SvendsenB14:esop,SieczkowskiBB15:itp,JungSSSTBD15:popl,BizjakBGCM16:fossacs}. 

The logics of type inhabitation of these systems are mostly extensions of $\iSLb$,  involving either first- or higher-order quantifiers (corresponding to dependent, polymorphic or impredicative types) or additional entities like \emph{clock variables} \cite{AtkeyMB13:icfp,Mogelberg14:lics,BizjakM15:entcs,BizjakBGCM16:fossacs}, or (a constructive analogue of) the universal modality \cite{CloustonBGB15:fossacs}. Nakano  \cite{Nakano00:lics,Nakano01:tacs} % introduced the very idea of using strong L\"ob $\Box$ as an
proposed using the axioms of $\iSLb$ for \emph{approximation modality}  crediting Sambin-de Jongh-style  results on elimination of fixpoints as one of his motivations (see \cite[\S~3]{Litak14:trends} for a detailed discussion of this point); more recent discussion of Nakano-style systems can be found in Abel and Vezzosi \cite{AbelV14:aplas} and Severi \cite{Severi2017}. The idea of using such modalities also in the metalanguage for reasoning about \emph{semantics} of programs has been popularized by Appel et al. \cite{AppelMRV07:popl}, who were nevertheless working  with the axiom $\biv$ rather than $\bvii$ seen in most later references.

As the above overview makes clear, the area has grown too large to allow an adequate summary in this paper. See \cite{Litak14:trends} for more information and \cite{MiliusL17:fi} for an overview of \emph{models of guarded \mbox{\textup{(}co\textup{)}}recursion}, i.e., from our point of view, categorical models for proof systems for fragments of such logics. %The main question to ask in this paper is 
 %We  ask, however, 
  Our question here is whether the Lewis arrow naturally occurs in this context. %such settings.   %and a tentative examination %of this question
 %leads to somewhat interesting conclusions. 
%\begin{itemize}
%\item 
 %Interestingly, despite the fact that---as discussed in \S~\ref{sec:comparr}---strong arrows are already rather well-recognized in an adjacent area of computer science, few references explicitly address the issue. %seem to either explicitly discuss potential connections between these two areas or clarify why such connections would be spurious. 

In fact, %it does occur in the literature. %happen in %one does not need to search very hard to see $\tto$ as a natural primitive in  type systems of
  %at least some of the above references. 
 starting from the original paper of Nakano \cite{Nakano00:lics} and even more so in references like Abel and Vezzosi \cite{AbelV14:aplas}, the introduction/elimination/inference rules governing the behaviour of  such an``approximation'' or ``delay'' modality are often formulated \emph{combining} $\Box$ and $\to$. This point is perhaps most explicitly addressed by Clouston and Gor\'{e} \cite{CloustonG15:fossacs}, a reference highly relevant from our point of view, as it does use $\tto$ (denoted therein as $\twoheadrightarrow$), claiming moreover:
\begin{smquote}
The main technical novelty of our sequent calculus is that we leverage the
fact that the intuitionistic accessibility relation is the reflexive closure of the
modal relation, by decomposing implication into a static (classical) component
and a dynamic `irreflexive implication' $\tto$ that looks forward along the modal
relation. In fact, this irreflexive implication obviates the need for $\Box$ entirely, as $\Box\phi$ is easily seen to be equivalent to $\top\tto\phi$. Semantically, the converse of this
applies also, as $\phi \tto \psi$   is semantically equivalent to $\Box(\phi \to \psi)$, but the $\tto$ connective is a necessary part of our calculus. We maintain $\Box$ as a first-class
connective in deference to the computer science applications and logic traditions
from which we draw, but we note that formulae of the form $\Box(\phi \to \psi)$ are
common in the literature---see Nakano's ($\to E$) rule  \cite{Nakano00:lics}, and even more directly
the $\circledast$ constructor of \cite{BirkedalM13:lics}. We therefore suspect that treating $\tto$ as a first-class
connective could be a conceptually fruitful side-benefit of this work  (\cite{CloustonG15:fossacs}, in a notation adjusted to this paper).
\end{smquote}
%\end{itemize}

\nosmurf
Clouston and Gor\'{e} \cite{CloustonG15:fossacs} provide a sequent calculus for a logic called here $\iKMlin$.  The focus on this logic is motivated by Litak's observation \cite{Litak14:trends} that $\iKMlinb$ is the propositional fragment of the Mitchell-B\`enabou logic of the \emph{topos of trees} proposed as a model of guarded (co)recursion  by Birkedal and coauthors \cite{BirkedalMSS12:lmcs} and used ever since \cite{Mogelberg14:lics,CloustonBGB15:fossacs,Severi2017}. %Strictly speaking, the logic of topos of trees was originally formulated using unary modality, but 

%While Clouston and Gor\'{e} \cite{CloustonG15:fossacs}   justify the move from $\iKMlinb$ to $\iKMlin$ simply by the fact that their semantic clause for $\tto$ makes \lb\ valid, %in hindsight we may identify a deeper reason: in our paper, it is given by Lemma \ref{lem:kmlin}.
% we have already learned in
  Let us note here that Lemma \ref{lem:kmlin} implies that \emph{any} semantics for $\iKMlin$ must make $\lb$ valid: in other words, $\iKMlin$ can be just seen as another syntactic presentation of $\iKMlinb$. 
  %, just like over classical classical logic, % Postulating $\iKMlinb$ means that
   %the only $\tto$ which satisfies \iP\ axioms and makes $\Box$ definable as $\top \tto \phi$ must satisfy \lb\ anyway. %\tadeusz{Mention the r\^ole of the converse of \lna{K} since Nakano. Seems I need to add it to our axiom list} 
 However, Lemma \ref{lem:kmlin} requires all the axioms of $\iKMlinb$ and when studying broader classes of models of guarded (co)recursion \cite{MiliusL17:fi}, more flexibility in adding $\tto$ is possible.

\begin{quest} \label{que:guarded}
Are there natural applications of $\tto$-logics not including the $\lb$ axiom in terms of guarded (co)recursion? And, more broadly,  do arithmetically relevant principles discussed in this paper have a computational interpretation?
\end{quest}

Let us add that, while Gentzen-style systems are not our main interest here, the above quote from  Clouston and Gor\'{e} \cite{CloustonG15:fossacs} hints at another motivation for studying constructive $\tto$. Namely, even in the setups which make it a definable connective, it can still prove a more convenient primitive from a proof-theoretical point of view than $\Box$ is. %---but this is certainly a motivation worth further investigation.

\subsection{Intuitionistic epistemic logic} \label{sec:intepi}
\nosmurfduo
Finally, let us briefly mention yet another recent area of research where strong intuitionistic modalities made a surprising 
appearance: in the work of Artemov and Protopopescu on  intuitionistic epistemic logic \cite{ArtemovP16:rsl}, presented also in 
this collection.\footnote{For other approaches to intuitionistic epistemic logic cf.  also Williamson \cite{Williamson92:jpl} or Proietti \cite{Proietti2012}
 and for a more dynamic take, see Kurz and Palmigiano \cite{KurzP13:lmcs}.} These authors work with unary $\Box$ and call $\lS$ the principle of ``co-reflection''. 
The minimal system denoted by these authors as \lna{IEL^-}\ corresponds to  $\iS$ in our notation, their  \lna{IEL}\ is obtained by 
adding\footnote{Litak \cite{Litak14:trends} denotes  $\neg\Box\bot$ as \lna{(nv)}---\emph{non-verum}.} $\neg\Box\bot$ and \lna{IEL^+} arises by adding \CF---i.e., is an extension of $\iPLL$.

A proof-theoretic justification for these systems is  presented in terms of the Brouwer-Heyting-Kolmogorov interpretation. 
This seems to provide a natural connection with references discussed  in \S~\ref{sec:comparr}---but, curiously, none of them 
seems to be mentioned by Artemov and Protopopescu, neither the extensive literature on %computational significance of 
 $\iPLL$, 
nor the r\^ole of $\iS$ as the logic of applicative functors (idioms, prenuclei \dots). %\emph{classical arrows}).
  We leave an epistemic interpretation of strong arrows and extensions of $\ws$ as a promising subject for future study.

%From our point of view, two remarks are in order. First, despite the fact 

%%%%%%%%%%%%%%%%%%%%%%%%%%%%%%%%%

\section{Applications of preservativity \qquad\protect\AIIp} 
\label{sec:apppre}

\nosmurfduo
Having briefly overviewed other motivations for studying constructive $\tto$, let us return to our main one. Preservativity has many applications. A number of these applications can be found in \cite{viss:eval85} and  \cite{viss:prop94}. 
We describe one of the main results of those papers in \Subsection~\ref{nnil}.
In \Subsection~\ref{sec:falsity}, we show how one can capture the invalidity of the law of excluded middle in terms of preservativity.
We illustrate how this result imposes a constraint on possible preservativity logics of theories.
%In \Subsection~\ref{dejo}, we provide yet another proof of De Jongh's Theorem for $\Sigma^0_1$-substitutions. This result has indeed
%many proofs, but each proof has its own interest. Moreover, each proof yields its own distinctive extra information.

\subsection{{\sf NNIL}}\label{nnil}
\nosmurfduo
 The {\sf NNIL}-formulas (\emph{No Nestings of Implications to the Left}  
\cite{viss:eval85,viss:prop94}) are %a class of propositional formulas. They are
  defined as follows:
 \begin{itemize}
%\item
%$\alpha ::=  p_0 \mid p_1 \mid \ldots$
\item
$\phi ::=  \bot \mid \top \mid p \mid  (\phi \wedge \phi) \mid (\phi \vee \phi)  \mid (p \to \phi)$
\end{itemize}
%
 %The {\sf NNIL}-formulas were introduced in
%\cite{viss:eval85}. See also  \cite{viss:prop94}. 

It is easy to see that there are only finitely many  nonequivalent {\sf NNIL}-formulas on finitely many variables. %modulo {\sf IPC}-provable equivalence.
Let $\vec p$ be the propositional variables of $\phi$ and define 
$\phi^\star$ as the disjunction of representatives of all {\sf IPC}-equivalence classes of {\sf NNIL}-formulas $\psi$ in the variables $\vec p$ such that ${\sf IPC}\vdash \psi \to \phi$.
 %We can make sense of this definition by picking one representative from each {\sf IPC}-equivalence class of the $\psi$. 
Using the Interpolation Theorem, we see that, for any {\sf NNIL}-formula $\chi$, we have ${\sf IPC} \vdash \chi \to \phi$ if and only if
${\sf IPC} \vdash \chi \to\phi^\star$. So, $\phi^\star$ is the best {\sf NNIL}-approximation from below of $\phi$. In more fancy terms, $(\cdot)^\star$ is
the right adjoint of the embedding functor of the preorder category of the {\sf NNIL}-formulas %ordered by {\sf IPC}-provable implication
 into
the preorder category of all propositional formulas, both preorders being {\sf IPC}-provable implication.

%In \cite{viss:eval85} (see also \cite{viss:prop94}), the following theorem was proved: 

\begin{theorem}[\cite{viss:eval85,viss:prop94}]\label{brilsmurf}
For any function $f$ from the propositional variables to $\Sigma^0_1$-sentences, 
$\phi^f \tto_{\sf HA} (\phi^\star)^f$. Hence, if ${\sf HA} \vdash \phi^f$, then ${\sf HA} \vdash (\phi^\star)^f$.
\end{theorem}

\nosmurf
The original aim of  \cite{viss:eval85} was to show: if ${\sf HA} \vdash \phi^f$, then ${\sf HA} \vdash (\phi^\star)^f$.
However, it turned out that the inductive assumption requires %needed to establlsh the desired result asks for 
 the stronger statement involving
preservativity. Thus, preservativity was discovered as a tool for induction loading.

\medent
Theorem~\ref{brilsmurf} can be reformulated in terms of admissible consequence.
We define:
\begin{itemize}
\item
 $\phi \vvdash_{{\sf HA},\Sigma^0_1} \psi$ if for any $\Sigma^0_1$-sub\-sti\-tu\-tion $f$, ${\sf HA} \vdash \psi^f$ whenever ${\sf HA} \vdash \phi^f$.
  \end{itemize}
 
 \medent
Thus, $\phi \vvdash_{{\sf HA},\Sigma^0_1} \psi$ means that $\phi / \psi$ is an admissible rule for
$\Sigma^0_1$-sub\-sti\-tu\-tions over {\sf HA}.
%\nosmurf
 Theorem~\ref{brilsmurf} now simply says: $\phi \vvdash_{{\sf HA},\Sigma^0_1} \phi^\star$. 
 It is optimal in the sense that, whenever
$\phi \vvdash_{{\sf HA},\Sigma^0_1} \psi$, we have ${\sf IPC} \vdash \phi^\star \to \psi$  \cite{viss:prop94}. Thus, 
\[ \phi \vvdash_{{\sf HA},\Sigma^0_1} \psi \text{ iff } \phi^\star  \vdash_{\sf IPC}   \psi.\]
If we view $\vvdash_{{\sf HA},\Sigma^0_1}$ and $\vdash_{\sf IPC}$ are pre-ordering categories, this says that
$(\cdot)^\star$ is the left adjoint of the embedding functor of  $\vvdash_{{\sf HA},\Sigma^0_1} $ in $\vdash_{\sf IPC}$.

\medent
The {\sf NNIL}-formulas play an important r\^ole in: the characterization of the provability logic of {\sf HA} for $\Sigma^0_1$-substitutions
by Ardeshir and Mojtahedi  \cite{arde:sigm14}, the study of infon logic \cite{cotr:tran13} and several other contexts \cite{rena:inte89,viss:nnil95,yang:intu08}. 

%\medent The class of {\sf NNIL}-formulas has many other interesting features. We refer the reader to e.g. 
%\cite{rena:inte89},  \cite{viss:nnil95}, \cite{yang:intu08}. %\tadeusz{THIS IS NOT FAN YANG'S PHD, IT'S MASTER'S THESIS}
%Finally, the  {\sf NNIL}-formulas play a r\^ole in the study of infon logic. See e.g. \cite{cotr:tran13}.

%%%%

\subsection{On the falsity of \emph{Tertium non Datur}} \label{sec:falsity}
\nosmurfduo
In intuitionistic propositional logic, we have the principle $\neg\,\neg\, (\phi \vee \neg \, \phi)$.
As a consequence, there is no direct logical expression of the constructive insight of the invalidity of the law of excluded 
middle.\footnote{We can consistently add $\neg\, \forall x\, (A(x) \vee \neg\, A(x))$ to constructive arithmetic for certain $A$.
E.g., {\sf HA} plus a weak version of Church's Thesis (cf. Appendix \ref{sec:real}) proves $\neg\, \forall x \, ({x\cdot x\downarrow} \vee {x \cdot x \uparrow})$.}
The connective  $(\cdot) \tto \bot$ is a  weaker form of negation, say $\sim$. Can we have,
 provably in $\iea$, that $\sim_{\sf HA} (A\vee \neg\, A)$, for some suitable $A$?
 
 \medent
 We will show that, for a wide range of theories $T$,
 we can indeed find such a sentence $A$, %such that, $\iea$ provably, $\sim_T (A\vee \neg\, A)$.
 including $T$ being {\sf HA}, ${\sf HA}+{\sf MP}$ or ${\sf HA}+{\sf ECT}_0$,
 ${\sf HA}^\ast$.
 %\medent
 We write:
 \begin{itemize}
 \item
 $T \leq U$ if $\iea$ verifies that $T$ is a subtheory of $U$.
 \end{itemize}
  
 \nosmurf
 Suppose  $\iea$ verifies  {\liv} for $U$, i.e. suppose that {\liv}  is in $\Lambda_{\InF{1}{U},\iea}$.
 We note that over $\iea$ we have  $(\Box_U\bot \vee \neg \,\Box_U \bot) \tto_U \Box_U\bot$. 
 This is in the desired direction since we can consider $\Box_U \bot$ as a weak form
 of falsity. However, we cannot get the desired result as long as we stay with $\Sigma^0_1$-sentences.
 
 \begin{theorem}
 Consider any consistent theory $U$.
There is, verifiably in $\iea+\Diamond_U\top$, no $\Sigma_1$-sentence $S$ such that $\sim_U(S\vee \neg\, S)$.
\end{theorem}

\begin{proof}
We work in $\iea+\Diamond_U\top$.
 Consider  a $\Sigma_1$-sentence $S$.  Suppose we have $\sim_U (S \vee \neg\, S)$.
 It follows that $(S \to (S\vee \neg\, S)) \tto_U (S \to \bot)$.  Thus, 
 $\top \tto_U \neg\, S$,  $\neg\, S \tto_U (S \vee \neg\, S)$ and $(S \vee \neg\, S) \tto_U \bot$. Ergo, $\Box_U \bot$. Quod non.
 \end{proof}
  
 \nosmurf
 To prepare the construction of the promised sentence, we first consider theories $V$ with ${\sf HA} \leq V$.
  Recall that $\Box_{V,x}A$ stands for (arithmetized) provability from the axioms of $V$ with G\"odel number $\leq x$.
 \begin{itemize}
 \item
 \emph{Feferman provability} for $V$ is defined by:
 $\triangle_V A := \exists x\, (\Box_{V,x} A \wedge \Diamond_{V,x}\top)$.
  \end{itemize}
 
 \noindent We have:\footnote{We do not present the principles for triangle as a schematic logic. 
 This is because of the occurrence of a variable over $\Sigma^0_1$-sentences. We would need a many-sorted
 propositional theory. Of course this is perfectly doable. We just did not develop it in this paper.}
 \begin{itemize}
 \item[{\sf Fe}1]
 $V \vdash A \;\;\To\;\; V \vdash \triangle_V A$.
 \item[{\sf Fe}2]
 $\iea\vdash \triangle_V (A\to B) \to (\triangle_V  A \to \triangle_V  B)$.
 \item[{\sf Fe}3]
$\iea \vdash S\to  \triangle_VS$, for $\Sigma^0_1$-sentences $S$.\\
We note that it follows that
$\iea \vdash \Box_VB \to  \triangle_V\Box_VB$.
\item[{\sf Fe4}]
$\iea \vdash \triangle_VB \to  \Box_VB$.
\item[{\sf Fe}5]
$\iea \vdash \Diamond_V\top \to (\triangle_V A \iff \Box_V A )$.
\item[{\sf Fe}6]
$\iea \vdash \triangledown_V\top$, where $\triangledown$ is $\neg\triangle \neg $.
 \end{itemize}
 
 \noindent We note that classically {\sf Fe}4 follows form {\sf Fe}5.
 
 Shavrukov \cite{shav:smar94} provides a complete axiomatization for the bimodal logic of ordinary provability and Feferman provability for
 {\sf PA}.

 \begin{quest}\label{que:shavrukov}
 Shavrukov employs a different interpretation of $\Box_{{\sf PA},x}$, to wit provability in $\mathrm{I}\Sigma_x$. It would be interesting to find
 a better analogue of the version of the Feferman predicate employed by Shavrukov for the case of (extensions of) {\sf HA}. Moreover, the principles given above provide a part of the principles given by Shavrukov for the classical case. We do not  get all Shavrukov's principles
 in the constructive case. It would be interesting to study how close we can get to his system.
 \end{quest}
 
 Note that, supposing that $V$ is consistent, we cannot get that, for all $A$, we have 
 $V \vdash \triangle_VA \to \triangle_V\triangle_VA$.
 Otherwise, we could reproduce the reasoning for G\"odel's Second Incompleteness Theorem. 
 This leads immediately to a contradiction with {\sf Fe}6.
 
 We remind the reader that the theory $V$ is $V$-verifiably \emph{essentially reflexive}.
 This means that both truly and $V$-provably, we have: for all $n$ and all $A$, we have $V \vdash \Box_{V,\underline n} A \to A$.\footnote{We have this even
 for formulas $A$, when we employ the usual convention for free variables under the box.} 
  
 \begin{theorem}
 Suppose ${\sf HA} \leq T$. We have 
  $\iea \vdash A \tto_V \triangle_V A$.
  \end{theorem}
  
  \begin{proof}
 Reason in $\iea$. Consider any $x$. We have, by essential reflexivity, 
 \[\Box_V ( \Box_{V,x} A \to ( \Box_{V,x} A \wedge \Diamond_{V,x}\top)).\]
 Hence, $\Box_V( \Box_{V,x} A \to  \triangle_V A)$. Ergo, by Theorem~\ref{tuinsmurf},
 $A \tto_V \triangle_VA$.
   \end{proof}
  
  \nosmurf
 Consider the G\"odel sentence $G_V$ of Feferman provability for $V$.
  We have then $\iea \vdash G_V \iff \neg\,\triangle_V G_V$.
  Whenever the intended theory is clear from the context, we write $G$ for $G_V$.
  
  \begin{theorem}\label{secretarissmurf}
  Suppose ${\sf HA}\leq V$.
   Then, $\iea \vdash G \tto_V \bot$ and
   $\iea \vdash \neg\, G \tto_V \bot$. 
  \end{theorem}
  
  \begin{proof}
 We reason in $\iea$.
 
 \medent
 We have $G \tto_V \triangle_V G$, and hence, $G \tto_V \neg\, G$. Since also
 $G \tto_V G$, it follows that $G \tto_V \bot$.
 
 \medent
 We have, $\neg\, G \tto_V  \triangle_V  G$ and $\neg\, G \tto_V  \triangle_V  \neg\, G$.
 Hence, $\neg\, G \tto_V    \triangle_V \bot$. So,
 by {\sf Fe}6, we have $\neg\, G\tto_V \bot$.
 \end{proof} 
  
  \begin{theorem}\label{directiesmurf}
   Suppose ${\sf HA} \leq T$ and that $T_0$ verifies {\liv} for $T$, i.e. that {\liv} is in $\Lambda_{\InF{1}{U},T_0}$.
   Then, we have
  $T_0 \vdash \sim_T (G_T\vee \neg \, G_T)$.
  \end{theorem}
  
  \noindent  This follows immediately from Theorem~\ref{secretarissmurf}. The reason why $T_0$ appears in the formulation is that we want the result both for $T_0 = \iea$ and for $T_0 = T$.
 
 %\begin{proof}

 %\end{proof} 
 
 \begin{quest} \label{que:extend}
 Can we extend Theorem~\ref{directiesmurf} to cases where we do not have ${\sf HA}\leq T$?
 \end{quest}
 
 \nosmurf
 We can now show that the preservativity logic of ${\sf PA}^\ast$ does not contain 
{\liv} and \mHCl. We first prove a purely modal result that delivers both cases. We can achieve it in two ways.

\begin{theorem}\label{zwierigesmurf}\
\begin{enumerate}[A.]
\item
$\loglh+ \mHCb \vdash (p \tto \bot \wedge \neg\, p \tto \bot ) \to \Box \bot$.
\item
$\loglg +  \mHCb \vdash \boxdot(p \tto \bot \wedge \neg\, p \tto \bot ) \to \Box \bot$.
%\item
%$\logla+ \bvi + \mHCl \vdash  (p \tto \bot \wedge \neg\, p \tto \bot ) \to \Box \bot$.
\end{enumerate}
\end{theorem}

\begin{proof}
(A): We reason in $\loglh+ \mHCb + (p \tto \bot \wedge \neg\, p \tto \bot )$.
By {\liv}, we have (a) $(p\vee \neg\, p) \tto \bot$. On the other hand, we have, by {\mHCb}, that
$\Box\bot \to (p\vee \neg\, p)$. By {\li}, we have  (b) $\Box\bot \tto (p\vee \neg\, p)$.
Combining (a) and (b), we find $\Box\bot \tto \bot$ and, hence, by {\lvii}, we obtain $\Box\bot$.

\medent
(B):  We reason in $\loglg +  \mHCb + \boxdot(p \tto \bot \wedge \neg\, p \tto \bot )$. 
By {\liv}, we have (a) $(p\vee \neg\, p) \tto \bot$. The principle {\mHCb} gives us
$\Box\bot \to (p\vee \neg\, p)$. It follows, by \li, that $\Box\Box \bot \to \Box(p\vee \neg\, p)$.
Ergo, we have $\Box\Box\bot \to \Box \bot$. We now apply the extended L\"ob's Rule, using that our assumption
$\boxdot(p \tto \bot \wedge \neg\, p \tto \bot ) $ is self-necessitating, to conclude that
$\Box \bot$.
%
%\medent
%Ad (C):
%We reason in $\logla+ \bvi + \mHCl +  (p \tto \bot \wedge \neg\, p \tto \bot ) $.
%We have $p \tto \bot$. Hence, by {\mHCl}, $p \vee \neg\, p$. It follows, by  {\bvi}, that
 %$\Box p \vee \Box \neg\, p$.   Combining this last insight with $p \tto \bot$ and $\neg\, p \tto \bot$,
 %we obtain $\Box\bot$.
\end{proof}

\noindent
As an immediate consequence of Theorems~\ref{secretarissmurf}, \ref{directiesmurf} and \ref{zwierigesmurf}, we have:

\begin{theorem}\label{fleurigesmurf}
Suppose ${\sf HA} \leq T$ and $T$ is $\Sigma_1^0$-sound. 
Then, 
we cannot have both {\liv} and {\mHCb} in $\Lambda^\circ_T$.
\end{theorem}

%\noindent Note that for the proof of Theorem~\ref{fleurigesmurf} we can use both Theorem~\ref{zwierigesmurf}(A) en (B). 

\begin{theorem}\label{th:nomhcpast}
Neither {\liv} nor {\mHCl} are in $\Lambda^\circ_{{\sf PA}^\ast}$.
\end{theorem}

\begin{proof}
Since ${\sf PA}^\ast$ is $\Sigma_1^0$-sound and validates {\mHCb}, by Theorem~\ref{fleurigesmurf},
it cannot validate {\liv}. 
%\medent
 Suppose now ${\sf PA}^\ast$ validates {\mHCl}. Then $\Lambda^\circ_{{\sf PA}^\ast}$ extends 
$\imHCl^- = \logla + \bvi +\mHCl$. It follows, by Lemma~\ref{th:mhcderiv}(\ref{mhcdi}), that $\Lambda^\circ_{{\sf PA}^\ast}$
contains {\liv}. Quod non, as we just saw.
\end{proof}

\nosmurf
Another salient consequence of Theorems~\ref{secretarissmurf}, \ref{directiesmurf} and \ref{zwierigesmurf} is the following result.

\begin{theorem}
For no $T \geq {\sf HA}$, we have: $\Lambda^\circ_T = \loglc +  \mHCb$.
\end{theorem}

\begin{proof}
Suppose ${\sf HA}\leq T$. Clearly, if $(\loglc +  \mHCb) \subseteq \Lambda^\circ_T$, it follows that
$T \vdash \Box_T\bot$. But then $\Box \bot\in \Lambda^\circ_T$. On the other hand, by a simple
Kripke model argument, we can show that $\loglc +  \mHCb\nvdash \Box \bot$.
\end{proof}

\noindent
Thus, not every extension of $\logld$ can be obtained as the preservativity logic of a $T\geq {\sf HA}$.

\medskip\nosmurf
We finish this subsection by giving a better condition under which {\mHCl} cannot be in the preservativity logic of a theory.
This condition will again imply that {\mHCl} is not in $\Lambda^\circ_{{\sf PA}^\ast}$.

\begin{theorem}
Suppose ${\sf HA}\leq T$, $T$ has the disjunction property and $T$ is consistent. Then, $\Lambda^\circ_T$ does not contain 
{\mHCl}. 
\end{theorem}

\begin{proof}
Suppose ${\sf HA}\leq T$, $T$ has the disjunction property and $T$ is consistent.
Moreover, suppose  $\Lambda^\circ_T$ contains 
{\mHCl}. We will derive a contradiction.

Let $G := G_T$. Since $T\vdash G \tto \bot$, it follows, by {\mHCl}, that $T \vdash G\vee \neg\, G$.
Hence, by the disjunction property, we find $T\vdash G$ or $T \vdash \neg\, G$. Hence $\top \tto_T G$ and
$\top \tto_T \neg\, G$. Ergo, $T \vdash \bot$.  
\end{proof}

\nosmurf
We note that {\sf PA} trivially satisfies {\mHCl}. Moreover,  ${\sf HA}\leq {\sf PA}$ and {\sf PA} is (hopefully) consistent.
However, {\sf PA} does not have the disjunction property. 

%%%%%%%%%%%%%%%%%%%%%%%%%%%%%%%%%%

\section{Conclusions}
\nosmurfduo
We are not nearly done, but our space is running out: if we did not stop now, we would have to turn this paper into a monograph. We hope to have convinced the reader that constructive $\tto$ provides a fascinating subject of research wherever it appears---be it computer science, philosophy or, especially, metatheory of arithmetic. %Not nearly enough energy has been devoted to its study as a different connective than intuitionistic $\Box$, still less so to creating unified theory underlying possible interpretations and applications. 
 This last context is particularly rife in challenges, despite decades of diligent research in the area. Let us highlight again several lists of unsolved problems regarding arithmetical interpretations: Open Questions \ref{que:main}, \ref{que:dichar}, \ref{que:hapre}, \ref{que:shavrukov}, \ref{que:extend} and (in Appendix \ref{sec:inter} below) \ref{que:bekle} and \ref{que:intea}. %It is our hope that this overview will convince the reader to try their hand at some of these. 
 %Compared with standard provability logics of $\lna{PA}$ and its extensions, it is amazing how little we understand, even regarding fairly basic problems. 

This,  however,  is not the only area where interesting open questions abound. As a simple example, consider the study of axiomatization and proof systems for various fragments  of \latto\ (e.g., Open Questions \ref{que:subi} and \ref{que:disfree}). Moreover, we have only briefly touched on the question of computational significance of $\tto$.  Extending category-, proof- and type-theoretic frameworks for ``strong arrows'' in computer science (\S~\ref{sec:arrows} and references therein) and providing Curry-Howard/computational interpretations of different axioms in Table \ref{tab:mainax} (cf. in particular Open Questions \ref{que:diarrow} and \ref{que:guarded}) 
  would seem a natural research direction.  %provides more examples of related problems to attack.
  %Finally, let us point out that computational and arithmetical interpretations could meet in the study of varieties of constructive $\tto$ in type theory.

A century after the publication of Lewis' first papers on $\tto$ and the \emph{Survey}, %we can safely say that 
 %his idea of introducing
  the full potential of the strict implication connective still remains to be exploited. 
 %still proves more %has been more
  %prescient than commonly acknowledged.
   %Its full potential could have been better exploited if
     It could have been otherwise if Lewis followed his evident interest in non-boolean logics (cf. \S\ref{sec:couldbrouw}).  %His main error, if one can call it an error, lied in
 %One factor %of conceptual barriers which 
  %preventing his ideas from full fruition was his reluctance to follow his open-minded  (cf. \S\ref{sec:couldbrouw}). 
 % lied in not breaking radically enough with classical logic, despite his open-minded mention of Brouwer's work (cf. \S\ref{sec:couldbrouw}). 
  %It was more than understandable in the context of his times, when the very idea of introducing intensional connectives was revolutionary enough.
   Another decision which in hindsight proved premature was to insist on principles like \aApp\ in even the weakest variant of his system (cf. Remark \ref{rem:lewapp}), which effectively rules out some of the most fruitful provability-motivated applications of $\tto$. %Ironically enough, these applications lie in the study of provability 
 %We, however, 
 With these conceptual blocks out of the way and having the advantage of an additional century worth of research on constructive logic, we have no excuse not to carry the torch further. 

\nosmurfduo
\subparagraph*{Acknowledgements.}
\nosmurfduo
%\begin{acknowledgments}
We would like to thank: Wesley H. Holliday and Merlin G\"otlinger for proof-reading parts of an earlier draft of this paper; Mark van Atten for pointing out references on the work of Orlov; the referees for their comments; and, especially, the editors for their successful effort to make the entire process as smooth as possible with a paper of this size. We are also grateful to Erwin R. Catesbeiana for his reflections on the reflection principle.  %and the intricate relationship between preservativity and provability.
%\end{acknowledgments}

%\tadeusz{Initial discussion with Rosalie?}

%{\scriptsize
%\section*{Bibliography}

%\tadeusz{The format of references is a total mess, not to mention duplicates. Guess we'll take care of it later.}\albert{I follow the style of not full first names
%but letters. Do you want first names? In the case I would be in favor of the second first name as a letter.}

%%%%%%%%%%%%%%%%%

%\bibliographystyle{elsarticle-num}
%\bibliographystyle{plainnat}
%\bibliographystyle{alpha}
%\bibliography{abbrs-full,intmodim,provint,loebmuintim,bisanitim,guardedim,strictlewishistory}
%}

%\ifbibt
%{\scriptsize
%\section*{Bibliography}
%\bibliographystyle{elsarticle-num}
%\bibliographystyle{plainnat}
%\bibliographystyle{alpha}
%\bibliography{IM_compactified_polished}
%\bibliography{IM_compactified_biber}
%\bibliography{abbrs-full,intmodim,provint,loebmuintim,bisanitim,guardedim,strictlewishistory}
%}
%}
%\else
\printbibliography
%\fi

\appendix

\renewcommand{\thesection}{\Alph{section}}

\section{A recap of realizability \AIIIp} \label{sec:real}

\nosmurfduo
We need Kleene's T-predicate: ${\sf T}(e,x,p)$ means $p$ is a halting computation for the partial recursive
function with index $e$ on input $x$. We write ${\sf U}(p)=y$ for: the result of computation $p$ is $y$.
We employ the usual assumptions that for at most one $p$ we have ${\sf T}(e,x,p)$ and that ${\sf U}(p) \leq p$.
Define:
\begin{itemize}
\item
$e \cdot x = y$ if $\exists p\, ({\sf T}(e,x,p) \wedge {\sf U}(p) = y)$.
\item
 $p:(e \cdot x = y)$ if ${\sf T}(e,x,p) \wedge {\sf U}(p) = y$.
\item
$e\cdot^z x = y$ if $\exists p \leq z\;\, p:(e\cdot x = y)$ or ($\forall q \leq z\, \neg{\sf T}(e,x,q) \wedge y=0$).
\item $e \cdot x\downarrow$ for (the partial recursive
function with index) $e$ being defined on $x$ and $e \cdot x\uparrow$ otherwise.
\end{itemize}

\noindent
 Sometimes we will need Kleene application for functions of several arguments. In such cases, we will write $x \cdot (\vec y\,)$. The tuple $(\vec y\,)$ is tacitly identified with a number, in particular we use $\varepsilon$ for the (code of the) empty sequence. 
 
We have several variants of the (intuitionistic)  Church's Thesis:
\begin{description}
\item[${\sf CT}_0$]
$\;\;\; \forall x \,\exists  y\, Axy \to \exists e\, \forall x\, \exists y\, (e\cdot x = y \wedge Axy)$. This is the standard arithmetical form of the Thesis, with only numerical quantifiers appearing (modulo a version of the choice principle), rather than an universally quantified function symbol \cite[1.11.7,p.95]{troe:meta73}, \cite[4.3,p.193]{troe:cons88vol1}.
\item[${\sf CT}_0!$]
$\;\;\; \forall x \,\exists ! y\, Axy \to \exists e\, \forall x\, \exists y\, (e\cdot x = y \wedge Axy)$. This slightly weakened form will play a central r\^ole in Appendix \ref{ctzero}, where more references are provided.
 \item[${\sf ECT}_0$] is the \emph{extended Church's Thesis}  \cite[4.4,p.199]{troe:cons88vol1}, \cite[3.2.14,p.195]{troe:meta73}:
\[\;\;\; \forall x \,(Bx \to \exists y\, Axy) \to \exists e\, \forall x\,(Bx \to \exists y\, (e\cdot x = y \wedge Axy))\]
\end{description}

\noindent
where $B$ ranges over \emph{almost negative} formulas: %, defined below.
\begin{itemize}
 \item
 $B ::= S \mid (B \wedge B) \mid (B \to B) \mid \forall v\, B$ 
 \end{itemize}

%Let $S$ range over $\Sigma_1^0$-formulas. An \emph{almost negative} $B$ is of the form: 
\noindent
and $S$ ranges over $\Sigma_1^0$-formulas. Almost negative formulas will play an important r\^ole in Appendix \ref{sec:interha}.

From \S~\ref{sec:iprel} on, we have been using the notion of \emph{q-realizability} \cite[\S~3.2.3, p. 189]{troe:meta73}, a variant of the usual Kleene realizability: %, defined as follows:

\newcommand{\qrel}{\widetilde q}

\medskip

\begin{tabular}{>{$}c<{$}@{$\,:=\;$}>{$}l<{$}}%@{\qquad\qquad}>{$}l<{$}@{$\,:=\,$}>{$}l<{$}}
x \qrel A & A \qquad \bro $A$ \text{\, atomic)} \\[\tbskip]
x \qrel (A \wedge B) & (j_1 x) \qrel A \wedge   (j_2 x) \qrel B \\[\tbskip]
x \qrel (A \vee B) & (j_1 x = 0 \to  (j_2 x) \qrel A) \vee (j_1 x \neq 0 \to  (j_2 x) \qrel B) \\[\tbskip]
x \qrel (A \to B) & (A \to B) \wedge \forall v (v \qrel A \to \exists z (x \cdot v = z \wedge z \qrel B)) \\[\tbskip]
x \qrel (\exists v Av) & (j_2 x) \qrel A(j_1 x) \\[\tbskip]
x \qrel (\forall v Av) & \forall v (Av \wedge \exists z (x \cdot v = z \wedge z \qrel Av))
\end{tabular}

\medskip

\noindent
where $j_1,j_2$ are the inverses of a chosen pairing function. Note that, unlike Troelstra \cite[\S~3.2.3]{troe:meta73}, we choose to plug additional conjuncts into  clauses for $\to$ and $\forall$, rather than for $\vee$ and $\exists$.

Apart for Troelstra \cite{troe:meta73} and Troelstra and van Dalen \cite{troe:cons88vol1}, another reference on realizability in {\sf HA} we recommend is Dragalin \cite{Dragalin88:trams}.

\section{$\Pi_1^0$-conservativity \qquad \AIIIp}\label{sec:picon}
%\nosmurfduo
\begin{foots}
In this appendix, we discuss both classical and constructive interpretability logic.
\end{foots}

\medent
An arithmetical theory $U$ is $\Pi^0_1$-conservative over a theory $T$ or $T \jump U$ if, for all $\Pi_1^0$-sentences $P$,
we have, if $U \vdash P$, then $T \vdash P$.\footnote{The use of the notation $\jump$ is just local in this paper. Often one uses
$\rhd_{\Pi^0_1}$.}${}^{,}$\footnote{Reflection of the general case, where we also consider non-arithmetical theories,
reveals that  $\Pi_1$-conservativity is `really' a relation between interpretations of a basic arithmetical theory in various theories.}
We write $A \jump_T B$ for $(T+A) \jump (T+B)$. 

We expand the language of propositional logic with the unary $\Box$ and the binary $\rhd$.
Consider any theory $T$.
We  set $\InF{2}{T}(\Box) := {\sf prov}_T(v_0)$ and $\InF{2}{T}(\rhd) := {\sf picon}_T(v_0,v_1)$.
\emph{Par abus de langage}, we write $\jump_T$ for $\rhd_{\InF{2}{T}}$, thus introducing an innocent ambiguity.
We write $\Lambda^{\bullet}_T$ for $\Lambda_{T,\InF{2}{T}}$. 

\medent
We note that a $\Pi^0_1$-sentence is constructively equivalent to the negation of $\Sigma^0_1$ sentence.
This implies that $A \to P$ is equivalent to $\neg\neg\,A \to P$. Thus, we find that $\neg\neg\, A $ and $A$
are mutually $\Pi^0_1$-conservative over $T$. This means that $\Box_T$ can only be defined from
$\jump_T$ for theories in which $\Box_TA$ and $\Box_T \neg\neg \, A$ are provably equivalent for all $A$. 
 Hence, in general provability cannot be defined from $\Pi_1^0$-conservativity over constructive theories.

\subsection{The classical case}
\nosmurfduo
Consider a classical theory $T$. We have 
 $T$-verifiably that $A \tto_T B$ iff $\neg\, B \jump_T \neg \, A$, and
$A \jump_T B$ iff $\neg\, B \tto_T \neg\, A$.
Thus, over $T$, $\Sigma_1^0$-preservativity and $\Pi^0_1$-conservativity are intertranslatable. This tells us that the $\Sigma^0_1$-preservativity
logic of $T$ can be found via a transformation of the $\Pi^0_1$-conservativity
logic of $T$.

%We give the axiomatization
%of {\sf ILM} and the corresponding preservativity principles. 
 The logic {\sf ILM} consists of  {\logbb} plus the following
principles.

 \renewcommand{\arraystretch}{1.5}
\[
\begin{tabular}{|ll|ll|}\hline
{\sf J}1 &	$ \Box (\phi \to \psi ) \to \phi \rhd \psi $ & \lx & $ \Box (\phi \to \psi ) \to \phi \tto \psi $ \\
{\sf J}2 &	$ (\phi \rhd \psi \wedge \psi \rhd \chi ) \to \phi \rhd \chi$ & \lii & $ (\phi \tto \psi \wedge \psi \tto \chi ) \to \phi \tto \chi$ \\ 
{\sf J}3 &	$ (\phi \rhd \chi \wedge \psi \rhd \chi ) \to (\phi \vee \psi )\rhd \chi $ &  \liii  & $ (\phi \tto \psi \wedge \phi \tto \chi ) \to \phi \tto (\psi \wedge \chi)$  \\
{\sf J}4 &	$ \phi \rhd \psi \to (\Diamond \phi \to \Diamond \psi )$ & \lxi & $ \phi \tto \psi \to (\Box \phi \to \Box \psi )$ \\
{\sf J}5 &	$ \Diamond \phi \rhd \phi$ & \lv & $\phi \tto \Box \phi$ \\
{\sf M} &   $ \phi \rhd \psi \to (\phi\wedge \Box \chi) \rhd (\psi \wedge \Box\chi)$ & \lviii & 
$ \phi \tto \psi \to ( \Box \chi \to \phi) \tto (\Box \chi \to\psi)$ \\ \hline
\end{tabular}
\]

\noindent
The list of principles for preservativity given above is equivalent to 
${\mathrm c}\hyph{\sf PreL} := {\mathrm i}\hyph{\sf PreL}^{-} + \lna{em}$.
See Lemma~\ref{th:basicia}, Fact~\ref{fact:necunnec}, Lemmas~\ref{lem:gla} and~\ref{lem:gwader}.

Theorem 12 of \cite{bekl:limi05} yields that the $\Pi^0_1$-conservativity
logic of $T$ is  {\sf ILM} whenever
$T$ is an extension of $\mathrm{I}\Pi_1^{-}+{\sf Exp}$. This class of theories contains
such salient theories as $\mathrm{I}\Sigma_1$ and {\sf PA}. 

Thus, we have justified Theorem~\ref{partijsmurf}, which tells us that
$\Lambda^\circ_T = {\mathrm c}\hyph{\sf PreL} $ if $T$ is a $\Sigma_1^0$-sound extension of 
$\mathrm{I}\Pi_1^{-}+{\sf Exp}$.

\medent
We note that the principle corresponding to {\lvi} would have been:
\[ (\dag)\;\;\;  (\phi \jump \psi \wedge \phi \jump \chi ) \to \phi \jump (\psi \wedge \chi).\]
Let $T$ be a $\Sigma_0^1$-sound theory with ${\sf PA} \leq T$. Consider the sentence
$G := G_T$ from \Subsection~\ref{sec:falsity}. Suppose $T$ satisfies (\dag). We have, in $T$, both
$\top \jump_T G$ and $\top\jump_T \neg\, G$. It follows that we have
 $\top\jump_T\bot$, i.e. $\Box_T\bot$. However, this contradicts $\Sigma_0^1$-soundness.
%%%

\subsection{The constructive case}
\nosmurfduo
In this subsection we zoom in on the case of {\sf HA}.
Here the situation for $\Pi_1^0$-conservativity is quite different. We still have, {\sf HA}-verifiably, $A \jump_{\sf HA} B$ iff
$\neg\, B \tto_{\sf HA} \neg\, A$. However, we do not have the equivalence of $A\tto_{\sf HA} B$ and $\neg\, B \jump_{\sf HA} \neg\, A$.
The equivalence fails in both directions.

We have $(\neg\neg\, \Box_{\sf HA} \bot \to \Box_{\sf HA}\bot) \tto \Box_{\sf HA} \bot$ \cite{viss:prop94}, but we do not have \\
 $\neg\, \Box_{\sf HA} \bot \jump_{\sf HA} \neg\, (\neg\neg\, \Box_{\sf HA} \bot \to \Box_{\sf HA}\bot)$,
 as this  is equivalent to $\Box_{\sf HA}\neg\neg\,\Box_{\sf HA}\bot$.
 
 In the other direction, trivially, we do have $\neg\, (\Box_{\sf HA}\bot\vee \neg\, \Box_{\sf HA}) \jump_{\sf HA}\bot$.
 But as shown in \S~\ref{sec:provo}, $\top\tto (\Box_{\sf HA}\bot\vee \neg\, \Box_{\sf HA}\bot)$ fails. %This follows from the principles treated below.

It is easily seen that the logic $\Lambda_{\sf HA}^\bullet$ contains i-{\sf ILM}, the theory  axiomatized by \logba + {\sf J}1-5 + {\sf M}.
 However, it contains more. As noted above, we have the principle $\vdash \neg\neg\, \phi \rhd \phi$.

%%%%%%%%
%%%%%%%%%%

\section{Interpretability \qquad \AIIIp}\label{vrolijkesmurf}\label{sec:inter}
%\nosmurfduo
\begin{foots}
In this appendix, we discuss both classical and constructive interpretability logic.
\end{foots}

\subsection{Basics}\label{ilo}
\nosmurfduo

\begin{foots}
 NB: The definitions of this subsection work for all theories in finite signature. So in this subsection the
theory need not be arithmetical and the axiom set can be just any set of axioms regardless of the complexity.
\end{foots}

\medent
As is well known, purely relational signatures can simulate signatures with terms via a term-unraveling procedure. %, that enables us to simulate signatures with terms
%in purely relational signatures.
 Thus, we can justify defining interpretations only for relational languages.
A \emph{one-dimensional translation} $\tau$ between relational signatures $\Xi$ and $\Theta$ provides a domain formula
$\delta_\tau(v_0)$ of signature $\Theta$ and assigns to each $n$-ary $\Xi$-predicate a $\Theta$-formula $P_\tau(v_0,\ldots, v_{n-1})$.
Here the variables of $\delta_\tau$ and $P_\tau$ are among those shown. We define a translation $A \mapsto A^\tau$ from
$\Xi$-formulas to $\Theta$-formulas as follows:
\begin{itemize}
\item
$P^\tau(x_0,\ldots, x_{n-1}) := P_{\tau}(x_0,\ldots,x_{n-1})$ (in case an $x_i$ is not free for $v_i$ in $P_\tau(v_0,\ldots,v_{n-1})$,
we employ the mechanism of renaming bound variables.) 
\item
$(\cdot)^\tau$ commutes with the propositional connectives.
\item
$(\forall x\, B)^\tau := \forall x\, (\delta_\tau(x) \to B^\tau)$, $(\exists x\, B)^\tau := \exists x\, (\delta_\tau(x) \wedge B^\tau)$.
\end{itemize}
\emph{Nota bene}: we also allow identity to be translated to a different formula.

We can  define the more complex notion of \emph{many-dimensional translation with parameters}. 
In the many-dimensional case a sequence of objects of the interpreting theory represents an object in the
interpreted theory. In the case with parameters allow a sequence of extra free variables, the parameters, to occur in the
domain formula and in the translations of the predicate symbols. 

Suppose $T$ has signature $\Theta$ and $U$ has signature $\Xi$. We define:
\begin{itemize}
\item
An interpretation $K:U \to T$ is a triple $\tupel{U,\tau,T}$, such that, 
for all $\Xi$-sentences $A$, if $U \vdash A$, then $T \vdash A^\tau$. 
\item
$T \rhd U$ if there is an interpretation $K:U \to T$.
\item
$A\rhd_T B$ if $(T+A) \rhd (T+B)$.
\end{itemize}

\noindent 
If we allow parameters, we add a parameter-domain $\alpha_K$ to the specification of $K$.
We demand that  $K:U \to T$ iff, $T$ proves that $\alpha_K$ is non-empty and that,
for all $\Xi$-sentences $A$, if $U \vdash A$, then $T \vdash \forall \vec w\, (\alpha_K(\vec w) \to A^{\tau,\vec w})$. 

We write $\delta_K$ for $\delta_{\tau_K}$ and $P_K$ for $P_{\tau_K}$.
For more information about the definition of an interpretation, see e.g. \cite{viss:cate06} and \cite{viss:whyR14}.

In the case of extensions of $\iea$ as the interpreting theory one can show that, for our
purposes, allowing many-dimensional interpretations makes no difference. We can eliminate the higher dimensions using Cantor pairing.
In case we have extensions of {\sf PA} as the interpreting theory, allowing parameters makes no difference.
We can eliminate parameters using the Orey-H\'ajek Characterization that guarantees an interpretation without parameters whenever there
is an interpretation.

In case we are not considering extensions of {\sf PA}, it is in most cases unknown whether 
 allowing parameters has an effect on the interpretability logic.
 
If the interpreting theory is an extension of {\sf PA} we can always eliminate domain relativization and
we can always replace an interpretation by an identity preserving equivalent. In case the interpreted theory
has {\sf PA}-provably infinitely many arguments, we even can do both at the same time.

If the interpreting theory is classical and does define one element in the interpreted theory, we can  eliminate the domain relativization
by setting all elements outside the original domain equal to the definable element. If we allow parameters we can eliminate the domain
relativization always as long as the interpreting theory is classical.
  
\subsection{Interpretability Logic introduced}
\nosmurfduo
The relation $\rhd_T$ can be arithmetized, say by ${\sf int}_T$.
We expand the language of propositional logic with the unary $\Box$ and the binary $\rhd$.
Consider any theory $T$ with a $\Delta_0({\sf exp})$-axiomatization.
We  set $\InF{3}{T}(\Box) := {\sf prov}_T(v_0)$ and $\InF{3}{T}(\rhd) := {\sf int}_T(v_0,v_1)$.
\emph{Par abus de langage}, we write $\rhd_T$ for $\rhd_{\InF{3}{T}}$, thus introducing an innocent ambiguity.
We write $\widetilde \Lambda_T$ for $\Lambda_{T,\InF{3}{T}}$.

In the classical case $\Box_TA$ is equivalent to $\neg\, A \rhd_T \bot$. Thus, classically, we also have the option
to  expand only with $\rhd$ and treat $\Box$ as a defined symbol. This equivalence can fail  intuitionistically.
One can see this, e.g., by taking $T := {\sf HA}$ and $A := (\Box_{\sf HA}\bot \vee \neg\,\Box_{\sf HA}\bot)$.
At present it is unknown whether $\Box_{\sf HA}A$ is {\sf HA}-provably equivalent to $\top \rhd_{\sf HA} A$, so we cannot exclude
that there would be a definition of the $\Box$ in terms of interpretability over {\sf HA}.

\subsection{Classical Interpretability Logic}\label{kalsbeek}
\nosmurfduo
Over {\sf PA} arithmetic interpretability and $\Pi^0_1$-conservativity coincide. Thus, the 
$\widetilde\Lambda_{{\sf PA}} = {\sf ILM}$.
The arithmetical completeness of {\sf ILM} for interpretability over {\sf PA} was proven by
Berarducci \cite{bera:inte90} and Shavrukov \cite{shav:rela88} proved that 
 this result also holds for all $\Sigma_0^1$-sound
extensions of {\sf PA}. \footnote{If we leave, for a moment, the context of arithmetical theories, we can say that the result holds
for all classical essentially reflexive sequential theories (with respect to some interpretation of arithmetic).}
%\medent
 The reader is referred to \cite{japa:logi98,viss:over98,arte:prov04} for more information about classical interpretability logic.

\medent
We know two further arithmetically complete interpretability logics. The first is {\sf ILP}.
This is the logic of $\Sigma_0^1$-sound finitely axiomatized extensions of ${\sf EA}^+$, also known as 
$\mathrm I\Delta_0+{\sf Supexp}$. If we take the contraposed preservativity-style version of {\sf ILP},
we obtain the logic $\loglf + \lxii+\lna{em}$ \cite{viss:inte90}.

\subsection{Constructive Interpretability Logic}
\nosmurfduo
In this subsection we treat constructive interpretability logic with the interpretability logic of {\sf HA} as our main focus. 
We need some preliminary material to get the discussion off
the ground.

\subsubsection{i-Isomorphism}
\nosmurfduo

\begin{foots}
The materials of the subsubsection work for any theories of finite signature.
\end{foots}

\medent
We will need the notion of  \emph{i-isomorphism} between interpretations.
Two interpretations $K,M:U\to T$ are {\emph i-isomorphic} if there is an \emph{i-isomorphism} $G$ between $K$ and $M$.
A $T$-formula $G$ is an \emph{i-isomorphism between $K$ and $M$} if the theory $T$ verifies that `$G$ is a bijection between $\delta_K$ and
$\delta_M$ that preserves the predicate symbols of $U$'. For example if $P$ is unary, we ask:
$T\vdash \forall u\forall v\, ((\delta_K(u) \wedge \delta_M(v) \wedge G(u,v)) \to (P^K(u) \iff P^M(v)))$. 

\medent
Let $T$ be any extension of {\sf HA}. Suppose $K: \iea \to T$. We also have the identical interpretation $\mathcal E: \iea \to T$
that translates all predicate symbols to themselves. E.g. ${\sf A}_\mathcal E(v_0,v_1,v_2) := {\sf A}(v_0,v_1,v_2)$, where {\sf A} is the relation
representing addition.
Then, by a special case of the Dedekind-Pudl\'ak Theorem,
%we find that
 there exists a formula $F$ such that $T$ proves that $F$ is an initial embedding of $\mathcal E$ in $K$.
Now it is easy to see that $\mathcal E$ is i-isomorphic to $K$ iff $T$ proves that $F$ is surjective.
Thus, there is a single fixed statement, say ${\sf C}_K$, that expresses that $\mathcal E$ is i-isomorphic to $K$.\footnote{We need
minor modifications of the formulation in case we have parameters.} 

\subsubsection{${\sf CT}_0!$}\label{ctzero}
\nosmurfduo
In this subsection, we present some basic facts about ${\sf CT}_0!$ (cf. Appendix \ref{sec:real}), which we will use to
derive a new interpretability principle over {\sf HA}.  % contains information and notation.

%\medent

\noindent
The theorem below is proven in  \cite{viss:pred06}. For completeness' sake, we repeat the proof here.
The proof is an adaptation of the proof of Tennenbaum's Theorem. Such proofs were used before to prove
the categoricity of i-{\sf EA} in constructive meta-theories under the assumption of Church's Thesis and
Markov's Principle. By taking some extra care we can avoid the assumption of Markov's Principle.

 \begin{theorem}\label{quorumsmurf}
 The theory i-{\sf EA} verifies the following. Suppose $T$ extends ${\sf HA}+{\sf CT}_0!$
 and $K:T \rhd \iea$.
 Then, $T \vdash {\sf C}_K$.
 \end{theorem}
 
 \begin{proof}
 We give the proof for the case without parameters. We need minor modifications to add parameters.
 
 \medent
 Suppose $T$ extends ${\sf HA}+{\sf CT}_0!$ and $K:T \rhd \iea$. 
 We note that i-{\sf EA} proves that $\lambda e\lambda x.(e\cdot^z x)$  is total.  
 Let ${\sf sig}(x) =1$ if $x>0$ and ${\sf sig}(x) =0$ if $x=0$.
  Let $F$ be the initial embedding of $\mathcal E$ in $K$. 
 
 \medent
 We work in $T$. Fix an element $z$ of $\delta_K$. We define the operation $\ast$ as follows.
 \begin{itemize}
 \item
$e \ast^z x =  y$ if $\exists e'\,\exists x'\,\exists y'\, (F(e,e') \wedge F(x,x') \wedge F(y,y') \wedge ({\sf sig}(e'\cdot^z x') = y')^K)$.
 \end{itemize}
 It is easy to see that $H_z := \lambda e.(1 - e\ast^z e)$ is a total 0,1-valued function. By ${\sf CT}_0!$, there is a recursive function
 that computes $H_z$, say with index $h$. Let $p:(h\cdot h = i)$.
Suppose $F(h,h')$ and $F(i,i')$ and $F(p,p')$. 
We have $H_z(h) =i$ and, hence, $h\ast^z h = 1-i$. This means that $({\sf sig}(h' \cdot^z h') = 1-i')^K$.
On the other hand, since $F$ is an initial embedding, we find $(p': (h'\cdot h' = i'))^K$. 

\medent We reason inside
$K$. In case $p'\leq z$, we have that $h'\cdot h' = h'\cdot^z h'$. Hence,
$i' = {\sf sig}(1-i')$. Quod non. Hence $z< p'$. We leave $K$.

\medent
Since $F$ is an initial embedding, we can find a $z^\ast< p$ such that $F(z^\ast,z)$. 
Since $z$ was an arbitrary element of $\delta_K$, we may conclude
that $F$ is surjective.    
 \end{proof}
 
 \noindent
   It follows that the interpretability logic of extensions $T$ of  ${\sf HA}+{\sf CT}_0!$ contains the following principle:
 \begin{itemize}
 \item
 $ \phi \rhd \psi \to \Box (\phi \to \psi)$. 
 \end{itemize}
 
 \begin{remark}
 The Tarski biconditionals {\sf TB} for the arithmetical language are all sentences of the form ${\sf True}(\gnum{A}) \iff A$.
 It is clear that every arithmetical theory locally interprets itself plus {\sf TB}. In the classical case it follows that
 ${\sf PA} \rhd ({\sf PA}+{\sf TB})$. However, we cannot have  ${\sf HA} \rhd ({\sf HA}+{\sf TB})$. If we had
 ${\sf HA} \rhd ({\sf HA}+{\sf TB})$, then we would have $K:({\sf HA} +{\sf CT}_0!)\rhd ({\sf HA}+{\sf TB})$, for some $K$.
 However, since the reduct of $K$ to the arithmetical language is i-isomorphic to $\mathcal E$, this
  would enable us to define truth for the arithmetical language in ${\sf HA}+{\sf CT}_0!$. By Tarski's Theorem
  on the undefinability of truth, we would find that ${\sf HA}+{\sf CT}_0!$ is inconsistent. Quod non. 
 \end{remark}
 
 \noindent
 For some further information about ${\sf CT}_0!$, see \cite{oost:lifs90}.
 
  \begin{remark}\label{smulsmurf}
With respect to interpretability, there is a certain analogy between ${\sf HA}+{\sf CT}_0!$ and ${\sf HA}^\ast$.

In \cite{viss:pred06}, the following result is proved. Let $\tau$ be translation from the arithmetical language to itself.
Consider the theory $T:= {\sf HA}^\ast + (\iea)^\tau$. Clearly, $\tau$ carries an interpretation of {\iea} in
$T$. Let ${\sf F}_\tau$ be the standard embedding of the $T$-numbers into the $\tau$-numbers. We have:
\[ {\sf HA}^\ast + (\iea)^\tau \vdash \forall y\, (\delta_\tau(y) \to (\exists x\, {\sf F}_\tau(x,y) \vee \Box_{\sf HA} \bot)).\footnote{In case $\tau$ has
parameters a slight adaptation of the formulation is needed.}\]
It is easy to see that we cannot generally eliminate the $\Box_{\sf HA} \bot$ from the result since ${\sf PA}+\Box_{\sf PA}\bot$ is an extension
of ${\sf HA}^\ast$. The theory ${\sf PA}+\Box_{\sf PA}\bot$ has many non-trivial interpretations of i-{\sf EA}. It has not been explored whether
the result described here throws any shadows on the interpretability logic of {\sf HA}.
\end{remark}

\subsubsection{The Interpretability Logic of {\sf HA}} \label{sec:interha}
\nosmurfduo
The interpretability logic of {\sf HA} has not yet been studied. It seems to us that there are some good reasons for this neglect, 
 the first being that
the more basic problem of the provability logic of {\sf HA} is still wide open. %It is better to first concentrate on that, since it is part of the full problem
%of the interpretability logic of {\sf HA}. 
 Unlike the case of the logic of $\Sigma^0_1$-preservativity, there are no indications that the study
of  the logic of interpretability will help in the study of provability logic. %In the case of interpretability, there are no such indications. 

Interpretability itself is intuitionistically significant, %and important in lots of cases.
e.g.,  the usual translations of elementary syntax in arithmetic work equally well
classically and intuitionistically. But---and here is our second reason---the usefulness of interpretations to compare
arithmetical theories is much diminished.
 %A symptom of this is that 
% the translation methods used to study  intuitionistic arithmetic are not translations.
 For example, the %double negation
   $\neg\neg$-translation does not commute with disjunction, and, thus, fails to carry an interpretation. 
 %One reason for the relative unimportance of interpretations is the fact that
 The demand of commutation with disjunction and existential quantification is much more restrictive intuitionistically than
classically. 

Still,  %has some interest. Certainly, 
  studying the differences between the interpretability logic of {\sf HA} and that of {\sf PA} highlights how %makes
%one more aware 
 the classical principles depend on the chosen logic. Also, the relevant methods are quite interesting.
Finally, a good friend makes an appearance here: Tennenbaum's Theorem plays a significant r\^ole.

\medent
Which of the axioms of {\sf ILM} remain in the interpretability logic of {\sf HA}? 
The principles of {\logba} and the principles {\sf J}1,2,4 and {\sf M} are valid over {\sf HA}. However, {\sf J}5 fails since, e.g., its instance
$\Diamond \Box \bot \rhd \Box \bot$  
fails.\footnote{\label{minismurf}The fact that $\Diamond \Box \bot \rhd \Box \bot$ is not valid for
{\sf HA}  follows, for example, from Theorem~\ref{lampjesmurf}  in combination with
what we already know about the provability logic of {\sf HA}.} 
The status of {\sf J}3 is unknown. We note that the classical argument for {\sf J}3 does yield following weakened version.
\begin{itemize}
\item
$ (\phi \rhd \chi \wedge \psi \rhd \chi) \to ((\phi \vee \psi) \wedge \neg\, (\phi \wedge \psi)) \rhd \chi$  
\end{itemize}
We define the modal $\Sigma^0_1$-formulas as follows:
\begin{itemize}
\item
$\sigma ::= \top \mid \bot \mid \Box \phi \mid (\sigma \vee \sigma)$
\end{itemize} 
The following valid principle was noted by Lev Beklemishev in conversation.
\begin{itemize}
\item
%Let \\
$ (\sigma \rhd \chi \wedge \sigma' \rhd \chi) \to (\sigma \vee \sigma') \rhd \chi$, with $\sigma$ and $\sigma'$ being modal $\Sigma_1^0$. 
\end{itemize}

\begin{quest} \label{que:bekle}
Let $A_0 := \forall S\in \Sigma^0_1\, ({\sf True}_{\Sigma^0_1}S \vee \neg\, {\sf True}_{\Sigma^0_1}S)$ and 
$A_1 := \forall S\in \Sigma^0_1\, (\Box_{\sf HA} S \to {\sf True}_{\Sigma^0_1}S)$.
%\medent
 As $ \Diamond_{\sf HA}A_1$ implies, by the Double Negation Translation,  $\Diamond_{\sf PA}A_1$, we have i-{\sf EA}-verifiably $(A_0\wedge \Diamond_{\sf HA}A_1) \rhd_{\sf HA} A_1$. We can do then  the Henkin construction for ${\sf PA}+A_1$ using the decidability for $\Sigma^0_1$-sentences.
We also have trivially $A_1 \rhd_{\sf HA} A_1$.
But do we have: 

\medskip

\qquad\qquad $((A_0\wedge \Diamond_{\sf HA}A_1)\vee A_1) \rhd_{\sf HA} A_1$ ?

\medent
%By similar considerations, 
Similarly, we have for any $B$ that 
%\medskip
%\qquad\qquad 
$(A_0\wedge \Diamond_{\sf HA}\top) \rhd_{\sf HA} (B\vee \neg\, B)$. 
 But do we have for all $B$ that 

\medskip

\qquad\qquad $((A_0\wedge \Diamond_{\sf HA}\top) \vee B) \rhd_{\sf HA} (B\vee \neg\, B)$ ? 
%\end{center}
\end{quest}

\noindent
Is the classically invalid principle
$\vdash (\phi \rhd \psi \wedge \phi \rhd \chi) \to \phi\rhd (\psi \wedge \chi)$ still invalid over {\sf HA}? 
We do not know that for $\phi = \top$.  However, the usual construction of 
Orey sentences for {\sf PA} can be adapted to give a sentence $O$ such that
$A \rhd_{\sf HA} O$ and $A \rhd_{\sf HA} \neg\, O$, where $A$ is the universal closure of
an instance of \emph{Tertium non Datur} that is sufficient to make the classical argument work.

\medent  
Theorem~\ref{quorumsmurf} throws a shadow downward on {\sf HA}.
 We need to define the \emph{$\Gamma_0$-formulas} %are defined as follows.
 %preparatory definitions
  to describe it.
 %The almost negative formulas are defined as in Appendix \ref{sec:real}.  %
  Let $S$ range over $\Sigma_1^0$-formulas and let 
 $A$ range over almost negative formulas, as defined in Appendix \ref{sec:real}: 
 \begin{itemize}
 \item
 $ B ::= S \mid (B\wedge B) \mid (B\vee B) \mid (A \to B) \mid \forall x\, B \mid \exists x\, B$
 \end{itemize}
 
 \noindent
  Anne Troelstra shows in \cite[\Section 3.6.6]{troe:meta73} that ${\sf HA}+ {\sf ECT}_0$ is $\Gamma_0$-conservative over
 {\sf HA}. \emph{A fortiori},  ${\sf HA}+ {\sf CT}_0!$ is $\Gamma_0$-conservative over
 {\sf HA}. Inspection of the proof shows that this fact is verifiable in i-{\sf EA}.
 We have:
 
 \begin{theorem}
 The theory i-{\sf EA} verifies the following. Suppose $C$ is in $\Gamma_0$. We have: if
 $\bigwedge_{i<n} (A_i \rhd_{\sf HA} B_i)$ and ${\sf HA}\vdash \bigwedge_{i<n} (A_i \to B_i) \to C$, then
 ${\sf HA} \vdash C$.
 \end{theorem}
 
 \begin{proof}
 Suppose $C$ is in $\Gamma_0$ and
 $\bigwedge_{i<n} (A_i \rhd_{\sf HA} B_i)$ and ${\sf HA}\vdash \bigwedge_{i<n} (A_i \to B_i) \to C$. It follows that 
 ${\sf HA}+ {\sf CT}_0! \vdash \bigwedge_{i<n} (A_i \to B_i)$ and ${\sf HA}+{\sf CT}_0!\vdash \bigwedge_{i<n} (A_i \to B_i) \to C$.
 Hence  ${\sf HA}+{\sf CT}_0!\vdash C$. Since
$C$ is in $\Gamma_0$, it follows that ${\sf HA} \vdash C$.
 \end{proof}
 
 \begin{corollary}
  The theory i-{\sf EA} verifies the following.
 Suppose $A$ is almost negative and $B$ is in $\Gamma_0$. Suppose further that $A\rhd_{\sf HA} B$. Then, ${\sf HA}\vdash A \to B$.
 \end{corollary}
 
  \begin{corollary}
   The theory i-{\sf EA} verifies the following: if $\top \rhd_{\sf HA} O$ and $\top \rhd_{\sf HA}\neg\, O$, then ${\sf HA}\vdash \bot$.
   Thus, if {\sf HA} is consistent, it has no Orey-sentences.
  \end{corollary}
 
\noindent
 We give counterparts of the above classes in the modal language, beginning
 with the almost negative ones.\footnote{By the Orey-H\'ajek characterization,
 $A \rhd_{\sf PA} B$ is a $\Pi_2^0$-relation. (It was shown to be complete $\Pi^0_2$ independently by Per Lindstr\"om and
 Robert Solovay.) No such reduction is known for  the relation $A \rhd_{\sf HA} B$. This relation is \emph{prima facie}
 $\Sigma^0_3$ and might, for all we know, be $\Sigma^0_3$-hard. We note that $\Pi^0_2$ is almost negative but
 $\Sigma^0_3$ is not. So we cannot take $\phi \rhd \psi$ as a modal almost negative formula. This does not exclude
 that further insight might allow us to include it at a later stage.}
 %of the  language as follows. 
  Let $\phi$ range over all formulas and 
 \begin{itemize}
 \item
 $\sigma :: = \bot \mid \top \mid \Box \phi \mid (\sigma\vee \sigma)$
 \item
 $\psi ::= \sigma \mid (\psi \wedge \psi) \mid (\psi \to \psi)$
  \end{itemize}
 We define the $\Gamma_0$-formulas of the bi-modal language as follows. Let $\phi$ range over all formulas and let
 $\psi$ range over the almost negative formulas.
  \begin{itemize}
    \item
 $\chi ::= \bot \mid \top \mid \Box \phi \mid (\phi \rhd \phi) \mid (\chi \wedge \chi) \mid (\chi \vee \chi) \mid (\psi \to \chi) $. 
 \end{itemize}
 %
%We find:

\begin{theorem}\label{lampjesmurf}
Let $\chi$ be in $\Gamma_0$.
The following principle is in the interpretability logic of {\sf HA}:
\[  (\bigwedge_{i<n} (\phi_i \rhd \psi_i) \wedge   \Box (\bigwedge_{i<n}(\phi_i \to \psi_i) \to \chi)) \to \Box \chi.\] 
 \end{theorem}
 
\begin{example}
The principle $\vdash (\neg\neg\, \Box\bot  \to \Box\bot)\rhd \neg\neg\Box \bot  \to \Box \neg\neg\, \Box\bot$ is valid over {\sf HA}.
 Since {\sf HA} is {\sf HA}-verifiably closed under the primitive recursive Markov's Rule, it follows that
 $\vdash ((\neg\neg\, \Box\bot  \to \Box\bot)\rhd \neg\neg\Box \bot)  \to \Box \Box\bot$ is valid over {\sf HA}.
 \end{example}
 
 \begin{remark}
 We note that the seemingly stronger principle
 \[  (\bigwedge_{i<n} (\phi_i \rhd \psi_i) \wedge    \bigwedge_{i<n}(\phi_i \to \psi_i) \rhd \chi) \to \Box \chi.\] 
 in fact follows from Theorem~\ref{lampjesmurf}.
 \end{remark}
 
 \begin{quest} \label{que:intea}
Is there an interpretation of i-{\sf EA} in {\sf HA} that is not i-isomorphic to $\mathcal E$? 
 
 \medent
 There are many strengthenings of our question. We can demand {\sf HA}-verifiability of the self-interpretation.
 We could ask whether there is an $A$ such that $\top\rhd_{\sf HA} A$, but ${\sf HA}\nvdash A$.  Etcetera.
 
 \medent
If one combines the proof of Theorem~\ref{quorumsmurf}  with q-realizability, one obtains the following.
In case the domain and the parameter domain of an interpretation $K$ of i-{\sf EA} in {\sf HA} are auto-q\footnote{See
\Subsection~\ref{sec:hastar} for the notion of \emph{auto-q}.}, then 
$K$ is i-isomorphic with $\mathcal E$. Thus, a non-trivial interpretation i-{\sf EA} in 
 {\sf HA} should either have a sufficiently complex domain or a sufficiently complex parameter domain.\footnote{We apologize for
the classical reasoning. However, since the relevant predicates are decidable, it can be constructively justified.} 
%\medent
Note also  that $\top \rhd_{\sf HA} A$ both implies that $\Box_{{\sf HA}+{\sf CT}_0!} A$ and that $\top\rhd_{\sf PA} A$, which
puts some severe constraints on the possible $A$.
 \end{quest}
 
\subsubsection{Interpretability and $\Pi_1^0$-Conservativity}
\nosmurfduo
We have seen that interpretability and $\Pi_1^0$-conservativity coincide over {\sf PA}.
Over other classical theories, interpretability and $\Pi_1$-conservativity part ways. For example, they come apart over
Primitive Recursive Arithmetic {\sf PRA}: we have
$\top \jump_{\sf PRA} \mathrm{I}\Sigma_1$, but not $\top \rhd_{\sf PRA}  \mathrm{I}\Sigma_1$. 

Over
{\sf HA},  interpretability and $\Pi_1$-conservativity likewise separate ways.
We still find that, {\sf HA}-verifiably, $\rhd_{\sf HA}$ is a sub-relation of $\jump_{\sf HA}$.
However, for example, we have $\Diamond \Box_{\sf HA}\bot \jump_{\sf HA} \Box_{\sf HA} \bot$, but not
$\Diamond \Box_{\sf HA}\bot \rhd_{\sf HA} \Box_{\sf HA} \bot$.  Also,
$\neg\neg \,\Box_{\sf HA}\bot \jump_{\sf HA} \Box_{\sf HA} \bot$, but not 
$\neg\neg \,\Box_{\sf HA}\bot \rhd_{\sf HA} \Box_{\sf HA} \bot$.\footnote{%Both results can be derived 
These results follow from Theorem~\ref{lampjesmurf} in combination
with facts about provability logic.}

\section{The problem of the \emph{Survey} \qquad \HIp
\label{sec:survprob}} %\protect\textpmhg{l}
\nosmurfduo
We are returning here to the issue briefly mentioned in the main text: the collapse of $\tto$ to $\to$ in Lewis' 
original system \cite{Lewis14:jppsm,Lewis18} discovered by Post and addressed by Lewis in a subsequent 
note \cite{Lewis20:jppsm}. This episode is instructive in illustrating how Lewis' own thinking about 
$\tto$ was often sabotaged by a combination of several factors, including:

\begin{itemize}
\item an insistence on boolean laws for ``material'' connectives, including in particular classical, involutive laws for negation;
\item especially in the 1910's, a certain carelessness in accepting deductive laws 
 for ``intensional'' connectives, especially those involving contraposition.
\end{itemize}

\nosmurf
The second problem was pretty much admitted by Lewis himself:

\begin{smquote}
In developing the system, I had worked for a month to avoid this principle, which later turned out to be false. 
Then, finding no reason to think it false, I sacrificed economy and put it in 
 (\cite{lewisincap}, via \cite[p.92]{Murphey05}).
\end{smquote}

\nosmurf
In hindsight, these problems are unsurprising, especially given the publication date of the \emph{Survey}. Not only %because %it would be ahistorical to criticize Lewis' for taking too much inspiration from boolean logic back 
 %in the 1910's
   were non-boolean systems in the prenatal stage, %not even fully developed yet. %there was little to none non-classical competition on offer. 
 %A broader, though obviously related problem is that 
  but also semantics of propositional logics %---algebraic or any other---
  was poorly understood at the time. \emph{Symbolic Logic} published in 1932 was already in a much better position, mostly thanks to efforts of Mordechaj Wajsberg and William Parry, who provided several crucial algebraic (counter)models used in Appendix II to establish independence results for axioms between  \lna{S1} and \lna{S5}. No such assistance was available to Lewis when writing the earlier \emph{Survey} and consequently, when deciding whether or not to adopt a specific axiom for $\tto$, he would mostly rely just on his philosophical intuitions, much like other authors in that period.\footnote{Cf. in this respect his remark \cite{Lewis20:jppsm}: %correcting the axiom system of the \emph{Survey}: 
  ``Mr. Post's example which demonstrates the falsity of 2.21 is not here reproduced, since it involves the use of a diagram and would require considerable explanation.'' A ``diagram'' is presumably a finite matrix/algebra  (which could indicate a largely overlooked inspiration Lewis' work provided for Post in developing non-classical logical ``matrices'', a.k.a. algebras or truth-tables!). In Appendix II to \emph{Symbolic Logic}, the counterexamples of Parry and Wajsberg were called ``groups''. It is worth mentioning that early Lewis' papers tended to have titles like \emph{Implication and the Algebra of Logic} \cite{Lewis12}, \emph{A New Algebra of Implications and Some Consequences} \cite{Lewis13:jppsm}, \emph{The Matrix Algebra for Implications} \cite{Lewis14:jppsm} or \emph{A Too Brief Set of Postulates for the Algebra of Logic} \cite{Lewis15:jppsm}, but this should not mislead us: the word ``algebra'' (or ``matrix'') is not taken here in the sense of modern model theory or universal algebra.}
 
 From our point of view, it is of particular interest to isolate the actual r\^ole played by classical logic with its involutive negation, the axiom \lb\ and redefinition of $\Box$ as \refeq{boxdef} in the collapse of the system of the \emph{Survey}. %The presentation of Post's derivation given by Lewis seems to indicate that these classical laws are necessary to collapse the system of the \emph{Survey}; it is interesting to see how much of classical logic is actually used. %they can be weakened somewhat.
  %the collapse of $\Box \phi$ to $(\top \strictif \phi)$ is central to the problem, but it is not exactly the case.
 
 The problematic axiom is the converse of the one which was latter baptised  \lna{A8}  in Appendix II to \emph{Symbolic Logic} :
 
\begin{description}
 \item[\lna{A8}] $(\phi \tto \psi) \tto (\neg\Diamond\psi \tto \neg\Diamond\phi)$.
\end{description}
 
 \noindent
 In the \emph{Survey}, this axiom was postulated as a strict \emph{bi-}implication, i.e., 
 
\begin{description}
 \item $(\phi \tto \psi) \ifff (\neg\Diamond\psi \tto \neg\Diamond\phi)$.
\end{description}
 
 \noindent
 In our setting, with $\Box$ rather than $\Diamond$ as the primitive and with $\phi$ not being equivalent to $\neg\neg\phi$, the missing half can be rendered as 
 
 \newcommand{\totoon}{\lna{2.21}}
  \newcommand{\totoonalt}{\lna{2.21'}}
 
\begin{description}
 \item[\totoon] $(\Box \psi \strictif \Box \phi) \strictif (\neg\phi \strictif \neg\psi)$,
\end{description}
 
 \noindent
  ``\totoon'' being Lewis' name for this axiom \cite{Lewis20:jppsm}. Of course, there are other conceivable variants, for example:
 
\begin{description}
 \item[\totoonalt] $(\Box \neg\phi \strictif \Box \neg\psi) \strictif (\psi \strictif \phi)$.
\end{description}
 
\newcommand{\auxpost}{\lna{Auxp}}
 
 \nosmurf
As it turns out, however, $\totoon$ is exactly what we need to reproduce Post's derivation over $\iP$  (together with a sub-boolean axiom $\auxpost$ introduced below). %rather than \totoonalt.
 %we need to assume both forms to proceed with our reasoning. %Let us note that just  \lna{2.21b}  would be enough in presence of \lb, as \bk\ makes the following two formulas theorems:
 
 %\begin{itemize}
 %\item $(\phi \tto \psi) \tto (\neg \psi \tto \neg\phi)$,
 %%\item $(\Box\phi \tto \Box\psi) \tto (\Box\neg\psi \tto \Box\neg\phi)$. <-false
 %\end{itemize} 
 
 To present further details, let us also recall that Lewis uses Modus Ponens for $\tto$, i.e., $\frac{\phi \quad \phi \tto \psi}{\psi}$ as the main  inference rule. %\avred{This
 %makes sense only in a reasoning system without assumptions, say Hilbert style.} \tamargin{Is this something we should mention 
 %explicitly as part of our criticism of Lewis' setup? If I understand correctly, this remark overlaps with the point made in the footnote below, maybe it's worth elaborating upon.}
 This in itself is telling:  in $\iP$, $\phi$ jointly with $\phi \tto \psi$ entails only $\Box\psi$. The rule $\frac{\Box\phi}{\phi}$ is \emph{admissible}, 
 but not \emph{derivable}, unless one postulates as an axiom explicitly $\Box\phi \to \phi$, something that Lewis' insistence on formulating all the 
 axioms with $\tto$ as the principal connective prevented him from doing; $\Box\phi \tto \phi$ is not quite the same 
 thing.\footnote{One can see here yet another instance of Lewis' peculiar paradox, pointed out by Ruth Barcan Marcus: 
 despite his insistence that ``the relation of strict implication expresses precisely that relation which holds when valid deduction is possible'' and 
 that ``the system of Strict Implication may be said to provide that canon and critique of deductive inference'' \cite[p. 247]{Lewis32:book}, his own systems tend to run into  problems with the relationship between $\to$, $\tto$, entailment 
 and deducibility (relevance logicians would point it out too, cf. Footnote \ref{ft:relevance}, but their own systems  have their own share of similar problems). }
 
 In the setup with Modus Ponens for $\tto$ as the central rule and $\ifff$ as the ``real'' equivalence or identity, 
 %proving the equivalence between $\Box\phi$ and $\phi$ does not require 
 instead of deducing  $\phi \leftrightarrow \Box\phi$ in the extension of our $\iP$ with Lewis' axioms we need to show both $\Box\phi \tto \phi$ (which is already a theorem for Lewis, cf. the discussion of $\aApp$ and Remark \ref{rem:lewapp} above) and 
 \begin{itemize}
 \item[\lv] $\phi \tto \Box\phi$,
 \end{itemize}
 deriving $\lv$ in turn requiring only finding another theorem $\chi$ s.t. $\chi \tto (\phi \tto \Box\phi)$ is also a theorem; in other words, to derive still weaker
\begin{description}
  \item[$\Box$\lv] $\Box(\phi \tto \Box\phi)$.
\end{description}
  
 \noindent
This in turn can be done if one has both \totoon\ and a law which is a mild consequence of excluded middle, namely

\begin{description}
\item[\auxpost] $(\neg\neg p\tto\neg(\Box p\to \Box\neg p)) \tto (p\tto\Box p)$.
\end{description}

\noindent
Note that to derive $\auxpost$, it is enough to have as an axiom scheme, e.g., 

\begin{description}
\item $(\neg\neg\Box p \wedge \neg\Box\neg p) \to \Box p$;
\end{description}

\noindent
this is why we call $\auxpost$ a mild consequence of boolean laws. 

Note also that in presence of \totoon, we have that 

\newcommand{\auxpt}{\lna{Auxp2}}

\begin{description}
\item[\auxpt] $\Box(\Box p \to \Box \neg p) \tto \Box \neg p$. 
\end{description}

\noindent
To get this formula, substitute $\bot$ for $\phi$ and $p$ for $\psi$ in $\totoon$, use $\lx$ and the fact that $\Box p \to \Box \neg p \eqdm \Box p \to \Box\bot$. 
%\nosmurf
 %at hand, even with $\iP$ as the propositional base, and 
  Now  we can redo in our setting the Post derivation as quoted by Lewis. Substituting $\Box p \to \Box\neg p$ for $\psi$ and $\neg p$ for $\phi$ in  \totoon\ yields  $$((\Box p \to \Box\neg p) \strictif \Box \neg p) \strictif (\neg\neg p\strictif \neg(\Box p \to \Box\neg p)).$$
%We are thus left with showing that  $\Box \neg\Box\psi \strictif \Box \neg\psi$ can be deduced from  \lna{2.21a}, but for this purpose it is enough to substitute $\bot$ for $\phi$ in this axiom and do a little reasoning in $\iP$, the details of which are left to the reader.
The antecedent  of this strict implication is precisely \auxpt\ and the consequent is the antecedent of $\auxpost$.

\begin{remark}
Of course, there are simpler ways of collapsing the system of the \emph{Survey} when full boolean logic and all Lewis axioms are assumed. Note that using classical logic and \lb\ (which is an axiom for Lewis, and as we established in Corollary \ref{cor:emderbox} can anyway be derived in $\iP + \cpc$), we can replace \totoon\ with

$$
\Box(\Box\psi \to \Box\phi) \to \Box(\psi \to \phi).
$$

\noindent
Classically, this axiom in turn can be replaced with

$$
\Diamond\psi \to \Diamond(\Box\psi \wedge\Diamond\psi).
$$

\noindent
Now, if $\Box\phi \to \phi$ (i.e., reflexivity) is also an axiom or a theorem (which, as  shown above, should be indeed the case in a modern representation of Lewis' original system, with $\frac{\Box\phi}{\phi}$ as an admissible or derivable rule), we can derive $\Box\phi \leftrightarrow \phi$, trivializing the modal operator.
\end{remark}

 \end{document}